\documentclass[12pt,letterpaper,hyper]{JHEP3}
\usepackage{amsmath}
\usepackage{amsthm}
\usepackage{amssymb,graphicx}
\usepackage{extpfeil}
\usepackage{capt-of}
\usepackage[leqno,fleqn,intlimits]{empheq} 

\usepackage{enumitem}

\usepackage{tabularx}
\newcommand{\thvline}{\vrule width 1pt}
\newcommand{\thhline}{\hrule height 1pt}

\usepackage{diagbox}

\renewcommand{\refname}{Bibliography}
\makeatletter
\renewenvironment{thebibliography}[1]
     {\bgroup\raggedright\small\section{\refname
        \@mkboth{\MakeUppercase\refname}{\MakeUppercase\refname}}%
      \list{\name{bib\@arabic\c@enumiv}
	    \@biblabel{\@arabic\c@enumiv}}%
           {\settowidth\labelwidth{\@biblabel{#1}}%
            \leftmargin\labelwidth
            \advance\leftmargin\labelsep
            \@openbib@code
            \usecounter{enumiv}%
            \let\p@enumiv\@empty
            \renewcommand\theenumiv{\@arabic\c@enumiv}}%
      \sloppy\clubpenalty4000\widowpenalty4000%
      \sfcode`\.\@m}
     {\def\@noitemerr
       {\@latex@warning{Empty `thebibliography' environment}}%
      \endlist\egroup}
\makeatother

\def\twoheaddownarrow{\ensuremath{\rotatebox[origin=c]{-90}{$\twoheadrightarrow$}}}

\newtheorem{thm}{Theorem}[section]
\newtheorem*{corollary}{Corollary}

\newtheorem{proposition}{Proposition}[section]

\newenvironment{myenumerate}{
\begin{enumerate}
   \setlength{\itemsep}{1pt}
   \setlength{\parskip}{0pt}
   \setlength{\parsep}{0pt}}{\end{enumerate}}
\newenvironment{myitemize}{
\begin{itemize}
   \setlength{\itemsep}{1pt}
   \setlength{\parskip}{0pt}
   \setlength{\parsep}{0pt}}{\end{itemize}}
\def\vth{\vartheta}
\renewcommand{\Im}{\mbox{Im}}
\renewcommand{\Re}{\mbox{Re}}

\newcommand{\pa}{\partial}

\newcommand{\IR}{\mathbb{R}}
\newcommand{\IZ}{\mathbb{Z}}

  \def\IN{\mathbb{N}}
\newcommand{\Tr}{\mbox{Tr}}
\newcommand{\sgn}{\mbox{sgn}}
\newcommand{\CC}{\cal{C}}
\newcommand{\LL}{{\cal L}}

\def\={\;  = \;}
\def\+{\, + \,}
\def\inn{\,\in\,}
\def\ceil#1{\lceil#1\rceil}  \def\flo#1{\lfloor#1\rfloor}
\def\ssm{\smallsetminus}

\def\dRS#1{\partial^{\rm RS}_{#1}}

 \def\Av#1#2{\text{Av}^{(#1)}\bigl[#2\bigr]}
 \def\AvB#1#2{\text{Av}^{(#1)}\Bigl[#2\Bigr]}
 \def\AvZ#1{\text{\rm Av}_IZ\Bigl[#1\Bigr]}
 \def\AvL#1{\text{Av}_{\IZ\t+\IZ}\Bigl[#1\Bigr]}
 \def\AvZ{\text{Av}_\IZ}
 \def\AvL{\text{Av}^{(m)}_{\IZ\t+\IZ}}

\def\AA{\bold{A}}
\def\BB{\bold{B}}
\def\AA{\mathcal{A}}
\def\BB{\mathcal{B}}

\def\mypmod#1{\; (\mod \; {#1})}

\def\wh{\widehat}
\def\wt{\widetilde}
\def\FR{\mathcal R}
\def\FRR{\mathcal R^{(2)}}
\def\erfc{\text{erfc}}

\def\d{\partial}
      \def\s{\sigma}  \def\g{\gamma}  
\def\t{\tau}   
\def\a{\alpha}  \def\b{\beta}
\def\m{\mu}   
\def\e{\epsilon}  \def\w{\omega}  \def\h{\eta}  \def\l{{\lambda}}  
\def\o{{\omega}}  \def\O{{\Omega}}  \def\D{\Delta}  \def\G{\Gamma}
\def\L{\rm I}

\def\CH{{\cal H}}
\def\CR{{\cal R}}
\def\CM{{\cal M}}
\def\CF{{\cal F}}
\def\CS{{\cal S}}

\def\CZ{{\cal Z}}
\def\CN{{\cal N}}
\def\CO{{\text O}}

\def\CQ{{\cal Q}}

\def\half{{\frac12}}
\def\fr#1#2{{\frac{#1\vphantom{X^Y}}{#2\vphantom{x_j}}}}

\def\NH{^{\rm OG}}
\def\IC{{\mathbb C}}

\def\IQ{{\mathbb Q}}
\def\CN{{\cal N}}
\def\CQ{{\cal Q}}

\def\CZ{{\cal Z}}

\def\1F1{{}_1\!F_1}
\def\2F0{{}_2\!F_0}

\def\ve{\varepsilon}
\def\bve{\pmb{\ve}}

\def\ga{\gamma}

\def\G{\Gamma}
\def\a{\alpha}

\def\h3{$\textrm{H}_3^+$}

\def\d{{\partial}}

\def\IC{{\mathbb C}}
\def\IR{{\mathbb R}}
\def\IZ{{\mathbb Z}}
\def\IM{{\mathbb M}}
\def\IJ{{\mathbb J}}

\newcommand{\E}{\mathbf{e}}
\newcommand{\sm}[4]{\bigl(\smallmatrix #1&#2\\ #3&#4\endsmallmatrix\bigr)}

\newcommand{\tK}{{\tilde K}}

\newcommand{\tS}{{\tilde S}}

\def\p{\partial}

\def\mod{{\rm mod}}

\def\CM{{\cal M}}
\def\CN{{\cal N}}
\def\CR{{\cal R}}

\def\CF{{\cal F}}

\def\CV{{\cal V}}
\def\CZ{{\cal Z}}
\def\CE{{\cal E}}

\def\CH{{\cal H}}
\def\CC{{\cal C}}
\def\CB{{\cal B}}
\def\CS{{\cal S}}
\def\CA{{\cal A}}
\def\CC{{\cal C}}
\def\CK{{\cal K}}

\def\CQ{{\cal Q }}
\def\CE{{\cal E }}
\def\CV{{\cal V }}
\def\CZ{{\cal Z }}
\def\CS{{\cal S }}

\def\Tr{{\rm Tr}}

\def\G{\Gamma}

\font\manual=manfnt
\def\dbend{\lower3.5pt\hbox{\manual\char127}}

\def\c{\cdot}

\def\p{\partial}

\def\bar{\overline}
\def\CS{{\cal S}}
\def\CN{{\cal N}}
\def\CH{{\cal H}}
\def\rt2{\sqrt 2}
\def\irt2{{1\over\sqrt 2}}

\def\t{\tilde}
\def\ndt{\noindent}
\def\s{\sigma}
\def\b{\beta}
\def\a{\alpha}
\def\w{\omega}
\def\g{\gamma}

\def\mod{{\rm mod}}

\font\cmss=cmss10
\font\cmsss=cmss10 at 7pt
\def\IL{\relax{\rm I\kern-.18em L}}
\def\IH{\relax{\rm I\kern-.18em H}}
\def\rlx{\relax\leavevmode}
\def\ZZ{\rlx\leavevmode\ifmmode\mathchoice{\hbox{\cmss Z
\kern-.4em Z}}
 {\hbox{\cmss Z\kern-.4em Z}}{\lower.9pt\hbox{\cmsss Z\kern-.
36em Z}}
 {\lower1.2pt\hbox{\cmsss Z\kern-.36em Z}}\else{\cmss Z\kern-.
4em
 Z}\fi}

\def\Tr{{\rm Tr}}

\def\G{\Gamma}

\def\rt2{\sqrt{2}}
\def\irt2{{1\over\sqrt{2}}}

\def\t{\tilde}

\def\ndt{\noindent}
\def\s{\sigma}

\newcommand{\Z}{{\mathbb Z}}
\newcommand{\R}{{\mathbb R}}
\newcommand{\N}{{\mathbb N}}
\newcommand{\C}{{\mathbb C}}
\newcommand{\M}{{\cal M}}
\newcommand{\Q}{{\mathbb Q}}

\def\Tr{{\rm Tr}}

\newcommand{\bea}{\begin{eqnarray}}
\newcommand{\eea}{\end{eqnarray}}
\newcommand{\be}{\begin{equation}}
\newcommand{\ee}{\end{equation}}
\newcommand{\bes}{\begin{equation*}}
\newcommand{\ees}{\end{equation*}}
\newcommand{\ben}{\begin{eqnarray*}}
\newcommand{\een}{\end{eqnarray*}}
\newcommand{\bem}{\begin{pmatrix}}
\newcommand{\eem}{\end{pmatrix}}
\newcommand{\bal}{\begin{align}}
\newcommand{\eal}{\end{align}}

\def\a{\alpha}
\def\b{\beta}
\def\c{\gamma}
\def\d{\delta}

\def\e{\epsilon}

\def\g{\gamma}
\def\h{\eta}

\def\im{\mathrm{Im}}

\def\k{\kappa}             
\def\l{\lambda}
\def\m{\mu}

\def\o{\omega}  \def\w{\omega}
\def\pa{\partial}
\def\s{\sigma}                                   
\def\t{\tau}

\def\v{\varphi}

\def\D{\Delta}

\def\G{\Gamma}

\def\L{\Lambda}
\def\O{\Omega}

\def\Th{\Theta}

\def\t{\tau}
\def\tbar{\bar{\tau}}

\def\vst{\v^{\text{stand\vphantom{p}}}}
\def\vstF{\Phi^{\rm stand}}
\def\AP#1{\mathcal A_{#1}}

\def\PHI{\v^{\text{opt}}}
\def\PHIF{\Phi^{\rm opt}}
\def\V{\mathcal V}

\def\DD{\Delta}

\def\elliptic{\CE}

\def\mytheta{\vartheta}

\title{Quantum Black Holes, Wall Crossing, and Mock Modular Forms}

\preprint{}
\author{Atish Dabholkar$^{1, 2}$, Sameer Murthy$^3$, and Don Zagier$^{4, 5}$\\

\vspace{0.2cm}

\it $^{1}${Theory Division, CERN, PH-TH Case C01600,  CH-1211,  23 Geneva, Switzerland}
\vspace{0.2cm}

\it $^2${Sorbonne Universit\'es, UPMC Univ Paris 06\\
\it  \hspace{0.15cm}UMR 7589, LPTHE, F-75005, Paris, France}
\vspace{0.2cm}

\it \hspace{0.05cm} {CNRS, UMR 7589, LPTHE, F-75005, Paris, France}
\vspace{0.2cm}

\ndt \it$^3${NIKHEF theory group, Science Park 105}, \it {1098 XG Amsterdam, The Netherlands} 
\vspace{0.2cm}

\ndt \it $^4${Max-Planck-Institut f\"ur Mathematik},\it{Vivatsgasse 7, 53111 Bonn, Germany}
\vspace{0.2cm}

\ndt \it $^5${Coll\`ege de France},  
\it{3 Rue d'Ulm, 75005 Paris, France}

\vspace{4mm}
\hspace{-5mm}
\leftline{Emails: 
\tt{atish@lpthe.jussieu.fr,  murthy.sameer@gmail.com, don.zagier@mpim-bonn.mpg.de}}
}
\abstract{
We show that the meromorphic Jacobi form that counts the  quarter-BPS states  in $\CN=4$ 
string theories can be  canonically decomposed as a sum of  a \textit{mock Jacobi form}  and an 
\textit{Appell-Lerch sum}. The  quantum degeneracies of single-centered black holes are Fourier 
coefficients of this mock Jacobi form, while the Appell-Lerch sum captures  the degeneracies of 
multi-centered black holes which decay upon wall-crossing. The \textit{completion} of the mock 
Jacobi form  restores the modular symmetries expected from $AdS_3/CFT_2$  holography but  has 
a holomorphic anomaly reflecting the non-compactness of the microscopic CFT.   For every positive 
integral value $m$ of  the magnetic charge invariant  of the black hole, our analysis leads 
to a  special mock Jacobi form of weight two  and  index $m$, which we characterize uniquely 
up to a  Jacobi cusp form. This family  of special forms and another closely related family 
of weight-one forms contain almost all the known mock modular forms  including the mock theta 
functions of Ramanujan, the generating function of Hurwitz-Kronecker class numbers, the mock modular forms 
appearing in the Mathieu and Umbral moonshine, as well as an infinite number of new examples.  
}

\begin{document}

\textit{``My dream is that I will live to see the day when our young physicists, struggling 
to bring the predictions of superstring theory into correspondence with the facts of nature, 
will be led to enlarge their analytic machinery to include not only theta-functions but mock 
theta-functions\,\dots\; But before this can happen, the purely mathematical exploration of the 
mock-modular forms and their mock-symmetries must be carried a great deal further.''}

\vspace{2mm}
\centerline{ \hspace{60mm} \textsl{ Freeman Dyson (1987 Ramanujan Centenary Conference)}}

\section{Introduction}

The quantum degeneracies associated with a black hole horizon are of central importance in 
quantum gravity. The counting function for these degeneracies for  a large class of black holes
in string theory is expected to be  modular from the perspective of holography. However, 
in situations with wall-crossing, there is an apparent loss of modularity and   it is far 
from clear if and how such a counting function can be modular.   In the context of quarter-BPS 
black holes in $\CN=4$ supersymmetric theories, we develop the  required analytic machinery 
that provides a complete answer to this question which leads naturally to the mathematics of
mock modular forms. We  present a number of new mathematical results motivated by but 
independent of these physical considerations.  

Since this paper is of possible interest to both theoretical physicists (especially
string theorists) and theoretical mathematicians (especially number theorists),
we give two introductions in their respective dialects.

\subsection{Introduction for mathematicians}

In the quantum theory of black holes in the context of string theory, the physical
problem of counting the dimensions of certain eigenspaces (``the number of
quarter-BPS dyonic states of a given charge'') has led to the study of Fourier
coefficients of certain {\it meromorphic} Siegel modular forms and to the
question of the modular nature of the corresponding generating functions.
Using and refining results of S. Zwegers \cite{Zwegers:2002}, we show that
these generating  functions belong to the recently discovered class of functions
called {\it mock modular forms}.

Since this notion is still not widely known, it will be reviewed in some detail
(in \S\ref{basicdefs}). Very roughly, a mock modular form of weight $k$ 
(more precisely, ``pure'' mock modular forms; we will also introduce
a somewhat more general notion of ``mixed'' mock modular forms in~\S\ref{MMMF}) 
is a holomorphic function~$f$ in the upper half plane to which is associated a 
holomorphic  modular form~$g$ of weight $2-k$, called the ``shadow'' of~$f$, such that the 
sum of~$f$ and a suitable non-holomorphic integral of~$g$ transforms like 
a holomorphic modular form of weight~$k$. Functions of this type occur 
in several contexts in mathematics and mathematical physics: as certain $q$-hypergeometric series 
(like Ramanujan's original mock theta functions), as generating functions of 
class numbers of imaginary quadratic fields~\cite{Zagier:1975}, or as characters of extended 
superconformal algebras~\cite{Ooguri:1989fd}, with special cases of the last class being conjecturally 
related to the Mathieu group~$M_{24}$~\cite{Eguchi:2010ej}. They also arise, as was shown 
by Zwegers in his thesis~\cite{Zwegers:2002},  as the Fourier coefficients of 
meromorphic Jacobi forms. It is this last occurrence which is  at the origin of the 
connection to black hole physics, because the Fourier coefficients of meromorphic Jacobi forms have  
the same wall-crossing behavior as that exhibited by the degeneracies of BPS states. 

The specific meromorphic Jacobi forms that are of interest for the black hole counting problem  
are the Fourier-Jacobi coefficients $\psi_m (\t,z)$ of  the meromorphic Siegel modular form
  \be\label{reciproigusa} \frac 1{\Phi_{10}(\O)} \= \sum_{m=-1}^{\infty} \psi_m (\t,z) \, p^m  \,, \qquad 
    \Bigl(\O = \left( \begin{array}{cc} \t & z \\ z & \sigma \\ \end{array} \right),\;\, p=e^{2\pi i \s} \Bigr) \,, \ee
the reciprocal of the Igusa cusp form of weight $10$, which arises as the partition function of 
quarter-BPS dyons in the type~II compactification on the product of a $K3$ surface and an elliptic curve \cite{Dijkgraaf:1996it,Shih:2005uc,David:2006yn}.
These coefficients, after multiplication by the discriminant function $\Delta(\tau)$, are meromorphic 
Jacobi forms of weight~$2$ with a double pole at $z=0$ and no others (up to translation by the period lattice). 

The new mathematical results of the paper are contained in Sections~\ref{MockfromJacobi}, \ref{flyJac}
and \ref{structure}. In~\S\ref{MockfromJacobi}, \-extending the results of~\cite{Zwegers:2002}, we show 
that any meromorphic Jacobi form $\v(\t,z)$ having poles only at torsion points $z=\a\t+\b$ ($\a,\,\b\in\Q$)
has a canonical decomposition into two pieces~$\v^{\rm F}(\t,z)$ and~$\v^{\rm P}(\t,z)$, called its ``finite" part and  
``polar" part, respectively, of which the first is a finite linear combination of classical theta series
with mock modular forms as coefficients, and the second is an elementary expression that is
determined completely by the poles of~$\v$. Again using the results of
Zwegers, we give explicit formulas for the polar part~$\v^P$ for all~$\v$ having only simple or double poles.  
In the particular case of~$\psi_{m}$, the polar part is given by the formula 
  \be \label{simfun} \psi_{m}^{\rm P}(\t,z) \= \frac{p_{24}(m+1)}{\Delta(\t)} \; \AP{2,m}(\t,z) \qquad 
  \big(q=e^{2\pi i \t} \,, \;  y=e^{2\pi i z} \big) \, , \ee
where $\Delta(\t)$ is the Ramanujan discriminant function, $p_{24}(m+1)$ the coefficient 
of~$q^{m}$ in $\Delta(\t)^{-1}$, and $\AP{2,m}(\t,z)$ the elementary function (Appell-Lerch sum) 
  \be \label{defBP} \AP{2,m}(\t,z)\= \sum_{s\in\Z} \frac{q^{ms^2 +s}y^{2ms+1}}{(1 -q^s y)^2} \,. \ee
Note that $\AP{2,m}$ exhibits wall-crossing: for  $0<\Im(z) < \Im(\t)$ it has the Fourier expansion 
  \be \label{AL2} \AP{2,m}(\t,z) \= \sum_{r\ge\ell>0 \atop r\,\equiv\ell\mypmod{2m}} \ell\,q^{\frac{r^{2}-\ell^{2}}{4m}}\,y^r\,, \ee
but the Fourier expansions are different in other strips $n<\Im(z)/\Im(\t)<n+1$.  
This wall-crossing is at the heart of both the mathematical and the physical theories. On the mathematical side
it explains the ``mockness" of the finite part $\psi_m^{\rm F}$ of~$\psi_m$. On the physical side it has
an interpretation in terms of counting two-centered black holes, with the integer $\ell$ in \eqref{AL2} being the  
dimension   of the $SU(2)$ multiplet with angular momentum $(\ell -1)/2$ contained in the electromagnetic 
field produced by the two centers.

Starting in \S\ref{flyJac} we focus attention on two particular classes of functions
$\{\v_{1,m}\}$ and $\{\v_{2,m}\}$, where $\v_{k,m}$ is a meromorphic Jacobi form 
of weight $k\in\{1,\,2\}$ and index~$m\in\N$ having singular part $(2\pi iz)^{-k}\+\text O(1)$ 
as $z\to0$ and no other poles except the translates of this one by the period lattice.  
These functions in the case~$k=2$ are related to the Fourier coefficients~$\psi_{m}$ defined 
in \eqref{reciproigusa} by 
\be \label{psifromphim}
\psi_{m}(\t,z) \= \frac{p_{24}(m+1)}{\Delta(\t)} \; \v_{2,m}(\t,z) \+ \text{$($weakly holomorphic Jacobi form$)$}.
\ee
The polar part of
$\v_{2,m}$ is the function~$\AP{2,m}$ defined above, and that of $\v_{1,m}$ an Appell-Lerch sum $\AP{1,m}$ 
with a similar definition (eq.~\eqref{AAm}), but the functions themselves are not unique, since we 
can add to them any (weak) Jacobi forms of weight~$k$ and index~$m$ without changing the defining property.
This will change the mock Jacobi form $\Phi_{k,m}=\v_{k,m}^{\rm F}$ in the corresponding way.  Much of our
analysis concerns finding ``optimal" choices, meaning choices for which the functions $\Phi_{k,m}$ have 
poles of as small an order as possible at infinity and consequently Fourier coefficients whose asymptotic 
growth is as small as possible.  Our main results concern the case~$k=2$, and are as follows:

\begin{myenumerate}
  \item If $m$ is a prime power, then the functions $\Phi_{2,m}$ can be chosen to be strongly 
holomorphic (i.e., with Fourier expansions containing only monomials $q^ny^r$ with $4nm-r^2\ge0$).
Their Fourier coefficients then have polynomial growth and are given by explicit linear
combinations of class numbers of imaginary quadratic fields.
 \item More generally, for arbitrary $m$, the function $\Phi_{2,m}$ can be chosen to be the sum
of the images under appropriate ``Hecke-like operators" of special mock Jacobi forms $\CQ_M$,
where~$M$ ranges over all divisors of~$m$ having an even number of distinct prime factors. 
These summands are the eigencomponents of~$\Phi_{2,m}$ with respect to the Atkin-Lehner-like operators
$W_{m_1}$ ($m_1|m$, $(m_1,m/m_1)=1$) acting on the space of mock Jacobi forms of index~$m$.
  \item The mock Jacobi form~$\CQ_1$ is the generating function of class numbers of imaginary quadratic fields,
and is strongly holomorphic. The other special mock Jacobi forms~$\CQ_M$ $(M=6$, 10, 14, \dots) 
can be chosen to have ``optimal growth" (meaning that their Fourier expansions contain only monomials 
$q^ny^r$ with $4nm-r^2\ge-1$). Their
Fourier coefficients then grow like $e^{\pi\sqrt{\D}/m}$ as $\D:=4nm-r^2$ tends to infinity.
  \item One can also choose $\Phi_{2,m}(\t,z)$ for arbitrary $m\ge1$ to be of the form 
$\sum\Phi_{2,m/d^2}^0(\t,dz)$, where $d$ ranges over positive integers with $d^2|m$ and each~$\Phi_{2,m/d^2}^0$
has optimal growth.
  \item There are explicit formulas for the polar coefficients (those with $\D<0$) of $\Phi_{2,m}^0$ 
and~$\CQ_M$.  In particular, the coefficient of $q^ny^r$ in~$\CQ_M$ vanishes if $4nM-r^2=0$ and equals
$\pm\v(M)/12$ if $4nM-r^2=-1$, where $\v(M)=\prod_{p|M}(p-1)$ is the Euler $\v$-function of~$M$.
\end{myenumerate}
The proofs of these properties are contained in~\S\ref{structure}, which gives a detailed description 
(based on difficult results on Jacobi forms proved in~\cite{SkoruppaZagier2, SkoruppaZagier}) of the way that the space of
holomorphic Jacobi forms is contained in the spaces of weak or of weakly holomorphic Jacobi forms.
This analysis contains several surprises, such as the result (Theorem~\ref{chooseKm}) that for
all $m\ge1$ the space of holomorphic Jacobi forms of index~$m$ has codimension exactly~1 in
the space of Jacobi forms of index~$m$ with optimal growth. 

For $\v_{1,m}$ the description is much less complete.  There is still a result like~2., with $M$ now 
ranging over the divisors of~$m$ with an~{\it odd}
number of distinct prime factors, but it is no longer possible in general to choose the~$\CQ_M$ to have
optimal growth.  The few cases where this is possible turn out to be related to mock theta functions
that have played a prominent role in the past. Thus $\CQ_{30}$ and $\CQ_{42}$ are essentially equal
to the mock theta functions of ``order~5" and ``order~7" in Ramanujan's original letter to Hardy, and 
$\CQ_2$ and several of the other $\CQ_M$ are related to the ``Mathieu moonshine" and 
``Umbral moonshine'' story~\cite{Eguchi:2010ej, Cheng:2012tq}.
The mock Jacobi forms~$\CQ_M$ for the $k=2$ case are also related to special mock theta functions, 
now of weight~$3/2$, e.g.~$\CQ_6$ is the mock theta function with shadow~$\eta(\tau)$ given in \cite{Zagier:2007}. 

The applications of the results (for the $k=2$ case, via the fact that $\psi_m$ is the sum of
$p_{24}(m+1)\v_{2,m}/\D$ and a weakly holomorphic Jacobi form of weight~$-10$ and index~$m$) to the
original physics problem are explained in~\S\ref{MockforDyons} and in the following ``second introduction."

\subsection{Introduction for physicists}

The microscopic quantum description of supersymmetric black holes in string
theory usually starts with a brane configuration of given
charges and mass at weak coupling, which is localized at a single point in the
noncompact spacetime. One then computes an appropriate indexed partition
function in the world-volume theory of the branes, which from the perspective of enumerative geometry 
computes topological invariants such as the  Donaldson-Thomas invariants. At strong coupling, the brane
configuration gravitates and the indexed partition function is expected to count the
microstates of the corresponding macroscopic gravitational configurations.
Assuming that the gravitational configuration is a single-centered black hole then
gives a way to obtain a statistical understanding of the entropy of the black hole in
terms of its microstates, in accordance with the Boltzmann relation\footnote{Under certain 
conditions the index equals the absolute number \cite{Sen:2008yk,Sen:2008vm, Dabholkar:2010rm}.}.

One problem that one often encounters is that the macroscopic configurations are no longer localized 
at a point and include not only a single-centered black hole of interest but also several multi-centered 
ones \cite{Denef:2000nb,Denef:2002ru,Bates:2003vx,Denef:2007vg}. Moreover, the indexed degeneracy of 
the multi-centered configurations typically jumps upon crossing walls of marginal stability in the 
moduli space where the multi-centered configuration breaks up into its single-centered constituents. 
These jumps are  referred to  as the  `wall-crossing phenomenon'.  

If one is interested in the physics of the horizon or the microstates of a single black
hole, the multi-centered configurations and the associated wall-crossings are thus something of a 
nuisance. It is  desirable to have a mathematical characterization that isolates the single-centered 
black holes directly at the microscopic level. One  distinguishing feature of single-centered 
black holes is that they are `\textit{immortal}' in that they exist as stable quantum states for 
all values of the moduli and hence their degeneracy does not exhibit the wall-crossing phenomenon. 
We will use this property later to   define the counting function for the  immortal black holes.

The wall-crossing phenomenon raises important conceptual questions regarding
the proper holographic formulation in this context. In many cases, the black hole can be viewed 
as an excitation of a black string. The near  horizon geometry of a black string is $AdS_3$ which 
is expected to be holographically dual to a two-dimensional conformal field theory $CFT_2$.  
The conformal boundary of  Euclidean $AdS_3$ is a 2-torus with a complex structure parameter $\tau$,
 and the physical partition function of $AdS_3$ and of the boundary $CFT_2$ is a function of $\tau$. 
The  $SL(2, \Z)$ transformations of $\t$ can be identified geometrically with global diffeomorphisms 
of the boundary of  $AdS_3$ space. The partition function is expected to have good modular properties
under this geometric  symmetry. This symmetry is crucial for the Rademacher-type expansions of the 
black hole degeneracies for understanding the quantum entropy of these black holes via holography \cite{Sen:2008yk,Sen:2008vm,Dabholkar:2011ec,Dijkgraaf:2000fq,deBoer:2006vg,Manschot:2007ha,Strominger:1998yg, 
Murthy:2009dq,Manschot:2009ia,DGMKloos}. Implementing the modular symmetries  presents several subtleties in situations when there is wall-crossing. 

The wall-crossing phenomenon has another important physical implication for the invariance of the 
spectrum under large gauge transformations of the antisymmetric tensor field.  Large gauge transformations 
lead to the `spectral flow symmetry' of the partition function of the black string \cite{deBoer:2006vg}. 
Since these transformations act both on the charges and the moduli, degeneracies of states with a 
charge vector $\Gamma$ at some point $\phi$ in the moduli space  get mapped to the degeneracies of states with charge
vector $\Gamma^{\prime}$ at some other point $\phi^{\prime}$ in the moduli space. Typically, there are many 
walls separating the point $\phi^{\prime}$  and the original point $\phi$. As a result, the degeneracies 
extracted from the black string at a \textit{given} point $\phi$ in the moduli space do not exhibit the 
spectral-flow symmetry.  On the other hand, the spectrum of immortal black holes is independent of 
asymptotic moduli and hence must exhibit the  spectral-flow symmetry.  This raises the question as to 
how to make the spectral-flow symmetry manifest for the degeneracies of immortal black holes in the 
generic situation when there is wall-crossing.

With these motivations, our objective will be to isolate the partition functions of the black string 
associated with immortal black holes and investigate their transformation properties  under the
boundary modular group and large gauge transformations. More
precisely, we would like to investigate the following four questions.
\begin{myenumerate}
  \item Can one define a \textit{microscopic}  counting function that cleanly
isolates  the microstates of immortal black holes from  those of  the multi-centered black configurations?
  \item What are the modular properties of  the counting function of  immortal black holes when the 
asymptotic spectrum exhibits the wall-crossing phenomenon?
  \item Can this counting function be related to a quantity that is properly modular as
might be expected from the perspective of near-horizon $AdS_3/CFT_2$   holography?
  \item Can one define a partition function of the immortal black holes that 
manifestly exhibits the spectral-flow symmetry resulting from large gauge transformations?
\end{myenumerate}

The main difficulties in answering these questions stem from the complicated  moduli dependence
of the black hole spectrum which is often extremely hard to compute. To address the central conceptual
issues in a tractable context, we consider  the compactification of Type-II on $K3 \times T^2$
with $\CN =4$ supersymmetry in four dimensions. The spectrum of  quarter-BPS dyonic states in this
model is exactly computable \cite{Dijkgraaf:1996it,Gaiotto:2005gf,Shih:2005uc,Shih:2005he,LopesCardoso:2004xf,David:2006yn} 
and by now is well understood  at \textit{all} points in the moduli space \cite{Sen:2007vb,Dabholkar:2007vk,
Cheng:2007ch} and for \textit{all} possible duality orbits \cite{Dabholkar:2007vk,
Banerjee:2008pv,Banerjee:2008pu,Banerjee:2008ri,Dabholkar:2008zy}. 
Moreover, as we will see, this particular model exhibits  almost all of the essential issues
that we wish to address.  The  ${\cal N}=4$ black holes have the remarkable property that even though
their spectrum is moduli-dependent, the partition function itself is moduli-independent. The  entire moduli
dependence of the black hole degeneracy is captured by the moduli dependence of
the choice of the Fourier contour \cite{Sen:2007vb,Dabholkar:2007vk, Cheng:2007ch}.
Moreover,  the only  multi-centered configurations that contribute to the supersymmetric index of quarter-BPS states
are the ones with only two centers,  each of which is half-BPS \cite{Dabholkar:2009dq}.  Consequently,  
the only way the index can jump at a wall is by the decay of  a two-centered configuration  into its half-BPS 
constituents \cite{Sen:2007vb,Dabholkar:2007vk, Cheng:2007ch}; this  is a considerable simplification compared 
to the general $\CN=2$ case where more complicated multi-centered decays  are possible. These features  
make the $\CN=4$ case  much more tractable.  

The number of microstates of quarter-BPS  dyonic states for the above-mentioned compactification is given by a Fourier
coefficient of a meromorphic Jacobi form $\psi_{m}(\t, z)$   with a moduli-dependent contour. 
The partition function (\ref{reciproigusa}) referred to earlier is the generating
function for these meromorphic Jacobi forms. Using this simplicity of the moduli
dependence and the knowledge of the exact spectrum, it is possible to give very
precise answers to the above questions in the $\CN =4$ framework, which  turn out
to naturally involve mock modular forms.
\begin{myenumerate}
  \item One can define a holomorphic  function for counting the
microstates of immortal black holes\footnote{We will use the terms  `immortal' and `single-centered' interchangeably.  
In general, the moduli-independent `immortal' degeneracies  can receive contributions not only from single black 
holes but also from scaling solutions \cite{Denef:2007vg}. They are not expected to contribute to the $\CN=4$ 
index that we consider \cite{Dabholkar:2009dq}. In addition, there can be  `hair' degrees of freedom \cite{Banerjee:2009uk, Jatkar:2009yd}, 
which are degrees of freedom localized outside the  black hole horizon that carry part of the charge of the black hole. 
In frames where the black hole is represented entirely in terms of D-branes, such hair modes are expected to be absent.} as a 
Fourier coefficient of the partition function of the black string for a specific choice of the Fourier 
contour \cite{Sen:2007vb,Dabholkar:2007vk,Cheng:2007ch}. The contour corresponds to
choosing the asymptotic moduli of the theory in the attractor region of the single-centered black hole.
  \item Because the asymptotic counting function  is a \textit{meromorphic} Jacobi form, the
near horizon counting function of immortal black holes is a \textit{mock} modular form in that it fails to be modular
but in a very specific way.  The failure of modularity is governed by a
\textit{shadow}, which is given in terms of another holomorphic modular form. 
  \item Given a mock modular form and its shadow, one can define its
\textit{completion} which is a non-holomorphic modular form. The failure of
holomorphy can be viewed as a `holomorphic anomaly' which is also governed by the
shadow.
\item The partition function of  immortal black holes with manifest spectral-flow invariance 
is a mock Jacobi form -- a new mathematical object  defined and  elaborated upon in \S{\ref{MockJacobi}}. 
\end{myenumerate}

The main physical payoff of the mathematics of mock modular forms in this context
is the guarantee that one can still define a completion as in (3) which \textit{is} modular albeit non-holomorphic. 
As mentioned earlier, the modular
transformations on the $\tau$ parameter can be identified with global
diffeomorphisms of the boundary of the near horizon $AdS_3$. This connection
makes the mathematics of mock modular forms physically very relevant for
$AdS_3/CFT_2$ holography in the presence of wall-crossing and holomorphic
anomalies. The required mathematical results concerning mock modular forms  
are developed in sections \S{\ref{Mock}}, \S\ref{MockfromJacobi}, 
\S\ref{flyJac}, and \S\ref{structure}. 

To orient the physics reader, we summarize the essential conclusions of this mathematical investigation from 
the perspective of the questions posed above.

\begin{enumerate}

\item  \textbf{Decomposition}: Given an asymptotic counting function,  the degeneracies of single-centered  
black holes can be isolated using the results in \S\ref{MockfromJacobi}. Applying Theorem~\ref{Thmdouble} to the 
 meromorphic Jacobi form  $\psi_{m}(\t, z)$ gives a unique decomposition
\be\label{decomp} \psi_{m}(\t, z)=\psi^{\rm F}_{m}(\t, z)\+\psi^{\rm P}_{m}(\t, z)\,, \ee
such that $\psi^{\rm P}_{m}(\t, z)$ is a simple function \eqref{simfun} with  the same pole structure in $z$ 
as $\psi_{m}(\t, z)$ and $\psi^{\rm F}_{m}(\t, z)$ has no poles. 
The elegant decomposition \eqref{decomp} is motivated partly by the choice of `attractor contour' for 
single-centered black holes and and has a direct physical interpretation:  $\psi_{m}(\t, z)$ is the counting 
function of all asymptotic states including both single and multi-centered configurations, $\psi^{\rm F}_{m}(\t, z)$ 
is the counting function of immortal black holes, whereas $\psi^{\rm P}_{m}(\t, z)$ is the counting function of multi-centered black holes. 

\hspace{6mm} Since both $\psi_{m}(\t, z)$ and $\psi^{\rm P}_{m}(\t, z)$ have poles in $z$,  their Fourier 
coefficients depend on the choice of the contour which in turn depends on the moduli. On the other hand,  
the Fourier coefficients of $\psi^{\rm F}_{m}(\t, z)$ are unambiguously defined without any contour or 
moduli dependence. This is what is expected for immortal black holes. 

\item \textbf{Modular Completion}: 
The immortal counting function $\psi^{\rm F}_{m}(\t, z)$  defined by the decomposition \eqref{decomp} is 
not modular.  However, theorem 8.3 ensures that by adding a specific nonholomorphic function  to $\psi^{\rm F}_{m}(\t, z)$,  
one can obtain its nonholomorphic completion $\wh \psi^{\rm F}_{m}(\t, z)$  which \textit{is} modular and transforms as 
a Jabobi form.  The failure of holomorphy of  the completion   $\wh \psi^{\rm F}_{m}(\t, z)$ is given by the equation
\be\label{holanom1}
 \tau_2^{3/2} \; \frac{\partial} {\partial \bar{\tau}}   \, \wh \psi_{m}^{\rm F}(\tau,z) \= 
\sqrt{\frac{m}{8 \pi i}} \; \frac{ p_{24}(m+1)}{ \Delta(\tau)} \,  
 \sum_{\ell \, \mod \, (2m)}  {\overline{\vth_{m,\ell}(\tau,0)}} \, \vartheta_{m,\ell} (\tau,z) \, .
\ee
Hence the counting function   of immortal black holes has a hidden modular symmetry and more specifically is
a mock Jacobi form as defined  in \S\ref{MockJacobi}. This is one of our main physics results and is described 
in more detail in \S\ref{MockforDyons}.

\item  \textbf{Holomorphic anomaly}:  The completion is a natural object to be identified with the indexed 
partition function of the superconformal field theory $SCFT_{2}$ dual to a single-centered $AdS_{3}$, which is expected  to be modular. 

\hspace{6mm} {}From this perspective,  the equation \eqref{holanom1} can be viewed  as a holomorphic anomaly and 
can in fact be taken as a defining property of a mock Jacobi form for physics applications. 
 Naively, an indexed partition function is expected to be holomorphic because of a   cancellation between 
right-moving bosons and fermions as for the elliptic genus \cite{Witten:1986bf}. However, if the $SCFT_{2}$ 
is noncompact, then the spectrum is continuous and this naive reasoning may fail leading  to an `anomaly'. 
The holomorphic anomaly can then arise as a consequence of the fact that for the right-movers in a noncompact 
$SCFT$, the density of states of bosons and fermions may be slightly different \cite{Troost:2010ud, Ashok:2011cy} 
and may not precisely cancel.  The detailed connection between the holomorphic anomaly and the noncompactness
of the $SCFT_{2}$ in this context needs to be understood better from a path-integral perspective.

\item 
 \textbf{Optimality}:  
 The Fourier coefficients of $\psi^{\rm F}_{m}(\t, z)$ grow exponentially rapidly as expected for a counting 
function of black hole degeneracies. It is clear from the anomaly equation \eqref{holanom1} that if  we add  
a holomorphic  true Jacobi form to $\psi^{\rm F}_{m}(\t, z)$ with the same weight and index, it will still 
admit a modular completion satisfying the same anomaly equation. This raises the question whether for a 
given holomorphic anomaly there is an `optimal'  mock  Jacobi form whose Fourier coefficients grow as slowly as possible.  
The  answer to this question (for the functions~$\v_{2,m}$ related to~$\psi_{m}$ by \eqref{psifromphim}) 
is  in the affirmative but is subtle and requires several new results in the theory of Jacobi forms 
developed in~\S\ref{structure}, motivated by and explaining the numerical experiments and 
observations described in~\S\ref{flyJac}.

\hspace{6mm} A practical implication of such an optimal choice is that the leading contribution to the black 
hole entropy is  then determined essentially by a Fourier coefficient of a true Jacobi form.  One can thus 
apply the  familiar Cardy formula  and  the Rademacher expansion of the Fourier coefficients of true modular 
forms for the leading answer. There will be exponentially subleading corrections to this leading answer
coming from the optimal mock Jacobi form. A Rademacher expansion for these corrections requires  a   
generalization~\cite{BringOno, BringMahlburg} applicable for mock rather than true modular forms.

\item  \textbf{Examples}:  Modular forms with slowly growing  Fourier coefficients are mathematically 
particularly interesting and many of the best-known examples of mock modular forms share this property.   
An `optimal' choice thus enables us in many  cases  to obtain a very  explicit expression for the immortal 
counting function $\psi^{\rm F}_{m}(\t, z)$ in terms of these known  mock modular forms. For example, 
for all $m$ prime, the optimal mock part  of $\psi^{\rm F}_{m}(\t, z)$  can be expressed in terms  
of the generating function of Hurwitz-Kronecker class numbers (see \S\ref{choosephi}).
On the other hand, for a nonprime $m$ our analysis leads to  new  mock modular forms with very 
slowly growing Fourier coefficients.

\hspace{6mm} The functions $\psi_{m}(\t, z)$ that arise in the black hole problem have a double 
pole at~$z=0$ and its translates, but our methods can also be applied to a related class of functions 
with just a single pole at $z=0$. This leads to a second infinite family of mock modular forms, this time 
of weight~1/2. (In the double pole case, the mock Jacobi forms had weight~2 and their coefficients 
were mock modular forms of weight~3/2. In both cases, the coefficients are in fact mock theta functions, i.e., 
mock modular forms whose shadows are unary theta series.) Unlike the first case, where we found that 
the mock modular forms occurring can always be chosen to have at most simple poles at the cusps, 
the pole order here is much greater in general, and there are only a handful of examples having only 
simple poles. It is remarkable that these include essentially all the most prominent examples of mock theta
functions, including the original ones of Ramanujan, the mock theta function conjecturally related 
to the Mathieu group $M_{24}$~\cite{Eguchi:2010ej} and the functions arising in the umbral moonshine 
conjecture~\cite{Cheng:2010pq}. 

\end{enumerate}

Modular symmetries are very powerful in physics applications because they relate
strong coupling to weak coupling, or high temperature to low temperature.   The  hidden modular symmetry 
of mock modular forms is therefore expected to be  useful in  diverse physics contexts.  As mentioned above,
mock modularity of the counting function in the present context of black holes is a consequence of meromorphy 
of the asymptotic counting function which  in turn is a consequence of noncompactness of the target space 
of the microscopic SCFT. Now, conformal field theories with a  noncompact target space occur naturally 
in several physics contexts. For example, a general class of four-dimensional BPS black holes is obtained 
as a supersymmetric D-brane configuration in Type-II compactification on a Calabi-Yau three-fold $X_6$. 
In the M-theory limit, these black holes can be viewed as excitations of the MSW black 
string \cite{Maldacena:1997de,Minasian:1999qn}. The microscopic theory describing the low energy 
excitations of the MSW string is the $(0, 4)$ MSW SCFT. The target space of this SCFT does not necessarily 
have to be compact in which case the considerations of this paper will apply. Very similar 
objects \cite{Zagier:1975, Hirzebruch:1976} have already made their appearance in the context of topological supersymmetric
Yang-Mills theory on $\mathbb{CP}{^2}$~\cite{Vafa:1994tf}. Other examples include 
the theory of multiple M5-branes~\cite{Alim:2010cf} quantum Liouville theory and E-strings~\cite{Minahan:1998vr}, 
and the $SL(2,\IR)/U(1)$ SCFT~\cite{Eguchi:2004yi,Troost:2010ud,Ashok:2011cy} where the CFT is noncompact. 
The appearance of a holomorphic anomaly in the regularized Poincar\'e series for the elliptic genus of $\rm CFT_{2}$
was noted in \cite{Manschot:2007ha} in the context of the $\rm AdS_{3}/CFT_{2}$ correspondence. 

We expect that the general framework of mock modular forms and Jacobi forms developed in this paper is likely 
to have varied physical applications in  the context of  non-compact conformal field theories, wall-crossings 
in enumerative geometry, and recently formulated  Mathieu and umbral moonshine conjectures~\cite{Eguchi:2010ej, Cheng:2010pq}.

\subsection{Organization of the paper}

In \S{\ref{ReviewK3}}, we review the physics
background concerning the string compactification on $K3 \times T^2$ and the
classification of BPS states corresponding to the supersymmetric black holes in
this theory. In sections \S{\ref{Modular}}, \S{\ref{Jacobi}}, and  \S{\ref{Siegel}}, we review the basic mathematical
definitions of various types of classical modular forms  (elliptic, Jacobi, Siegel)
and illustrate an application to the physics of quantum black holes in each
case by means of an example. In \S{\ref{Contours}}, we review the moduli
dependence of the Fourier contour prescription for extracting the degeneracies of quarter-BPS black holes 
in the ${\cal N}=4$ theory from the partition function which is a meromorphic Siegel modular form. 
In \S{\ref{MockfromJacobi}}, we refine results due to Zwegers to show that 
any meromorphic Jacobi form with poles only at the sub-multiples of the period lattice
can be decomposed canonically into two pieces, one of which is a mock Jacobi form. We give explicit 
formulas for this decomposition in the case when the poles have at most order~2,
and again give several examples. 
In \S\ref{flyJac} we give a detailed description of the experimental results for the 
Fourier coefficients of the two families of mock Jacobi forms~$\{\Phi_{2,m}\}$ and~$\{\Phi_{1,m}\}$, 
and formulate the main mathematical results concerning them. 
The proofs of these results are given in \S\ref{structure}, after we have formulated 
and proved a number of structural results about holomorphic and weakly holomorphic 
Jacobi forms that are new and may be of independent interest. In \S{\ref{MockforDyons}, we apply these results 
in the physical context to determine the mock Jacobi form that counts the  degeneracies of single-centered black holes
and discuss the implications for $AdS_{2}/CFT_{1}$ and $AdS_{3}/CFT_{2}$ holography.

\section{Review of Type-II superstring theory on $K3 \times T^2$  \label{ReviewK3}}

Superstring theories are naturally formulated in ten-dimensional Lorentzian
spacetime $ \M_{10}$. A `compactification' to four-dimensions is obtained by
taking $\M_{10}$ to be a product manifold $ \R^{1, 3} \times X_6$ where
$X_6$ is a compact Calabi-Yau threefold and $\R^{1, 3}$ is the noncompact
Minkowski spacetime. We will focus in this paper on a compactification of Type-II
superstring theory when $X_6$ is itself the product $X_6 = K3 \times T^2$. A
highly nontrivial and surprising result from the 90s is the statement that this
compactification is quantum equivalent or `dual' to a compactification of
heterotic string theory on $T^4 \times T^2$ where $T^4$ is a four-dimensional
torus \cite{Hull:1994ys, Witten:1995ex}. One can thus describe the theory either
in the Type-II frame or the heterotic frame.

The four-dimensional theory in $\R^{1,3}$ resulting from this
compactification has $ \mathcal{N}=4$ supersymmetry\footnote{This supersymmetry is a super Lie algebra
containing $ISO(1, 3) \times SU(4)$ as the even subalgebra where $ISO(1,3)$ is the Poincar\'e symmetry
of the $\R^{1, 3}$ spacetime and $SU(4)$ is an internal symmetry usually referred to as R-symmetry. 
The odd generators of the superalgebra are called supercharges. With $ \CN=4$ supersymmetry,
there are eight complex supercharges which transform as a spinor of $ISO(1,3)$ and a fundamental 
of~$SU(4)$.}. The massless fields in the theory consist of $22$ vector multiplets in addition to the
supergravity multiplet. The massless moduli fields consist of the $S$-modulus $\l$ taking values in the coset
  \be\label{Smoduli} SL(2,\Z)\backslash SL(2, \R)/ O(2, \R), \ee
and the $T$-moduli $\mu$ taking values in the coset
  \be\label{Narainmoduli} O(22, 6, \Z) \backslash O(22, 6, \R) /O(22, \R) \times O(6, \R). \ee
The group of discrete identifications $SL(2,\Z)$ is called the $S$-duality group. In the heterotic frame,
it is the electro-magnetic duality group \cite{Sen:1994yi, Sen:1994fa}, whereas in the
type-II frame, it is simply the group of area-preserving global diffeomorphisms
of the $T^2$ factor. The group of discrete identifications $O(22, 6, \mathbb{Z})$ 
is called the $T$-duality group. Part of the $T$-duality group $O(19, 3, \mathbb{Z})$ 
can be recognized as the group of geometric identifications on the
moduli space of K3; the other elements are stringy in origin and have to do with mirror symmetry.

At each point in the moduli space of the internal manifold $K3 \times T^2$, one
has a distinct four-dimensional theory. One would like to know the spectrum of
particle states in this theory. Particle states are unitary irreducible
representations, or supermultiplets, of the $\CN =4$ superalgebra. The
supermultiplets are of three types which have different dimensions in the rest
frame. A long multiplet is $256$-dimensional, an intermediate multiplet is
$64$-dimensional, and a short multiplet is $16$-dimensional. A short multiplet
preserves half of the eight supersymmetries ({\it i.e.} it is annihilated by four
supercharges) and is called a half-BPS state; an intermediate multiplet preserves
one quarter of the supersymmetry ({\it i.e.} it is annihilated by two
supercharges), and is called a quarter-BPS state; and a long multiplet does not
preserve any supersymmetry and is called a non-BPS state. One consequence of the
BPS property is that the spectrum of these states is `topological' in that it
does not change as the moduli are varied, except for jumps at certain walls in
the moduli space \cite{Witten:1978mh}.

An important property of a BPS states that follows from the superalgebra is
that its mass is determined by its charges and the moduli \cite{Witten:1978mh}.
Thus, to specify a BPS state at a given point in the moduli space, it suffices to
specify its charges.  The charge vector in this theory transforms in the vector
representation of the $T$-duality group $O(22, 6, \mathbb{Z})$ and in the
fundamental representation of the $S$-duality group $SL(2, \mathbb{Z})$. It is thus
given by a vector $\G^{\a I}$  with integer entries
  \be\label{chargevector} \G^{\a I} \= \left( \begin{array}{c} N^I \\ M^I \\ \end{array} \right)
   \qquad \text{where} \qquad {  \a = 1,2\, ; \quad {I=1,2, \ldots 28}} \, \,  \ee
transforming  in the $(2, 28)$ representation of $SL(2, \mathbb{Z}) \times O(22, 6, \mathbb{Z})$.
The vectors $N$ and $M$ can be regarded as the quantized electric and magnetic
charge vectors of the state respectively. They both belong to an even, integral,
self-dual lattice $\Pi^{22, 6}$. We will assume in what follows that $ \Gamma =
(N, M)$ in (\ref{chargevector}) is primitive in that it cannot be written as an
integer multiple of $ (N_0, M_0)$ for $N_0$ and $M_0$ belonging to $\Pi^{22, 6}$.
A state is called purely electric if only $N$ is non-zero, purely magnetic if
only $M$ is non- zero, and dyonic if both $M$ and $N$ are non-zero.

To define $S$-duality transformations, it is convenient to represent the $S$-modulus
as a complex field $S$ taking values in the upper half plane. An $S$-duality transformation
  \be\label{Sgroup} \ga \equiv\left( \begin{array}{cc} a&b\\ c&d \end{array} \right) \in SL(2;\Z) \ee
acts simultaneously on the charges and the $S$-modulus by
  \be\label{stransform}\left(\begin{array}{c} N \\ M \\ \end{array} \right) \rightarrow
    \left( \begin{array}{cc} a&b\\ c&d \end{array} \right)
    \left( \begin{array}{c} N \\ M \\ \end{array} \right)\;, \qquad S\to\frac{a S+b}{c S+d}\;. \ee

To define $T$-duality transformations, it is convenient to represent the $T$-moduli
by a $28 \times 28$ matrix $\mu^A_{\,\,I}$ satisfying 
  \be \mu^{t} \, L \, \mu = L \ee 
with the identification that $\mu\sim k\mu$ for every $k\in O(22;\R)\times O(6;\R)$. Here $L$ is the $28\times28$ matrix
  \be\label{lorentzian}  L_{IJ} = \left( \begin{array}{ccc} - \textbf{C}_{16} & \textbf{0} & \textbf{0}\\
    \textbf{0} & \textbf{0} & \textbf{\rm I}_6  \\ \textbf{0} &   \textbf{\rm I}_6 & \textbf{0}  \\ \end{array} \right), \ee
with $\textbf{\rm I}_s$ the $s\times s$ identity matrix and $ \textbf{C}_{16}$ is the Cartan matrix of
$E_8\times E_8$ . The $T$-moduli are then represented by the matrix 
  \be \CM \=\mu^t \mu \ee
which satisifies
  \be \CM^t \= \CM\,, \qquad \CM^t L \CM = L\,. \ee
In this basis, a $T$-duality transformation can then be represented by a $28\times 28 $ matrix $R$ with integer entries satisfying
  \be\label{defRtduality}  R^t L R \= L\,, \ee
which acts simultaneously on the charges and the $T$-moduli by
  \be\label{ttransform}  N \rightarrow R N; \quad M \rightarrow R M\,; \quad \mu \rightarrow \mu R^{-1} \ee

Given the matrix $\mu^{A}_I$, one obtains an embedding $ \Lambda^{22, 6}\subset
\mathbb{R}^{22, 6}$ of $\Pi^{22, 6}$ which allows us to define the
moduli-dependent charge vectors $Q$ and $P$ by
  \be\label{physcharges} Q^A =  \mu^A_I N_I\,,   \qquad P^{A} =  \mu^A_I M_I \, . \ee
The matrix $L$ has a $22$-dimensional eigensubspace with eigenvalue $-1$ and a
$6$- dimensional eigensubspace with eigenvalue $ +1$. Given $Q$ and $P$, one can
define the `right-moving' and `left-moving' charges\footnote{The right-moving charges couple to
the graviphoton vector fields associated with the right-moving chiral currents in
the conformal field theory of the dual heterotic string.} $Q_{R,L}$ and $P_{L,R}$ 
as the projections
\be 
Q_{R, L} \=\frac{(1 \pm L)}2\,Q \, ; \qquad P_{R, L} \= \frac{(1 \pm L)}2\,P \, . 
\ee

If the vectors $N$ and $M$ are nonparallel, then the state is quarter-BPS. On the
other hand, if $N= p N_0$ and $ M =q N_0$ for some $N_0 \in \Pi^{22, 6}$ with $p$
and $q$ relatively prime integers, then the state is half-BPS.

An important piece of nonperturbative information about the dynamics of the
theory is the exact spectrum of all possible dyonic BPS-states at all points in
the moduli space. More specifically, one would like to compute the number
$d(\Gamma)|_{S, \m}$ of dyons of a given charge $\Gamma$ at a specific point
$(S, \m)$ in the moduli space. Computation of these numbers is of course a very
complicated dynamical problem. In fact, for a string compactification on a
general Calabi-Yau threefold, the answer is not known. One main reason for
focusing on this particular compactification on $K3 \times T^2$ is that in this
case the dynamical problem has been essentially solved and the exact spectrum of
dyons is now known. Furthermore, the results are easy to summarize and the
numbers $d(\Gamma)|_{S, \m}$ are given in terms of Fourier coefficients of
various modular forms.

In view of the duality symmetries, it is useful to classify the inequivalent
duality orbits labeled by various duality invariants. This leads to an
interesting problem in number theory of classification of inequivalent duality
orbits of various duality groups such as $SL(2, \mathbb{Z}) \times O(22, 6; \Z)$
in our case and more exotic groups like $E_{7,7} (\Z)$ for other choices of
compactification manifold $X_6$. It is important to remember though that a
duality transformation acts simultaneously on charges and the moduli. Thus, it
maps a state with charge $\Gamma$ at a point in the moduli space $ (S, \mu)$ to
a state with charge $ \Gamma'$ but at some other point in the moduli space $(S',
\mu')$.  In this respect, the half-BPS and quarter-BPS dyons behave differently.
\begin{itemize}
  \item For half-BPS states, the spectrum does not depend on the moduli. Hence
$d(\Gamma)|_{S', \m'} = d(\Gamma)|_{S, \m}$. Furthermore, by an $S$-duality
transformation one can choose a frame where the charges are purely electric with
$M=0$ and $N \neq 0$. Single-particle states have $N$ primitive and the number of
states depends only on the $T$-duality invariant integer $n \equiv N^2/2$. We can
thus denote the degeneracy of half-BPS states $d(\Gamma)|_{S', \m'}$ simply by $d(n)$.
  \item For quarter-BPS states, the spectrum does depend on the moduli, and
$d(\Gamma)|_{S', \m'} \neq d(\Gamma)|_{S, \m}$. However, the partition function
turns out to be independent of moduli and hence it is enough to classify the
inequivalent duality orbits to label the partition functions. For the specific
duality group $SL(2, \mathbb{Z}) \times O(22, 6; \Z)$ the partition functions are
essentially labeled by a single discrete invariant \cite{Dabholkar:2007vk,Banerjee:2007sr,Banerjee:2008ri}.
      \be\label{gcd}  I \= \gcd (N \wedge M) \,,       \ee
The degeneracies themselves are Fourier coefficients of the partition
function. For a given value of $I$, they depend only on\footnote{There is an additional 
dependence on  arithmetic $T$-duality invariants but the degeneracies for states with nontrivial
values of these $T$-duality invariants can be obtained from the degeneracies discussed here
by demanding $S$-duality invariance \cite{Banerjee:2008ri}. } the moduli and the three
$T$-duality invariants $(m, n, \ell) \equiv (M^2/2, N^2/2, N \cdot M)$. Integrality
of $ (m, n, \ell)$ follows from the fact that both $N$ and $M$ belong to $\Pi^{22, 6}$. 
We can thus denote the degeneracy of these quarter-BPS states $d(\Gamma)|_{S, \m}$ simply
by $d(m, n, l)|_{S, \m}$.  For simplicity, we consider only $I=1$ in this paper.
\end{itemize}

Given this classification, it is useful to choose a representative set of charges
that can sample all possible values of the three $T$-duality invariants. For this
purpose, we choose a point in the moduli space where the torus $T^2$ is a product
of two circles $S^1 \times \tS^1$ and choose the following charges in a Type-IIB frame.
\begin{myitemize}
  \item For electric charges, we take $n$ units of momentum along the circle $S^1$, 
and $\tilde K$  Kaluza- Klein monopoles associated with the circle $\tS^1$.
  \item For magnetic charges, we take $Q_1$ units of D1-brane charge wrapping $S^1$, $Q_5$
D5-brane wrapping $K3 \times S^1$ and $l$ units of momentum along the $\tS^1$ circle.
\end{myitemize}
We can thus write 
\be\label{final charges} \G \= \left[ \begin{array}{c}  N \\ M \end{array} \right] 
   \= \left[ \begin{array}{cccc} 0,& n; & 0, & \tilde K \\ Q_1, & \tilde n; & Q_5, & 0\\ \end{array} \right] \,. \ee
The $T$-duality quadratic invariants can be computed using a restriction of the matrix
(\ref{lorentzian}) to a $ \Lambda^{(2,2)}$ Narain lattice of the form
  \be\label{lorentzintwo}  L \= \left( \begin{array}{cc} \textbf{0} & \textbf{\rm I}_2  \\ 
                  \textbf{\rm I}_2 & \textbf{0} \\ \end{array} \right) , \ee
to obtain 
  \be\label{invts} M^2/2 \=  Q_1Q_5\,, \quad N^2/2 \= n \tilde K\,, \quad N \cdot M =\tilde n \tilde K\,. \ee
We can simply the notation further by choosing $\tilde K=Q_5=1$, $Q_1= m$, $\tilde n = l$ to obtain
  \be\label{invts}  M^2/2 \= m, \quad N^2/2 \= n, \quad N \cdot M =   l\,. \ee

For this set of charges,  we can focus our attention on a subset of $T$-moduli associated with the torus $T^2$ parametrized by
  \be\label{moduli-matrix} \CM \=  \left( \begin{array}{ccc} G^{-1}\,\quad  &G^{-1}B \\
   -BG^{-1}\,\quad &G -B G^{-1} B\end{array}\right) \,,\ee
where $G_{ij}$ is the metric on the torus and $B_{ij} $ is the antisymmetric tensor field. Let $U = U_1 + i U_2$ be the
complex structure parameter, $A$ be the area,  and $\e_{ij}$ be the Levi-Civita symbol with $\e_{12} = -\e_{21} =1$, then
  \be\label{metric-B} G_{ij} = \frac{A}{U_2} \left(\begin{array}{cc} 1& U_1 \\ U_1 & |U|^2 \\ \end{array}\right)
   \quad \text{and} \quad B_{ij} = A B \e_{ij} \,, \ee
and the complexified K\"ahler modulus $U = U_{1 }+ iU_2$ is defined as $U:= B + i A$. The $S$-modulus $S = S_1 + S_2$
is defined as 
\be\label{Sdef}
S:= a + i\exp{(-2\phi)}
\ee
 where $a$ is the axion and $\phi$ is the dilaton field in the four dimensional
heterotic frame.  the relevant moduli can be parametrized by three complex scalars $S, T, U$ which define the so-called
`STU' model in $\CN =2$ supergravity.  Note that these moduli are labeled naturally in the heterotic frame which 
are related to the $S_{B}$, $T_{B}$, and $U_{B}$ moduli in the Type-IIB frame by
  \be\label{hetBrel} S \= U_B, \quad T = S_B,  \quad U = T_B \,. \ee

\section{Modular forms in one variable \label{Modular}}

Before discussing mock modular forms, it is useful to recall the variety of
modular objects that have already made their appearance in the context of
counting black holes. In the following sections we give the basic
definitions of modular forms, Jacobi forms, and Siegel forms, using the notations
that are standard in the mathematics literature, and then in each case illustrate
a physics application to counting quantum black holes by means of an example.

In the physics context, these modular forms arise as generating functions for
counting various quantum black holes in string theory. The structure of poles of
the counting function is of particular importance in physics, since it determines
the asymptotic growth of the Fourier coefficients as well as the contour
dependence of the Fourier coefficients which corresponds to the wall crossing
phenomenon. These examples will also be relevant later in \S{\ref{MockforDyons}} in
connection with mock modular forms. We suggest chapters I and III of
\cite{Mod123} respectively as a good general reference for classical and Siegel
modular forms and \cite{Eichler:1985ja} for Jacobi modular forms.

\subsection{Basic definitions and properties} \label{modularbasic}
Let $\mathbb{H}$ be the upper half plane, \textit{i.e.}, the set of complex numbers $\tau$
whose imaginary part satisfies $\Im(\t)>0$. Let $SL(2, \mathbb{Z})$ be the group of matrices  $\sm abcd$
 with integer entries such that $ad - bc =1$.

A \emph{modular form} $f (\t)$ of weight $k$ on $SL(2,\Z)$ is a holomorphic
function on $\IH$, that transforms as
  \be\label{modtransform0} f(\frac{a\t + b }{c\t + d}) = (c\t + d)^k\,f(\t) \qquad \forall \;\left(
       \begin{array}{cc} a & b \\ c & d \\  \end{array} \right) \in SL(2,\Z)\,, \ee
for an integer $k$ (necessarily even if $f(0) \neq 0$). It follows from the definition
that $f(\t)$ is periodic under $\tau \to \tau +1$ and can be written as a Fourier series
  \be\label{holmod} f(\t) \= \sum_{n= -\infty}^\infty a(n)\,q^n   \qquad \bigl(q :=  e^{2 \pi i \tau}\bigr)\,,  \ee
and is bounded as $\rm{Im}(\t) \to \infty$.
If $a(0) =0$, then the modular form vanishes at infinity and is called a \emph{cusp form}.
Conversely, one may weaken the growth condition at $\infty$ to $f(\t) = \CO(q^{-N}) $
rather than $ \CO(1)$ for some $N \ge 0$; then the Fourier coefficients of $f$ have the behavior
$a(n)=0$ for $n < -N$. Such a function is called a \emph{weakly holomorphic modular form}.

The vector space over $\mathbb{C}$ of holomorphic modular forms of weight $k$ is
usually denoted by~$M_k$. Similarly, the space of cusp forms of weight $k$ and
the space of weakly holomorphic modular forms of weight $k$ are denoted by $S_k$
and $M^{\,!}_k$ respectively. We thus have the inclusion
  \be\label{in}  S_k \; \subset M_k \; \subset \;M^{\,!}_k \;.  \ee
The Fourier coefficients of the modular forms in these spaces have different growth properties:
\begin{enumerate}
\item $f \in S_{k}\;\Rightarrow \; a_{n} \= \CO(n^{k/2})$ \ as $n \to \infty\,$;
\item $f \in M_{k}\; \Rightarrow \; a_{n} \= \CO(n^{k-1})$ \ as $n \to \infty\,$;
\item $f \in M_{k}^{\,!}\;\Rightarrow\; a_{n} \= \CO(e^{C \sqrt{n}})$ \ as $n \to \infty$ for some $C>0\,$.
\end{enumerate}

Some important modular forms on $SL(2,\Z)$ are:
\begin{enumerate}
\item The {\it Eisenstein series} $E_{k} \in M_{k}$ ($k \ge 4$).  The first two of these are
\begin{eqnarray}\label{eisen}
   E_4 (\t) &\=& 1 \+ 240\,\sum_{n=1}^\infty \frac{n^3q^n}{1- q^n} \= 1+ 240 q + 2160q^2 + \cdots \, , \\
   E_6 (\t) &\=& 1 \,-\, 504\,\sum_{n=1}^\infty \frac{n^5q^n}{1- q^n} \= 1 -504q -16632q^2 - \cdots \,.
\end{eqnarray}
\item The {\it discriminant function} $\Delta$. It is given by the product expansion
  \be\label{discrim}  \D(\t)  \=  q\,\prod_{n=1}^\infty {(1 - q^n)^{24}} \= q - 24 q^2 + 252 q^3 + ... \ee
or by the formula $\Delta =  \left( E_4^3 - E_6^2 \right)/1728$.  We mention for later use that the 
function $E_2(\t)=\dfrac1{2\pi i}\dfrac{\D'(\t)}{\D(\t)}=1-24\sum\limits_{n=1}^\infty\dfrac{nq^n}{1-q^n}$ is also an Eisenstein series, 
but is not modular. (It is a so-called {\it quasimodular form}, meaning in this case that the non-holomorphic function 
$\wh E_2(\t)=E_2(\t)-\dfrac3{\pi\,\Im(\t)}$ transforms like a modular form of weight~2.) This function can be used to form the {\it Ramanujan-Serre derivative}
  \be\label{RSderiv} \dRS{k}\,:\,M_k\to M_{k+2}\,, \qquad  \dRS{k}f(\t)\,:=\,\frac1{2\pi i}\,f'(\t) \,-\,\frac k{12}\,E_2(\t)\,f(\t)\,. \ee
\end{enumerate}

The ring of modular forms on $SL(2,\Z)$ is generated freely by $E_4$ and $E_6$, so any modular form of weight $k$ can be 
written (uniquely) as a sum of monomials $E_4^\a E_6 ^\b$ with $4\a+6\b =k$.  We also have $M_{k}=\C\cdot E_k\oplus S_k$ 
and $S_{k} = \Delta \cdot M_{k-12}$, so any $f \in M_{k}$ also has a unique expansion as $\sum_{n=0}^{[k/12]} \a_n\, E_{k-12n} \,\Delta^{n}$ 
(with $E_0=1$ and $E_2$ replaced by~0). From either representation, we see that a modular form is uniquely determined by its weight and first few Fourier coefficients.

Given two modular forms $(f,g)$ of weight $(k,l)$, one can produce a sequence of modular forms of 
weight $k+l+2n, \, n \ge 0$ using the {\it Rankin-Cohen bracket} 
\be\label{RCbrac}   [f,g]_{n} \= [f,g]^{(k,l)}_{n} \= 
   \sum_{r+s=n} (-1)^s \left( {k+n-1 \atop r} \right) \left( {\ell+n-1 \atop s} \right) f^{(s)}(\t) g^{(r)}(\t) \ee
where $f^{(m)} := \left(\frac 1{2 \pi i} \frac{d}{d\tau} \right)^{m} f$. For $n=0$, this is simply the product 
of the two forms, while for $n>0$ we always have $[f,g]_n \in S_{k+l+2n}$.  The first two non-trivial examples are 
  \be\label{rcbraceg} [E_4, E_6]_1 \= -3456 \,\Delta \ , \qquad [E_4, E_4]_2 = 4800 \, \Delta \ . \ee

As well as modular forms on the full modular group $SL(2,\Z)$, 
one can also consider modular forms on subgroups of finite index, with the same transformation law
\eqref{modtransform0} and suitable conditions on the Fourier coefficients to define the notions of 
holomorphic, weakly holomorphic and cusp forms. The weight $k$ now need no longer be even, 
but can be odd or even half integral, the easiest way to state the transformation property when 
$k \in \Z + \half$ being to say that $f(\t)/\theta(\t)^{2k}$ is invariant under some congruence subgroup
of $SL(2,\Z)$, where $\theta(\t)=\sum_{n\in\Z}q^{n^2}$. The graded vector space of modular forms 
on a fixed subgroup $\G \subset SL(2,\Z)$ is finite dimensional in each weight, finitely generated 
as an algebra, and closed under Rankin-Cohen brackets.  Important examples of modular forms of half-integral
weight are the {\it unary theta series}, i.e., theta series associated to a quadratic form in one variable. 
They come in two types:
\be\label{unary1}  
  \sum_{n\in\Z}\ve(n)\,q^{\l n^2} 
  \qquad\text{for some $\l\in\Q_+$ and some even periodic function $\ve$} 
\ee
and
\be\label{unary2} 
 \sum_{n\in\Z}n\,\ve(n)\,q^{\l n^2} 
 \qquad\text{for some $\l\in\Q_+$ and some odd periodic function $\ve$}\,, 
\ee
the former being a modular form of weight 1/2 and the latter a cusp form of weight 3/2.  
A theorem of Serre and Stark says that in fact every modular form of weight 1/2 is a linear 
combination of form of the type \eqref{unary1}, a simple example
being the identity
\be\label{defeta}
  \eta(\tau)\;:=\; q^{1/24} \prod_{n=1}^{\infty} (1-q^{n}) 
  \= \sum_{n=1}^{\infty} \chi_{12}(n)\,q^{n^2/24}\,,
\ee
proved by Euler for the so-called {\it Dedekind eta function} $\eta(t)=\D(\t)^{1/24}$. 
Here $\chi_{12}$ is the function of period 12 defined by
\be\label{defchi12} 
 \chi_{12}(n) \=  \left\{ \begin{array}{lll} + 1 \quad & {\rm if} \; n \equiv \pm 1 \, ({\rm mod} \, 12) \\  
  - 1 \quad & {\rm if} \; n \equiv \pm 5 \, ({\rm mod} \, 12) \\  0 \quad & {\rm if}  \; (n,12)>1 \ . \end{array} \right. 
\ee

Finally, we recall the definition of the Petersson scalar product. If $f(\t)$ and $g(\t)$
are two modular forms of the same weight $k$ and the same ``multiplier system'' on some subgroup $\G$ of 
finite index of $SL_{2}(\IZ)$ (this means that the quotient $f/g$ is invariant under $\G$), and if either $k<1$ 
or else at least one of $f$ and $g$ is a cusp form, then we can define the (normalized) 
\emph{Petersson scalar product}  of $f$ and $g$ by 
\be \label{Peterkabeta}
\big(f,g\big) \=  \int_{\CF_{\G}} \, f(\t) \, \overline{g(\t)} \, \t_{2}^{k} \, d \mu(\t) 
\left/  \int_{\CF_{\G}} \, d \mu(\t) \right.
  \qquad\Bigl(\,d \mu(\t)\,:=\, \frac{d\t_{1} d\t_{2}}{\tau_{2}^{2}} \Bigr)\, , 
\ee
where $\CF_{\G}$ is a fundamental domain for $\G$. 
This definition is independent of the choice of the subgroup $\G$ and the fundamental 
domain $\CF_{\G}$.  By the Rankin-Selberg formula (see, for example,~\cite{Zagier:1982}),
we have that if $f=\sum_\l a_\l q^\l$ and $g=\sum_\l b_\l q^\l$ (where $\l\ge0$ may have a denominator), then 
\be \label{RS} (f,g)\=\frac{\Gamma(k)}{(4\pi)^k}\;
\text{Res}_{s=k}\biggl(\,\sum_{\l>0}\frac{a_\l \, \overline{b_\l}}{\l^s}\biggr)\,,\ee
a formula that will be used later. For instance, for $f=g=\eta$ we have $k=1/2$ and 
$\sum_{\l} a_{\l}  b_{\l}  \l^{-s} = 24^{s}(1-2^{-2s})(1-3^{-2s}) \, \zeta(2s)$, and hence  
$(\eta, \eta) = 1/\sqrt{6}$.

\subsection{Quantum black holes and modular forms}

Modular forms occur naturally in the context of counting the Dabholkar-Harvey  (DH)
states \cite{Dabholkar:1989jt,Dabholkar:1990yf}, which are states in the string
Hilbert space that are dual to perturbative BPS states. The spacetime helicity
supertrace counting the degeneracies reduces to the partition function of a
chiral conformal field theory on a genus-one worldsheet.  The $\tau $ parameter
above becomes the modular parameter of the genus one Riemann surface. The
degeneracies are given by the Fourier coefficients of the partition function.

A well-known simple example is the partition function $Z(\tau) $ which counts the
half-BPS DH states for the Type-II compactification on $K3 \times T^2$ considered
here. In the notation of (\ref{chargevector}) these states have zero magnetic
charge $M=0$, but nonzero electric charge $N$ with the $T$-duality invariant $N^2
=2n$, which can be realized for example by setting $Q_1 =Q_5=l =0$ in (\ref{final
charges}). They are thus purely electric and perturbative in the heterotic
frame\footnote{Not all DH states are half-BPS. For example, the states that are
perturbative in the Type-II frame correspond to a Type-II string winding and
carrying momentum along a cycle in $T^2$. For such states both $M$ and $N$ are
nonzero and nonparallel, and hence the state is quarter- BPS.}. The partition
function is given by the partition function of the chiral conformal field theory
of $24$ left-moving transverse bosons of the heterotic string. The Hilbert space
$\cal H$ of this theory is a unitary Fock space representation of the commutation algebra
\be\label{osci}  [a_{in},\,a_{j m}^\dagger]  \= \delta_{ij}\,\delta_{n+m,0} 
  \qquad (i,\,j = 1,\ldots,24\,, \quad n,\,m = 1,\,2, \ldots,\infty)  \ee
of harmonic modes of oscillations of the string in $24$ different directions. The Hamiltonian is
  \be\label{ham}  H \= \sum_{i=1}^{24} n\,a_{in}^\dagger\,a_{in} \;-\;1\,, \ee
and the partition function is
  \be\label{defpart}  Z(\t) \= \Tr_{\cal H} (q^H) \,. \ee
This can be readily evaluated since each oscillator mode of energy $n$ contributes to the trace
  \be\label{harm}  1 + q^n + q^{2n} + \ldots  \= \frac 1{1-q^n} \,. \ee
The partition function then becomes
  \be\label{partsmall}  Z(\t) \=  \frac 1{\D(\t)}\,,  \ee
where $\Delta$ is the cusp form (\ref{discrim}). Since $\Delta$
has a simple zero at $q=0$, the partition function itself has a pole at $q=0$,
but has no other poles in $\mathbb{H}$. Hence, $Z(\t)$ is a weakly holomorphic
modular form of weight $-12$. This property is essential in the present physical
context since it determines the asymptotic growth of the Fourier coefficients.

The degeneracy $d(n)$ of the state with electric charge $N$ depends only on the
$T$-duality invariant integer $n$ and is given by
  \be\label{partsmall2}  Z(\t) \=   \sum_{n=-1}^\infty d(n)\,q^n  \,. \ee
For the Fourier integral
  \be\label{dn}  d(n)\ = \int_C e^{-2\pi i\t n} Z(\t) d\t  \,,  \ee
one can choose the contour $\CC$ in $\mathbb{H}$ to be
  \be\label{cont1}  0 \leq \Re(\t) <1 \,,  \ee
for a fixed imaginary part $\Im(\tau) $. Since the partition function has no poles in
$ \mathbb{H}$ except at $q =0$, smooth deformations of the contour do not change
the Fourier coefficients and consequently the degeneracies $d(n)$ are uniquely
determined from the partition function. This reflects the fact that the half-BPS states are
immortal and do not decay anywhere in the moduli space. As a result, there is no wall crossing
phenomenon, and no jumps in the degeneracy.

In number theory, the partition function above is well-known in the context of
the problem of partitions of integers. We can therefore identify
  \be\label{colorpartition}  d(n) \= p_{24}(n +1)  \qquad (n \geq 0)\,.  \ee
where $p_{24}(I)$ is the number of colored partitions of a positive integer $I$ using
integers of $24$ different colors.

These states have a dual description in the Type-II frame where they can be
viewed as bound states of $Q_1$ number of D1-branes and $Q_5$ number of D5-branes
with $M^2/2 = Q_1 Q_5 \equiv m $. This corresponds to setting $n= \tK = l =0$ in
(\ref{final charges}). In this description, the number of such bound states
$d(m)$ equals the orbifold Euler character $ \chi(\textrm{Sym}^{m+1}(K3))$ of the
symmetric product of $(m+1)$ copies of $K3$-surface \cite{Vafa:1994tf}. The
generating function for the orbifold Euler character
  \be\label{Zeuler} \hat Z(\s) = \sum_{m=-1}^\infty  \chi(\textrm{Sym}^{m+1}(K3))\, p^m
  \qquad \bigl(p:= e^{2\pi i\s}\bigr)\ee
can be evaluated \cite{Goettsche:1990go} to obtain
  \be\label{Zeuler2} \hat Z(\s) = \frac 1{\Delta(\sigma)} \,. \ee
Duality requires that the number of immortal BPS-states of a given charge must equal the
number of BPS-states with the dual charge. The equality of the two partition functions
(\ref{partsmall}) and (\ref{Zeuler2}) coming from two very different counting
problems is consistent with this expectation. This fact was indeed one of the
early indications of a possible duality between heterotic and Type-II strings
\cite{Vafa:1994tf}.

The DH-states correspond to the microstates of a small black hole 
\cite{Sen:1995in,Dabholkar:2004yr,Dabholkar:2004dq} for large $n$. The macroscopic entropy
$S(n)$ of these black holes should equal the asymptotic growth of the degeneracy
by the Boltzmann relation \be\label{bolt} S(n) = \log d(n); \quad n
\gg 1 \,. \ee In the present context, the macroscopic entropy can be
evaluated from the supergravity solution of small black holes 
\cite{LopesCardoso:1998wt,LopesCardoso:1999ur,LopesCardoso:1999cv,
LopesCardoso:1999xn,Dabholkar:2004yr,Dabholkar:2004dq}. 
The asymptotic growth of the microscopic degeneracy can
be evaluated using the Hardy-Ramanujan expansion (Cardy formula). There is a
beautiful agreement between the two results \cite{Dabholkar:2004yr,Kraus:2005vz}
  \be\label{bolt2} S(n) \= \log d(n) \sim 4 \pi \sqrt{n}\quad n \gg 1 \,.   \ee
Given the growth properties of the Fourier coefficients mentioned above, it is clear
that, for a black hole whose entropy scales as a power of $n$ and not as
$\log(n)$, the partition function counting its microstates can be only weakly
holomorphic and not holomorphic.

These considerations generalize in a straightforward way to congruence subgroups
of $SL(2, \mathbb{Z})$ which are relevant for counting the DH-states in various
orbifold compactifications with $\CN=4$ or $ \CN=2$ supersymmetry
\cite{Dabholkar:2005by,Sen:2005ch,Dabholkar:2005dt}.

\section{Jacobi forms \label{Jacobi}}

\subsection{Basic definitions}

Consider a holomorphic function $\v(\tau, z)$ from $\mathbb{H} \times\C$ to $\C$ which 
is ``modular in $\tau$ and elliptic in $z $'' in the sense that it transforms under the modular group as
  \be\label{modtransform}  \v\Bigl(\frac{a\t+b}{c\t+d},\frac{z}{c\t+d}\Bigr) \= 
   (c\t+d)^k\,e^{\frac{2\pi imc z^2}{c\t+d}}\,\v(\t,z)  \qquad \forall \quad
   \Bigl(\begin{array}{cc} a&b\\ c&d \end{array} \Bigr) \in SL(2; \Z) \ee
and under the translations of $z$ by $\mathbb{Z} \tau + \mathbb{Z}$ as
  \be\label{elliptic}  \v(\t, z+\l\tau+\mu)\= e^{-2\pi i m(\l^2 \t + 2 \l z)} \v(\t, z)
  \qquad \forall \quad \l,\,\m \in \Z \, , \ee
where $k$ is an integer and $m$ is a positive integer.

These equations include the periodicities $\v(\t+1,z) = \v(\t,z)$ and $\v(\t,z+1) = \v(\t,z)$, so $\v$ has a Fourier expansion
  \be\label{fourierjacobi} \v(\t,z) \= \sum_{n, r} c(n, r)\,q^n\,y^r\,, \qquad\qquad
   (q :=e^{2\pi i \t}, \; y := e^{2 \pi i z}) \ . \ee
Equation \eqref{elliptic} is then equivalent to the periodicity property
  \be\label{cnrprop}  c(n, r) \= C(4 n m - r^2 , \, r) \ ,
  \qquad \mbox{where} \; C(\DD, r) \; \mbox{depends only on} \; r \mypmod{2m} \ . \ee
(We will sometimes denote $c(n,r)$ and $C(\DD,r)$ by $c_\v(n,r)$ and $C_{\v}(\DD,r)$ 
or by $c(\v;n,r)$ and $C(\v\, ;\DD,r)$ when this is required for emphasis or clarity.) 
The function $\v(\tau, z)$ is called a \emph{holomorphic Jacobi form} (or simply a \emph{Jacobi form})
of weight $k$ and index $m$ if the coefficients $C(\DD,r)$ vanish for $\DD<0$, {\it i.e.} if 
  \be\label{holjacobi}  c(n, r) \= 0 \qquad \textrm{unless} \qquad 4mn \ge r^2\,. \ee
It is called a \emph{Jacobi cusp form} if it satisfies the stronger condition that
$C(\DD,r)$ vanishes unless $\DD$ is strictly positive, {\it i.e.}
  \be\label{cuspjacobi}  c(n, r) = 0 \qquad \textrm{unless} \qquad 4mn > r^2 \ , \ee
and it is called a \emph{weak Jacobi form} if it satisfies the weaker condition
  \be\label{weakjacobi} c(n, r) \= 0\qquad   \textrm{unless}  \qquad n \geq 0 \, \ee
rather than \eqref{holjacobi}, whereas a merely \emph{weakly holomorphic Jacobi form} satisfies only the yet weaker condition
that $c(n,r)=0$ unless $n\ge n_0$ for some possibly negative integer~$n_0$ (or equivalently $C(\D,r)=0$ unless $\D\ge\D_0$
for some possibly negative integer~$\D_0$).  The space of all holomorphic (resp.~cuspidal, weak, or weakly holomorphic) Jacobi
forms of weight~$k$ and index~$m$ will be denoted by $J_{k,m}$ (resp.~$J_{k,m}^0$,  $\wt J_{k,m}$, or $\wt J^{\,!}_{k,m}$).

Finally, the quantity $\D=4mn-r^2$, which by virtue of the above discussion is the crucial invariant of a monomial
$q^ny^r$ occurring in the Fourier expansion of~$\v$, will be referred to as its {\it discriminant}.  
(It would be mathematically more correct to use this word for the quantity~$-\D$, but $\D$ is usually positive 
and it is more convenient to work with positive numbers.)

\subsection{Theta expansion and Taylor expansion \label{ThetaTaylorExp}}

A Jacobi form has two important representations, the {\it theta expansion} and the {\it Taylor expansion}. 
In this subsection, we explain both of these and the relation between them. 

If $\v(\t, z)$ is a Jacobi form, then the transformation property (\ref{elliptic}) implies its
Fourier expansion with respect to $z$ has the form
  \be\label{jacobi-Fourier} \v(\t, z) \= \sum_{\ell\inn \Z} \;q^{\ell^2/4m}\;h_\ell(\t) \; e^{2\pi i\ell z} \ee
where $h_\ell(\tau)$ is periodic in $\ell$ with period $2m$.  In terms of the coefficients \eqref{cnrprop} we have
  \be\label{defhltau}  h_{\ell}(\t) \= \sum_{\DD} C(\DD,\ell) \,  q^{\DD/4m} \, \qquad \qquad (\ell \inn \Z/2m \Z)\;.  \ee
Because of the periodicity property, equation \eqref{jacobi-Fourier} can be rewritten in the form 
  \be\label{jacobi-theta} \v(\t,z) = \sum_{\ell\inn \Z/2m\Z} h_\ell(\t) \, \vth_{m,\ell}(\t, z)\,, \ee
where $\vth_{m,\ell}(\t,z)$ denotes the standard index $m$ theta function 
  \begin{eqnarray} \label{thetadef} \vth_{m,\ell}(\t, z) 
   \;:=\; \sum_{{r\inn\Z} \atop {r\,\equiv\,\ell\,(\mod\,2m)}} q^{r^2/4m} \, y^r \, \,
   \= \sum_{n\inn\Z} \,q^{(\ell+2mn)^2/4m} \,y^{\ell+2mn}  \end{eqnarray}
(which is a Jacobi form of weight $\half$ and index $m$ on some subgroup of
$SL(2, \mathbb{Z})$). This is the theta expansion of $\v$.  The coefficiens $h_{\ell}(\t)$ are  modular forms of weight $k-\frac12$
and are weakly holomorphic, holomorphic or cuspidal if $\v$ is a weak Jacobi form, a Jacobi form or a Jacobi cusp form, respectively. 
More precisely, the vector $h := ( h_1, \ldots, h_{2m})$ transforms like a modular form of weight $k-\frac 12$ under $SL(2,\Z)$.

The theta decomposition \eqref{jacobi-theta} leads to the definition of a differential operator on Jacobi forms as follows.  Let
$$ \LL_m \= \frac{4m}{2\pi i}\,\frac\partial{\partial\t} \,-\,\frac1{(2\pi i)^2}\,\frac{\partial^2}{\partial z^2}  $$
be the index $m$ heat operator, which sends $\v=\sum c(n,r)q^ny^r$ to $\sum(4nm-r^2)c(n,r)q^ny^r$. Here the Fourier coefficients
have the same periodicity property~\eqref{cnrprop} as for~$\v$, so $\LL_m\v$ has the same elliptic transformation properties as a Jacobi
form of index~$m$. Moreover, since $\LL_m$ annihilates all the $\vth_{m,\ell}$, we have 
$\LL_m(\sum h_\ell\vth_{m,\ell})=4m\sum h'_\ell\vth_{m,\ell}$, so the modified heat operator
\be\label{LRS} \LL_{k,m} = \LL_m -\frac{m(k-\frac12)}3\,E_2  \;\,:\;\;
   \sum_{\ell} h_\ell(\t)\,\vth_{m,\ell}(\t,z) \;\mapsto\; 4m\,\sum_\ell \dRS{k-\frac12}h_\ell(\t)\,\vth_{m,\ell}(\t,z) \,, \ee
where $\dRS{*}$ is the Ramanujan-Serre derivative defined in~\eqref{RSderiv}, sends $J_{k,m}$ to $J_{k+2,m}$ 
(and also $\wt J_{k,m}$ to $\wt J_{k+2,m}$ and $J^{\,!}_{k,m}$ to $J^{\,!}_{k+2,m}$).  These operators will be used later.

A Jacobi form $\v \in \wt J_{k,m}$ also has a Taylor expansion in $z$ which for $k$ even takes the form 
  \be\label{defxin}  \v(\t,z) \= \xi_0(\t) \+ \left(\frac{\xi_1(\t)}2 + \frac{ m \xi'_0(\t)}{k} \right) (2\pi iz)^2   
  \+ \left(\frac{\xi_2(\t)}{24} + \frac{m \xi'_1(\tau) }{2\,(k+2)}+ \frac{m^2  \xi^{''}_0 (\tau)}{2 k(k+1)} \right) \, (2 \pi i z)^4 + \cdots \ee
with $\xi_{\nu} \in M_{k+2 \nu} (SL(2,\Z))$ and the prime denotes $\frac 1{2 \pi i} \frac{d}{d \tau} $ as before. 
In terms of the Fourier coefficients of $\v$, the modular form $\xi_\nu$ is given by
  \be\label{xicnrsrel}  \frac{(k+2\nu -2)!}{(k+ \nu -2)!} \, \xi_{\nu} (\tau) 
  \= \sum_{n=0}^{\infty} \left( \sum_{r} P_{\nu,k} (nm,r^2) c(n,r) \right) q^{n}\,,  \ee
where $P_{\nu,k}$ is a homogeneous polynomial of degree $\nu$ in $r^2$ and $n$ 
with coefficients depending on $k$ and $m$, the first few being
  \bea\label{Pnu012}  P_{0,k} & \= &  1 \ , \cr   P_{1,k} & \= &   k r^2 \,-\, 2nm  \ , \cr
  P_{2,k} & \= &  (k+1)(k+2)r^4 \,-\, 12 (k+1) r^2 m n  \+ 12 m^2  n^2 \ .  \eea
The Jacobi form $\v$ is determined by the first $m+1$ coefficients $\xi_{\nu}$, and the map   $\v \mapsto (\xi_0,\dots,\xi_m)$
is an isomorphism from $\wt J_{k,m}$ to $M_{k} \oplus M_{k+2} \oplus \cdots \oplus M_{k+2m}$.  
For $k$ odd, the story is similar except that \eqref{defxin} must be replaced by 
  \be\label{defxinodd}
  \v(\t,z) \= \xi_0(\t)\;(2\pi iz) \+ \left(\frac{\xi_1(\t)}6+\frac{m \xi'_0(\t)}{k+2} \right) \, (2\pi iz)^3 \+ \cdots     \ee
with $\xi_{\nu} \in M_{k+2\nu+1}(SL(2,\Z))$, and the map $\v \mapsto (\xi_0, \dots, \xi_{m-2})$ gives an 
isomorphism from $\wt J_{k,m}$ to $M_{k+1} \oplus M_{k+3} \oplus \cdots \oplus M_{k+2m-3}$.

We also observe that even if $\v$ is weak, so that the individual coefficients $c(n,r)$ grow like $C^{\sqrt{4nm - r^2}}$, 
the coefficients  $ \sum_{r} P_{\nu,k} (nm,r^2) \, c(n,r) $ of $\xi_{\nu}$ still have only polynomial growth. We thus have the
following descriptions (analogous to those given for classical modular forms in \S\ref{modularbasic}) 
of holomorphic and weak Jacobi forms in terms of their asymptotic properties: \\

\vspace{0.2cm}

\begin{tabular}{ccl}
  Holomorphic Jacobi & $\Longleftrightarrow$ &  $c(n,r) =0$ for $4mn-r^2 < 0$ \\
   & $\Longleftrightarrow$ & the function $q^{m\a^2} \v(\t,\a \t + \b)$ (which is a modular form of weight $k$  \\
   && and some level) is bounded as $\tau_2 \to \infty$ for every $\a,\b \in \IQ$\\
   & $\Longleftrightarrow$ &  all $h_{j}(\t)$ in \eqref{jacobi-theta} are bounded as $\tau_2 \to \infty$ \\
   & $\Longleftrightarrow$ &  $c(n,r)$ have polynomial growth. \\   &&\\
  Weak Jacobi & $\Longleftrightarrow$ &  $c(n,r) =0$ for $n < 0$ \\
   & $\Longleftrightarrow$ &   $\v(\t,\a \t + \b)$ is bounded as $\tau_2 \to \infty$ for any fixed $z \in \C$ \\
   & $\Longleftrightarrow$ & all $h_{j}(\t) = O(q^{-j^2/4m})$ as $\tau_2 \to \infty$ \\
   & $\Longleftrightarrow$ &  $\sum_{r} P_{\nu,k} (nm,r^2) c(n,r)$ have polynomial growth. 
\end{tabular}

Finally, the relation between the Taylor expansion \eqref{defxin} of a Jacobi form and its theta expansion \eqref{jacobi-theta} 
is given by   
  \be\label{xinhrel} \xi_\nu(\t) \= (4m)^\nu \left({k+2\nu -2 \atop \nu }\right)^{-1}  \,
   \sum_{\ell \, ({\rm mod} \, 2m)} \bigl[ h_{\ell}(\tau), \, \vth^0_{m,\ell}(\tau) \bigr]_\nu \ , \ee
where $[\; , \; ]_{\nu} = [\; , \; ]_{\nu}^{(k-\frac 12,\frac 12)}$ denotes the Rankin-Cohen bracket 
(which, as we mentioned above, also works in half-integral weight), and $\vth^0_{m,\ell}(\tau) = \vth_{m,\ell}(\tau,0)$ (Thetanullwerte).
There is a similar formula in the odd case, but with $\vth^0_{m,\ell}(\t)$ replaced by 
\be\label{defth1}
\vth^1_{m,\ell}(\tau) \= \frac 1{2 \pi i} \, \frac{\p}{\p z} \vth_{m,\ell}(\tau,z) \Big|_{z=0} \= \sum_{r\equiv \ell \, (\mod \; 2m)} r\,q^{r^2/4m} \, . 
\ee

\subsection{Example: Jacobi forms of index $1$} \label{MJF1}

If $m=1$, \eqref{cnrprop} reduces to $c(n,r) = C(4n-r^2)$ where $C(\DD)$ is a
function of a single argument, because $4n - r^2$ determines the value of $r \, (\mod \, 2)$. 
So any $\v \in J_{k,1}^{\textrm{weak}} $ has an expansion of the form 
  \be\label{phiC}   \v(\tau, z) \= \sum_{n, r \in \Z} C(4n-r^2) \, q^n\,y^r \,.  \ee
It also follows that $k$ must be even, since in general, $C(\DD,-r) = (-1)^{k} C(\DD,r)$.

One has the isomorphisms $J_{k,1} \cong M_{k} \oplus S_{k+2}$ and $J_{k,1}^{\textrm{weak}} \cong M_{k} \oplus M_{k+2}$. If
$\v \in J_{k,1}^{\textrm{weak}}$ with an expansion as in  \eqref{phiC}, then 
  \be\label{phiAB} \v(\tau, 0) \=  \sum_{n=0}^\infty a(n)\,q^n\,, \qquad  
  \frac1{2(2\pi i)^2} \v''(\tau, 0) \= \sum_{n=0}^\infty b(n)\,q^n\,,  \ee
where 
  \be\label{defAB}  a(n) \= \sum_{r\in\Z} C(4n-r^2)\,, \qquad  b(n) \= \sum_{r>0} r^2\,C(4n-r^2)\,, \ee
and the isomorphisms are given (if $k>0$) by the map $\v \mapsto (\CA, \CB)$ with 
 \be\label{ABmod} \CA(\t) \= \sum a(n)\,q^n \in M_{k}\,, \qquad \CB(\t) \=  \sum\bigl(kb(n)-na(n)\bigr)\,q^{n} \in M_{k+2}\,.\ee
For $J_{k,1}$ one also has the isomorphism $J_{k,1} \cong M_{k-\frac 12}^{+}\left( \G_0(4)\right)$ given by 
\be\label{phitog} \v(\t,z) \leftrightarrow g(\t) \= \sum_{{\DD\ge 0} \atop \DD \equiv 0,\,3 \, {\rm mod}4} 
C(\DD)\,q^{\DD}\,. \ee

We have four particularly interesting examples $ \v_{k, 1}$
\be\label{phik} \v_{k, 1}(\tau, z) \= \sum_{n,\,r\in\Z} C_k(4n -r^2)\,q^n\,y^r\,,\quad k=-2,\,0,\,10,\,12\,, \ee
which have the properties (defining them uniquely up to multiplication by scalars)
\begin{myitemize}
  \item $\v_{10,1}$ and $\v_{12,1}$ are the two index $1$ Jacobi cusp forms of smallest weight;
  \item $\v_{-2,1}$ and $\v_{0,1}$ are the unique weak Jacobi forms of index $1$ and weight $ \le 0$;
  \item $\v_{-2,1}$ and $\v_{0,1}$ generate the ring of weak Jacobi forms of even weight freely over the ring of
      modular forms of level~$1$, so that
   \be\label{jacovermod}  J^{\textrm{weak}}_{k,m} \= \bigoplus_{j=0}^{m} M_{k+2j}\left(SL(2,\Z) \right) 
      \cdot\v_{-2,1}^j\, \v_{0,1}^{m-j}\qquad\text{($k$ even)}\,; \ee
 \item $\v_{-2,1} = \v_{10,1}/\Delta$, $\v_{0,1} = \v_{12,1}/ \Delta$,
  and the quotient $ \v_{0,1}/ \v_{-2,1} = \v_{12,1}/ \v_{10,1}$ is a multiple of the Weierstrass $\wp$ function.
\end{myitemize}
The Fourier coefficients of these functions can be computed from the above recursions, since the pairs $(\CA,\CB)$ for $\v=\v_{-2,1}$, 
$\v_{0,1}$, $\v_{10,1}$ and $\v_{12,1}$ are proportional to $(0,1)$, $(1,0)$, $(0,\Delta)$ and $(\Delta,0)$, respectively\footnote{For $k=0$, the second formula in~\eqref{ABmod} must be modified, and the function
$\sum b(n) q^{n}$ for $\v_{0,1}$ is in fact~$E_{2}(\t)$.}.  
The results for 
the first few Fourier coefficients are given in Table \ref{tablefcoeffs} below. In particular, the Fourier expansions of $\v_{-2,1}$ and $\v_{0,1}$ begin
  \bea\label{Fourirphi02}  \v_{-2,1} & \= & \frac{(y-1)^2}{y} \,-\, 2\,\frac{(y-1)^4}{y^2}\,q \+  \frac{(y-1)^4(y^2-8y+1)}{y^3}\,q^2 \+ \cdots \ , \\
  \v_{0,1} & \= & \frac{y^2 + 10 y +1}{y} \+ 2 \frac{(y-1)^2\,(5 y^2 - 22 y + 5)}{y^2}\,q \+ \cdots  \ .   \eea

\begin{table}[h]     \centering
   \begin{tabular}{|c|ccccccccc|ccccc}   \hline   
   $k$ &  $C_{k}(-1)$ & $C_{k}(0)$ &$C_{k}(3) $& $C_{k}(4)$& $C_{k}(7)$ & $C_{k}(8)$ &$C_{k}(11)$& $C_{k}(12)$& $C_{k}(15)$ \\
   \hline $-2$ & 1 & $-2$ &8 &$-12$&39&$-56$&152&$-208$&513\\   $0$ &  1  & 10 &$-64$&108&$-513$&808&$-2752$&4016&$-11775$\\
     $10$ & 0 & 0 &1&$-2$&$-16$&36&99&$-272$&$-240$\\   $12$ & 0 & 0 &1&10&$-88$&$-132$&1275&736&$-8040$\\  
   \hline \end{tabular} \label{tablefcoeffs} \captionof{table}{ \small {Some Fourier coefficients}}   \end{table}
\smallskip

The functions $\v_{k,1}$ ($k=10,\,0,\,-2$) can be expressed in terms of the Dedekind eta function 
\eqref{defeta} and the Jacobi theta functions $\vth_1,\,\vth_2,\,\vth_3,\,\vth_4$ by the formlas
  \be\label{phi10} \v_{10, 1}(\t, z) \= \eta^{18}(\t) \, {\vth_1^2(\t, z)}\,, \ee
   \be\label{phiminus2}  \v_{-2,1}(\t, z) \= \frac{\vth_1^2(\t, z)}{\eta^6(\t)}
          \= \frac{\v_{10, 1}(\t, z)} {\Delta(\tau)} \;. \ee
 \be\label{phizero}  \v_{0, 1} (\t, z) \= 4 \left( \frac{\vth_2(\t, z)^2}{\vth_2(\t)^2} +
    \frac{\vth_3(\t, z)^2}{\vth_3(\t)^2} +\frac{\vth_4(\t, z)^2}{\vth_4(\t)^2} \right) \, , \ee

Finally, we say a few words about Jacobi forms of odd weight.  Such a form cannot have index~1, as we saw.
In index~2, the isomorphisms $J_{k,2} \cong S_{k+1}$ and $J_{k,2}^{\textrm{weak}} \cong M_{k+1}$ show
that the first examples of holomorphic and weak Jacobi forms occur in weights 11 and $-1$, respectively,
and are related by $\v_{-1,2}=\v_{11,2}/\D$. The function $\v_{-1,2}$ is given explicitly by
 \be\label{phiminus1}\v_{-1,2}(\t, z) \= \frac{\vth_1(\t,2z)}{\eta^3(\t)} \;, \ee
with Fourier expansion beginning
 \be  \label{Fourirphi1}  \v_{-1,2} \= \frac{y^2-1}{y} \,-\, \frac{(y^2-1)^3}{y^3}\,q \,-\,3\,\frac{(y^2-1)^3}{y^3}\,q^2 \+ \cdots \,,  \ee
and its square is related to the index~1 Jacobi forms defined above by 
\be \label{weierstrass}
  432\,\v_{-1,2}^2 \= \v_{-2,1}\,\bigl(\v_{0,1}^3 \,-\, 3\,E_4\,\v_{-2,1}^2\,\v_{0,1} \+ 2\,E_6\,\v_{-2,1}^3\bigr)\;. 
\ee
(In fact, $\v_{-1,2}/\v_{-2,1}^2$ is a multiple of $\wp'(\t,z)$ and this equation, divided by $\v_{-2,1}^4$, is just the
usual equation expressing $\,{\wp'}^{\,2}\,$ as a cubic polynomial in~$\wp$.)  It is convenient to introduce abbreviations
  \be\label{ABC}  A\=\v_{-2,1}\,,  \qquad B\=\v_{0,1}\,,\qquad  C\=\v_{-1,2}\;. \ee
With these notations, the structure of the full bigraded ring of weak Jacobi forms is given by
  \be\label{WJF}  J_{*,*}^{\textrm{weak}} \= \C[E_4,E_6,A,B,C]\,\bigl/\,(432C^2-AB^3 + 3E_4A^3B - 2E_6A^4)\;.  \ee

\subsection{Hecke-like operators \label{hecke}}

In~\cite{Eichler:1985ja} Hecke operators $T_{\ell}$ acting on $J_{k,m}$ were introduced, but also various ``Hecke-like'' operators, 
again defined by the action of certain combinations of elements of $GL(2,\IQ) \rtimes \IQ^2$, which send Jacobi forms 
to Jacobi forms, but now possibly changing the index. We describe these operators here.  
This subsection is more technical than the preceding three, but will be essential later.

The first is the very simple operator $U_{s}$ ($s \ge 1$) which sends $\v(\tau,z)$ to $\v(\tau,s z)$, i.e.,
  \be\label{defUl} U_s\,: \quad \sum_{n,r} c(n,r)\,q^n\,y^r  \quad \mapsto  \quad \sum_{n,r} c(n,r)\,q^n\,y^{sr} \ ,  \ee
This operator maps $J_{k,m}$ to $J_{k,s^2m}$. 

The second operator, $V_{k,t}$ ($t \ge 1$), sends $J_{k,m}$ to $J_{k,t m}$. It is given in terms of its action on Fourier coefficients by 
\be\label{defVl}
c(\v|V_{k,t}\,; n,r) \= \sum_{d|(n,r,t)} d^{k-1} c\Big(\v \,; \frac{nt}{d^2},\frac{r}{d}\Big) \, , \qquad 
C(\v|V_{k,t}\,; \DD,r) \= \sum_{d|(\frac{r^{2}+\DD}{4mt},r,t)} d^{k-1} C\Big(\v \,; \frac{\DD}{d^2},\frac{r}{d}\Big) \, .
\ee
In~\S\ref{flyJac}, we will also need the modified operator $\CV_{k,t}^{(m)}:J_{k,m} \mapsto J_{k,tm}$ defined by 
\be\label{Vkt1}  \CV_{k,t}^{(m)}\= \sum_{s^2|t \atop (s,m)=1}\mu(s)\,s^{k-1}\,V_{k,t/s^2}\,U_s\, ,  \ee
so that, for instance, $\CV_{k,t}^{(m)}=V_{k,t}$ if~$t$ is square-free, but~$\CV_{k,4}^{(m)}=V_{k,4}-2^{k-1}U_{2}$ if $m$ is odd.  


Next, for each positive integer $m_{1}$ with $m_{1} \| m$ (this means that $m=m_1 m_2$ with $m_{1}$
and~$m_{2}$ coprime) we have an involution $W_{m_1}$ on $J_{k,m}$,  which we will 
call an \emph{Atkin-Lehner involution,}\footnote{This terminology was not used in~\cite{Eichler:1985ja}, 
where these involutions were introduced, but is justified by the results of~\cite{SkoruppaZagier}
where it is shown that they correspond
in a precise sense to the Atkin-Lehner involutions on the space of modular forms of  weight~$2k-2$ and 
level~$m$.} that is defined in terms of the theta expansion of Jacobi forms by 
  \be\label{defW} W_{m_1}\,: \quad  \sum_{\ell \; (\mod \, 2m)} h_{\ell} (\tau) \, \vth_{m,\ell} (\tau,z)  \quad \mapsto  \quad   
  \sum_{\ell \; (\mod \, 2m)} h_{\ell^*} (\tau) \, \vth_{m,\ell} (\tau,z) \ee
(or equivalently by $C(\DD,r) \mapsto  C(\DD,r^*)$), where the involution $\ell \mapsto \ell^*$ on $\Z/2m\Z$ is defined by 
  \be\label{definvol}  \ell^* \equiv - \ell \, ({\rm mod} \, 2m_1) , \qquad  \ell^* \equiv   + \ell \, ({\rm mod} \, 2m_2) \ . \ee
These operators commute and satisfy $W_{m/m_1} = (-1)^{k} W_{m_1}$, so that we get an eigenspace decomposition 
  \be\label{Jkmeigen} J_{k,m} \= \bigoplus_{\bve=(\ve_1,\dots ,\ve_t) \,\in\,\{ \pm 1\}^t  \atop \ve_1 \cdots \ve_t \= (-1)^k}  J_{k,m}^{\bve} \ ,  \ee
where $m=p_1^{r_1} \cdots p_{t}^{r_{t}}$ is the prime power decomposition of $m$ and $\ve_i$ is the eigenvalue of $W_{p_{i}^{r_{i}}}$. 
One can also introduce the ``$p$-adic Atkin-Lehner involution'' $W_{p^{\infty}}$, which acts on any~$J_{k,m}$
as~$W_{p^{\nu}}$ where~$p^{\nu}||m$. (In particular, it acts as the identity if~$p \nmid m$.) This operator 
simply changes the second index~$r$ in~$C(\DD,r)$ to~$r^{*}$, where~$r$ is equal to the negative 
of~$r^{*}$ $p$-adically and to~$r^{*}$ $p'$-adically for any~$p' \neq p$. 
In particular,~$W_{p^{\infty}}$ commutes with all~$U_{s}$, $V_{k,t}$, and~$\CV_{k,t}^{m}$.

We now describe a further Hecke-like operator, which was not given in \cite{Eichler:1985ja} but is mentioned 
briefly in \cite{SkoruppaZagier}, p.~139, and  
which we will also need in the sequel. If $t$ is a positive integer whose square divides $m$, then we 
define $u_{t}: \; J_{k,m} \; \longrightarrow J_{k,m/t^{2}}$ by 
\be\label{Ucech}
u_t\,: \quad \sum_{n,r} c(n,r)\,q^n\,y^r \quad \mapsto \quad 
    \sum_{n,r} \left( \sum_{a \mypmod{t}} c \bigl( n+ ar+\frac{m}{t^{2}} a^{2}, \, tr + \frac{2ma}{t}  \bigr)\right) q^n\,y^r 
\ee
or equivalently by 
\be\label{Ucech2}
C_{{\v|u_{t}}}(\DD,\,\ell \mypmod{2m/t^{2}}) 
\= \sum_{r \mypmod{2m/t} \atop r \equiv \ell \mypmod{2m/t^{2}}} C_{\v}(\DD t^{2}, \, rt  \mypmod{2m} ) \; .
\ee
It is easily checked that this map does indeed send Jacobi forms to Jacobi forms of the index stated, 
and that 
for any integers $m,\,t >0$, the composite map $J_{k,m/t^{2}} \stackrel{U_{t}}{\longrightarrow} 
J_{k,m} \stackrel{u_{t}}{\longrightarrow} J_{k,m/t^{2}}$ is simply multiplication by~$t$ 
(because, if $\v = \v_{1} \! \mid \! {U_{t}}$ with $\v_{1}$ of index $m/t^{2}$, then each term 
on the right of \eqref{Ucech2} equals $C_{\v_{1}}(\DD,\ell)$). It follows that we have 
\be\label{primdecomp}
J_{k,m} \= \bigoplus_{t^{2}\mid m} J^{\rm prim}_{k,\frac{m}{t^{2}}} \mid {U_{t}} \, ,
\ee
where
\be\label{primdecomp2}
J^{\rm prim}_{k,m}  \=  \bigcap_{t^{2}\mid m \atop t >1} \ker \big(J_{k,m} \stackrel{u_{t}}{\longrightarrow} J_{k,m/t^{2}} \big) 
\= \bigcap_{p^{2}\mid m \atop \text{$p$ prime}} \ker \big(J_{k,m} \stackrel{u_{p}}{\longrightarrow} J_{k,m/p^{2}} \big) 
\ee
(the equivalence of the two definitions follows because $u_{t} \, u_{t'} = u_{tt'}$ for any $t,\,t'$),
with the projection map~$\pi^{\rm prim}: J_{k,m} \mapsto J_{k,m}^{\rm prim}$ onto the first factor in~\eqref{primdecomp} 
given by
\be \label{primproj}
\pi^{\rm prim}(\v) \= \sum_{t^{2}|m} \frac{\mu(t)}{t} \, \v \! \mid \!  u_{t}\! \mid \! U_{t} \, . 
\ee


%

We will sometimes find it useful to work, not with the space of Jacobi forms, but 
with the larger space of holomorphic functions on $\mathbb{H} \times\C$
that obey the elliptic transformation law \eqref{elliptic} of Jacobi forms, but are not required to 
obey the modular transformation property \eqref{modtransform}. 
We shall call such functions \emph{elliptic forms}.  (The word \emph{Jacobi-like forms} has been
used in the literature to denote functions which, conversely, satisfy \eqref{modtransform} 
but not \eqref{elliptic}.) Any elliptic form has a Fourier expansion as in \eqref{fourierjacobi}
with the periodicity condition \eqref{cnrprop}. We will denote by $\elliptic_{m}$, $\elliptic^{0}_{m}$, $\wt \elliptic_{m}$ 
the spaces of all elliptic forms of index $m$ satisfying the growth conditions 
\eqref{holjacobi}, \eqref{cuspjacobi}, and \eqref{weakjacobi}, respectively, and we will call them
\emph{holomorphic, cuspidal}, and \emph{weak} respectively. 
We also denote by 
$\elliptic_{\pm,m}$, $\elliptic^{0}_{\pm,m}$, $\wt \elliptic_{\pm,m}$ the $(\pm 1)-$eigenspaces 
of theta functions under the involution $\v(\t,z) \mapsto \v(\t,-z)$, i.e., those functions 
whose Fourier coefficients $c(n,r)$ satisfy the parity condition $c(n,-r) = \pm c(n,r)$. 
Because any Jacobi form of weight $k$ satisfies this parity condition with $\pm 1 = (-1)^{k}$
(by \eqref{modtransform} applied to $-{\bf 1}_{2}$), we have 
\be \label{PrimParts}
J_{k,m} \subset \elliptic_{\pm,m} \, , \quad J^{0}_{k,m} \subset \elliptic^{0}_{\pm,m} \, , \quad  
\wt J_{k,m} \subset \wt \elliptic_{\pm,m} \,  \qquad \big(\, (-1)^{k} = \pm 1 \, \big) \, . 
\ee
It is now clear from the definitions above that all of the Hecke-like operators 
$V_{t}$, $U_{t}$, $W_{m_1}$ for $(m_1,m/m_1)=1$ and $u_t$ for $t^2|m$ defined 
in this section extend from the space of Jacobi forms to the space of elliptic forms, 
and that the decompositions \eqref{Jkmeigen} and \eqref{primdecomp} remain true if we replace all the 
$J_{k}$s by $\elliptic_{+}$s.
 
Finally, we should warn the reader explicitly that, although the operators $U_{t}$ and $V_{t}$ also act on
the spaces $\wt J_{k,*}$ and $\wt \elliptic_{\pm,*}$ of weak Jacobi or weak elliptic forms, the operators 
$W_{m_1}$ and $u_{t}$ do not: the condition \eqref{weakjacobi} is equivalent to saying that $C(\DD,\ell)$ 
for $|\ell| \le m$ vanishes whenever $\DD < -\ell^{2}$, and this condition is preserved neither under 
the permutation $\; \ell \mypmod{2m} \mapsto \ell^{*}  \mypmod{2m}$ defining $W_{m_{1}}$,
nor by the summation conditions in  \eqref{Ucech2}.  As a concrete example, the 7-dimensional space $\wt J_{2,6}$
contains a 6-dimensional subspace of forms invariant under the involution $W_2$, but no anti-invariant form,
and its image under the projection operator $\pi_2^-=\frac12(1-W_2)$ to anti-invariant forms is the 1-dimensional
space spanned by 
\be \label{phi26min}  \v_{2,6}^-(\t,z) \= j'(\t)\,\frac{\vth_{1}(\t,4z)}{\vth_{1}(\t,2z)} \= 
 \sum_{n,\,r\in\IZ } \chi_{12}(r)\,C_{2,6}(24n-r^2)\;q^n\,y^r \,,
\ee
with 
$$
\begin{tabular}{|c|ccccccc|}  \hline     
$\DD$ & $\;\,\le-49$ & $-25$& $-1$& $23$ & 47 & 71 & $\cdots$ \\
       \hline  $C_{2,6}(\DD)$ & \;0 & $-1$ & $-1$ & 196882 & 43184401 & 2636281193 & $\cdots$ \\
\hline  \end{tabular}  \quad ,
$$
which is a weakly holomorphic but not a weak Jacobi form.  (Specifically, for any $\v\in\wt J_{2,6}$ we have
$\pi^-_2(\v)=6\,c_\v(0,5)\,\v_{2,6}^-\,$.)



%

\subsection{Quantum black holes and  Jacobi forms}\label{BHandJF}

Jacobi forms usually arise in string theory as elliptic genera of two-dimensional superconformal field theories (SCFT) 
with $(2, 2)$ or more worldsheet supersymmetry\footnote{An SCFT with $(r, s) $ supersymmetries has $r$  
left-moving and $s$ right-moving supersymmetries.}.  We denote the superconformal field theory by $\sigma(\CM)$ when it corresponds
to a sigma model with a target manifold $\mathcal{M}$.  Let $H$ be the Hamiltonian in the Ramond sector, and $J$ be the left-moving
$U(1)$ $R$-charge. The elliptic genus $\chi(\tau, z; \mathcal{M})$ is then defined as \cite{Witten:1986bf,Alvarez:1987wg, Ochanine:1987}
a trace over the Hilbert space $\mathcal{H}_R$ in the Ramond sector
  \be\label{ell}  \chi(\tau, z; \mathcal{M}) \= \Tr_{\mathcal{H}_R} \bigl( (-1)^F q^H y^J  \bigr) \,, \ee
where $F$ is the fermion number.

An elliptic genus so defined satisfies the modular transformation property
(\ref{modtransform}) as a consequence of modular invariance of the path integral.
Similarly, it satisfies the elliptic transformation property (\ref{elliptic}) as
a consequence of spectral flow.  Furthermore, in a unitary SCFT, the positivity
of the Hamiltonian implies that the elliptic genus is a weak Jacobi form. The
decomposition (\ref{jacobi-theta}) follows from bosonizing the $U(1)$ current in
the standard way so that the contribution to the trace from the momentum modes of
the boson can be separated into the theta function (\ref{thetadef}). See, for
example, \cite{Kawai:1993jk,Moore:2004fg} for a discussion. This notion of the
elliptic genus can be generalized to a $(0, 2)$ theory using a left-moving $U(1)$
charge $J$ which may not be an R- charge. In this case spectral flow is imposed
as an additional constraint and follows from gauge invariance under large gauge
transformations \cite{deBoer:2006vg,Gaiotto:2006wm,Kraus:2006nb,Denef:2007vg}.

A particularly useful example in the present context is $ \sigma(K3)$, which is a $(4,4)$ SCFT whose target space
is a $K3$ surface. The elliptic genus is a topological invariant and is independent of the moduli of the $K3$. Hence, it can
be computed at some convenient point in the $K3$ moduli space, for example, at the orbifold point where the $K3$ is the Kummer
surface. At this point, the $\sigma(K3)$ SCFT can be regarded as a $\mathbb{Z}_2$ orbifold of the $\sigma(T^4)$ SCFT, which 
is an SCFT with a torus $T^4$ as the target space. A simple computation using standard techniques of orbifold conformal field theory
\cite{Dixon:1985jw, Ginsparg:1988ui} yields 
  \be\label{k3}  {\chi}(\t,z; K3) \= 2\,\v_{0,1}(\t,z) \= 2 \sum \, C_0(4n- l^2)\, q^n\,y^l \,.  \ee
Note that for $z=0$, the trace (\ref{ell}) reduces to the Witten index of the SCFT and correspondingly the elliptic genus reduces to
the Euler character of the target space manifold. In our case, one can readily verify from (\ref{k3}) and (\ref{phizero}) that
$\chi(\t, 0; K3)$ equals $24$, which is the Euler character of $K3$.

A well-known physical application of Jacobi forms is in the context of the five-dimensional Strominger-Vafa black 
hole\cite{Strominger:1996sh}, which is a bound state of $Q_1$ D1-branes, $Q_5$ D5-branes, $n$ units of momentum and $l $
units of five-dimensional angular momentum \cite{Breckenridge:1996is}. The degeneracies $d_m(n,l)$ of such black holes depend 
only on $m =Q_1 Q_5$. They are given by the Fourier coefficients $c(n, l)$ of the elliptic genus ${\chi}(\t, z;
\textrm{Sym}^{m+1} (K3))$ of symmetric product of $(m+1)$ copies of $K3$-surface.

Let us denote the generating function for the elliptic genera of symmetric products of K3 by
\be\label{defZ5d}  \wh Z(\s,\t,z) := \sum_{m=-1}^\infty \chi(\t,z; \textrm{Sym}^{m+1}(K3))\,p^m \ee
where $\chi_m(\t,z)$ is the elliptic genus of $ \textrm{Sym}^m(K3)$. A standard orbifold computation \cite{Dijkgraaf:1996xw} gives
  \be\label{ellprzero}  \wh Z(\s,\t,z) = \frac 1{p}\prod_{r>0,\; s\ge0,\; t} \frac 1 {(1- q^s y^t p^r)^{2C_0(rs, t)}} \ee
in terms of the Fourier coefficients $2 C_o$ of the elliptic genus of a single copy of $K3$.

For $z=0$, it can be checked that,  as expected, the generating function (\ref{ellprzero}) for elliptic genera
reduces to the generating function (\ref{Zeuler2}) for Euler characters
  \be\label{Ztoeta}  \wh Z(\s,\t, 0) = \hat Z(\sigma) = \frac 1{\Delta(\s)} \, .  \ee

\section{Siegel modular forms \label{Siegel}}

\subsection{Definitions and examples of Siegel modular forms}

Let $Sp(2, \mathbb{Z})$ be the group of $(4\times 4)$ matrices $g$ with integer entries satisfying $g J g^t =J$ where
  \be    J \equiv \left( \begin{array}{cc} 0 & -I_2 \\ I_2 & 0 \\ \end{array} \right)  \ee
is the symplectic form. We can write the element $g$ in block form as
  \be\label{sp}  \left(  \begin{array}{cc}  A & B \\ C & D \\ \end{array} \right),\ee
where $A, B, C, D$ are all $(2\times 2)$ matrices with integer entries.
Then the condition $g J g^t =J$ implies
  \be\label{cond}  AB^t=BA^t, \qquad CD^t=DC^t, \qquad AD^t-BC^t= \mathbf 1\, ,  \ee
Let $\mathbb{H}_2$ be the (genus two) Siegel upper half plane, defined as the set
of $(2\times 2)$ symmetric matrix $\Omega$ with complex entries
  \be\label{period}   \Omega = \left( \begin{array}{cc} \t & z \\  z & \sigma \\ \end{array} \right) \ee
satisfying
  \be\label{cond1}   \textrm{Im} (\t) > 0, \quad \textrm{Im} (\sigma) > 0, \quad \det(\Im(\O)) > 0 \, .  \ee
An element $g \in Sp(2, \mathbb{Z})$ of the form (\ref{sp}) has a natural action on $\mathbb{H}_2$ under which it is stable:
  \be\label{trans} \Omega \to (A \Omega + B )(C\Omega + D ) ^{-1}.  \ee
The matrix $\Omega$ can be thought of as the period matrix of a genus two 
Riemann surface\footnote{See \cite{Gaiotto:2005hc,Dabholkar:2006xa,Banerjee:2008yu} for a
discussion of the connection with genus-two Riemann surfaces.} on which there is a natural symplectic action of $Sp(2, \mathbb{Z})$.

A Siegel form $F(\Omega)$ of weight $k$ is a holomorphic function $ \mathbb{H}_2\rightarrow \mathbb{C}$ satisfying
  \be\label{trans2}   F\bigl((A\Omega + B)(C\Omega + D) ^{-1}\bigr) \=  \det(C\Omega + D)^k \, F (\Omega)\,.  \ee
A Siegel modular form can be written in terms of its Fourier series
  \be\label{siegelfourier}   F(\Omega) \= \sum_{n,\,r,\,m\inn\Z \atop r^2\le4mn} a(n,r, m)  \, q^n\,y^r\,p^m \,.  \ee
If one writes this as 
  \be\label{siegeljacobi}    F(\Omega) = \sum_{m=0}^{\infty}  \v^{\rm F}_{m} (\t, z)  \, p^m   \ee
with 
  \be\label{siegeljacfourier}   \v^F_m(\t, z) \= \sum_{n,\,r} a(n,r, m)  \, q^n\,y^r  \, ,  \ee
then each $\v^{\rm F}_{m} (m \geq 0)$ is a Jacobi form of weight $k$ and index $m$.

An important special class of Siegel forms were studied by Maass which he called the
{\it Spezialschar}.  They have the property that $a(n,r,m)$ depends only on the
discriminant $4 mn -r^2$ if $ (n,r,m)$ are coprime, and more generally
  \be\label{Cnrm}  a(n, r, m) \= \sum_{d \mid (n,r,m),\;\,d>0} d^{k-1} C\Bigl(\frac{4 mn - r^2}{d^2}\Bigr)  \ee 
for some coefficients $C(N)$. Specializing to $m=1$, we can see that these numbers are simply the coefficients associated 
via \eqref{phiC} to the Jacobi form $\v = \v^{\rm F}_1 \in J_{k,1}$, and that \eqref{Cnrm} says precisely that the 
other Fourier-Jacobi coefficients of $F$ are given by $\v_{m}^{\rm F} = \v_1^{\rm F}|V_{m}$ with $V_{m}$ as in \eqref{defVl}. 
Conversely, if $\v$ is any Jacobi form of weight $k$ and index $1$ with Fourier expansion \eqref{phiC}, then the function 
$F(\Omega)$ defined by \eqref{siegelfourier} and \eqref{Cnrm} or by $F(\O) = \sum_{m=0}^{\infty} \bigl( \v|V_{m} \bigr)(\t,z)\,p^{m}$
is a Siegel modular form of weight $k$ with $\v_1^{\rm F} = \v$.   The resulting map from $J_{k,1}$ to the Spezialschar is called the 
Saito-Kurokawa lift or \textit{additive lift} since it naturally gives the sum representation of a Siegel form using the Fourier 
coefficients of a Jacobi form as the input. (More information about the additive lift can be found in \cite{Eichler:1985ja}.)

The example of interest to us is the Igusa cusp form $\Phi_{10}$ (the unique cusp form of weight $10$) which is the Saito-Kurokawa
lift of the Jacobi form $ \v_{10, 1}$ introduced earlier, so that
\be\label{igusa-additive}  
  \Phi_{10}(\Omega) \= \sum_{m=1}^{\infty} \bigl(\v_{10,1}|V_{m}\bigr) (\t,z) \, p^{m}   
  \= \sum_{n,\,r,\,m} a_{10}(n,r,m) \,q^n\,y^r\,p^m \,, 
 \ee
where $a_{10}$ is defined by (\ref{Cnrm}) with $k=10$ in terms of the coefficients $C_{10}(d)$ given in Table \ref{tablefcoeffs}.

A Siegel modular form sometimes also admits a product representation, and can be obtained as
Borcherds lift or {\it multiplicative lift} of a weak Jacobi form of weight zero
and index one. This procedure is in a sense an exponentiation of the additive
lift and naturally results in the product representation of the Siegel form using
the Fourier coefficients of a Jacobi form as the input. Several examples of
Siegel forms that admit product representation are known but at present there is
no general theory to determine under what conditions a given Siegel form admits a product representation. 

For the Igusa cusp form $\Phi_{10}$, a product representation does exist. It was obtained by Gritsenko and Nikulin 
\cite{GritNik1, GritNik2}
as a multiplicative lift of the elliptic genus $ \chi(\t,z; K3) =2\v_{0,1}(\t,z) $ and is given by
  \be \label{final2}  \Phi_{10}(\Omega) \= q y p \prod_{(r, s, t) >0}\bigl( 1 - q^s y^t p^r\bigr)^{2C_0(4rs-t^2)}, \ee
in terms of $C_0$ given by (\ref{phizero}, \ref{phik}).  Here the notation $(r,s,t)>0$ means that $r,\,s,\,t\,\in\Z$ with
either $r>0$ or $r=0$, $s >0$, or $r=s=0$, $t<0$.  In terms of the $V_{m}$, this can be rewritten in the 
form $\; \Phi_{10}= p \, \v_{10,1} \, \exp\bigl(-2\sum_{m=1}^{\infty} (\v_{0,1}|V_{m}) \, p^{m}\bigr)$.

\subsection{Quantum black holes and Siegel modular forms}

Siegel forms occur naturally in the context of counting of quarter-BPS dyons.  The partition function for these dyons depends on three
(complexified) chemical potentials $(\sigma, \t, z)$, conjugate to the three $T$-duality invariant integers $ (m, n, \ell)$ respectively 
and is given by
\be\label{igusa} 
 \textrm{Z}(\Omega) \= \frac 1{\Phi_{10}(\Omega)} \ .  
\ee

The product representation of the Igusa form is particularly useful for the
physics application because it is closely related to the generating function for
the elliptic genera of symmetric products of $K3$ introduced earlier.  This is a
consequence of the fact that the multiplicative lift of the Igusa form is
obtained starting with the elliptic genus of $K3$ as the input. Comparing the
product representation for the Igusa form (\ref{final2}) with \eqref{ellprzero}, we get the relation:
  \be\label{4d5drel}  \textrm{Z}(\t,z,\s) \= \frac 1{\Phi_{10}(\t,z,\s)} \= \frac{\wh Z (\t,z,\s)} {\v_{10,1}(\t,z)} \ .  \ee

This relation to the elliptic genera of symmetric products of $K3$ has a deeper
physical significance based on what is known as the 4d-5d lift
\cite{Gaiotto:2005gf}. The main idea is to use the fact that the geometry of the
Kaluza-Klein monopole in the charge configuration (\ref{final charges}) reduces
to five-dimensional flat Minkowski spacetime in the limit when the radius of the
circle $\tS^1$ goes to infinity. In this limit, the charge $l$ corresponding to
the momentum around this circle gets identified with the angular momentum $l$ in
five dimensions. Our charge configuration (\ref{final charges}) then reduces
essentially to the Strominger-Vafa black hole \cite{Strominger:1996sh} with
angular momentum \cite{Breckenridge:1996is} discussed in the previous subsection.
Assuming that the dyon partition function does not depend on the moduli, we thus
essentially relate $Z(\Omega)$ to $\hat Z(\O)$. The additional factor in
(\ref{4d5drel}) involving $\Phi_{10} (\s,\t,z)$ comes from bound states of
momentum $n$ with the Kaluza-Klein monopole and from the center of mass motion of
the Strominger-Vafa black hole in the Kaluza-Klein geometry \cite{Gaiotto:2005hc,David:2006yn}.

The Igusa cusp form has double zeros at $z=0$ and its $Sp(2,\IZ)$ images. The partition
function is therefore a \textit{meromorphic} Siegel form (\ref{trans2}) of weight
$-10$ with double poles at these divisors. This fact is responsible for much of
the interesting physics of wall-crossings in this context as we explain in the next section.

Using  \eqref{igusa-additive} or \eqref{final2}, one can compute the coefficients $\psi_{m}$ 
in the Fourier-Jacobi expansion~\eqref{reciproigusa}. The first few, multiplied by $\D$ and a suitable integer, are given by 
\bea\label{firstpsim}
\D\,\psi_{-1} & \= & A^{-1}\;, \nonumber \\ 
\D\,\psi_{0\;}  & \= & 2\,A^{-1}B\;, \nonumber \\
4\,\D \,\psi_{1\;} &\=& 9\,A^{-1}B^2 \+ 3 E_4 A\;, \\
27 \, \D \, \psi_{2\;}  & \= & 50 \, A^{-1} B^3 \+  48 E_4 AB \+ 10 E_6 A^2\;,  \nonumber  \\
384 \, \D \, \psi_{3\;}  & \= & 475 \,  A^{-1} B^4 \+ 886 E_4 AB^2 \+ 360 E_6 A^2 B 
\+ 199 E_4^2 A^3\;, \nonumber \\
72 \, \D \, \psi_{4\;}  & \= & 51 \,  A^{-1} B^5 \+ 155 E_4 A B^3 \+  93 E_6 A^2 B^2 
\+ 102 E_4^2 A^3 B \+ 31 E_4 E_6 A^4\;. \nonumber 
\eea
(Here $E_4$, $E_6$, $\D$ are as in \S\ref{Modular} and $A=\v_{-2,1}$, $B=\v_{0,1}$ as in \S\ref{Jacobi}.) 
The double zero of $\Phi_{10}$ at $z=0$ is reflected by the denominator $A$ in the $A^{-1} B^{m+1}$ terms in these formulas. 
Note that, by \eqref{4d5drel}, the right-hand side of \eqref{firstpsim} multiplied by $A$ and divided by 
a suitable denominator are equal to the elliptic genus of symmetric products of a $K3$ surface. For example,
\be
\chi({\rm Sym}^{4} K3; \t,z) \= \frac{1}{384} \big(475 \,  B^4 \+ 886 E_4 A^{2} B^2 \+ 360 E_6 A^3 B 
\+ 199 E_4^2 A^4\; \big) \, .
\ee

\section{Walls and contours \label{Contours}}

Given the partition function \eqref{igusa}, one can extract the black hole
degeneracies from the Fourier coefficients. The three quadratic $T$-duality
invariants of a given dyonic state can be organized as a $2 \times 2$ symmetric matrix
\be\label{matrix_charge_vector}  
\L \=\bem N\!\cdot\! N & N\!\cdot\! M \\  M\!\cdot\!N & M\!\cdot\! M \eem \= \bem 2n& \ell \\ \ell & 2 m \eem \; , 
\ee
where the dot products are defined using the $O(22, 6; \mathbb{Z})$ invariant metric $L$. The matrix $ \Omega$ in (\ref{igusa}) 
and (\ref{period}) can be viewed as the matrix of complex chemical potentials conjugate to the charge matrix $\L$.
The charge matrix $\L$ is manifestly $T$-duality invariant. Under an $S$-duality transformation (\ref{Sgroup}),  it transforms as
\be\label{lambdatransform} 
 \L \;\to\; \gamma \L \gamma^t  
\ee
There is a natural embedding of this physical $S$-duality group $SL(2,\Z)$ into $Sp(2, \mathbb{Z})$:
\be\label{sembed} 
 \left( \begin{array}{cc}  A & B \\ C & D \\ \end{array} \right)
    \= \left( \begin{array}{cc} (\g^t)^{-1} & \textbf{0} \\ \textbf{0} & \g   \\ \end{array} \right)  
    \= \left( \begin{array}{cccc} d & -c & 0 & 0 \\ -b & a & 0 & 0 \\ 0 & 0 & a & b \\ 0 & 0 & c & d \\ 
    \end{array} \right) \; \in\, Sp(2, \mathbb{Z})\,.  
\ee
The embedding is chosen so that $\Omega \to (\g^T)^{-1} \Omega \g^{-1}$ and $\Tr(\Omega \cdot \Lambda)$ in the Fourier integral 
is invariant. This choice of the embedding ensures that the physical degeneracies extracted from the Fourier
integral are $S$-duality invariant if we appropriately transform the moduli at the same time as we explain below.

To specify the contours, it is useful to define the
following moduli-dependent quantities. One can define the matrix of right-moving $T$-duality invariants
  \be\label{rightmatrix_charge_vector} \L_{R}  =\bem Q_R\cdot Q_R & Q_R\cdot P_R \\ P_R\cdot Q_R & P_R\cdot P_R \eem \;, \ee
which depends both on the integral charge vectors $N, M$ as well as the $T$-moduli $\mu$.  One can then define 
two matrices naturally associated to the $S$-moduli $S = S_1+iS_2$ and the $T$-moduli $\mu$ respectively by
  \be\label{ST_moduli} {\cal S} = \frac 1{S_2} \bem |S|^2 & S_1\\\ S_1& 1 \eem\quad, \quad {\cal T} 
  \= \frac{ \L_{R}}{{|\det (\L_{R})|^\frac 12}}\;.  \ee
Both matrices are normalized to have unit determinant. In terms of them, we
can construct the moduli-dependent `central charge matrix'
  \be\label{def_Z_vec} {\mathcal Z} = {{|\det(\L_{R})|^\frac 14}}\, \big({\cal S}+{\cal T} \big)\,, \ee
whose determinant equals the BPS mass
  \be M_{Q,P} \= |\det{\mathcal Z}|\;.  \ee
We define 
  \be \wt \Omega \equiv \left( \begin{array}{cc} \sigma &  -z \\  -z & \tau \\ \end{array} \right)\;. \ee
This is related to $\Omega$ by an $SL(2,\Z)$ transformation
  \be  \tilde\Omega \=  \wh S \Omega  \wh S^{-1}  \quad \text{where} \quad 
   \wh S \=  \left( \begin{array}{cc} 0 &  1 \\  -1 & 0 \\ \end{array} \right)\,, \ee
so that, under a general $S$-duality transformation $\gamma$, we have the transformation $\tilde\Omega\rightarrow \g\wt\Omega \g^T$
as  $\Omega \to (\g^T)^{-1} \Omega \g^{-1}$.

With these definitions,  $\Lambda, \Lambda_{R}, {\cal Z}$ and $\tilde \Omega$ all transform as $X \to \g X \gamma^T$ under an 
$S$-duality transformation (\ref{Sgroup}) and are invariant under $T$-duality transformations. 
The moduli-dependent Fourier contour can then be  specified in a duality-invariant fashion by\cite{Cheng:2007ch}
  \be\label{contour}  {\cal C} = \{\im  (\wt \O) = \ve^{-1} {\cal Z};\quad 0 \leq \Re(\t), \, \Re(\s), \, \Re(z) < 1 \},   \ee
where $\ve \rightarrow 0^+$. For a given set of charges, the contour depends on the moduli $S, \mu$ through the definition of the central
charge vector (\ref{def_Z_vec}). The degeneracies $d(n, \ell,m)\lvert_{S, \m}$ of states with the $T$-duality invariants $(n,\ell,m)$
at a given point $(S,\mu)$ in the moduli space are then given by\footnote{The physical degeneracies have an additional multiplicative 
factor of  $(-1)^{\ell +1}$ which we omit here for simplicity of notation in later chapters.} 
  \be\label{inverse}  d(n, \ell, m)\lvert_{S, \m} = \int_{\CC} \, e^{-i\pi \Tr (\Omega \cdot \Lambda)}\, \, {\textrm{Z}(\Omega)} \, d^3 \Omega \, . \ee

This contour prescription thus specifies how to extract the degeneracies from the partition 
function for a given set of charges and in any given region of the moduli space. In particular, 
it also completely summarizes all wall-crossings as one moves around in the moduli space for a fixed set of charges. 
Even though the indexed partition function has the same functional form throughout the moduli space, the spectrum is moduli
dependent because of the moduli dependence of the contours of Fourier integration and the pole structure of the partition function. 
Since the degeneracies depend on the moduli \textit{only} through the dependence of the contour ${\cal C}$, moving around in the 
moduli space corresponds to deforming the Fourier contour. This does not change the degeneracy except when
one encounters a pole of the partition function. Crossing a pole corresponds to crossing a wall in the moduli 
space. The moduli space is thus divided up into domains separated by ``walls of marginal stability.'' In each domain the degeneracy
is constant but it jumps upon crossing a wall as one goes from one domain to the other. The jump in the degeneracy has a nice mathematical 
characterization. It is simply given by the residue at the pole that is crossed while deforming the Fourier contour in going from one domain to the other.

We now turn to the degeneracies of single-centered black holes. Given the
$T$-duality invariants $\Lambda$, a single centered black hole solution is known to
exist in all regions of the moduli space as long as $\det (\Lambda)$ is large and
positive. The moduli fields can take any values $(\l_ \infty, \mu_ \infty)$ at
asymptotic infinity far away from the black hole but the vary in the black hole geometry.  Because of the attractor phenomenon
\cite{Ferrara:1995ih,Strominger:1996kf}, the moduli adjust themselves so that near the horizon of the black hole of charge
$\L$ they get attracted to the values $(\l_*(\L), \mu_*(\L))$ which are determined by the requirement that the
central charge $ \mathcal Z_*(\L)$ evaluated using these moduli becomes proportional to $\L$. These attractor values are
independent of the asymptotic values and depend only on the charge of black hole. We call these moduli the
attractor moduli. This enables us to define the \textit{attractor contour} for a
given charge $\L$ by fixing the asymptotic moduli to the attractor values corresponding to this charge. In this case
  \be\label{att}  \CZ(\l_\infty, \mu_\infty) \=  \CZ(\l_*(\L),\mu_*(\L)) \sim \L \ee
and we have the attractor contour
  \be\label{contourattract} {\cal C}_* \= \{\im (\wt \O) \= \ve^{-1} {\Lambda};\quad 0 \leq \Re(\t), \, \Re(\s), \, \Re(z) < 1 \} \ee
which depends only on the integral charges and not on the moduli. The significance of the attractor
moduli in our context stems from the fact if the asymptotic moduli are tuned to these values for given
$(n,\ell,m)$, then only single-centered black hole solution exists.  The degeneracies $d^*(n,\ell,m)$ obtained using the attractor contour
  \be\label{inverse3}   d^*(n, \ell,m) = \int_{ \CC_*}  \, e^{-i\pi \Tr (\Omega \cdot \Lambda)}\, \, {\textrm{Z}(\Omega)} \, d^3 \Omega \ee
are therefore expected to be the degeneracies of the immortal single-centered black holes.

\section{Mock modular forms \label{Mock}}

Mock modular forms are a relatively new class of modular objects, although individual examples had been
known for some time. They were first isolated explicitly by S. Zwegers in his thesis \cite{Zwegers:2002}
as the explanation of the ``mock theta functions'' introduced by Ramanujan in his famous last letter to
Hardy.  An expository account of this work can be found in \cite{Zagier:2007}. 

Ramanujan's mock theta functions are examples of what we will call ``pure'' mock modular forms. 
By this we mean a holomorphic function in the upper half plane which transforms under
modular transformations almost, but not quite, as a modular form. The
non-modularity is of a very special nature and is governed by another holomorphic
function called the \textit{shadow} which is itself an ordinary modular form. We will describe this 
more precisely in \S\ref{basicdefs}  and give a number of examples.  In \S{\ref{MockJacobi}}, 
we introduce a notion of mock Jacobi forms (essentially, holomorphic functions of  $\t$ and $z$ 
with theta expansions like that of usual Jacobi forms, but in which the coefficients $h_{\ell}(\t)$ 
are mock modular forms) and show how all the examples given in the first subsection 
occur naturally as pieces of mock Jacobi forms. Finally in \S\ref{MMMF} we define a more 
general notion of mock modular forms in which the shadow is replaced by a sum of products 
of holomorphic and anti-holomorphic modular forms. This has the advantage that there are now 
many examples of strongly holomorphic mock modular forms, whereas pure mock modular forms 
almost always have a negative power of $q$ at at least one cusp, and also that almost all of the examples 
occurring naturally in number theory, algebraic geometry, or physics are of this more general 
``mixed'' type, with the pure mock modular forms arising as quotients of mixed mock modular forms 
by eta products.

\subsection{Pure mock modular forms \label{basicdefs}}

We define a (weakly holomorphic) {\it pure mock modular form} of weight $k\in \frac12 \Z$ as the first member 
of a pair $(h,g)$, where
\begin{enumerate}
\item $h$ is a holomorphic function in $\IH$ with at most exponential growth at all cusps, 
\item the function $g(\t)$, called the \textit{shadow} of $h$, is a holomorphic\footnote{One can also consider the case where the shadow is allowed
to be a weakly holomorphic modular form, but we do not do this since none of our examples will be of this type.} modular form of weight $2-k\,$, 
and 
\item the sum $\wh h := h\+g^*$, called the {\it completion} of $h$, transforms like a holomorphic modular form of weight $k$,
   {\it i.e.} $\wh h (\tau)/\theta (\tau)^{2k}$ is invariant under $\t \to \g \t$  for all $ \t \in \IH $ and for all
   $\g$ in some congruence subgroup of $SL(2,\Z)$. 
\end{enumerate} 
Here $g^*(\t)$, called the {\it non-holomorphic Eichler integral}, is a solution of the differential equation
\be\label{starinv} (4\pi\t_2)^k\,\frac{\pa g^*(\t)}{\pa \bar{\tau}} \= -2\pi i\;\bar{g(\tau)} \, . \ee
If $g$ has the Fourier expansion $g(\t)=\sum_{n\ge0}\;b_n\,q^n$, we fix the choice of $g^*$ by setting 
\be\label{defstar} 
  g^*(\t)  \=  \bar b_0\,\frac{(4\pi\t_2)^{-k+1}}{k-1} \+ \sum_{n>0} \;n^{k-1}\,\bar b_n \;\G(1-k,4\pi n\t_2)\;q^{-n}\, , 
\ee
where $\t_2 = \rm{Im}(\t)$ and $\G(1-k,x) = \int_x^{\infty} t^{-k}\,e^{-t}\,dt$ denotes the incomplete gamma function, and where 
the first term must be replaced by $ -\bar b_0\,\log(4\pi\t_2)\,$ if $k=1$. Note that the series in~\eqref{defstar} converges despite 
the exponentially large factor $q^{-n}$ because $\G(1-k,x)= O(x^{-k}e^{-x})\,$. If we assume either that $k>1$ or that $b_{0}=0$,
then we can define $g^*$ alternatively by the integral  
\be\label{Lkintrep} g^*(\t) = \biggl(\frac i{2\pi}\Bigr)^{k-1} \int_{-\bar{\t}}^{\infty} (z+ \t)^{-k} \ \bar{g(-\bar{z})}\; dz \;. \ee
(The integral is independent of the path chosen because the integrand is holomorphic in $z$.) 
Since $h$ is holomorphic, \eqref{starinv} implies that the completion of $h$ is related to its shadow by
    \be\label{ddtbarh}   (4\pi\t_2)^k\,\,\frac{\pa \wh h(\t)}{\pa \bar{\tau}} \= -2\pi i\;\bar{g(\tau)}\;.  \ee

We denote by $ \IM_{k|0}^{\,!}$ the space of weakly holomorphic pure mock modular forms
of weight $k$ and arbitrary level. (The ``0'' in the notation corresponds to the word ``pure'' 
and will be replaced by an arbitrary integer or half-integer in \S\ref{MMMF}.)  
This space clearly contains the space $M_{k}^{\,!}$ of ordinary weakly holomorphic
modular forms (the special case $g=0$, $h=\wh h$) and we have an exact sequence  
  \be\label{exseqmmf} 0 \longrightarrow M_{k}^{\,!} \longrightarrow \IM_{k|0}^{\,!} \stackrel{\CS}{\longrightarrow} \bar {M_{2-k}}  \; \, ,\ee
where the {\it shadow map} $\CS$ sends $ h$ to $\bar{g} $.\footnote{We will use the word ``shadow"  to denote either $g(\t)$
or $\overline{g(\t)}$, but the shadow {\it map}, which should be linear over $\C$, always sends $h$ to $\bar g$, the complex conjugate 
of its holomorphic shadow.  We will also often be sloppy and say that the shadow of a certain mock modular form ``is'' some modular
form $g$ when in fact it is merely proportional to $g$, since the constants occurring are arbitrary and are often messy.}

In the special case when the shadow $g$ is a unary theta series as in \eqref{unary1} or \eqref{unary2} (which can only
hapen if $k$ equals 3/2 or 1/2, respectively), the mock modular form $h$ is called a {\it mock theta function}.
All of Ramanujan's examples, and all of ours in this paper,  are of this type.   In these cases the incomplete gamma functions
in~\eqref{defstar} can be expressed in terms of the complementary error function $\,\erfc(x)=2\pi^{-1/2}\int_x^\infty e^{-t^2}dt\;$:
\be \label{erfc}  \G\bigl(-\frac12,\,x\bigr)\= \frac2{\sqrt x}\,e^{-x}\;-\;2\,\sqrt\pi\;\erfc\bigl(\sqrt x\bigr)\;, \qquad
     \G\bigl(\frac12,\,x\bigr)\= \sqrt\pi\;\erfc\bigl(\sqrt x\bigr)\;.  \ee

A very simple, though somewhat artificial, example of a mock modular form is the weight~2 Eisenstein series $E_2(\t)$ 
mentioned in~\S\ref{modularbasic}, which is also a quasimodular form. Here the shadow $g(\t)$ is a constant and its 
non-holomorphic Eichler integral $g^*(\t)$ simply a multiple of~$1/\t_2\,$.  
This example, however, is atypical, since most quasimodular forms (like $E_2(\t)^2$) are not mock modular forms.
We end this subsection by giving a number of less trivial examples. Many more will occur later. 

\smallskip 
\ndt {\bf Example 1.}  In Ramanujan's famous last letter to Hardy in 1920, he gives 17 examples of mock theta functions, though
without giving any complete definition of this term.  All of them have weight $1/2$ and are given as $q$-hypergeometric
series. A typical example (Ramanujan's second mock theta function of ``order 7"\,---\,a notion that he also does not define) is 
  \be\label{Ramaeg}  \CF_{7,2}(\t) = - q^{-25/168}  \, \sum_{n=1}^{\infty} \frac{q^{n^2}}{(1-q^{n}) \cdots (1 - q^{2n-1})}
  \= - q^{143/168} \, \bigl( 1 + q + q^2 + 2 q^3 + \cdots \bigr) \ .  \ee
This is a mock theta function of weight $1/2$ on $\Gamma_0(4) \cap \Gamma(7) $ with shadow the unary theta series 
  \be\label{nequiv2} \sum_{n\equiv 2 \, (\mod\;7)}  \chi_{12}(n)\,n\,q^{n^2/168}\,, \ee
with $\chi_{12}(n)$ as in \eqref{defchi12}.  The product $\eta(\t)\CF_{7,2}(\t)$ is a strongly holomorphic mixed mock modular
form of weight $(1,1/2)$, and by an identity of Hickerson~\cite{Hickerson} is equal to an indefinite theta series:
  \be\label{Hickerson} \eta(\t) \, \CF_{7,2} (\t) \= \sum_{r,\,s\,\in\,\Z + \frac5{14}} 
  \frac12\bigl({\rm sgn}(r) + {\rm sgn}(s)\bigr)\, (-1)^{r-s} \, q^{(3r^2 +8rs+3s^2)/2} \ . \ee

\smallskip
\ndt {\bf Example 2.} 
The second example is the generating function of the Hurwitz-Kronecker class numbers $H(N)$. 
These numbers are defined for $N>0$ as the number of $PSL(2,\Z)$-equivalence classes of 
integral binary quadratic forms of discriminant $-N$, weighted by the reciprocal of the number of 
their automorphisms (if $-N$ is the discriminant of an imaginary quadratic field $K$ other than $\IQ(i)$ or 
$\IQ(\sqrt{-3})$, this is just the class number of $K$), and for other values of $N$ by 
$H(0) = -1/12$ and $H(N)=0$ for $N<0$. It was shown in  \cite{Zagier:1975} that the function 
  \be\label{classnum}  {\bf H}(\t) \;:=\; \sum_{N=0}^{\infty} H(N)\,q^{N} 
 \= -\frac 1{12} \+ \frac13 q^3 \+ \frac12 q^4 \+ q^7 \+ q^8 \+ q^{11} \+ \cdots  \ee
is a mock  modular form of weight $3/2$ on $\Gamma_0(4)$, with shadow the classical theta function $\theta (\tau)  = \sum q^{n^2}$. 
Here ${\bf H}(\tau)$ is strongly holomorphic, but this is very exceptional. In fact, up to minor variations it is essentially the {\it only} 
known non-trivial example of a strongly holomorphic pure mock modular form, which is why we will introduce mixed mock modular forms below.
As mentioned earlier, this mock modular form has appeared in the physics literature in the context of topologically twisted supersymmetric
Yang-Mills theory on $\mathbb{CP}{^2}$ \cite{Vafa:1994tf}.

\smallskip
\ndt {\bf Example 3. }
This example is taken from \cite{Zagier:2007}.  Set
  \be\label{defF62}  F_2^{(6)}(\tau) \= - \sum_{r>s>0} \chi_{12}(r^2-s^2) \, s \, q^{rs/6} \= q + 2 q^2 +  q^3 + 2 q^4 - q^5 + \cdots \ee
with $\chi_{12}$ as in \eqref{defchi12}, and let $E_2(\t)$ be the quasimodular Eisenstein series of weight~2 on $SL(2,\Z)$ as defined
in \S\ref{modularbasic}.  Then the function 
   \be\label{defh6}  h^{(6)}(\tau) \= \frac{12 F_2^{(6)}(\t) - E_2(\t)}{\eta(\tau)} 
   \= q^{-1/24} \, \bigl(-1 + 35 q + 130q^2 + 273q^3 + 595q^4 + \cdots \bigr) \ee
is a weakly holomorphic mock modular form of weight $3/2$ on $SL(2,\Z)$ with shadow proportional to $\eta(\tau)$. More generally, if we define 
  \be\label{defF6k}  F_k^{(6)}(\t) \=  - \sum_{r>s>0} \chi_{12}(r^2-s^2) \, s^{k-1} q^{rs/6}  \qquad (\text{$k>2$, $k$ even})\,,  \ee
then for all $\nu \ge 0$ we have 
  \be\label{hetaRC6} \frac{24^{\nu}}{\left({2\nu \atop \nu }\right)}\,[h^{(6)},\eta]_\nu 
   \=12 F_k^{(6)} \,-\, E_k \+ {\rm cusp \; form \; of \; weight \;} k \ , \quad k = 2\nu +2 \ , \ee
where $ [h^{(6)}, \eta ]_{\nu}$ denotes the $\nu$th Rankin-Cohen bracket of the mock modular form $h^{(6)}$ and the modular 
form $\eta$ in weights $(3/2,1/2)$. This statement, and the similar statements for other mock modular forms which come later, 
are proved by the method of holomorphic projection, which we do not explain here, and are intimately connected with 
the mock Jacobi forms introduced in the next subsection. 
That connection will also explain the superscript ``6" in~\eqref{defF62}--\eqref{hetaRC6}.

\smallskip
\ndt {\bf Example 4.}  Our last example is very similar to the preceding one. Set 
\be\label{defF22}  F_2^{(2)}(\t) \=  \sum_{r>s>0 \atop r-s \, \rm odd} (-1)^r\,s\,q^{rs/2} \= q + q^2 -  q^3 +  q^4 - q^5 + \cdots \, . \ee
Then the function
  \be\label{defh2}  h^{(2)}(\t) \= \frac{24 F_2^{(2)}(\t) - E_2(\t)}{\eta(\t)^3}  
    \= q^{-1/8} \, \bigl(-1 + 45 q + 231q^2 + 770q^3 + 2277q^4  + \cdots \bigr)\ee
is a weakly holomorphic mock modular form of weight $1/2$ on $SL(2,\Z)$ with shadow proportional to $\eta(\t)^3$ and, as in Example~3, if we set  
  \be\label{defF2k}   F_k^{(2)}(\t) \=  - \sum_{r>s>0 \atop r-s \, \rm odd} (-1)^{r} \, s^{k-1} \, q^{rs/2}  \qquad (\text{$k>2$, $k$ even})\,,\ee
then for all $\nu \ge 0$ we have 
  \be\label{hetaRC2} \frac{8^\nu}{\left({2\nu \atop \nu }\right)}\,[h^{(2)},\eta^3]_{\nu} 
   \= 24 F^{(2)}_k \,-\, E_k \+ {\rm cusp \; form \; of \; weight \;} k \ , \quad k = 2\nu +2 \ , \ee
where $[h^{(2)}, \eta^3]_\nu$ denotes the Rankin-Cohen bracket in weights $(1/2,3/2)$.

In fact, the mock modular form $h^{(2)}$, with the completely different definition 
\be\label{oogfor}
h^{(2)}(\t) \= \frac{\vth_2(\t)^4 - \vth_4(\t)^4}{\eta(\t)^3} \, - \, \frac{24}{\vth_3(\t)} \, 
\sum_{n\in\IZ} \frac{q^{n^2/2-1/8}}{1+q^{n-1/2}} \, , 
\ee
arose from the works of Eguchi-Taormina \cite{Eguchi:1988af} and 
Eguchi--Ooguri--Taormina--Yang \cite{EOTY:1989} in the late 1980's 
in connection with characters of the $\CN=4$ superconformal algebra in two dimensions and the elliptic genus of $K3$ surfaces, and
appears explicitly in the work of Ooguri \cite{Ooguri:1989fd} (it is $q^{-1/8}(-1+F(\t)/2)$ in Ooguri's notation) and
Wendland~\cite{Wendland:2000}. It has recently aroused considerable interest because of the ``Mathieu moonshine" discovery by Eguchi, 
Ooguri and Tachikawa (see~\cite{Eguchi:2010ej, Cheng:2010pq}) that the coefficients 45, 231, 770, \dots\ in~\eqref{defh2} are 
the dimensions of certain representations of the Mathieu group~$M_{24}\,$.  
The equality of the right-hand sides of \eqref{defh2} and \eqref{oogfor}, which is not a priori 
obvious, follows because both expressions define mock modular forms whose shadow is the 
same multiple of $\eta(\t)^{3}$, and the first few Fourier coefficients agree. 


\subsection{Mock Jacobi forms \label{MockJacobi}}

%
%
By a (pure) {\it mock Jacobi form}\footnote{We mention that there are related constructions and definitions in the 
literature. The weak Maass-Jacobi forms of~\cite{BringRicht}
include the completions of our mock Jacobi forms, but in general are non-holomorphic in $z$ as 
well as in $\tau$. See also~\cite{Choie} and~\cite{BringmannRaum}. The functions discussed in this paper
will be holomorphic in~$z$.} 
(resp.~{\it weak mock Jacobi form}) of weight $k$ and index $m$  we will mean a 
holomorphic function $\v$ on $\IH \times \C$ that satisfies the elliptic transformation property~\eqref{elliptic},
and hence has a Fourier expansion as in~\eqref{fourierjacobi} with the periodicity property~\eqref{cnrprop}
and a theta expansion as in~\eqref{jacobi-theta}, and that also satisfies the same cusp conditions~\eqref{holjacobi}
(resp.~\eqref{weakjacobi}) as in the classical case, but in which the modularity property with respect to the
action of $SL(2,\Z)$ on $\IH\times\Z$ is weakened: the coefficients $h_{\ell}(\tau)$ in \eqref{jacobi-theta}
are now mock modular forms rather than modular forms of weight $k-\frac12$, and the modularity property of
$\v$ is that the {\it completed} function 
  \be\label{defphihat} \wh \v(\t,z) \;=\!\sum_{\ell\inn\Z/2m\Z} \wh h_\ell(\t) \, \vth_{m,\ell}(\t,z)\,,  \ee
rather than $\v$ itself, transforms according to \eqref{modtransform}.  If $g_\ell$ denotes the shadow of $h_l$, then we have
  $$ \wh\v(\t,z) \;= \! \v(\t,z)\+\sum_{\ell\inn\Z/2m\Z} g^*_\ell(\t)\,\vth_{m,\ell}(\t,z)$$
with $g^*_\ell$ as in \eqref{defstar} and hence, by \eqref{starinv},   
 \be \psi(\t,z)\;:=\; \t_2^{k-1/2}\,\frac{\pa}{\pa\overline\t}\wh\v(\t,z)\;\doteq\;\sum_{\ell\inn\Z/2m\Z} \overline{g_\ell(\t)}\,\vth_{m,\ell}(\t,z)\,. \ee
(Here $\doteq$ indicates an omitted constant.)  The function $\psi(\t,z)$ is holomorphic in $z$, satisfies the same elliptic
transformation property~\eqref{elliptic} as $\v$ does (because each $\vth_{m,\ell}$ satisfies this), satisfies the heat equation
$\,\bigl(8\pi im\,\frac{\pa}{\pa\t}-\frac{\pa^2}{\pa z^2}\bigr)\psi=0$  (again, because each $\vth_{m,\ell}$ does), and, by virtue
of the modular invariance property of $\wh\v(\t,z)$, also satisfies the transformation property
 \be\label{modtransformpsi}  \psi(\frac{a\t+b}{c\t+d}, \frac{z}{c\t+d}) \=   |c\t+d|\,(c\bar\t+d)^{2-k}\,e^{\frac{2\pi i m c z^2}{c\t +d}}
  \, \psi(\t,z)  \qquad \forall\;\left(\begin{array}{cc} a&b\\ c&d \end{array} \right) \inn SL(2;\Z) \ee
with respect to the action of the modular group.  These properties say precisely that $\psi$ is a {\it skew-holomorphic Jacobi form} of weight~$3-k$
and index~$m$ in the sense of Skoruppa \cite{Skoruppa1,Skoruppa2}, 
and the above discussion can be summarized by saying that we have an exact sequence
  \be\label{exseqJmf} 0 \longrightarrow J_{k,m}^{\text{weak}} \longrightarrow \IJ_{k|0,m}^{\text{weak}}
          \stackrel{\CS}{\longrightarrow} J_{3-k,m}^{\textrm{skew}} \ee
(and similarly with the word ``weak" omitted), where $\IJ_{k|0,m}^{\textrm{weak}}$ (resp.~$\IJ_{k|0,m}$) denotes the space of weak 
(resp.~strong) pure mock Jacobi forms and the ``shadow map" $\CS$ 
sends\footnote{This shadow map was introduced in the context of weak Maass-Jacobi forms in~\cite{BringRicht}.}
 $\v$ to $\psi$.

It turns out that most of the classical examples of mock theta functions occur as the components of a
vector-valued mock modular form which gives the coefficients in the theta series expansion of a mock Jacobi
form.  We illustrate this for the four examples introduced in the previous subsection.

\smallskip
\ndt {\bf Example 1.}   The function $\CF_{7,2}(\t)$ in the first example of \S\ref{basicdefs} is actually one of three 
``order~7 mock theta functions" $\{\CF_{7,j}\}_{j=1,2,3}$ defined by Ramanujan, each given by a 
$q$-hypergeometric formula like \eqref{Ramaeg}, each having a shadow  $\Theta_{7,j}$ like in \eqref{nequiv2}  
(but now with the summation over $n \equiv j$ rather than $n \equiv 2$ modulo~7), and each  satisfying an indefinite 
theta series identity like \eqref{Hickerson}.  We extend $\{\CF_{7,j}\}$ to all $j$ by defining it to be an odd periodic
function of $j$ of period $7$, so that the shadow of $\CF_{7,j}$ equals $\Theta_{7,j}$ for all~$j\in\Z$. Then the function 
  \be \CF_{42}(\t,z) \= \sum_{\ell\;(\mod\;84)} \chi_{12}(\ell) \, \CF_{7,\ell} (\t) \, \vth_{42,\ell} (\t,z) \ \ee
belongs to $\IJ_{1,42}^{\text{weak}}$. The Taylor coefficients $\xi_{\nu}$ as defined in equation \eqref{defxinodd} are 
proportional to $\sum_{j=1}^3 \bigl[\CF_{7,j},\Theta_{7,j}\bigr]_{\nu}$ and have the property
that their completions  $\widehat \xi_{\nu} = \sum_{j=1}^3 \bigl[\widehat \CF_{7,j},\Theta_{7,j}\bigr]_{\nu}$ 
transform like modular forms of weight $2 \nu +2$ on the full modular group $SL(2,\Z)$.

\smallskip
\ndt {\bf Example 2.}  Set $\CH_0(\t)=\sum_{n=0}^\infty H(4n)q^n$ and 
$\CH_1(\t)=\sum_{n=1}^\infty H(4n-1)q^{n-\frac14}$ so that $\CH_{0}(\t)+\CH_{1}(t) = {\bf H}(\t/4)$. 
Then the function
  \be\label{classH}  \CH(\t,z) \= \CH_0(\t)\vth_{1,0}(\t,z)\+\CH_1(\t)\vth_{1,1}(\t,z) 
   \= \sum_{\substack{n,\,r\in\IZ \\ 4 n - r^2 \ge 0}}  H(4n-r^2) \, q^n\,y^r\,,  \ee 
is a mock Jacobi form of weight $2$ and index $1$ with shadow $\overline{\vth_{1,0}(\t,0)}\,\vth_{1,0}(\t,z)\+
\overline{\vth_{1,1}(\t,0)}\,\vth_{1,1}(\t,z)$.  The $\nu$th Taylor coefficient $\xi_\nu$ of $\CH$ is given by
  \be\label{thetaHRC} \frac{4^\nu}{\binom{2\nu}\nu}\,\sum_{j=0}^1 [\vth_{1,j},\CH_j]_\nu
  \= \delta_{k,2}\,E_{k} \,-\,F_k^{(1)}\+ \text{(cusp form of weight $k$ on $SL(2,\Z)$)}\,, \ee
where $k=2\nu+2$ and
  \bea\label{defF1k} F_k^{(1)}(\t) \;:=\; \sum_{n>0}\Bigl(\sum_{d|n} \min\bigl(d,\frac nd\bigr)^{k-1}\Bigr) q^{n}\qquad\quad \text{($k$ even, $k\ge2$)}\,. \eea 
In fact the cusp form appearing in \eqref{thetaHRC} is a very important one, namely (up to a factor $-2$) the sum of the
normalized Hecke eigenforms in $S_k(SL(2,\Z))$, and equation~\eqref{thetaHRC} is equivalent to the famous formula of Eichler
and Selberg expressing the traces of Hecke operators on~$S_k(SL(2,\Z))$ in terms of class numbers of imaginary quadratic fields.


\smallskip
\ndt {\bf Example 3.}  Write the function $h^{(6)}$ defined in~\eqref{defh6} as 
\be \label{defC6}
h^{(6)} (\t) \= \sum_{\substack{\DD \ge - 1 \\ \DD \equiv -1 \; (\mod \; 24)}} C^{(6)}(\DD) \, q^{\DD/24}
\ee
with the first few coefficients $C^{(6)}(\DD)$ given by
$$
\begin{tabular}{|c|ccccccccccc|c}  \hline     $\DD$ & $-1$& $23$& $47$ & $71$ &$95$& $119$& $143$ & 167 & 191 & 215 & 239 \\
       \hline  $C^{(6)}(\DD)$ &$-1$ & 35 & 130 & 273 & 595 & 1001 & 1885 & 2925 & 4886 & 7410 & 11466 \\
\hline         \end{tabular}  \quad .
$$
Then the function 
\be \label{defCF6} \CF_6(\t,z) \= \sum_{\substack{n,\,r\in\IZ \\ 24n-r^2 \ge -1}} \chi_{12}(r)\,C^{(6)}(24n-r^2)\;q^n\,y^r \ee
is a mock Jacobi form of index~6 (explaining the notation $h^{(6)}$). Note that,  
surprisingly, this is even {\it simpler} than the expansion of the index~1 mock Jacobi form just discussed,  
because its twelve Fourier coefficients $h_{\ell}$ are all multiples of just one of them, while the two Fourier 
coefficients $h_{\ell}$ of $\CH(\t,z)$ were not proportional. (This is related to the fact that the shadow
$\eta(\t)$ of $h^{(6)}(\t)$ is a modular form on the full modular group, 
while the shadow $\theta(\t)$ of $\CH(\t)$ is a modular form on a congruence subgroup.)
Specifically, we have $h_{\ell}(\t) = \chi_{12}(\ell) h^{(6)}(\t)$ for all $\ell\in\Z/12\Z$,
where $\chi_{12}$ is the character defined in~\eqref{defchi12}, so that $\CF_6$ has the factorization 
\be\label{F6splits} 
 \CF_6(\t,z) \= h^{(6)}(\t)\,\sum_{\ell\;(\mod 12)} \chi_{12}(\ell)\,\vth_{6,\ell}(\t,z) \= \eta(\t)\,h^{(6)}(\t)\,\frac{\vth_{1}(\t,4z)}{\vth_{1}(\t,2z)} \, .  
\ee
Combining this with \eqref{xinhrel} and noting that $\vth^{0}_{6,1}-\vth^{0}_{6,5}=\vth^{0}_{6,11}-\vth^{0}_{6,7}=\eta$,
we see that the functions described in~\eqref{hetaRC6} are proportional to the Taylor coefficients $\xi_{\nu}(\t)$ of $\CF_6(\t,z)$.

\smallskip
\ndt {\bf Example 4.} The fourth example is very similar. Write the mock modular form~\eqref{defh2} as 
\be
h^{(2)} (\t) \=  \sum_{\substack{\DD \ge - 1  \\ \DD \equiv -1 \; (\mod \; 8)}} C^{(2)}(\DD) \, q^{\DD/8} 
\ee
with initial coefficients given by
$$
\begin{tabular}{|c|cccccccccc|cc}   \hline    $\DD$ & $-1$& 7& 15 & 23 & 31& 39 & 47 & 55 & 63 & 71  \\
       \hline  $C^{(2)}(\DD)$ &$-1$ & 45 & 231 & 770 & 2277 & 5796 & 13915 & 30843 & 65550 & 132825 \\
 \hline        \end{tabular}  \quad . 
$$
Then the function 
\bea\label{defCF2}
\CF_2(\t,z) & \= & \sum_{\substack{n,\,r\in\IZ \\ 8 n - r^2 \ge -1}} \chi_4(r) \, C^{(2)}(8n-r^2)\,q^n \, y^r 
\= h^{(2)}(\t)\,\bigl(\vth_{2,1}(\t,z) - \vth_{2,3}(\t,z)\bigr) \cr
& \= &  h^{(2)}(\t)\,\vth_ 1(\t,2z) =  \eta(\t)^{3} \, h^{(2)}(\t) \, C(\t,z)
\eea
(where $\chi_4(r) = \pm 1$ for $r \equiv \pm 1 \, (\mod \; 4)$ and $\chi_4(r) =0$ for $r$ even) 
is a mock Jacobi form of index~2. (Here $C$ is the Jacobi form of weight $-1$ and index 2 
introduced in \S\ref{MJF1}.)  
As in Example 3, the functions given in~\eqref{hetaRC2} are proportional to the Taylor coefficients
$\xi_{\nu}$ of $\CF_2$, because $\vth_{2,1}^ 1 = - \vth_{2,3}^ 1 = \eta^3$, where $\vth_{2,\ell}^ 1$ is defined by~\eqref{defth1}.

\subsection{Mock modular forms: the general case \label{MMMF}}

If we replace the condition ``exponential growth'' in $1.$ of~\ref{basicdefs} by ``polynomial growth,"  we get
the class of {\it strongly holomorphic mock modular forms}, which we can denote $\IM_{k|0}$, and an exact 
sequence $0 \to M_{k} \to \IM_{k|0} \to \bar {M_{2-k}} $. This is not very useful, however, because there are 
almost no examples of  ``pure'' mock modular forms that are strongly holomorphic, essentially the only ones being 
the function $\CH$ of Example~2 of \S\ref{basicdefs} and its variants. It becomes useful if we generalize to {\it mixed 
mock modular forms} of mixed weight $k|\ell$. (Here $k$ and $\ell$ can be integers or half-integers.) These are
holomorphic functions $h(\tau)$, with polynomial growth at the cusps, that have completions $\widehat h$ of the
form $\wh h = h + \sum_j f_j\,g^*_j$ with $f_j \in M_{\ell}$, $g_j \in M_{2-k+\ell}$ 
that transform like modular
forms of weight $k$. The space $\IM_{k|\ell}$ of such forms thus fits into an exact sequence 
  \be\label{exseqmmf2} 0 \longrightarrow M_{k} \longrightarrow \IM_{k|\ell}  \stackrel{\CS}{\longrightarrow} 
  M_{\ell}    \otimes    \bar {M_{2-k+\ell}} \,, \ee
where the shadow map $\CS$ now sends $h$ to $\sum_{j} f_{j}\,\bar g_{j}$. If $\ell = 0$ this reduces to the 
previous definition, since $M_0=\C$, but with the more general notion of mock modular forms 
there are now plenty of strongly holomorphic examples, and, just as for ordinary modular forms, they have 
much nicer properties (notably, polynomial growth of their Fourier coefficients) than the weakly holomorphic ones. 
If the shadow of a mixed mock modular form $h \in \IM_{k|\ell}$ happens to contain only one term 
$f(\t) \overline{g(\t)}$, and if $f(\t)$ has no zeros in the upper half-plane, then $f^{-1} \, h $ is a weakly holomorphic 
pure mock modular form of weight $k-\ell$ (and in fact, all weakly holomorphic pure mock modular forms arise in this way). 
For functions in the larger space $\IM_{k|\ell}^{\,!}$ of weakly holomorphic mixed mock modular forms, defined in the 
obvious way, this always happens, so $\IM_{k|\ell}^{\,!}$  simply coincides with $M_{\ell}\otimes\IM_{k-\ell}^{\,!}$ and does
not give anything new, but the smaller space $\IM_{k|\ell}$ does not have such a decomposition and is more interesting.
Moreover, there are many examples.  In fact, apart from Ramanujan's original examples of mock theta functions, which 
were defined as $q$-hypergeometric series and were weakly holomorphic pure mock modular forms of weight 1/2, most of
the examples of mock modular forms occurring ``in nature,'' such as Appell-Lerch sums or indefinite theta series,
are strongly holomorphic mixed modular forms.

We can also define ``even more mixed" mock modular forms by replacing $\IM_{k|\ell}$ by the space $\IM_{k}=\sum_{\ell}\IM_{k|\ell}$,
i.e., by allowing functions whose shadow is a finite sum of products $f_j(\t)\overline{g_j(\t)}$ with the $f_j$ of
varying weights~$\ell_j$ and $g_j$ of weight $2-k+\ell_j$. This space fits into an exact sequence 
  \be\label{exseqmmf3} 0 \longrightarrow M_{k} \longrightarrow \IM_{k}  \stackrel{\CS}{\longrightarrow} 
 \bigoplus_\ell M_\ell \otimes \bar {M_{2-k+\ell}} \;. \ee
In fact, the functions in this space are the most natural objects, and from now on we will use ``mock modular forms'' to
refer to this more general class, including the adjectives ``mixed'' and ``pure'' only when needed for emphasis. To justify this
claim, we consider the space $\mathfrak M_{k}$ of all {\it real-analytic modular forms of weight~$k$}, i.e.,  
real-analytic functions that transform like modular forms of weight~$k$ and have polynomial 
growth at all cusps. More generally, one can consider the space $\mathfrak M_{k,\ell}$ of real-analytic functions 
of polynomial growth that transform according to $F(\g \t) = (c\t+d)^{k} (c \bar \tau + d)^{\ell} F(\t)$ for all 
matrices $\g = \sm abcd$ belonging to some subgroup of finite index of $SL(2,\IZ)$. Since the function $\t_2=\Im (\t)$
transforms by $(\g \t)_2 = \t_2/|c\t+d|^2$, the spaces  $\mathfrak M_{k,\ell}$ and $\t_2^{r} \, \mathfrak M_{k+r,\ell+r}$ 
coincide for all~$r$. (Here, $k$, $\ell$ and $r$ can each be integral or half-integral.) Obviously the space $\mathfrak M_{k,\ell}$
contains the space $M_{k} \otimes \overline{M_{\ell}}$, so by the comment just made, it also contains the direct sum 
$\bigoplus_{r \in \IZ} \t_2^{r} \,  M_{k+r} \otimes \overline{M_{\ell+r}}$. Elements of this subspace will be called {\em decomposable}. 
Now, if a function $F \in \mathfrak M_{k} = \mathfrak M_{k,0}$ has the property that its $\bar \t$-derivative (which automatically belongs to 
$\mathfrak M_{k,2}\,$) is decomposable, say as $\sum_{j} {\t_2}^{r_{j}} \, f_j \,\overline{g_j}$ (where the weights of $f_{j}$ and $g_{j}$ 
are necessarily equal to $k+r_{j}$ and $2+r_{j}$, respectively), then $F$ is the sum of a holomorphic function 
$h$ and a multiple of $\sum_{j} f_{j} g^*_{j}$ and the function $h$ is a (mixed) mock modular form in the sense just explained. 
Conversely, a holomorphic function $h$ is a (mixed) mock modular form of weight $k$ if and only if there exist numbers $r_{j}$ and 
modular forms $f_{j} \in M_{k+r_{j}}$ and $g_{j} \in M_{2 + r_{j}}$ such that the sum $\wh h:= h+ \sum_{j} f_{j} g^*_{j}$ 
belongs to~$\mathfrak M_{k}$. We can thus summarize the entire discussion by the following very simple definition: 
\begin{quote}
{\it A mock modular form is the holomorphic part of a real-analytic modular form whose $\bar \t$-derivative is decomposable.}  
\end{quote}

As examples of this more general class we mention the functions $F_2^{(6)}$ and $F_2^{(2)}$ defined in equations~\eqref{defF62}
and~\eqref{defF22}, which are both (strong) mixed mock modular forms of total weight~2 and trivial character on the full modular group $SL(2,\Z)$,
but with different mixed weights: $12F_2^{(6)}$ is the sum of the functions $E_2(\t)$ and $\eta(\t)\,h^{(6)}(\t)$ of weights $2|0$ and
$\frac32|\frac12$, while $24F_2^{(2)}$ is the sum of the functions $E_2(\t)$ and $\eta(\t)^3h^{(2)}(\t)$ of weights $2|0$ and $\frac12|\frac32$.

We remark that in the pure case we have $\t_2^{k}\cdot\p_{\bar \t}F \,\in\,\overline{M_{2-k}}=\textrm{Ker}\bigl(\p_\t\!:
\mathfrak M_{0,2-k} \to \mathfrak M_{2,2-k} \bigr)$,  so that the 
subspace $\IM_{k,0}$ of $\IM_{k}$ can be characterized by a {\it single} second-order differential equation $\Delta_k \wh h=0$,
where $\D_k=\t_2^{2-k}  \p_{\t}  \t_2^{k}  \p_{\bar \t}$ is the Laplacian in weight~$k$. For this reason, (weak) pure 
mock modular forms can be, and often are, viewed as the holomorphic parts of (weak) harmonic Maass forms, i.e., of functions 
in $\mathfrak M_{k}^{\,!}$ annihilated by~$\D_{k} \,$.  But in the general case, there is no connection between mock modular 
forms or their completions and any second-order operators.  Only the Cauchy-Riemann operator~$\p/\p\bar \t$ is needed,
and the new definition of mock modular forms as given above is both more flexible and more readily generalizable than the
more familiar definition of pure mock modular forms in terms of harmonic Maass forms. 

Finally, we can define mixed mock Jacobi forms in exactly the same way as we defined pure ones (i.e., they have
theta expansions in which the coefficients are mixed mock modular forms, with completions that transform like ordinary Jacobi forms).
We will see in the next section that such objects arise very naturally as the ``finite parts" of meromorphic Jacobi forms.

\subsection{Superconformal characters and mock modular forms\label{FinitePart}} 

As already mentioned, mock modular forms first made their appearance in physics some ten years before the work of Zwegers, 
in connection with the characters of superconformal algebas~\cite{Eguchi:1988af, EOTY:1989}.  
The representations of the $\mathcal{N}=4$ superconformal algeba with central charge~$c=6k$ are labeled by two quantum numbers, the
energy~$h$ and spin~$l$ of the lowest weight state~$|\Omega\rangle$, which satisfy $h>k/4$, $l\in\{1/2,\,1,\dots,k/2\}$  for the non-BPS 
(massive) case and $h=k/4$, $l\in\{0,1/2,\dots,k/2\}$ for the BPS (massless) case. The corresponding character, defined by
 $$ \text{ch}^{\tilde R}_{k,h,l}(\t,z) \= \text{Tr}\bigl(q^{L_0-c/24}\,e^{4\pi izJ_0}\bigr)
     \qquad\bigl(\,L_0|\Omega\rangle=h|\Omega\rangle,\quad J_0|\Omega\rangle=l|\Omega\rangle\bigr)\,, $$
turns out to be, up to a power of $q$, a Jacobi form in the non-BPS case but a ``mock" object in the BPS case.  In 
particular, for $k=1$ the non-BPS characters are given by
  \be\label{nBPSchar}\text{ch}^{\tilde R}_{1,h,\frac12}(\t,z) \= -\,q^{h-\frac38}\,\frac{\vth_1(\t,z)^2}{\eta(\t)^3}   
     \= -\,q^{h-\frac38}\,\eta(\t)^3\,\v_{-2,1}(\t,z)   \qquad\bigl(h>\frac14\bigr) \ee
and the BPS characters by
  \be\label{BPSchar} \text{ch}^{\tilde R}_{1,\frac14,0}(\t,z) \,=\, -\,\frac{\vth_1(\t,z)^2}{\eta(\t)^3}\,\mu(\t,z)\,,  
    \;\quad  \text{ch}^{\tilde R}_{1,\frac14,\frac12}(z,\t)+2\text{ch}^{\tilde R}_{1,\frac14,0}(\t,z) 
   \,=\, -\, q^{-\frac18}\,\frac{\vth_1(\t,z)^2}{\eta(\t)^3}\,, 
\ee
where $\mu(\t,z)$ is the function defined by
  \be\label{mudef}  \mu(\t,z)\=\frac{e^{\pi iz}}{\vth_1(z,\t)}\;\sum_{n\in\Z}\frac{(-1)^nq^{\frac{n^2+n}2}e^{2\pi inz}}{1\,-\,q^n\,e^{2\pi iz}}
   \qquad(q=e^{2\pi i\t},\;y=e^{2\pi iz})\,, \ee
an ``Appell-Lerch sum" of a type that also played a central role in Zwegers's work\footnote{More precisely, the function $\mu(\t,z)$
is the specialization to $u=v=z$ of the 2-variable function $\mu(u,v;\t)$ considered by Zwegers.} and that we will discuss in the next 
section.  The ``mock" aspect is the fact that $\mu(\t,z)$ must be completed by an appropriate non-holomorphic correction 
term in order to transform like a Jacobi form.

The mock modular form $h^{(2)}$ arises in connection with the decomposition of the elliptic genus of a $K3$ surface as a linear
combination of these characters.  This elliptic genus is, as discussed in \ref{BHandJF}, equal to the weak Jacobi form $2B=2\v_{0,1}$.
It can be represented as a specialization of the partition  function (= generating function of states) of superstrings compactified
on a $K3$ surface~\cite{EOTY:1989} and as such can be decomposed into pieces corresponding to the representations of the $N=4$ superconformal 
algebra with central charge~$6k=6$.  The decomposition has the form 
\cite{Ooguri:1989fd, Wendland:2000, EST:2007}
  \be\label{SCA} 2\v_{0,1}(\t,z)\=20\,\text{ch}^{\tilde R}_{1,\frac14,0}(z,\t)
    \+A_0\,\text{ch}^{\tilde R}_{1,\frac14,\frac12}(z,\t)
     \,+\, \sum_{n=1}^\infty A_n\, \text{ch}^{\tilde R}_{1,n+\frac14,\frac12}(z,\t) \,, \ee
where $A_0=-2$, $A_1=90$, $A_2=462$, \dots\ are the Fourier coefficients $2C^{(2)}(8n-1)$ of the mock modular form $2h^{(2)}(\t)$.
Inserting~\eqref{nBPSchar} and~\eqref{BPSchar}, we can rewrite this as
 \be \label{K3decomp} \eta(\t)^{-3}\,\frac{\v_{0,1}(\t,z)}{\v_{-2,1}(\t,z)} \= -\,12\,\mu(\t,z) \,-\, h^{(2)}(\t)\,. \ee
In the next section we will see how to interpret this equation as an example of a certain canonical 
decomposition of meromorphic Jacobi forms \eqref{CoverA}.  

\section{From meromorphic Jacobi forms to mock modular forms \label{MockfromJacobi}}

In this section we consider Jacobi forms $\v(\t,z)$ that are meromorphic with respect to the variable~$z$.  It was discovered
by Zwegers~\cite{Zwegers:2002} that such forms, assuming that their poles occur only at points $z\in\IQ\t+\IQ\,$ (i.e., at torsion
points on the elliptic curve $\C/\Z\t+\Z\,$),  have a modified theta expansion related to mock modular forms.  
Our treatment is based on his, but the presentation is quite different and the results go further in one key respect. 
We show that $\v$ decomposes canonically into two pieces, one constructed directly from its poles and consisting of a finite
linear combination of Appell-Lerch sums with modular forms as coefficients, and the other  being a mock Jacobi form in the sense
introduced in the preceding section.  Each piece separately
transforms like a Jacobi form with respect to elliptic transformations.  Neither piece separately transforms like a Jacobi form with
respect to modular transformations, but each can be completed by the addition of an explicit and elementary non-holomorphic correction
term so that it does transform correctly with respect to the modular group.

In \S\ref{FinitePart} we explain how to modify the Fourier coefficients $h_\ell$ defined in \eqref{jacobi-Fourier}
when $\v$ has poles, and use these to define a ``finite part" of~$\v$ by the theta decomposition~\eqref{jacobi-theta}. 
In \S\ref{PolarPart} we define (in the case when $\v$ has simple poles only) a ``polar part" of $\v$ as a finite linear 
combination of standard Appell-Lerch sums times modular forms arising as the residues of~$\v$ at its poles, and show that
$\v$ decomposes as the sum of its finite part and its polar part.  Subsection \ref{MockDecomposition} gives the proof
that the finite part of~$\v$ is a mock Jacobi form and a description of the non-holomorphic correction term needed
to make it transform like a Jacobi form.  This subsection also contains a summary in tabular form of the various functions that 
have been introduced and the relations between them.  In \S\ref{double} we describe the modifications needed in the case of double poles
(the case actually needed in this paper) and in \S\ref{MJExamples}
we present a few examples to illustrate the theory.  Among the ``mock" parts of these are two of the most interesting mock
Jacobi forms from \S\ref{Mock} (the one related to class numbers of imaginary quadratic fields and the one conjecturally
related to representations of the Mathieu group~$M_{24}$).  Many other examples will be given in \S\ref{flyJac}.

Throughout the following sections, we use the convenient notation $\E(x) := e^{2\pi i x}$.

\subsection{The Fourier coefficients of a meromorphic Jacobi form\label{FinitePart}} 
As indicated above, the main problem we face is to find an analogue of the
theta decomposition~\eqref{jacobi-theta} of holomorphic Jacobi forms in the meromorphic case.  We will approach this
problem from two sides: computing the Fourier coefficients of $\v(\t,z)$ with respect to~$z$, and computing the 
contribution from the poles.  In this subsection we treat the first of these questions.

Consider a meromorphic Jacobi form $\v(\t, z)$ of weight $k$ and index $m$. We assume
that $\v(\t, z)$ for each $\t \in \IH$ is a meromorphic function of $z$ which has poles only at points $z = \a \t + \b$ wit\
$\a$ and $\b$ rational. In the case when $\v$ was holomorphic, we could write its Fourier expansion in the form \eqref{jacobi-Fourier}.
By Cauchy's theorem, the coefficient $h_\ell(\t)$ in that expansion could also be given by the integral formula
 \be\label{hAl} h^{(z_0)}_{\ell}(\t) \= q^{-\ell^2/ 4m}  \int_{z_0}^{z_0+1} \v(\t,z)\;\E(-\ell z)\;dz\,, \ee
where $z_0$ is an arbitrary point of $\C$.  From the holomorphy and transformation properties of~$\v$ it follows that the value
of this integral is independent of the choice of $z_0$ and of the path of integration and depends only on $\ell$ modulo~$2m$
(implying that we have the theta expansion \eqref{jacobi-theta}) and that each $h_\ell$ is a modular form of weight~$k-\frac12$.
Here each of these properties fails: the integral \eqref{hAl} is not independent of the path of integration (it jumps when
the path crosses a pole); it is not independent of the choice of the initial point $z_0$; it is not periodic in $\ell$
(changing $\ell$ by $2m$ corresponds to changing $z_0$ by $\tau$); it is not modular; and of course the expansion \eqref{jacobi-theta}
cannot possibly hold since the right-hand-side has no poles in $z$. 

To take care of the first of these problems, we specify the path of integration in \eqref{hAl} as the horizontal line
from $z_0$ to $z_0+1$.  If there are poles of $\v(\t,z)$ along this line, this does not make sense; in that case, we define
the value of the integral as the average of the integral over a path deformed to pass just above the poles and the
integral over a path just below them.  (We do not allow the initial point~$z_0$ to be a pole of $\v$, so this makes sense.)
To take care of the second and third difficulties, the dependence on~$z_0$ and the non-periodicity in~$\ell$, we play one of
these problems off against the other.  From the elliptic transformation property~\eqref{elliptic} we find that
\ben &h_\ell^{(z_0+\t)}(\t) &\= q^{-\ell^2/4m}\,\int_{z_0}^{z_0+1}\v(\t,\,z+\t)\,\E(-\ell\,(z+\t))\;dz \\ 
  & &\=   q^{-(\ell+2m)^2/4m}\,\int_{z_0}^{z_0+1}\v(\t,\,z)\,\E(-(\ell+2m)\,z)\;dz  \= h_{\ell+2m}^{(z_0)}(\t)\,,  \een
i.e., changing $z_0$ by $\t$ corresponds to changing $\ell$ by $2m$, as already mentioned.  It follows that if we
choose $z_0$ to be $-\ell\t/2m$ (or $-\ell\t/2m+B$ for any~$B\in\R$, since it is clear that the value of the integral~\eqref{hAl}
depends only on the height of the path of integration and not on the initial point on this line), then the quantity
\be\label{Defofhl}  h_\ell(\t)\;:=\;h_\ell^{(-\ell\t/2m)}(\t)\,,  \ee
which we will call the $\ell$th {\bf canonical Fourier coefficient} of~$\v$, depends only on the value of $\ell$ (mod $2m$).
This in turn implies that the sum 
  \be\label{phiF}  \v^F(\t,z)\;:=\; \sum_{\ell\inn \Z/2m\Z} h_\ell(\t) \, \vth_{m,\ell}(\t, z)\,, \ee
which we will call the {\bf finite} (or {\bf Fourier}) {\bf part} of $\v$, is well defined.  If $\v$ is holomorphic, then
of course $\v^F=\v$, by virtue of \eqref{jacobi-theta}.
Hence if we define the {\bf polar part}~$\v^{P}$ of~$\v$ to be~$\v-\v^{F}$, then~$\v^{P}$ only depends
on the principal part of~$\v$ at its singularities.  We will give explicit formulas in \S\ref{PolarPart} and \S\ref{double}
in the cases when all poles are single or double.

Note that the definiton of $h_\ell(\t)$ can also be written
\be\label{Defofhl2} h_\ell(\t) \= q^{\ell^2/4m}\,\int_{\R/\Z} \v(\t,z-\ell\t/2m)\;\E(-\ell z)\;dz\,, \ee
with the same convention as above if $\v(\t,z-\ell\t/2m)$ has poles on the real line.

Finally, we remark that none of these definitions have anything to do with the modular transformation
property \eqref{modtransform}, and would also work if our initial function $\v(\t,z)$ satisfied only the elliptic 
transformation property \eqref{elliptic}. (In that case, however, the functions $h_{\ell}$ in general need not 
be modular, and in fact will be mock modular when $\v(\t,z)$ is a meromorphic Jacobi form.)
What we have said so far can be summarized by a commutative diagram 
\bea 
\CE^{\rm Mer}_{\pm,m}  & \stackrel{F}{\longrightarrow} & \CE_{\pm,m}  \cr 
\cup \; \; & & \;  \cup  \qquad , \cr
J^{\rm Mer}_{k,m}  & \stackrel{F}{\longrightarrow} & \IJ_{k,m} 
\eea
where $J^{\rm Mer}_{k,m}$ (resp.~$\CE^{\rm Mer}_{\pm,m}$) denote the subspace of meromorphic
Jacobi (resp.~elliptic) forms having poles only at torsion points $z \in \IQ \t + \IQ$, and the map~$F$
sends $\v \mapsto \v^{\rm F}$.

\subsection{The polar part of $\v$ (case of simple poles) \label{PolarPart}}
We now consider the contribution from the poles. To present the results we first need to
fix notations for the positions and residues of the poles of our meromorphic function~$\v$.  
We assume for now that the poles are all simple.

By assumption, $\v(\t,z)$ has poles only at points of the form $z=z_s=\a\t+\b$ for $s=(\a,\b)$ belonging to some subset $S$ 
of~$\IQ^2$.  The double periodicity property~\eqref{elliptic} implies that $S$ is invariant under translation by~$\Z^2$, and
of course $S/\Z^2$ must be finite.  The modular transformation property~\eqref{modtransform} of $\v$ implies that  $S$ is $SL(2,\Z)$-invariant.
For each $s=(\a,\b) \in S$, we set
  \be\label{defDs} D_s(\t) \= 2\pi i\;\E(m\a z_s)\, {\rm Res}_{z=z_s}\bigl(\v(\t,z)\bigr) \qquad  (s=(\a,\b)\in S,\;z_s=\a\t+\b)\;.  \ee
The functions $D_{s}(\t)$ are holomorphic modular forms of weight $k-1$ and some level, and only finitely many of them are distinct. More precisely, they satisfy
  \bea\label{Dell} \bullet && \;\, D_{(\a+\l,\b+\mu)} (\t)\;\=\;\E(m(\mu\a-\l\b))\,D_{(\a,\b)}(\t) \quad {\rm for} \quad(\l,\mu)\inn\Z^2\,, \\
   \label{Dmod}  \bullet && D_s \Bigl(\dfrac{a\t+b}{c\t+d}\Bigr) \= (c\t+d)^{k-1}\,D_{s\g}(\t) \quad {\rm for} \quad 
      \g = \bigl(\begin{array}{cc} a & b \\ c & d \\ \end{array} \bigr) \inn SL(2,\Z) \;, \eea
as one sees from the transformation properties of $\v$.  (The calculation for equation~\eqref{Dmod} is given in \S\ref{double}.)  It is precisely to obtain the simple
transformation equation~\eqref{Dmod} that we included the factor $\E(m\a z_s)$ in~\eqref{defDs}.  Since we are assuming for the time
being that there are no higher-order poles, all of the information about the non-holomorphy of $\v$ is contained in these functions.

The strategy is to define a ``polar part" of $\v$ by taking the poles $z_s$ in some fundamental parallelogram for the action of
the lattice $\Z\t+\Z$ on $\C$ (i.e., for $s=(\a,\b)$ in  the intersection of $S$ with some square box  $[A,A+1)\times[B,B+1)\,$)
and then averaging the residues at these poles over all translations by the lattice.  But we must be careful to do this in just
the right way to get the desired invariance properties.  For each $m\in\IN$ we introduce the {\it averaging operator}
\be \label{defAvm}  \Av m{F(y)} \;:=\; \sum_{\l} \,q^{m\l^2}y^{2m\l}F(q^\l y)  \ee
which sends any function of $y$ (= $\Z$-invariant function of $z$) of polynomial growth in~$y$ to a function of $z$
transforming like an index~$m$ Jacobi form under translations by the full lattice $\Z\t+\Z$.  For example, we have
\be \label{ThetAv} q^{\ell^2/4m}\,\Av m{y^\ell}\= \sum_{\l\inn\Z} \,q^{(\ell+2m\l)^2/4m}\,y^{\ell+2m\l} \= \vth_{m,\ell}(\t, z) \ee
for any $\ell\in\Z\,$.  If $F(y)$ itself is given as the average
\be\label{defAvZ} F(y)\=\AvZ\bigl[f(z)\bigr] \;:=\; \sum_{\mu\in\Z} f(z+\mu)   \qquad(z\in\C,\;\, y=\E(z)) \ee
of a function $f(z)$ in $\C$ (of sufficiently rapid decay at infinity), then we have
\be\label{defAvL} \Av m{F(y)} \=  \AvL\bigl[f(z)\bigr] \;:=\; \sum_{\l,\,\mu\in\Z} e^{2\pi im(\l^2\t+2\l z)}\,f(z+\l\t+\mu)\;. \ee

We want to apply the averaging operator \eqref{defAvm} to the product of the function $D_s(\t)$ with a standard rational
function of $y$ having a simple pole of residue~1 at $y=y_s=\E(z_s)$, but the choice of this rational function is not obvious.
The right choice turns out to be $\FR_{-2m\a}(y)$, where $\FR_c(y)$ for $c\in\R$ is defined by
\be\label{defRc}   \FR_c(y)\=\begin{cases} \quad\dfrac12\,y^c\,\dfrac{y+1}{y-1} & \text{if $c\in\Z$\,,} \\
    \quad\, y^{\ceil c}\,\dfrac1{y-1} &\text{if $c\in\R\ssm\Z$\,.} \end{cases}  \ee
(Here $\ceil c$ denotes the ``ceiling" of $c$, i.e., the smallest integer $\ge c$. The right-hand side can also be written  
more uniformly as $\dfrac12\;\dfrac{y^{\flo c+1}+y^{\ceil c}}{y-1}\,$, where $\flo c=-\ceil{-c}$ denotes the ``floor" of~$c$,
i.e., the largest integer $\le c$.)  This function looks artificial, but is in fact quite natural.  First of all, by
expanding the right-hand side of~\eqref{defRc} in a geometric series  we find
\be\label{Rcseries}  \FR_c(y)\=\begin{cases} -\sum^*_{\ell\ge c} y^\ell &\text{if $|y|<1$,} \\
                            \phantom{-}\sum^*_{\ell\le c} y^\ell &\text{if $|y|>1$,}\end{cases}\ee
where the asterisk on the summation sign means that the term $\ell=c$ is to be counted with multiplicity 1/2 when it occurs 
(which happens only for $c\in\Z$, explaining the case distinction in~\eqref{defRc}).  This formula, which can be seen as the 
prototypical example of wall-crossing, can also be written in terms of $z$ as a Fourier expansion (convergent for all $z\in\C\ssm\R$)
\be\label{Rcseries2} \FR_c(\E(z))\=-\sum_{\ell\in\Z}\frac{\sgn(\ell-c)\+\sgn(z_2)}2\;\E(\ell z) \qquad(y=\E(z),\;\, z_2=\Im(z)\ne0)\,, \ee
without any case distinction.  Secondly, $R_c(y)$ can be expressed in a  natural way as an average:
\begin{proposition}\label{RcAsAv} For $c\in\R$ and $z\in\C\ssm\Z$  we have  
\be  \FR_c(\E(z))\=\AvZ\Bigl[\frac{\E(cz)}{2\pi iz}\Bigr]\;.\ee
\end{proposition} 
\ndt {\it Proof:} If $c\in\Z$, then
$$ \AvZ\Bigl[\frac{\E(cz)}{2\pi iz}\Bigr]\=\frac{y^c}{2\pi i}\;\sum_{n\in\Z}\frac1{z-n}
   \=\frac{y^c}{2\pi i}\;\frac\pi{\tan\pi z}\=\frac{y^c}2\;\frac{y+1}{y-1}$$
by a standard formula of Euler.  If $c\notin\Z$ then the Poisson summation formula and \eqref{Rcseries2} give
\begin{align*} \AvZ\Bigl[\frac{\E(cz)}{2\pi iz}\Bigr] &\=\sum_{n\in\Z}\,\frac{\E(c(z+n))}{2\pi i(z+n)} 
     \=\sum_{\ell\in\Z}\,\biggl(\int_{iz_2-\infty}^{iz_2+\infty}\frac{\E(c(z+u))}{2\pi i(z+u)}\,\E(-\ell u)\,du\biggr)\,\E(\ell z) \\
   &\= -\,\sum_{\ell\in\Z}\frac{\sgn(\ell-c)\+\sgn(z_2)}2\;\E(\ell z)\=\FR_c(\E(z))  \end{align*}
if $z_2\ne0$, and the formula remains true for $z_2=0$ by continuity. An alternative proof can be obtained by noting
that $\E(-cz)\FR_c(\E(z))$ is periodic of period~1 with respect to~$c$ and expanding it as a Fourier series in~$c$,
again by the Poisson summation formula.

\medskip

For $m\in\N$ and $s=(\a,\b)\in\Q^2$ we now define a {\it universal Appell-Lerch sum} $\AA_m^s(\t,z)$ by
  \be \label{defFms} \AA_m^s(\t,z)\= \E(-m\a z_s)\,\Av m{\FR_{-2m\a}(y/y_s)}\qquad(y_s=\E(z_s)=\E(\b)q^\a)\,. \ee
It is easy to check that this function satisfies
  \be \AA_m^{(\a+\l,\b+\mu)}\= \E(-m(\mu\a-\l\b))\,\AA_m^{(\a,\b)} \qquad(\l,\,\mu\in\Z) \ee
and hence, in view of the corresponding property \eqref{Dell} of $D_s\,$, that the product $D_s(\t)\AA_m^s(\t,z)$ depends
only on the class of~$s$ in $S/\Z^2$.  We can therefore define 
\be \label{phiP}  \v^P(\t,z)\;:=\;\sum_{s\in S/\Z^2}D_s(\t)\,\AA_m^s(\t,z)\,, \ee
and it is obvious from the above discussion that this function, which we will call the {\bf polar part} of $\v$,
 satisfies the index~$m$ elliptic transformation property~\eqref{elliptic}
and has the same poles and residues as~$\v$, so that the difference $\v-\v^P$ is holomorphic and has a theta expansion. In fact, we have:

\begin{thm}\label{Thmsingle}  Let $\v(\t,z)$ be a meromorphic Jacobi form with simple poles at $z=z_s=\a\t+\b$ for
$s=(\a,\b)\in S\subset\IQ^2$, with Fourier coefficients $h_\ell(\t)$ defined by \eqref{hAl} and \eqref{Defofhl} or by~\eqref{Defofhl2}
and residues $D_s(\t)$ defined by~\eqref{defDs}.  Then $\v$ has the decomposition
\be\label{FplusP}  \v(\t,z)\=\v^F(\t,z)\+\v^P(\t,z)\,, \ee
where $\v^F$ and $\v^P$ are defined by equations \eqref{phiF} and \eqref{phiP}, respectively.
\end{thm}

\ndt {\it Proof:} Fix a point $P=A\t+B\in\C$ with $(A,B)\inn\R^2\ssm S$.  Since $\v$, $\v^F$ and $\v^P$ are meromorphic, it
suffices to prove the decomposition~\eqref{FplusP} on the horizontal line $\,\Im(z)=\Im(P)=A\t_2$.  On this line we have the Fourier expansion
$$ \v(\t,z)\=\sum_{\ell\inn\Z} q^{\ell^2/4m}\,h_\ell^{(P)}(\t)\,y^\ell  \,, $$
where the coefficients $h_\ell^{(P)}$ are defined by~\eqref{hAl} (modified as explained in the text there if $A=\a$ for
any $(\a,\b)\in S$, but for simplicity we will simply assume that this is not the case, since we are free to choose~$A$ any
way we want).  Comparing this with~\eqref{phiF} gives
\be\label{v-vF} \v(\t,z)\,-\,\v^F(\t,z) \= \sum_{\ell\inn\Z}\bigl(h_\ell^{(P)}(\t)-h_\ell(\t)\bigr)\,q^{\ell^2/4m}\,y^\ell \qquad(\Im(z)=\Im(P))\,.\ee
But $q^{\ell^2/4m}(h_\ell^{(P)}(\t)-h_\ell(\t))$ is just $2\pi i$ times the sum of the residues of $\v(\t,z)\E(-\ell z)$ in the parallelogram
with width~1 and horizontal sides at heights $A\t_2$ and $-\ell\t_2/2m$, with the residues of any poles on the latter line being
counted with multiplicity~1/2 because of the way we defined $h_\ell$ in that case, so
\bal  q^{\ell^2/4m}\bigl(h_\ell^{(P)}(\t)-h_\ell(\t)\bigr) 
  &\=2\pi i\sum_{s=(\a,\b)\inn S/\Z} \frac{\sgn(\a-A)-\sgn(\a+\ell/2m)}2\;\text{Res}_{z=z_s}\bigl(\v(\t,z)\E(-\ell z)\bigr) \notag \\
    &\= \sum_{s=(\a,\b)\inn S/\Z} \frac{\sgn(\a-A)-\sgn(\ell+2m\a)}2\;D_s(\t)\,\E(-(\ell+m\a)z_s)\,. \notag  \end{align} 
(Here ``$(\a,\b)\inn S/\Z$" means that we consider all $\a$, but $\b$ only modulo~1, which is the same by periodicity as
considering only the $(\a,\b)$ with $B\le \b<B+1$.)  Inserting this formula into~\eqref{v-vF} and using~\eqref{Rcseries2}, we find
\bal   \v(\t,z)\,-\,\v^F(\t,z)
   &\= -\sum_{s=(\a,\b)\inn S/\Z} \E(-m\a z_s)\,D_s(\t)\sum_{\ell\in\Z} \frac{\sgn(\Im(z-z_s))+\sgn(\ell+2m\a)}2\,\bigr(\frac y{y_s}\bigl)^\ell \notag\\  
  &\= \;\,\sum_{s=(\a,\b)\inn S/\Z} \,\E(-m\a z_s)\;D_s(\t)\;\FR_{-2m\a}(y/y_s) \notag\\
  &\= \sum_{s=(\a,\b)\inn S/\Z^2}\,\sum_{\l\in\Z}\E\bigl(-m(\a-\l)(z_s-\l\t)\bigr)\,D_{(\a-\l,\b)}(\t)\;\FR_{-2m(\a-\l)}(q^\l y/y_s) \notag \\
   &\= \sum_{s=(\a,\b)\inn S/\Z^2}\,D_s(\t)\,\E(-m\a z_s)\;\sum_{\l\in\Z}q^{m\l^2}\,y^{2m\l}\,\FR_{-2m\a}(q^\l y/y_s)\,, \notag  \end{align}
where in the last line we have used the periodicity property~\eqref{Dell} of $D_s(\t)$ together with the obvious periodicity property
$\FR_{c+n}(y)=y^n\FR_c(y)$ of $\FR_c(y)$.  But the inner sum in the last expression is just $\Av m{\FR_{-2m\a}(y/y_s)}$, so from the
definition~\eqref{defFms} we see that this agrees with $\v^P(\t,z)$, as claimed.

\subsection{Mock modularity of the Fourier coefficients \label{MockDecomposition}}

In subsections \S\ref{FinitePart} and \S\ref{PolarPart} we introduced a canonical splitting of a meromorphic Jacobi
form $\v$ into a finite part $\v^F$ and a polar part $\v^P$, but there is no reason yet (apart from the simplicity of
equation~\eqref{Defofhl}) to believe that the choice we have made is the ``right" one:  we could have defined
periodic Fourier coefficients $h_\ell(\t)$ in many other ways (for instance, by taking $P=P_0-\ell/2m\t$ with any
fixed $P_0\in\C$ or more generally $P=P_\ell-\ell\t/2m$ where $P_\ell$ depends only on $\ell$ modulo~$2m$)
and obtained other functions $\v^F$ and~$\v^P$. What makes the chosen decomposition special is that, as we will now
show, the Fourier coefficients defined in~\eqref{Defofhl} are (mixed) mock modular forms and the function $\v^F$
therefore a (mixed) mock Jacobi form in the sense of \S\ref{MMMF}.  The corresponding shadows will involve
theta series that we now introduce.

For  $m\in\N$, $\,\ell\in\Z/2m\Z$ and $s=(\a,\b)\in\Q^2$  we define the unary theta series
\be\label{deftheta3/2} 
  \Theta^s_{m,\ell}(\t) \= \E(-m\a\b)\,\sum_{\l\inn\Z+\a+\ell/2m}\l\,\E(2m\b\l)\,q^{m\l^2}  \ee
of weight 3/2 and its Eichler integral\footnote{Strictly speaking, the Eichler integral as defined by equation~\eqref{defstar}
with $k=1/2$ would be this multiplied by $2\sqrt{\pi/m}$, but this normalization will lead to simpler
formulas and, as already mentioned, there is no good universal normalization for the shadows of mock modular forms.}   
\be\label{defEichler3/2} 
 \Theta^{s\,*}_{m,\ell}(\t) \= \frac{\E(m\a\b)}2 \,\sum_{\l\inn\Z+\a+\ell/2m}
      \sgn(\l)\,\E(-2m\b\l)\,\erfc\bigl(2|\l|\sqrt{\pi m\t_2}\bigr)\,q^{-m\l^2}  \ee
(cf.~\eqref{defstar} and \eqref{erfc}). One checks easily that these functions transform by
\bea    \Th^{(\a+\l,\b+\mu)}_{m,\ell}(\t) \=& \E(m(\mu\a-\l\b))\,\Th^{(\a,\b)}_{m,\ell}(\t)  &\qquad(\l,\,\mu\in\Z)\,, \\
     \Th^{(\a+\l,\b+\mu)\,*}_{m,\ell}(\t) \=& \E(-m(\mu\a-\l\b))\,\Th^{(\a,\b)\,*}_{m,\ell}(\t) &\qquad(\l,\,\mu\in\Z)\,. \eea
with respect to translations of $s$ by elements of $\Z^2$.
From this and \eqref{Dell} it follows that the products $D_s\overline{\Th^s_{m,\ell}}$ and $D_s\Th^{s\,*}_{m,\ell}$ 
depend only on the class of $s$ in $S/\Z^2$, so that the sums over~$s$ occurring in the following theorem make sense.

\begin{thm}\label{Thmmock} Let $\v$, $h_\ell$ and $\v^F$ be as in Theorem~\ref{Thmsingle}.  Then each $h_\ell$ is a mixed
mock modular form of weight $k-\frac12\,|\,k-1$, with shadow $\sum_{s\in S/\Z^2} D_s(\t)\,\overline{\Th^s_{m,\ell}(\t)}$, and the
function $\v^F$ is a mixed mock Jacobi form.  More precisely, for each $\,\ell\in\Z/2m\Z\,$ the completion of $h_\ell$ defined by
\be\label{defhcompleted} \wh h_\ell(\t)\;:=\; h_\ell(\t)\; - \; \sum_{s\in S/\Z^2} D_s(\t)\,\Th^{s\,*}_{m,\ell}(\t)\,, \ee
with $\Th^{s\,*}_{m,\ell}$ as in~\eqref{defEichler3/2}, tranforms like a modular form of weight $k-1/2$ with
respect to some congruence subgroup of $SL(2,\Z)$, and the completion of $\v^F$ defined by
\be\label{defphiFhat}  \wh\v^F(\t,z)\;:=\;   \sum_{\ell\mypmod{2m}} \wh h_\ell(\t)\,\vth_{m,\ell}(\t,z) \ee
transforms like a Jacobi form of weight~$k$ and index~$m$ with respect to the full modular group.
\end{thm}

The key property needed to prove this theorem is the following proposition, essentially due to Zwegers, which says that the
functions $\AA_m^s(\t,z)$ defined in \S\ref{PolarPart} are (meromorphic) mock Jacobi forms of weight 1 and index~$m$, with shadow 
$\sum_{\ell\;\text{(mod $2m$)}}\overline{\Theta^s_{m,\ell}(\t)}\,\vth_{m,\ell}(\t,z)$ (more precisely, that each $\AA_m^s$ is
a meromorphic mock Jacobi form of this weight, index and shadow with respect to some congruence subgroup of $SL(2,\Z)$ 
depending on~$s$ and that the collection of all of them is a vector-valued mock Jacobi form on the full modular group):  
\begin{proposition}\label{Sander} For $m\in\N$ and $s\in\Q^2$ the completion $\wh \AA_m^s$ of $\AA_m^s$ defined by
  \be \label{defFsmhat} \wh \AA^s_m(\t,z) \;:=\; \AA^s_m(\t,z) \; + \sum_{\text{\rm $\ell$ (mod $2m$)}} \Th^{s\,*}_{m,\ell}(\t)\,\vth_{m,\ell}(\t,z)\;. \ee
satisfies
\bal\label{FsmPeriodic}    \wh\AA^{(\a+\l,\b+\mu)}_m(\t,z) &\=\E(-m(\mu\a-\l\b))\,\wh\AA^{(\a,\b)}_m(\t)\qquad(\l,\,\mu\in\Z)\,,   \\
  \label{FsmElliptic}       \wh \AA^s_m(\t,z+\l\t+\mu) &\= \E(-m(\l^2\t+2\l z))\,\wh \AA^s_m(\t) \qquad\quad(\l,\,\mu\in\Z)\,, \\
  \label{FsmModular}  \wh \AA^s_m\Bigl(\frac{a\t+b}{c\t+d},\,\frac z{c\t+d}\Bigr)\!&\= 
     (c\t+d)\;\E\Bigl(\frac{mcz^2}{c\t+d}\Bigr)\,\wh\AA^{s\gamma}_m(\t,z) \quad\;
    \Bigl(\g=\Bigl(\begin{matrix} a&b\\c&d\end{matrix}\Bigr)\inn SL(2,\Z)\Bigl)\,.      \end{align}
\end{proposition}
\ndt {\it Proof:} The first two properties are easy to check because they hold for each term in \eqref{defFsmhat} separately.
The modular transformation property, considerably less obvious, is essentially the content of Proposition~3.5 of~\cite{Zwegers:2002},
but the functions he studies are different from ours and we need a small calculation to relate them.  Zwegers defines two
functions $f_u^{(m)}(z;\t)$ and $\tilde f_u^{(m)}(z;\t)$ ($m\in\N$, $\t\in\mathbb H$, $\,z,\,u\in\C$) by
$$ f_u^{(m)}(z;\t)= \AvB m{\frac1{1-y/\E(u)}},\quad\;
   \tilde f_u^{(m)}(z;\t) =  f_u^{(m)}(z;\t)-\frac12\sum_{\text{$\ell$ (mod $2m$)}} R_{m,\ell}(u;\t)\,\vth_{m,\ell}(\t,z) $$
(here we have rewritten Definition 3.2 of~\cite{Zwegers:2002} in our notation), where 
\be\label{defRml}  R_{m,\ell}(u;\t) \= \sum_{r\in\ell+2m\Z}\Bigl\{\sgn\bigl(r+\frac12\bigr)
   \,-\,\text{erf}\Bigl(\sqrt\pi\,\frac{r\t_2+2mu_2}{\sqrt{m\t_2}}\Bigr)\Bigr\}\;q^{-r^2/4m}\;\E(-ru)\,, \ee
and shows (Proposition 3.5) that $\tilde f_u^{(m)}$ satisfies
the modular transformation property 
\be\label{ZwegersMod}  \tilde f_{u/(c\t+d)}^{(m)}\Bigl(\frac z{c\t+d};\,\frac{a\t+b}{c\t+d}\Bigr)
   \= (c\t+d)\;\E\Bigl(\frac{mc(z^2-u^2)}{c\t+d}\Bigr)\,\tilde f_u^{(m)}(z;\t)  \ee
for all $\g= \Bigl(\begin{matrix} a&b\\c&d\end{matrix}\Bigr)\in SL(2,\Z)$.  Noting that $\,\text{erf}(x)=\sgn(x)(1-\erfc(|x|))$, we find that 
$$  \frac12\,R_{m,\ell}(z_s;\t)\;=\sum_{r\equiv\ell\text{ (mod $2m$)}} \frac{\sgn(r+\frac12)-\sgn(r+2m\a)}2\,q^{-r^2/4m}\,y_s^{-r}
    \+ \E(m\a z_s)\,\Th^{s\,*}_{m,l}(\t) $$
in our notation.  On the other hand, from~\eqref{defRc} we have
\be  \FR_{-2m\a}(y)\=\frac1{y-1} \+\sum_{r\in\Z}\frac{\sgn(r+\frac12)-\sgn(r+2m\a)}2\,y^r \ee
(note that the sum here is finite).  Replacing $y$ by $y/y_s$ and applying the operator $\text{Av}^{(m)}$, we find (using~\eqref{ThetAv})
$$ \E(m\a z_s)\,\AA_m^s(\t,z)\=-f^{(m)}_{z_s}(z;\t)\+ \sum_{r\in\Z}\frac{\sgn(r+\frac12)-\sgn(r+2m\a)}2\,q^{-r^2/4m}\,y_s^{-r}\,\vth_{m,r}(\t,z)\,. $$
Combining these two equations and rearranging, we obtain 
$$ \wh \AA_m^s(\t,z) \= -\,\E(-m\a z_s)\,\tilde f^{(m)}_{z_s}(z;\t)\,,$$
and the modularity property~\eqref{FsmModular} then follows from~\eqref{ZwegersMod} after a short calculation.

\smallskip
The proof of Theorem~\ref{Thmmock} follows easily from Proposition~\ref{Sander}. We define the completion of the function $\v^P$ studied in~\S\ref{PolarPart} by
 \be\label{defPolarhat} \wh\v^P(\t,z)\;:=\; \sum_{s\in S/\Z^2} D_s(\t)\;\wh \AA^s_m(\t,z)\,. \ee
The sum makes sense by \eqref{FsmPeriodic}, and from the transformation equations \eqref{FsmElliptic}--\eqref{FsmModular}
together with the corresponding properties \eqref{Dell}--\eqref{Dmod} of the residue functions $D_s(\t)$ it follows that 
$\wh\v^P(\t,z)$ transforms like a Jacobi form of weight~$k$ and index~$m$ with respect to the full modular group.
Comparing equations~\eqref{defPolarhat} and~\eqref{defFsmhat} with equations~\eqref{defphiFhat} and~\eqref{defhcompleted}, we find that
$$ \wh\v^P(\t,z) \,-\, \v^P(\t,z) 
    \= \sum_{\ell\inn\Z/2m\Z}\sum_{s\inn S/\Z^2} D_s(\t)\;\Th^{s\,*}_{m,\ell}(\t)\,\vth_{m,\ell}(\t,z) 
    \=  \v^F(\t,z)\,-\,\wh\v^F(\t,z)   $$
and hence, using Theorem~\ref{Thmsingle}, that
$$  \wh\v^F(\t,z) \+ \wh\v^P(\t,z) \= \v^F(\t,z)\+\v^P(\t,z) \= \v(\t,z)\,. $$
Since both $\v(\t,z)$ and $\wh\v^P(\t,z)$ transform like Jacobi forms of weight~$k$ and index~$m$, it follows that $\wh\v^F(\t,z)$
also does, and then the fact that each $\wh h_\ell$ transforms like a modular form of weight $k-1/2$ (and hence that each $h_\ell$ is 
a mixed mock modular form with the weight and shadow given in the theorem) follows by the same argument that proves the modularity of
the coefficients $h_\ell$ in the theta expansion~\eqref{jacobi-theta} in the classical case.

\bigskip
 {\bf Summary.}  For the reader's convenience, we give a brief summary of the results given up to now.  We have the following six
functions of two variables $(\t,z)\in\mathbb{H}\times\C\,$:
\begin{myitemize}
\item  $\v(\t,z)$, a meromorphic Jacobi form of weight $k$ and index~$m$, assumed to have only simple poles
   at $z=z_s=\a\t+\b$ for $s=(\a,\b)$ in some discrete subset $S\subset\Q^2\,$;
\item $\v^F(\t,z)$, the {\bf finite} part of $\v$, defined by the theta expansion $\sum_{\text{$\ell$ (mod $2m$)}}h_\ell(\t)\vth_{m,\ell}(\t,z)$
  where $h_\ell(\t)$ is $q^{\ell^2/4m}$ times the $\ell$th Fourier coefficient of $\v(\t,z-\ell\t/2m)$ on the real line;
\item $\v^P(\t,z)$, the {\bf polar} part of $\v$, defined as $\sum_{s\in S/\Z^2} D_s(\t)\AA_m^s(\t,z)$, where $\AA_m^s$ is an 
   explicit Appell-Lerch   sum having simple poles at $z\in z_s+\Z\t+\Z\,$;
\item $\v^{C}(\t,z)$, a non-holomorphic {\bf correction} term, defined as $\sum_s\sum_\ell D_s(\t)\Th_{m,\ell}^{s\,*}(\t)\vth_{m,\ell}(\t,z)$ where
  the $\Th_{m,\ell}^{s\,*}$ are the Eichler integrals of explicit unary theta series of weight $3/2\,$;
\item $\wh \v^F(\t,z)$, the completed finite part, defined as $\v^F(\t,z)-\v^{C}(\t,z)\,$;
\item $\wh \v^P(\t,z)$, the completed polar part, defined as $\v^P(\t,z)+\v^{C}(\t,z)\,$.
\end{myitemize}
These six functions are related by
\be\label{relations}  \v^F\+\v^P \=\v\= \wh\v^F\+\wh\v^P\;,\qquad \v^F\,-\,\wh\v^F\=\v^{C}\=\wh\v_P\,-\,\v_P\;.  \ee
Each of them is real-analytic in $\t$, meromorphic in~$z$, satisfies the elliptic transformation property~\eqref{elliptic} with respect to~$z$,
 and has precisely two of four further desirable properties of such a function (note that $6=\bigl({4\atop2}\bigr)$), as shown in the following table
   \begin{center}     \begin{tabular}{r|c|c|c|c|c|c|}
   \multicolumn 1{c}{} & \multicolumn 1{c}{$\;\;\v\;\;$} & \multicolumn 1{c}{$\;\v^F\;$} 
      & \multicolumn 1{c}{$\;\v^{\rm P}\;$} & \multicolumn 1{c}{$\;\,\v^{C}\;\,$} & \multicolumn 1{c}{$\;\wh\v^{\rm F}\;$} & \multicolumn 1{c}{$\;\wh\v^{\rm P}\;$}\\
      \cline{2-7}        holomorphic in $\t\,?\;$ & \checkmark & \checkmark & \checkmark & $-$  & $-$  & $-$ \\
               transforms like a Jacobi form\,?\; & \checkmark & $-$  & $-$  & $-$  & \checkmark & \checkmark\\
                   holomorphic in $z\,?\;$  & $-$ & \checkmark & $-$  & \checkmark & \checkmark& $-$ \\
     determined by the poles of $\v\,?\;$ & $-$  & $-$ & \checkmark  & \checkmark  & $-$ & \checkmark  \\
   \cline{2-7}      \end{tabular} \qquad\qquad\qquad   \end{center}   \bigskip
in which the three checked entries in each row  correspond to one of the equations~\eqref{relations}.  Each Fourier
coefficient $h_\ell$ of $\v$ is a mixed mock modular form of weight $(k-1,1/2)\,$, and the finite part $\v^F$ is a mixed mock Jacobi form. 
In the holomorphic case, the functions $\v$, $\v^F$ and $\wh\v^F$ coincide and the functions $\v^P$, $\v^{C}$ and~$\wh\v^P$ vanish.  

We end by mentioning one further property of the canonical decomposition $\v=\v^F+\v^P$ that seems of interest.  The finite part $\v^P$
of~$\v$ is a linear combination of terms $h_\ell(\t)\,\vth_{m,\ell}(\t,z)$ where $h_\ell(\t)$ is a Laurent power series in~$q$ and
$\vth_{m,\ell}(\t,z)$ a linear combination of terms $q^{r^2/4m}y^r$, so it satisfies the same growth condition as a weak Jacobi form,
viz., that it contains only monomials $q^ny^r$ with discriminant $4mn-r^2$ bounded from below.  But the Appell-Lerch sums $\AA_m^s(\t,z)$
occurring in $\v^P(\t,z)$ have the opposite property: substituting the Fourier expansion~\eqref{Rcseries2} of $\FR_c$ into the
definition~\eqref{defFms}, we find that they contain only terms $q^ny^r$ with $4mn-r^2=-\d^2$ where the numbers $\d$ occurring
are rational (more precisely, $\d\in r+2m(\Z+\a)$) and unbounded (more precisely, $\d$ lies between~0 and~$r+\text O(1)$).

\subsection{The case of double poles \label{double}}

In this subsection we extend our considerations to the case when $\v$ is allowed to have double poles,
again assumed to be at points $z=z_s=\a\t+\b$ for $s=(\a,\b)$ belonging to some discrete subset $S$ of~$\Q^2$.
The first thing we need to do is to generalize the definition~\eqref{defDs} to this case. For $s\in S$ we
define functions $E_s$ and $D_s$ on $\mathbb H$ by the Laurent expansion
  \be\label{defDE}    \E(m\a z_s)\,\v(\t,z_s+\ve) 
    \= \frac{E_s(\t)}{(2\pi i\ve)^2} \+ \frac{D_s(\t)-2m\a\,E_s(\t)}{2\pi i\ve} \+ O(1) \qquad\text{as $\ve\to 0$.}  \ee
(Notice that $D_s(\t)$ is the same function as in \eqref{defDs} when the pole is simple.)
It is easily checked that the behavior of these functions under translations of $s$ by elements of~$\Z^2$ is given by
equation~\eqref{Dell} and its analogue for $E_s$.  For the modular behavior, we have:
\begin{proposition}  The functions $E_s(\t)$ and $D_s(\t)$ defined by~\eqref{defDE} are modular forms of weight $k-2$ and $k-1$,
respectively. More respectively, for all $s\in S$ and $\sm abcd\in SL(2,\Z)$ we have
  \be\label{modEDdoub}   E_s\Bigl(\frac{a\t+b}{c\t+d}\Bigr)\=(c\t+d)^{k-2}E_{s\g}(\t)\,, 
     \qquad D_s\Bigl(\frac{a\t+b}{c\t+d}\Bigr)\=(c\t+d)^{k-1}D_{s\g}(\t)\,.   \ee
\end{proposition}
\ndt {\it Proof:} We rewrite~\eqref{defDE} as
 $$\E(m\a(z_s+2\ve))\,\v(\t,z_s+\ve) \=  \frac{E_s(\t)}{(2\pi i\ve)^2} \+ \frac{D_s(\t)}{2\pi i\ve} \+ \text O(1)\,,$$
and also write $\a_s$ and $z_s(\t)$ instead of just $\a$ and $z_s$. Then using the easily checked identities
$$  z_s(\g\t)\=\frac{z_{s\g}(\t)}{c\t+d}\,,\qquad \a_{s\g}\,z_{s\g}(\t)\,-\,\a_s\,z_s(\g\t)\=\frac{c\,z_{s\g}(\t)^2}{c\t+d}  
     \qquad\Bigl(\g\t:=\frac{a\t+b}{c\t+d}\;\Bigr)\,, $$
and the modular transformation equation~\eqref{modtransform}, we find 
\bal   \frac{(c\t+d)^2\,E_s(\g\t)}{(2\pi i\ve)^2} &\+ \frac{(c\t+d)\,D_s(\g\t)}{2\pi i\ve}  
    \equiv\; \E\Bigl(m\a\Bigl(z_s(\g\t)+\frac{2\ve}{c\t+d}\Bigr)\Bigr)\,\v\Bigl(\g\t,\,z_s(\g\t)+\frac\ve{c\t+d}\Bigr)    \notag \\
  &\qquad\equiv\;  \E\Bigl(m\a\frac{z_{s\g}(\t)+2\ve}{c\t+d}\Bigr)\,\v\Bigl(\g\t,\,\frac{z_{s\g}(\t)+\ve}{c\t+d}\Bigr) \notag \\ 
  &\qquad\equiv\;(c\t+d)^k\,\E\Bigl(m\a\frac{z_{s\g}+2\ve}{c\t+d}+\frac{mc(z_{s\g}+\ve)^2}{c\t+d}\Bigr)\,\v\Bigl(\t,\,z_{s\g}+\ve\Bigr) \notag\\
  &\qquad\equiv\;(c\t+d)^k\,\E(m\a_{s\g}(z_{s\g}+2\ve))\,\v(\t,\,z_{s\g}+\ve) \notag\\
  &\qquad\equiv\;(c\t+d)^k\,\Bigl[\frac{E_{s\g}(\t)}{(2\pi i\ve)^2} \+ \frac{D_{s\g}(\t)}{2\pi i\ve} \Bigr]\;,\notag  \end{align}
where ``$\,\equiv\,$" means equality modulo a quantity that is bounded as $\ve\to0$. The claim follows.
\medskip

Next, we must define a standard Appell-Lerch sum with a double pole at $z=z_s$.  We begin by defining a rational
function $\FRR_c(y)$ with a double pole at $y=1$ for each $c\in\R$. Motivated by Proposition~\ref{RcAsAv}, we require
\be\label{RRcAsAv}   \FRR_c(\E(z))\= \AvZ\biggl[\frac{\E(cz)}{(2\pi iz)^2}\biggr] 
   \= \frac1{(2\pi i)^2}\,\sum_{n\in\Z} \frac{\E(c(z-n))}{(z-n)^2}\,. \ee
To calculate this explicitly as a rational function, we could imitate the proof of Proposition~\ref{RcAsAv}, but it is
easier to note that $\FRR_c= \bigl(-\dfrac1{2\pi i}\dfrac d{dz}\+c\bigr)\FR_c= \bigl(-y\dfrac d{dy}\+c\bigr)\FR_c$ and hence
from equations~\eqref{Rcseries2}, \eqref{Rcseries} and~\eqref{defRc} we get three alternative definitions
\bal\label{defRRseries} \FRR_c(y) &\=  \sum_{\ell\in\Z}\frac{|\ell-c| \+ \sgn(z_2)\,(\ell-c)}2\;y^\ell  \\
 \label{defRRc2}  &\= \begin{cases} \sum_{\ell\ge c} (\ell-c)\,y^\ell &\text{if $|y|<1$} \\
                            \sum_{\ell\le c} (c-\ell)\,y^\ell &\text{if $|y|>1$}\end{cases}  \\
          &\=    y^{\flo c+1}\,\biggl(\frac 1{(y-1)^2} \+ \frac{c\,-\,\flo c}{y-1}\biggr)\,. \end{align}
of $\FRR_c(y)$. (Notice that in these equations neither the asterisk on the summation sign nor the case distinction for $c\in\Z$ and $c\notin\Z$
are needed as before, and that the function $\FRR_c(y)$, unlike $\FR_c(y)$, is continuous in~$c$.)  For $s=(\a,\b)\in\Q^2$ and $m\in\N$ we set 
  \be \label{defGms} \BB_m^s(\t,z)\= \E(-m\a z_s)\,\Av m{\FRR_{-2m\a}(y/y_s)}\,, \ee
in analogy with \eqref{defFms}.  If we then define the polar part $\v^P$ of $\v$ by
\be \label{phiPdoub}  \v^P(\t,z)\= \sum_{s\in S/\Z^2}\Bigl( D_s(\t)\,\AA_m^s(\t,z)\+ E_s(\t)\,\BB_m^s(\t,z) \Bigr)\,; \ee
then the definitions of the functions $D_s$, $E_s$, $\AA_m^s$ and $\BB_m^s$ immediately imply that $\v^P$ has the same
singularities as~$\v$, so that the difference 
  \be\label{NewFinitePart}  \v^F(\t,z) \= \v(\t,z)\,-\,\v^P(\t,z) \ee
is a holomorphic function of $z$.

As before, the key property of the Appell-Lerch sums is that they are again mock Jacobi forms, of a somewhat
more complicated type than before.  We introduce the unary theta series
\be\label{deftheta1/2} 
  \mytheta^s_{m,\ell}(\t) \= \E(-m\a\b)\,\sum_{\l\inn\Z+\a+\ell/2m}\E(2m\b\l)\,q^{m\l^2}  \ee
of weight 1/2 and its (again slightly renormalized) Eichler integral 
\be\label{defEichler1/2} 
  \mytheta^{s\,*}_{m,\ell}(\t) \= \frac{\overline{\mytheta^s_{m,\ell}(\t)}}{2\pi\sqrt{m\t_2}} \,-\, 
   \E(m\a\b)\!\sum_{\l\inn\Z+\a+\ell/2m} |\l|\,\E(-2m\b\l)\,\erfc\bigl(2|\l|\sqrt{\pi m\t_2}\bigr)\,q^{-m\l^2}  \ee
(cf.~\eqref{defstar} and \eqref{erfc}).  Then we can define the completion $\wh\BB_m^s$ of $\BB_m^s$ by 
\be \label{defGsmhat} \wh\BB^s_m(\t,z)\;:=\;\BB^s_m(\t,z)\;+\,m\,\sum_{\text{\rm $\ell$ (mod $2m$)}}
\mytheta^{s\,*}_{m,\ell}(\t)\,\vth_{m,\ell}(\t,z)\;.\ee

\begin{proposition}\label{Sander2} For $m\in\N$ and $s\in\Q^2$ the completion $\wh \BB_m^s$ of $\BB_m^s$ defined
by \eqref{defGsmhat} satisfies
\bal\label{GsmPeriodic}    \wh\BB^{(\a+\l,\b+\mu)}_m(\t,z) &\=\E(-m(\mu\a-\l\b+\l\m))\,\wh\BB^{(\a,\b)}_m(\t)\qquad(\l,\,\mu\in\Z)\,,   \\
  \label{GsmElliptic}       \wh\BB^s_m(\t,z+\l\t+\mu) &\= \E(-m(\l^2\t+2\l z))\,\wh\BB^s_m(\t) \qquad\quad(\l,\,\mu\in\Z)\,, \\
  \label{GsmModular}  \wh\BB^s_m\Bigl(\frac{a\t+b}{c\t+d},\,\frac z{c\t+d}\Bigr)\!&\= 
     (c\t+d)^2\;\E\Bigl(\frac{mcz^2}{c\t+d}\Bigr)\,\wh\BB^{s\gamma}_m(\t,z) \quad\;
    \bigl(\g=\Bigl(\begin{matrix} a&b\\c&d\end{matrix}\Bigr)\in SL(2,\Z)\bigl)\,.      \end{align}
\end{proposition}

The proof is exactly similar to that of~\ref{Sander}.  We define functions $g_u^{(m)}(z;\t)$ and $\tilde g_u^{(m)}(z;\t)$ by
applying the operator $\,\dfrac1{2\pi i}\dfrac\partial{\partial u}-2m\dfrac{u_2}{\t_2}\,$ to $f_u^{(m)}(z;\t)$ and $\tilde f_u^{(m)}(z;\t)$;
then the transformation equation~\eqref{ZwegersMod} of   $\tilde f_u^{(m)}$ implies the same transformation equation for  $\tilde g_u^{(m)}$ but with
the initial factor $(c\t+d)$ replaced by $(c\t+d)^2$, and a calculation exactly similar to the one given before shows that $\E(m\a z_s)\BB_m^s(\t,z)$
differs from $g_{z_s}^{(m)}(z;\t)$ by a finite linear combination of functions $\vth_{m,r}(\t,z)$ and that 
$\E(m\a z_s)\wh\BB_m^s(\t,z)=\tilde g_{z_s}^{(m)}(z;\t)$.  We omit the details.

\begin{thm}\label{Thmdouble}  
 Let $\v$ be as above, with singularities at $z=z_s$ $(s\in S\subset\Q^2)$ given by~\eqref{defDE}.  Then the finite part $\v^F$
as defined by~\eqref{NewFinitePart} coincides with the finite part defined by the theta expansion~\eqref{phiF}, the coefficients $h_\ell(\t)$
in this expansion are mixed mock modular forms, with completion given by
 $$  \wh h_\ell(\t) \=   h_\ell(\t)\+\sum_{s\in S/\Z^2} \Bigl(D_s(\t)\,\Th^{s\,*}_{m,\ell}(\t) \+ E_s(\t)\,\mytheta^{s\,*}_{m,\ell}(\t)\Bigr)\,,$$
and the completion $\wh\v^F$ defined by~\eqref{defphiFhat} transforms like a Jacobi form of weight~$k$ and index~$m$.\end{thm}

\smallskip
The proof follows the same lines as before: the equivalence of~\eqref{phiF} and~\eqref{NewFinitePart} is proved by expanding $\v(\t,z)$ as a Fourier
series along the horizontal line $\Im(z)=\Im(P)$ for some generic point $P\in\C$ and calculating the difference $\v-\v^F$ as
a sum of residues, and the mock modularity is proved by decomposing $\v$ as $\wh\v^F+\wh\v^P$ with $\wh\v^F$ as in~\eqref{defphiFhat}  and
$\wh\v^P=\sum_{s\in S/\Z^2}\bigl(D_s\wh \AA_m^s+E_s\wh\BB_m^s\bigr)$, which transforms like a Jacobi form by virtue of Proposition~\ref{Sander2}.
Again the details are left to the reader.  Note that here the mock modular forms $h_\ell(\t)$ are of the ``even more mixed" variety mentioned
at the end of \S\ref{MockJacobi}, since they now have a shadow that is a linear combination of two terms $\sum_sD_s\,\overline{\Th_{m,l}^s}$ and
$\sum_sE_s\,\overline{\mytheta_{m,l}^s}$ belonging to two different tensor products $M_{k-1}\otimes\overline{M_{3/2}}$ and $M_{k-2}\otimes\overline{M_{1/2}}$
and hence two different mixed weights $k-\frac12\,|\,k-1$ and $k-\frac12\,|\,k-2$ rather than a single mixed weight $k-\frac12\,|\,k-1$
as in the case of simple poles.

\subsection{Examples \label{MJExamples}}  We end this section by giving five examples of meromorphic Jacobi forms
and their canonical decompositions into a mock Jacobi form and a finite linear combination of Appell-Lerch sums.  We  systematically
use the notations $A=\v_{-2,1}$, $B=\v_{0,1}$, $C=\v_{-1,2}$ for the three basic generators of the ring of weak Jacobi forms
as described in equations~\eqref{Fourirphi02}--\eqref{WJF}.
 
\medskip
\noindent{\bf Example 1: Weight 1, index 1, simple pole at $z=0$.}  As our first example we take the Jacobi form $\v=C/A\in J_{1,1}^{\text{mer}}$,
which has a simple pole of residue $1/\pi i$ at $z=0$ and a Fourier expansion beginning
$$\frac{y+1}{y-1}\,-\,(y^2-y^{-2})\,q\,-\,2(y^3-y^{-3})\,q^2\,-\,2(y^4-y^{-4})\,q^3 \,-\,(2y^5+y^4-y^{-4}-2y^{-5})\,q^4\,-\,\cdots\;.$$  
Calculating the Fourier expansion of the polar part $\v^P=\Av1{\frac{y+1}{y-1}}$, we find that it begins the same way, and indeed, we must have
$\v=\v^P$ because the Fourier coefficients $h_\ell$ all vanish (we have $h_{-\ell}=-h_\ell$ because the weight is odd and $h_{\ell+2}=h_\ell$
because the index is~1).  So here there is no mock Jacobi form at all, but only the polar correction term given by the Appell-Lerch sum, a
kind of a ``Cheshire cat" example which is all smile and no cat.

\medskip \noindent
{\bf Example 2: Weight 1, index 2, simple pole at $z=0$.} As a second example take $\v=BC/A\in J_{1,2}^{\text{mer}}$.  Here we find
\ben \v^P&\=&12\,\AvB2{\frac{y+1}{y-1}}\=12\,\frac{y+1}{y-1}\,-\,12(y^4-y^{-4})q^2\+24(y^5-y^{-5})\,q^3\+\cdots \\
\v^F&\=&(y-y^{-1})\,-\,(y^3+45y-45y^{-1}-y^{-3})q\+(45y^3-231y+231y^{-1}-45y^{-3})q^2\+\cdots \, , \een
and by comparing the shadows and the first few Fourier coefficients (cf.~eq.~\eqref{defCF2}), we see 
that $\v^F$ is the negative of the mock Jacobi form $\CF_2(\t,z)=h^{(2)}(\t)\,\vth_ 1(\t,2z)$ 
discussed in Example~4 of \S\ref{MockJacobi} that is conjecturally related to representations of~$M_{24}\,$. 
Using~\eqref{phiminus1} and dividing by $\vth_1(\t,2z)$, we can rewrite the identity $\v=\v^P+\v^F$ as
  \be \label{CoverA} 
  \eta(\t)^{-3}\,\frac{\v_{0,1}(\t,z)}{\v_{-2,1}(\t,z)} \= -\,\frac{12}{\vth_1(\t,2z)}\,\AvB2{\frac{1+y}{1-y}}\,-\,h^{(2)}(\t)\,. \ee
One can check that $\vth_1(\t,2z)\mu(\t,z)=\Av2{\frac{1+y}{1-y}}$, so this agrees with the splitting~\eqref{K3decomp} coming 
from the decomposition of the elliptic genus of $K3$ into superconformal characters.

\medskip \noindent
{\bf Example 3: Weight 1, index 1/2, simple pole at $z=0$.} We can also multiply~\eqref{K3decomp} by $\vth_1(\t,z)$ and write it in the form
  \be \label{BoversqrtA}
  \frac B{\sqrt A} \= -12\,\AvB{\frac12,-}{\frac{\sqrt y}{1-y}} \;-\;h^{(2)}(\t)\,\vth_1(\t,z)\,, \ee
where $\,\text{Av}^{(\frac12,-)}[F(y)]:=\sum_{n\in\Z}(-1)^nq^{n^2/2}\,y^n\,F(q^ny)$.  This decomposition also fits into the scheme explained in
this section, but slightly modified because the index here is half-integral.

\medskip \noindent
{\bf Example 4: Weight 2, index 0, double pole at $z=0$.} Now take $\v=B/A$, which, as we saw in \S\ref{MJF1}, is just 
a multiple of the Weierstrass $\wp$-function.This example does not quite fit in with 
the discussion given so far, since we did not allow $m=0$ in the definition~\eqref{defAvm} of the averaging operator 
($m<0$ doesn't work at all, because the sum defining the average diverges, and $m=0$ is less interesting since a form
of index~0 is clearly determined up to a function of $\t$ alone by its singularities, so that in our discussion we 
excluded that case too), but nevertheless it works.  The decomposition $\v=\v^P+\v^F$ takes the form
  \be \frac{\v_{0,1}(\t,z)}{\v_{-2,1}(\t,z)} \= 12\,\AvB0{\frac{y}{(1-y)^2}}\+ E_2(\t)\,,\ee
with the finite part being simply the quasimodular Eisenstein series $E_2(\t)$, which
is also a mock modular form. (It is automatically independent of~$z$ since the index is~0.)

\medskip \noindent
{\bf Example 5: Weight 2, index 1, double pole at $z=0$.}  The next example is a special case of an infinite family
that will be treated in detail in \S\ref{flyJac}.  Take $\v=B^2/A\in J_{2,1}^{\text{mer}}$. Then we find
\ben
\v &\=& \Bigl(\frac{144\,y}{(1-y)^2}\+y+22+y^{-1}\Bigr) \+ \Bigl(22y^2+152y-636+152y^{-1}+ 22y^{-2}\Bigr)\,q \\
&&  \+\Bigl(145y^3-636y^2+3831y-7544+3831y^{-1}-636y^{-2}+145y^{-3}\Bigr)\,q^2 \+\cdots \, , \\ 
\v^{\rm P} &\=& \AvB 1{\frac{144\,y}{(1-y)^2}} \= \frac{144\,y}{(1-y)^2} \+ 0\,q \+ 144\,(y^3 + y^{-3})q^2 \+ 288\,(y^4+y^{-4})q^3 \+ \cdots 
\een
and hence that $\v^F=\v-\v^P=\sum\limits_{4n-r^2\ge-1} C(4n-r^2)\,q^n\,y^r$ with the first few $C(\DD)$ given by 
$$ \begin{tabular}{|c|ccccccccc|}   \hline    $\DD$ & $-1$& 0&3 & 4 & 7& 8 & 11 & 12 & 15  \\
       \hline  $C(\DD)$ &1 & 22 & 152 & $-636$ & 3831& $-7544$ & 33224 & $-53392$ & 191937\\  \hline  
  \end{tabular}  \;\, . 
$$
On the other hand, we have the weak Jacobi form $E_4(\t)\v_{-2,1}(\t,z)=\sum C^*(4n-r^2) q^n y^r$ with 
$$
\begin{tabular}{|c|ccccccccc|}   \hline  $\DD$ & $-1$& 0&3 & 4 & 7& 8 & 11 & 12 & 15 \\ 
 \hline  $C^*(\DD)$ &1 & $-2$ & 248 & $-492$ & 4119 & $-7256$ & 33512& $-53008$ & 192513  \\  \hline  
\end{tabular}  \;\, , 
$$
We see that $C(\DD)$ and $C^*(\DD)$ are very close to each other (actually, asymptotically the same) 
and that their difference is precisely $-288$ times the Hurwitz-Kronecker 
class number $H(d)$. We thus have $\v^F= E_4A -288\CH$, and we have found~$\CH(\t,z)$, the simplest of
all mock Jacobi forms, as the finite part of the meromorphic Jacobi form $(E_4A^2-B^2)/288A$.  We
remark that the numbers $-C^*(\DD)$ have an interesting interpretation as ``traces of singular moduli"~\cite{Zagier:2002}.

\medskip \noindent
{\bf Example 6: Weight $-$5, index 1, simple poles at the 2-torsion points.} 
Finally, we give an example with more than one pole. Take $\v=A^3/C\in J_{1,-5}$.  This function has three poles, all
simple, at the three non-trivial 2-torsion points on the torus $\C/(\Z\t+\Z)$.  The three corresponding modular forms,
each of weight~$-6$, are given by 
$$ D_{(0,\frac12)}(\t) \= 16\,\frac{\eta(2\t)^{12}}{\eta(\tau)^{24}}\,, \quad
   D_{(\frac12,0)}(\t) \= -\frac14\,\frac{\eta\bigl(\frac\t2\bigr)^{12}}{\eta(\tau)^{24}}\,, \quad
   D_{(\frac12,\frac12)}(\t) \= \frac i4\,\frac{\eta\bigl(\frac{\t+1}2\bigr)^{12}}{\eta(\tau)^{24}} \, , $$
and one finds 
\ben 
&&\v \= \v^{\rm P} \= D_{(0,\frac12)}(\t)\,\AvB1{\frac12 \, \frac{y-1}{y+1}}
  \+ q^{1/4}\, D_{(\frac12,0)}(\t)\,\AvB1{\frac 1{2y} \, \frac{y+\sqrt q}{y-\sqrt q}} \\
 &&\qquad\qquad\quad  \+ q^{1/4}\, D_{(\frac12,\frac12)}(\t)\,\AvB1{\frac 1{2y} \, \frac{y-\sqrt q}{y+\sqrt q}} \,,\qquad \v^F\;\equiv\;0\,, \een
another ``Cheshire cat" example (of necessity, for the same reason as in Example~1,
since again $m=1$ and $k$ is odd).

\section{Special mock Jacobi forms \label{flyJac}}

We now turn to the study of certain families of meromorphic Jacobi forms of weights 1 or 2 and arbitrary 
positive index and of their related mock Jacobi forms, the weight 2 case being the one relevant for the application to 
black hole microstate counting. In this section, we introduce these families and formulate a nunmber of properties
of them that were found experimentally and that we will prove in~\S\ref{structure}.  In particular, it will turn
out that all of these functions can be expressed,
using the Hecke-like operators introduced in \S\ref{hecke}, in terms of a collection of mock Jacobi forms $\CQ_{m}$
defined (though not quite uniquely: there is some choice involved starting at $m = 91$) for all square-free~$m$, 
of weight either~1 or 2 depending on whether~$m$ has an odd or even number of prime factors.  These functions, which
seem to have particularly nice arithmetic properties, include all of the examples studied in \S\ref{Mock}, 
and several other special mock modular forms that have appeared in the literature.  

\subsection{The simplest meromorphic Jacobi forms \label{simplest}}

The two simplest situations for the theory described in \S\ref{MockfromJacobi} are when the meromorphic 
Jacobi form $\v=\v_{k,m}$ satisfies either 
\newline\noindent\null\quad
$\bullet$ $k=1$ and $\v(\t,z)$ has a simple pole at $z=0$ and no other poles in $\IC/(\IZ\t + \IZ)$ or 
\newline\noindent\null\quad
$\bullet$ $k=2$ and $\v(\t,z)$ has a double pole at $z=0$ and  no other poles in $\IC/(\IZ\t + \IZ)$, \\
since in these two cases  the modular forms $D_{s}(\t)$ (resp.~$E_{s}(\t)$) defined in \S\ref{PolarPart} (resp.~\S\ref{double}) 
are simply constants and the canonical Fourier coefficients $h_{\ell}$ of $\varphi$ are pure, rather than mixed, mock modular forms.
If we normalize $\v$ by $\v(\t,z) \sim (2\pi i z)^{-1}$ in the first case and $\v(\t,z) \sim (2\pi i z)^{-2}$ in the second case 
as $z \to 0$, then any two choices of $\v$ differ by a (weakly) holomorphic Jacobi form whose Fourier coefficients 
are true modular forms and are therefore well understood. It therefore suffices in principle to make some specific
choice of $\v_{1,m}$ and $\v_{2,m}$ for each $m$ and to study their finite parts $\v_{k,m}^{\rm F}$, but, as we shall
see, the problem of choosing a good representative $\v$ is very interesting and has several subtleties. 
For any choice of $\v_{k,m}$, we will use the notation~$\Phi_{k,m}:=\v_{k,m}^{\rm F}$, with decorations like 
$\Phi_{k,m}^{\rm stand}$ or $\Phi_{k,m}^{\rm opt}$ corresponding to the choice of $\v_{k,m}$. 
The shadow of the mock Jacobi form $\Phi_{k,m}$  is independent of the choice of representative and is a multiple of 
$\sum_{\ell\mypmod{2m}}\overline{\vth_{m,\ell}^{2-k}(\t)}\,\vth_{m,\ell}(\t,z)$,
where $\vth_{m,\ell}^1$ and $\vth_{m,\ell}^0$ are defined as at the end of \S\ref{ThetaTaylorExp}.

The Jacobi forms $\v_{2,m}$ are the ones related to the Fourier coefficients $\psi_m$ of 
$\Phi_{10}(\O)^{-1}$ that appear in the application to black holes (eq.~\eqref{reciproigusa}). 
Indeed, since the principal part of the pole of $\psi_{m}(\t,z)$ at $z=0$ equals $p_{24}(m+1)\,\Delta(\t)^{-1}\,(2\pi iz)^{-2}$,
we can write $\Delta(\t)\psi_m(\t,z)$ as the sum of $p_{24}(m+1)\,\v_{2,m}$ and a weak Jacobi form.  
If we make the simplest choice for $\v_{2,m}$, namely 
  \be\label{defphi2m}  \vst_{2,m} \= \frac{B^{m+1}}{12^{m+1}A} \,, \ee
where $ A = \v_{-2,1}$, $B = \v_{0,1}$ are the special Jacobi forms introduced in \S\ref{Jacobi} 
(which has the correct principal part because $A(\t,z) \sim (2\pi i z)^2$ and $B(\t,z) \to 12$ 
as $z \to 0$), then the first few values of  $\D(\t)\,\psi_{m}(\t,z)-p_{24}(m+1)\,\vst_{2,m}(\t,z)$ can be read off from~\eqref{firstpsim}. For instance, we have $p_{24}(3)=3200$ and the fourth equation of~\eqref{firstpsim} says that 
$\D\psi_2=3200\,\vst_{2,2}+\frac{16}9E_4AB+\frac{10}{27}E_6A^2$.

In the weight one case we can also make a simple standard choice of $\v$.  It has to be divisible
by $C$ because the weight is odd, and since $C\sim4\pi iz$ for $z\to0$ we can choose
  \be\label{defphi1m}  \vst_{1,m} \= \frac{B^{m-1}\,C}{2\cdot 12^{m-1}\,A}\= \frac C2 \, \v^{\rm stand}_{2,m-2}\;.   \ee
as our standard functions. However, we emphasize that \eqref{defphi2m} and \eqref{defphi1m} are only the starting points 
for our discussion and will be replaced by better choices as we go along.

We will apply the decomposition of \S\ref{MockfromJacobi} to the functions $\v$.  The polar part is independent of the 
choice of $\v$ and is given by $\v_{k,m}^{\rm P}=\AP{k,m}$, 
where $\AP{k,m}$ is given
(in the notation of equations~\eqref{defFms} and~\eqref{defGms}) by
  \be\label{AP1AP2}  \AP{1,m} \= \AA_{1,m}^{(0,0)} \= \AvB m{\frac12 \, \frac{y+1}{y-1} } \, , \qquad   \AP{2,m} \= \AA_{2,m}^{(0,0)}
  \= \AvB m{ \frac{y}{(y-1)^2} }  \ee
or---written out explicitly---by
  \be\label{AP1AP2again}
  \AP{1,m}(\t,z) \=  -\frac12 \, \sum_{s \in \mathbb{Z}} q^{ms^2\,}y^{2ms} \, \frac{1 +q^s y}{1 -q^s y} \, ,  \qquad 
  \AP{2,m}(\t,z) \=  \sum_{s \in \mathbb{Z}} \frac{q^{ms^2 +s\,}y^{2ms+1}}{(1 -q^s y)^2} \, .  \ee
Also independent of the choice of $\v$ is the correction term which must be added to the finite 
part and subtracted from the polar part of $\v_{2,m}$ in order to make these functions transform
like Jacobi forms. In particular, for the double pole case, we find from \eqref{defGsmhat} that the completion 
$\wh \Phi_{2,m}$ of $\Phi_{2,m} = \v^{\rm F}_{2,m}$ for any choice of $\v_{2,m}$ is given by 
\be\label{phi2mhat}
\wh \Phi_{2,m}(\t,z) \= \Phi_{2,m}(\t,z) \, - \, \sqrt{\frac{m}{4\pi}} \, \sum_{\ell \mypmod{2m}} 
\big( \vth_{m,\ell}^{0} \big)^{*} (\t)\, \vth_{m,\ell} (\t,z) \, , 
\ee
where $\vth_{m,\ell}^{0}(\t) = \vth_{m,\ell}(\t,0)$ and  
$(\vth_{m,\ell}^{0})^{*}$ is normalized\footnote{Note that the function $\vth_{m,\ell}^{(0,0)\,*}$ 
as defined in \eqref{defEichler1/2} is equal to $1/\sqrt{4\pi m}$ times $(\vth_{m,\ell}^{0})^{*}$.}
as in equations~\eqref{defstar} and~\eqref{erfc}.

The functions $\AP{1,m}$, $\AP{2,m}$ have wall-crossings, so their Fourier expansions 
depend on the value of $z_2 = \Im(z)$. In the range $0 < z_2 < \t_2$, or $|q| < |y| < 1$, they are given by 
\bea  \label{AAm}
\AP{1,m} \=  \biggl(- \sum_{s \ge 0}  \sideset{}{^{*}} \sum_{\ell \ge 0} \+  \sum_{s < 0} \sideset{}{^{*}} \sum_{\ell \le 0} \biggr) \, 
q^{ms^2 + \ell s} \, y^{2ms+\ell} \, , \\ \label{BBm}
\AP{2,m} \=  \biggl(\sum_{s \ge 0}  \sideset{}{^{*}} \sum_{\ell \ge 0} \, - \, \sum_{s < 0}  \sideset{}{^{*}} \sum_{\ell \le 0} \biggr) \, 
\ell \, q^{ms^2 + \ell s} \, y^{2ms+\ell} \, , 
\eea
where the asterisk on the summation sign has the usual meaning (count the term $\ell=0$ with multiplicity $1/2$).
Hence the coefficient of $q^{n} y^{r}$ is zero unless $\DD=4mn-r^2$ is the negative of the square of an 
integer congruent to $r$ modulo $2m$, in which case it is $0$, $\pm \frac12$, $\pm1$ or $\pm 2$ if $k=1$ and is
0, $\pm \sqrt{-\DD}$ or $\pm 2 \sqrt{-\DD}$ if $k=2$. (Compare the remark at the end of \S\ref{MockDecomposition}.)
If we expand $\AP{1,m}$ and $\AP{2,m}$ in Fourier series for $z_2$ in some other interval $(N,N+1)\t_2$, then the formulas~\eqref{AAm}
and~\eqref{BBm} are unchanged except for a shift in the inequalities defining the summation.  In summary, the Fourier coefficients
of the polar part of $\v$, although not quite like the coefficients of a usual Jacobi form---they do not only depend 
on $\DD$ and on $r$ mod $2m$, they are non-zero for infinitely many negative $\DD$, and there is wall-crossing---are
extremely simple and have very slow growth (at most like $\sqrt{|\DD|}$). 

The interesting part of the Fourier expansion of $\v_{k,m}$ is therefore the mock part~$\Phi_{k,m}$. Here the choice 
of $\v_{k,m}$ makes a big difference. For instance, Examples 1--4 of  \S\ref{MJExamples} say that the mock Jacobi forms $\vstF_{k,m}$
corresponding to the standard choices~\eqref{defphi2m} and~\eqref{defphi1m} for small~$m$ are
\be
\vstF_{1,1} \= 0, \quad\; \vstF_{1,2} \= \frac1{24} \, \CF_2 \,, \quad\; \vstF_{2,0} \= \frac 1{12} \, E_2 \, ,
\quad\; \vstF_{2,1} \= \frac 1{144} \, E_4\,A \,-\, 2\,\CH \,,
\ee
 and we see that this choice is a good one in the first three cases (indeed, in these cases there {\it is} no choice!) 
but that, as we have already observed, a better choice of $\v$ in the fourth case would have been the function 
$\PHI_{2,1}=\vst_{2,1} - \frac 1{144} E_4 A$, with mock part $\PHIF_{2,1}=-2 \CH$,
because it has no terms $q^{n} y^{r}$ with discriminant $\DD=4n-r^2$ smaller than zero and (consequently) 
has Fourier coefficients that grow polynomially in~$\DD$ (actually, like $\DD^{1/2+\varepsilon}$) rather than
exponentially in~$\sqrt{\DD}$. Our next main objective, to which we devote the next subsection, 
is therefore to replace our initial choice of $\v_{k,m}$
as given by \eqref{defphi2m} or \eqref{defphi1m} by an optimal choice $\PHI_{k,m}$, where
``optimal'' means ``with the negative discriminants $d=4mn-r^2$ that occur being as little negative 
as possible'' or, equivalently, ``with Fourier coefficients growing as slowly as possible.'' 
This choice turns out  to have Fourier coefficients of very much 
smaller growth than those of our initial choice $\vst_{k,m}$ in all cases, and to be unique in many cases.

\subsection{Choosing the function $\v$ for weight 2 and small index \label{choosephi}}
We now consider in more detail the meromorphic Jacobi forms of weight 2 and index $m$
with a double pole at the origin introduced in~\S\ref{simplest}. Any such form will have the form 
\be
\v_{2,m} \= \frac{\psi_{0,m+1}}{12^{m+1}A}\,, \qquad
\psi_{0,m+1} \;\, \in \;\,  B^{m+1} \+ A\cdot \wt J_{2,m} \;\,\subset\;\, \wt J_{0,m+1} \;.
\ee
For each $m$, we want to choose a meromorphic Jacobi form $\PHI_{2,m}$ differing from $\vst_{2,m}$
by a weak Jacobi form of weight~2 and index~$m$ whose coefficients grow as slowly as possible.  
We are interested both in the form of the corresponding~$\psi_{0,m+1}=\psi_{0,m+1}^{\rm opt}$ as a 
polynomial in $A$ and $B$ with modular form coefficients, and in the nature of the finite 
part~$\Phi_{2,m}^{\rm opt} = \PHI_{2,m} - \AP{2,m}$ as a mock Jacobi form. 
In this subsection we will look at the optimal choice for $0\le m \le 11$ to see what 
phenomena occur, while in the next subsection, we report on the calculations for higher values up to
about~250 and formulate the general results. 

For {\bf m=0}, we have to choose 
  $$ \PHI_{2,0} \= \vst_{2,0} \= \frac{B}{12A} \; , \qquad \psi^{\rm opt}_{0,1} \= B  $$
because there are no weak Jacobi forms of weight~2 and index~0. The mock part $\Phi^{\rm opt}_{2,0}$
here is $E_2(\t)/12$, as we have already seen (Example~3 of \S\ref{MJExamples}). 

For {\bf m=1} the only freedom is the addition of a multiple of $E_4 A$ and (as recalled in \S\ref{simplest}) 
the optimal choice is 
\be\label{PHI21} 
 \PHI_{2,1} \= \vst_{2,1}\,-\,\frac{E_4 A}{144}\,,\qquad \psi^{\rm opt}_{0,2}\=B^2\,-\,E_4 A^2\,, 
\ee
because then $\PHIF_{2,1}=-2\CH$ has no terms $q^{n} y^{r}$ with $4n-r^2 <0$  and its Fourier coefficients 
$-2 H(4n-r^2)$ have only polynomial growth as $4n-r^2 \to \infty$. 

Now we look at {\bf m=2}. Here we can modify $12^3\,\Phi^{\rm stand}_{2,2} = B^3/A - 12^3 \AP{2,2}$ 
by adding to it any linear combination of the two weak 
Jacobi forms $E_4 A B$ and $E_6 A^2$. Each of these three functions has Fourier coefficients of the form 
$c(n,r) = C_{\v}(8n-r^2)$  (because the $(n,r)$-coefficient of a Jacobi form of prime or prime power index $m$ and 
even weight always depends only on the discriminant $4mn-r^2$), with the first few given by 
  \begin{center} \begin{tabular}{|c|ccccccccc|ccc}  
\hline
       $\DD$ & $-4$& $-1$& $0$ & $4$ &$7$& $8$& $12$ & $15$ &$16$ \\  \hline
       $C(12^3\Phi^{\rm stand}_{2,2};\DD)$ &1&32& $366$ & $528$ &1056 &$-10144$&$-14088$&$92832$&$-181068$\\
      $C(E_4AB;\DD)$ &1 &8& $-18$  & 144 &$2232$&$-4768$&$-22536$&96264&$-151884$\\
      $C(E_6A^2;\DD)$ &1 & $-4$ & 6  & $-480$ &$1956$&$-2944$&$-28488$&96252&$-139452$\\
\hline
       \end{tabular} \end{center}

\ndt It follows from this table that the linear combination 
\be\label{defX2} 
\PHI_{2,2} \=  \vst_{2,2} \, - \, \frac 1{576} \, E_4 A B \+ \frac 1{864} E_6 A^2 
 \= \frac{B^3 \, - \, 3 E_4  A^2 B \+ 2 E_6 A^3 }{12^3 \, A} 
\ee
has a mock Jacobi part  $ \Phi^{\rm opt}_{2,2} =  \PHI_{2,2} - \AP{2,2}$
with no negative terms and much smaller coefficients:
\begin{center}
\begin{tabular}{c|cccccccccccc}  
       $d$ & $-4$& $-1$& $0$ & $4$ &$7$& $8$& $12$ & $15$ &$16$ \\
       \hline       $C(\Phi^{\rm opt}_{2,2};\DD)$ &0&0& 1/4 & $-1/2$ &$-1$ &$-1$&$-2$& $-2$& $-5/2$ \\
        \end{tabular}  
\end{center}
\ndt On inspection of these numbers and the next few values, we recognize them as simple linear combinations of 
Hurwitz class numbers, namely:
  $$  C(\Phi^{\rm opt}_{2,2}; \DD) =  \begin{cases} - H(\DD) \quad  &  {\rm if} \, \DD \equiv 7 \mypmod 8 \, ,  \\ 
  -  H(\DD) - 2 H(\DD/4)  \quad &   {\rm if} \, \DD \equiv 0 \mypmod 4 \, ,  \end{cases}   $$
which says that $\Phi^{\rm opt}_{2,2}$ is simply $-\CH|V_2$, with $\CH$ as in \eqref{classH}  and $V_2$ as in \eqref{defVl}.
(Note, by the way, the similarity  between \eqref{defX2} and \eqref{weierstrass}. Together they say that 
$\PHI_{2,2} = (C/2A)^2$.)

The cases {\bf m=3}, {\bf m=4}, {\bf m=5}, {\bf m=7}, {\bf m=8}, {\bf m=9} and {\bf m=11} are similar to the case $m=2$. For each
of these values of~$m$ the coefficient of $q^ny^r$ in $\vstF_{2,m}$ depends only on $\DD=4nm-r^2$, the number of negative $\DD$ for
which this coefficient is non-zero equals the dimension of the space of weak Jacobi forms of weight~2 and index~$m$, and we
can subtract a unique weak Jacobi form from $\vst_{2,m}$ to obtain a meromorphic Jacobi form $\PHI_{2,m}$ whose finite part is
a strongly holomorphic mock Jacobi form with very small Fourier coefficients.  For example, for $m=3$ the optimal choice of
$\PHI_{2,m}$ is $\PHI_{2,3}=\psi^{\rm opt}_{0,4}/12^4A$ with 
  \be\label{PSIF4} \psi^{\rm opt}_{0,4}   \= B^4\,-\, 6 E_4 A^2 B^2 \+ 8 E_6 A^3 B \,-\, 3 E_4^2 A^4 \ee
and the first few Fourier coefficients of $12^4\,\vstF_{2,3}$ and $\PHIF_{2,3}$ are given by
  \begin{center} \begin{tabular}{|c|ccccccccc|ccc}   \hline
       $\DD$ & $-9$& $-4$& $-1$& $0$ & $3$ &$8 $& $11$& $12$ & $15$  \\  \hline
       $12^4\,C(\vstF_{2,3};\DD)$ & 1 & 42 & 687 & 5452 & 1104 & 8088 & $-1488$ & $-139908$ & $-19629$\\
      $C(\PHIF_{2,3};\DD)$ & 0 & 0 & 0  & $1/3$ & $-1/3 $& $-1$  & $-1$ & $-4/3$ & $-2$ \\
       \hline \end{tabular} \end{center}
Again we recognize the coefficients $C(\PHIF_{2,3};\DD)$ as $-H(\DD)$ if $9\!\nmid\!\DD$ and $-H(\DD)-3H(\DD/9)$ 
if $9\!\mid\!\DD$, so that
just as for $m=2$ we find that the mock part of $\PHI_{2,3}$ equals $-\CF_ 1|V_3$.  The same thing also happens for $m=5$, 7 and 11, 
suggesting that for prime values of $m$ we can always choose
 \be\label{PHI2p}    \PHIF_{2,p} \= -\,\CH|V_p  \qquad \text{($p$ prime).} \ee
This will be proved in \S\ref{structure}.  
For the prime-power values $m=4$, 8 and 9 this formula no longer holds, but the
coefficients of $\PHIF_{2,m}$ are still simple linear combinations of Hurwitz-Kronecker class numbers and we find the expressions
$$  \PHIF_{2,4}=-\CH|V_4 + 2 \CH|U_2\,, \quad \PHIF_{2,8}=-\CH|V_8 + 2 \CH|U_2V_2\,,  \quad \PHIF_{2,9}=-\CH|V_9 + 3 \CH|U_3\,, $$
where $U_t$ is the Hecke-like operator defined in \S\ref{hecke} that multiplies the index of a Jacobi form (or mock Jacobi form) by~$t^2$.
This suggests that for prime power values of $m$ we can choose
 \be\label{PHI2pp}    \PHIF_{2,m} \= -2\,\CH|\,\CV^{(1)}_{2,m}  \qquad \text{($m$ a prime power)} \ee
where $\CV^{(1)}_{k,t}$ for $t\in\N$ is the Hecke-like operator from Jacobi (or weakly holomorphic Jacobi, or mock Jacobi) forms of weight~$k$ and index~$1$ to forms of weight~$k$ and index~$t$ defined 
in \eqref{Vkt1}. 
Formula \eqref{PHI2pp} will also be proved in \S\ref{structure}. 

The data for the values discussed so far ($0\le m\le11$, $m\ne6,\,10$) is summarized in the table below, which expresses
the weak Jacobi form $\psi^{\rm opt}_{0,m+1}=12^{m+1}A\,\PHI_{2,m}$ for each of these values of~$m$ as a polynomial in~$A$ and~$B$.
In this table we notice many regularities, e.g., the coefficients of $E_4A^2B^{m-2}$, $E_6A^3B^{m-3}$ and $E_4^2A^4B^{m-4}$  
in $\psi^{\rm opt}_{0,m}$ are always given by $-\bigl({m\atop2}\bigr)$, $2\bigl({m\atop3}\bigr)$ and $-3\bigl({m\atop4}\bigr)$, respectively,  
and more generally, if we write $\psi^{\rm opt}_{0,m}$ as $\sum_{i=0}^m f_{m,i} \, A^{i} \, B^{m-i}$ with 
$f_{m,i}=f_{m,i}(\t) \in M_{2i}$, $f_{m,0}=1$, then the constant terms $f_{m,i}(\infty)$ of the modular forms $f_{m,i}$ (which 
can be read off from the table by setting $E_4=E_6=1$ and $D=0$) are given in each case by 
$f_{m,i}(\infty)=(-1)^{i-1}(i-1)\bigl({m\atop i}\bigr)$ or equivalently by the generating function   
  \be\label{coeffCm} \sum_{i=0}^m  f_{m,i}(\infty) \, X^{m-i}\=  (X-1)^m \+m \, (X-1)^{m-1}  \= (X-1)^{m-1}(X+m-1)\;, \ee
e.g.~$X^8 - 28X^6 + 112X^5 - 210X^4 + 224X^3 - 140X^2 + 48X - 7 = (X-1)^7(X+7)$.  This formula has a simple interpretation: the
constant term with respect to~$q$ of $\sum f_{m,i}\,A^{i}\,B^{m-i}$ is $\sum f_{m,i}(\infty)(y^{-1}-2+y)^i(y^{-1}+10+y)^{m-i}$
and this Laurent polynomial---subject to the condition $f_{m,1}(\infty)=0$, which comes from the fact that there are no modular forms of weight~2 on $SL(2,\Z)$---has minimal growth $\,\text O\bigl(y+y^{-1}\bigr)\,$ precisely for the coefficients defined by~\eqref{coeffCm}. 

\vspace{0.2cm}

\begin{minipage}{\linewidth}
    \vspace{5pt}
  \begin{centering}
    \begin{tabular}{!{\thvline}   r  !{\thvline}   r   r   r   r   r   r   r   r   r  r  !{\thvline}      }
       \noalign{\thhline}
      $m$ & $1$  & $2$ & $3$  & $4$ & $5$& $6$ & $8$ & $9$ & $10$ & $12$ \\
       \noalign{\thhline}
      $B^m$ & 1 & $1$  & $1$ & $1$  & $1$ & $1$& $1$ &$1$& $1$ & $1$\\
       \hline
      $E_4 \cdot A^2 B^{m-2}$ & 0 & $-1$  & $-3$ & $-6$  & $-10$ & $-15$& $-28$& $-36$ &$-45$ & $-66$\\
       \hline
      $E_6 \cdot A^3 B^{m-3} $ & 0 & $0$  & $2$ & $8$  & $20$ & $40$ &$112$& $168$ &$240$ & $440$\\
       \hline
      $E_4^2  \cdot  A^4 B^{m-4}$ & 0 & $0$  & $0$ & $-3$  & $-15$ & $-45$& $-210$& $-378$ &$-630$ & $-1485$\\
       \hline
      $E_4 E_6 \cdot A^5 B^{m-5} $ & 0 & $0$  & $0$ & $0$  & $4$ & $24$& $224$& $504$ &$1008$ & $3168$\\
       \hline
      $E_4^3  \cdot A^6 B^{m-6} $ & 0 & $0$  & $0$ & $0$  & $0$ & $-5$&  $-140$& $-420$ &$-1050$&$-4620$\\
      $D  \cdot A^6 B^{m-6} $ & 0 & $0$  & $0$ & $0$  & $0$ & $1$& $10$& $30$ &$48$ & $222$\\
       \hline
      $E_4^2 E_6 \cdot A^7 B^{m-7} $ & 0 & $0$  & $0$ & $0$  & $0$ & $0$ &$48$& $216$ &$720$ & $4752$\\
       \hline
      $E_4^4 \cdot  A^8 B^{m-8}$ & 0 & $0$  & $0$ & $0$  & $0$ &  $0$ &$-7$& $-63$ &$-315$&$-3465$\\
      $E_4 D \cdot A^8 B^{m-8} $ & 0 & $0$  & $0$ & $0$  & $0$ &  $0$ &$-4$& $-36$ &$-18$&$-360$\\
       \hline
      $E_4^3 E_6 \cdot  A^9 B^{m-9}$ & 0 & $0$  & $0$ & $0$   & $0$& $0$ &$0$& $8$ & $80$ & $1760$\\
      $ E_6 D \cdot A^9 B^{m-9} $ & 0 & $0$  & $0$ & $0$  & $0$& $0$ &$0$& $20$ & $-16$ & $80$\\
       \hline
      $E_4^5  \cdot A^{10} B^{m-10} $ & $0$  & $0$ & $0$  & $0$ & $0$& $0$ &$0$& $0$ &$-9$ & $-594$\\
      $E_4^2 D \cdot A^{10} B^{m-10} $ & $0$  & $0$ & $0$  & $0$ & $0$& $0$ &$0$& $0$ &$18$&$702$\\
       \hline
      $E_4^4 E_6 \cdot  A^{11} B^{m-11}$ & $0$  & $0$ & $0$  & $0$ & $0$& $0$ &$0$& $0$ &$0$ & $120$\\
      $E_4 E_6 D \cdot A^{11} B^{m-11} $ &  $0$  & $0$ & $0$  & $0$ & $0$& $0$ &$0$& $0$ &$0$ & $-672$\\
       \hline
      $E_4^6 \cdot  A^{12} B^{m-12}$ & $0$  & $0$ & $0$  & $0$ & $0$& $0$ &$0$& $0$ &$0$ & $-11$\\
      $E_4^3 D \cdot A^{12} B^{m-12} $ &  $0$  & $0$ & $0$  & $0$ & $0$& $0$ &$0$& $0$ &$0$ & $188$\\
      $D^2 \cdot A^{12} B^{m-12} $ &  $0$  & $0$ & $0$  & $0$ & $0$& $0$ &$0$& $0$ &$0$ & $-8$\\
       \noalign{\thhline}
     \end{tabular}
   \end{centering}
    \captionof{table}{
            \small Coefficients of $\psi^{\rm opt}_{0,m}$ in terms of the standard basis. Here $D := 2^{11} 3^3 \Delta$.}
 \label{Cmcoeffs}
\end{minipage}
 
\vspace{0.6cm}
 

For {\bf m=6} several new things happen, all 
connected with the fact that~6 is the first value of~$m$ having more than one prime factor.
First of all, the formula \eqref{PHI2pp}, which worked for prime powers, now fails: 
$-2 \CH|\CV_{2,6}^{(1)}$
is not a possible choice for $\Phi_{2,6}$. The reason is that, as we will see later in general, the form
$\CH|\CV_{2,m}^{(1)}$ (for any~$m$) is invariant under all Atkin-Lehner involutions, but $\Phi_{2,m}$ is 
not when $m$ is not a prime power. The expression $-2\CH|\CV_{2,6}^{(1)}$   
does, however, give the optimal choice of the eigencomponent $\Phi_{2,6}^{++}$ of $\Phi_{2,6}$ with eigenvalues 
$W_{2} = W_{3} = +1$. The other eigencomponent, corresponding to $W_{2}=W_{3}=-1$, is obtained by subtracting a 
multiple of the non-weak Jacobi form $\v_{2,6}^{-}$ given in \eqref{phi26min} from $\pi_{2}^{-}(\PHIF_{2,6})$ 
(which is weakly holomorphic but not weak), and turns out to be simply a multiple of the mock Jacobi form 
$\CF_{6}$ defined in \S\ref{MockJacobi} (Example~3), giving as our first choice of the ``optimal'' mock modular form 
\be\label{PhiA26}
\Phi^{\rm I}_{2,6} \= -2 \, \CH\,|\,\CV_{2,6}^{(1)} \,-\, \frac{1}{24}\,\CF_{6} \; .
\ee

The second new phenomenon for $m=6$ is that requiring the order of the poles to be as small as possible 
no longer fixes a unique optimal choice. The function~$\Phi_{2,6}^{\rm I}$ defined above has an even 
part that is strongly holomorphic (because $\CH$ is) and an odd part that has discriminant~$\DD \ge -1$ 
(because~$C^{(6)}(\DD)=0$ for $\DD < -1$), but we can change this choice without disturbing 
the discriminant bound $\DD \ge -1$ by adding to it an arbitrary multiple of the function
\be \CK_{6} = \frac{E_4A B^5-5E_6A^2B^4+10E_4^2A^3B^3-10E_4E_6A^4B^2+(5E_4^3-\frac14D)A^5B-E_4^2E_6A^5}{12^5} \ee
which is (up to a scalar factor, fixed here by requiring $\CK_{6} = y -2 +y^{-1} + \text O(q)$) 
the unique Jacobi form of weight 2 and index 6 satisfying the same discriminant bound. 

Two particularly nice choices besides the function $\Phi_{2,6}^{\rm I}$ are the functions 
\be \Phi_{2,6}^{\rm II} \=  \Phi_{2,6}^{\rm I} \,-\, \frac1{24} \, \CK_{6} \, \qquad  \Phi_{2,6}^{\rm III} \= \Phi_{2,6}^{\rm I} \+ \frac14 \, \CK_{6} \ee
  
The special properties of these two choices are that $\Phi_{2,6}^{\rm II}$ is the unique choice with $c(\Phi;0,1)=0$
and (hence) has the form $12^{-7} \sum_{i=0}^7 f_{7,i} A^{i-1} B^{7-i}$ with $f_{7,i}$ satisfying 
\eqref{coeffCm}, while $\Phi_{2,6}^{\rm III}$ is the unique choice for which $c(n,r)=0$
whenever $24n=r^{2}$.  (This last choice is from some points of view the most canonical one.)
Each of these three special choices, as we shall see below, has analogues for all~$m$, although in general
none of them is canonical, since for large~$m$ we can change them by adding cusp forms without 
affecting their defining properties and we do not know any criterion that selects a unique form.

The following table, which gives the first few Fourier coefficients $C(\Phi;\DD,\,\text{$r\;$(mod 12)})$ for the Jacobi form $\CK_6$ 
and each of the five mock Jacobi forms $-2\CH|\CV_{2,6}^{(1)}$, $\CF_6$, $\Phi_{2,6}^{\rm I}$, $12\Phi_{2,6}^{\rm II}$ and 
$\Phi_{2,6}^{\rm III}$, summarizes the above discussion. In this table we have set in boldface the Fourier coefficients~$C(\D,r)$
whose special properties determine the choice of $\Phi$ (namely, $C(-1,\pm1)+C(-1,\pm5)=0$ for~$\Phi^{\rm I}$, $C(-1,\pm1)=0$ 
for $\Phi^{\rm II}$, and $C(0,0)=0$ for $\Phi^{\rm III}$).
\begin{table}[h]
 \begin{center} \begin{tabular}{!{\thvline}c!{\thvline}cc|c|c|c|c|c|cc|c|c|c|c!{\thvline}}  
 \noalign{\thhline}
       $\DD$ & $-1$& $-1$& $0$& $8$ & $12$ & $15$ & $20$ & $23$ & $23$ & 24 & 32 & 36  & 39 \\  
       $r\!\pmod{12}$ & $\pm1$ & $\pm5$ & $0$ & $\pm4$ & $6$ & $\pm3$ & $\pm2$ & $\pm1$ & $\pm5$ & $0$ & $\pm4$ & $6$ & $\pm3$ \\  \noalign{\thhline}
       $\CK_6$ & $1$ & $1$ & $-2$ & $10$ & $28$ & $10$ & $-12$ & $13$ & $13$ & $-44$ & $58$  & $104$ & $2$ \\ 
       $-2\CH\,|\,\CV_{2,6}$ & $0$& $0$& $\fr12$ & $-\fr12$ & $-1$ & $-1$ & $-1$ & $-\fr32$ & $-\fr32$ & $-1$ & $-\fr52$ & $-2$ & $-2$ \\  
       $\CF_6$ & $-1$ & $1$& $0$& $0$ & $0$ &$0$& $0$& $35$ & $-35$ & $0$ & $0$  & $0$ & $0$ \\        \hline
     $\Phi_{2,6}^{\rm I}$ & $\fr{\bold 1}{\bold 2\bold 4}$ & $\bold{-}\fr{\bold 1}{\bold 2\bold 4}$ 
          & $\fr12$ & $-\fr12$ & $-1$ & $-1$ & $-1$ & $-\fr{71}{24}$ & $-\fr1{24}$ & $-1$ & $-\fr52$  & $-2$ & $-2$ \\ 
     $12\,\Phi_{2,6}^{\rm II}$ & $\bold 0$ & $-1$ & $7$ & $-11$ & $-26$ & $-17$ & $-6$ & $-42$ & $-7$ & $10$ & $-59$  & $-76$ & $-25$ \\ 
     $\Phi_{2,6}^{\rm III}$ & $\fr7{24}$ & $\fr5{24}$ & $\bold 0$ & $2$ & $6$ & $\fr32$ & $-4$ & $\fr7{24}$ & $\fr{77}{24}$ & $-12$ & $12$  & $24$ & $-\fr32$ \\ 
 \noalign{\thhline}
    \end{tabular}\quad \end{center}
\label{wt6tab}  \caption{\small{Jacobi and mock Jacobi forms of weight~2 and index~6}} 
\end{table}
We mention in passing that the function $\psi_{0,7}^{\rm I} = 12^{7} A \v_{2,6}^{\rm I}$ corresponding to \eqref{PhiA26} 
is precisely equal to $B$ times $\psi^{\rm opt}_{0,6}$ as given in Table~2. We do not know if this has any special significance. It may be connected 
with the fact that $\psi^{\rm opt}_{0,6}$, which is equal to the quotient $\vth_{1}(\t,4z)/\vth_{1}(\t,2z)$ 
(Watson quintuple product identity), is a $(-1)$--eigenvector of the involution~$W_{2} = W_{3}$ on~$\wt J_{0,6}$.

For {\bf m=10} the situation is similar. 
Again, we can achieve $\DD \ge -1$ but not uniquely because there is a weak Jacobi form $\CK_{10} \in \wt J_{2,10}$ 
(unique up to a normalizing factor, which we fix in the same way as before) satisfying the same
discriminant bound and we can change $\v_{2,10}$ an arbitrary multiple of $\CK_{10}$.  We again have three distinguished choices 
\be
\v_{2,10}^{\rm II} \= \frac{\psi_{0,11}^{\rm II}}{12^{11}A} \, , \qquad 
\v_{2,10}^{\rm I} \= \v_{2,10}^{\rm II}  \+  \frac1{12} \,\CK_{10} \, \qquad 
\v_{2,10}^{\rm III} \= \v_{2,10}^{\rm I} \+ \frac38 \,\CK_{10} \,, 
\ee
where $\psi_{0,11}^{\rm II} \in \wt J_{0,11}$ is given (with $D=32(E_4^3-E_6^2)=2^{11}3^3\D$ as in Table~2) by
\bea \Psi^{\rm II}_{11} & = & B^{11} - 55 E_{4} A^{2} B^{9} + 330 E_{6} A^{3} B^{8} 
-990 E_{4}^{2} A^{4} B^{7} + 1848 E_{4} E_{6} A^{5} B^{6} -(2310 E_{4}^{3}-111 D) A^{6} B^{5} \cr
&& +1980 E_{4}^{2} E_{6} A^{7} B^{4} -(1155 E_{4}^{4} +120 E_{4} D) A^{8} B^{3} +(440 E_{4}^{3} E_{6} +20 E_{6} D) A^{9} B^{2} \cr
&& -(99 E_{4}^{5} -117 E_{4}^{2} D) A^{10} B + (10 E_{4}^{4} E_{6} -56 E_{4} E_{6} D) A^{11} 
\eea
and is the unique choice satisfying \eqref{coeffCm}, $\Phi_{2,10}^{\rm III}$ is the unique 
choice with no $\DD=0$ coefficients, and $\Phi_{2,10}^{\rm I}$ 
is the unique 
choice whose invariant part under $W_{2}$ (or $W_{5}$) is strongly holomorphic. Explicitly, we have 
\be \label{PHI2ten}  \Phi_{2,10}^{\rm I} \=  -\,2\,\CH|\V_{2,10}^{(1)} \,-\, \frac1{12} \,\CF_{10}\,  \ee
where $\CF_{10}$ is a mock Jacobi form of weight~2 and index~10 with Fourier coefficients given by
  \be\label{c10}  c(\CF_{10};n,r) =  \begin{cases} \phantom{-\,}C^{(10)}(40n-r^2) \quad  
                 & \text{$r\equiv\pm1 \mypmod{20}$ or $r\equiv\pm3 \mypmod{20}$}  \\ 
  -\, C^{(10)}(40n-r^2) \quad  
                 & \text{$r\equiv\pm7 \mypmod{20}$ or $r\equiv\pm9 \mypmod{20}$}  \end{cases}    \ee
for a function $C^{10}(\DD)$ whose first few values are
  \begin{center} \begin{tabular}{|c|ccccccccccccc|ccc}  
\hline
       $\DD$ & $-1$& $31$& $39$& $71$ & $79$ &$111 $& $119$& $151$ & $159$ & $191$ & $199$ & $231$ & $239$  \\  \hline
       $C^{(10)}(\DD)$ & $-1$ & 9 & 21 & 35 & 63 & 57 & $112$ & $126$ & $207$ & $154$ & $306$ & $315$ & $511$  
       \\ \hline
\end{tabular}\;\,, \end{center}
or equivalently, where $\CF_{10}$ has a theta expansion of the form
\be\label{F10exp} \CF_{10}(\t,z)\=\sum_{\ell \mypmod{20} \atop (\ell,10)=1} h_\ell(\t)\,\vth_{10,\ell}(\t,z) \ee
with mock modular forms $h_\ell=h_\ell(\t)$ having expansions beginning
\bea &&h_1\=-h_9\=-h_{11}\=h_{19}\=q^{-1/40}\,\bigl(-1\+21q\+63q^2\+112q^3\+\cdots\bigr)\,, \\
     &&h_3\=-h_7\=-h_{13}\=h_{17}\=q^{31/40}\,\bigl(9\+35q\+57q^2\+126q^3\cdots\bigr)\,. \eea

\subsection{Further data and formulation of main results for the weight 2 family \label{Statements}}

In \S\ref{choosephi}, we found that the functions $\Phi_{2,m}$ for $m \le 11$ could all be obtained, using
the Hecke-like operators from \S\ref{hecke}, from just three functions $\CH$, $\CF_{6}$ and $\CF_{10}$. 
Continuing, we find that for ${\bf m=12}$ the same statement holds: the $(+1,+1)$ eigencomponent of this form with respect
to the involutions $W_4$ and $W_3$ can be chosen to be $-2\CH|\CV_{2,12}^{(1)}$ and the $(-1,-1)$ eigencomponent
can be chosen to be $\frac1{12}\CF_6|\CV_{2,2}^{(6)}$, with $\CV_{k,t}^{(M)}$ as in~\eqref{Vkt1} 
(so $V_{2,12}^{(1)} = \frac{1}{4} (V_{12} - 2 V_{2} U_{3})$ and $\CV_{2,2}^{(6)} = \frac12 V_{2}$). Thus 
\be
\Phi_{2,12} \= - 2 \CH\,|\,\CV^{(1)}_{2,12} \+ \frac1{12}\; \CF_6\,|\,\CV^{(6)}_{2,2}
\ee
is a possible choice for $\Phi_{2,12}$. Similarly, for ${\bf m=18}$ and ${\bf m=24}$ we can take the functions
\be
\Phi_{2,18} \= - 2\CH\,|\,\CV^{(1)}_{2,18} \+ \frac1{24}\; \CF_6\,|\,\CV^{(6)}_{2,3}\,, 
\qquad \Phi_{2,24} \= - 2\CH\,|\,\CV^{(1)}_{2,24} \+ \frac1{24}\; \CF_6\,|\,\CV^{(6)}_{2,4}
\ee
as our representatives of the coset $\vstF_{2,m}+\wt J_{2,m}$. 

On the other hand, the indices ${\bf m=14}$ and ${\bf m=15}$ are similar to the indices $m=6$ and $m=10$ looked at previously: here
we find optimal mock Jacobi forms of the form
\be
\Phi_{2,14}\= -2\CH\,|\,\CV^{(1)}_{2,14} \,-\, \frac18\,\CF_{14}\,, 
\qquad \Phi_{2,15}\= -2\CH\,|\,\CV^{(1)}_{2,15} \,-\, \frac16\,\CF_{15}\,, 
\ee
where $\CF_{14}$ and $\CF_{15}$ are forms anti-invariant under $W_p$ ($p|m$) with Fourier expansions of the form 
  \be\label{c14}  c(\CF_{14};n,r) =  \begin{cases} \phantom{-\,}C^{(14)}(56n-r^2) \quad & \text{for $r\equiv\pm1$, $\pm3$ or $\pm 5$ (mod~28),}  \\ 
    -\, C^{(14)}(56n-r^2) \quad  & \text{for $r\equiv\pm9$, $\pm11$, or $\pm 13$ (mod~28),\phantom{XXX}}  \end{cases}    \ee
  \be\label{c15}  c(\CF_{15};n,r) =  \begin{cases} \phantom{-\,}C^{(15)}(60n-r^2) \quad  
                 & \text{for $r\equiv\pm1$, $\pm2$, $\pm4$, or $\pm 7$ (mod~30),}  \\ 
  -\, C^{(15)}(60n-r^2) \quad & \text{for $r\equiv\pm8$, $\pm11$, $\pm13$, or $\pm 14$ (mod~30),}  \end{cases}  \ee
with coefficients $C^{14}(\DD)$ and $C^{15}(\DD)$ given for small $\DD$ by
 \begin{center} \begin{tabular}{|c|ccccccccccccc|}  
 \hline
       $\DD$ & $-1$& $31$& $47$& $55$ & $87$ &$103 $& $111$& $143$ & $159$ & $167$ & $199$ & $215$  & $223$ \\  \hline
       $C^{(14)}(\DD)$ & $-1$ & 3 & 10 & 15 & 15 & 30 & 42 & $20$ & $51$ & $65$ & $51$ & $85$ & $120$   \\
       \hline
\end{tabular} \;, \end{center}
 \begin{center} \begin{tabular}{|c|ccccccccccccc|}  
 \hline
       $\DD$ & $-1$& $11$& $44$& $56$ & $59$ &$71 $& $104$& $116$ & $119$ & $131$ & $164$ & $176$ & $179$  \\  \hline
       $C^{(15)}(\DD)$ & $-1$ & 2 & 10 & 4 & 14 & 9 & 28 & $24$ & $29$ & $26$ & $40$ & $28$ & $54$  \\
       \hline
 \end{tabular} \;. \end{center}

This pattern continues: as we continue to higher values, we find that for each index~$M$ which is a product of an even number of distinct primes
(i.e., for which $\mu(M)=1$, where $\mu(\,\cdot\,)$ denotes the M\"obius function), we need to introduce a new function $\CQ_M$ that is anti-invariant 
under~$W_{p}$ for all $p$ dividing $M$, but that these functions and their images under Hecke-like operators then suffice to give possible 
choices for our special mock Jacobi forms in all indices, in the sense that the function
  \be \label{Qdecomp} 
  \Phi_{2,m} \= \frac{1}{2^{\w(m)-1}} \! \sum_{M\mid m \atop \mu(M)\=+1} \CQ_{M} \, | \, \CV^{(M)}_{2,\,m/M} \quad \qquad 
  (\, \w(m) := \# \{p\mid \text{$p$ prime, $p|M$}\} \, )
  \ee
belongs to the coset $\vstF_{k,m}+\wt J_{2,m}$ for all $m\ge1$.
(The factor~$2^{1-\w(m)}$ will be explained below.)
In fact more is true.  The functions $\CQ_M$ in~\eqref{Qdecomp} are simply the $(-1,\dots,-1)$ eigencomponents of $\Phi_{2,M}$ for some choice of 
these latter functions within their defining cosets, and the various terms in~\eqref{Qdecomp} are the eigencomponents of $\Phi_{2,m}$ for
the involutions $W_{m_1}$ ($m_1\|m$).  In other words, we have the following theorem:
  \begin{thm} \label{phimQm}  For any choice of the mock Jacobi forms $\Phi_{2,m}\in\vstF_{2,m}+\wt J_{2,m}$ 
  $(m \in \IN)$, the eigencomponents $\Phi_{2,m}^{\ve_1,\dots,\ve_r}$  $(m=\prod_{i} p_{i}^{\nu_{i}}$, 
   $\ve_{i} = \pm 1$, $\prod \ve_{i} = +1)$
  with respect to the decomposition\footnote{or rather, its analogue for elliptic forms}~\eqref{Jkmeigen}  are given by 
    \be \label{eigenpieces}  
  \frac12  \Phi_{2,m} \, \Bigl|\; \prod_{i=1}^{r} \big(1+\ve_{i} W_{p_{i}^{\nu_{i}}} \big) \= \CQ_{M} \,\bigl|\,\CV_{2,m/M}^{\,(M)} \quad\bigl(\text{mod }(J_{2,m}^{\,!}\bigr)^{\ve_1,\dots,\ve_r}\bigr)\;, \ee
where $M$ is the product of the $p_i$ with $\ve_i=-1$, $\CQ_{M}$ is defined by 
\be  \label{defQmeven}
\CQ_{M} \= \frac12 \Phi_{2,M} \, \Bigl|\; \prod_{p|M} \big(1- W_{p} \big) \= 2^{\w(M)-1} \, 
\Phi_{2,M}^{-,\dots,-} \, \qquad (\mu(M)\=+1), 
\ee
and $\CV_{2,m/M}^{(M)}$ is the Hecke-like operator defined in~\eqref{Vkt1}. 
In particular, the decomposition~\eqref{Qdecomp} is always true modulo $\wt J_{2,m}$, 
and if we choose the functions~$\Phi_{2,M}$ and consequently~$\CQ_{M}$ arbitrarily for 
all~$M$ with~$\mu(M)=1$, then the function defined by~\eqref{Qdecomp} gives an admissible  
choice  of $\Phi_{2,m}$ for all indices~$m\ge 1$. 
\end{thm}

\noindent{\bf Examples.} 
The functions $\CQ_{M}$ for the first few values of $M$ with $\mu(M)=1$ are
\be \label{Qlist}
\CQ_{1} = -  \CH \,, \; \; \CQ_6 =-\frac{\CF_6}{12}\,, \; \;  \CQ_{10} =-\frac{\CF_{10}}{6}\,, \; \;  
\CQ_{14}=-\frac{\CF_{14}}{4} \,, \; \; \CQ_{15}=-\frac{\CF_{15}}{3} \, , \; \; \dots \, , 
\ee
where $\CH, \CF_{6}, \dots , \CF_{15}$ are the functions introduced in \S\ref{MockJacobi}, \S\ref{choosephi},
and this subsection. (The reasons for the various normalizations, and for having two symbols $\CF_{M}$
and $\CQ_{M}$ to denote proportional functions, will be discussed later.)

We will prove Theorem \ref{phimQm} towards the end of this section  
by calculating the shadows of the mock Jacobi forms on the two sides 
of~\eqref{eigenpieces}, since two mock Jacobi forms differ by a true 
Jacobi form if and only if they have the same shadow. 

\smallskip

\noindent{\bf Remarks.} {\bf 1.}
We wrote $\bigl(J^{\,!}_{2,m}\bigr)^{\ve_1,\dots,\ve_r}$ rather than $\wt J_{2,m}^{\,\ve_1,\dots,\ve_r}$ in 
the theorem because of the fact mentioned in~\S\ref{hecke} that the involutions $W_{m_1}$ do not preserve 
the space of weak Jacobi forms (or weak elliptic forms), so that the left-hand side of eq.~\eqref{eigenpieces} 
for an arbitrary choice of $\Phi_{2,m}$ may contain monomials $q^ny^r$ with $n<0$.  In fact, it will follow 
from the results given below that the functions $\Phi_{2,M}$ can be chosen in such a way that the eigencomponents 
$\Phi_{2,M}^{-,\cdots,-}$ do not contain any negative powers of~$q$, in which case the expression on 
the right-hand side of~\eqref{Qdecomp} really does lie in $\vstF_{k,m}+\wt J_{2,m}$ as asserted, and not merely 
in $\vstF_{k,m}+J^{\,!}_{2,m}$.\\
\noindent{\bf 2.} An exactly similar statement is true for the functions $\Phi_{1,m}$, except that in that case we 
have $\prod_i\ve_i=-1$ (because in odd weight $W_m$ acts as multiplication by~$-1$), so that the $M$ that 
occur in the decomposition of $\Phi_{1,m}$ are the divisors of~$m$ with~$\mu(M)=-1$, and 
the functions~$\CQ_{M}$ that occur in the decomposition of $\Phi_{1,M}$ are defined by
\be \label{defQmodd}
\CQ_{M} \= \frac12 \, \Phi_{1,M} \, \Bigl|\; \prod_{p|M} \big(1- W_{p} \big) \= 2^{\w(M)-1} \; 
\Phi_{2,M}^{-,\cdots,-} \, \qquad (\mu(M)\=-1) 
\ee
instead of equation~\eqref{defQmeven}.
The reason for the factor~$2^{\w(m)-1}$ in these two equations is now apparent:
it is the minimal factor making~$\CQ_{M}$ an integral linear combination of images of~$\Phi_{k,m}$
under Atkin-Lehner involutions. (We could remove one factor~2 because for each decomposition~$M=M_{1}M_{2}$,
the two terms~$\mu(M_{1}) \Phi_{k,m}|W_{M_{1}}$ and~$\mu(M_{2}) \Phi_{k,m}|W_{M_{2}}$  
are equal.)

\smallskip

Theorem \ref{phimQm} explains the first part of our observations, those connected with the eigencomponent 
decomposition of the mock Jacobi forms $\Phi_{2,m}$ and the fact that the collection of all of these 
functions can be expressed in terms of the special functions~\eqref{Qlist} and their images under Hecke-like operators.
But it does not say anything about the orders of the poles of the functions in question, which was our original criterion for optimality.
On the other hand, from the tables of Fourier coefficients of the functions~$\CQ_{1}, \dots, \CQ_{15}$ that we have given 
we see that, although none of these functions except $\CQ_{1}$ is strongly holomorphic, each of them 
satisfies the only slightly weaker discriminant bound $\D := 4mn-r^{2} \ge-1$.
This property will turn out to play a prominent role. We will call Jacobi or elliptic forms 
satisfying it {\bf forms with optimal growth}\footnote{An alternative name would be ``forms with 
only simple poles'' because the forms $h_{\ell}(\t) = \sum_{\DD} C(\DD,\ell) \, q^{\DD/4m}$ have no singularity 
worse than $q^{-1/4m}$ as $q \to 0$, i.e., they have at most simple poles with respect to the 
local parameter $q^{1/4m}$ at $\t=i \infty$. However, this terminology would be dangerous because 
it could also suggest simple poles with respect to~$z$. 
Another reasonable name would be ``nearly holomorphic,'' but this terminology also carries a risk of 
confusion since ``almost holomorphic'' is already a standard notation for modular forms which 
are polynomials in~$1/\t_{2}$ with holomorphic coefficients.}, because their Fourier coefficients 
grow more slowly than forms that do not obey this discriminant bound.
We observe that any Jacobi or elliptic form having optimal growth is automatically a weak 
form, since $4nm-r^2\ge-1$ implies $n\ge0$.

As we continued our computations to higher values of $M$, 
we were surprised to find that this ``optimal growth'' property continued to hold for each 
$M$ of the form $pq$  ($p$ and $q$ distinct primes) up to 200.  It followed that the full function $\Phi_{2,M}$ 
could also be chosen to have optimal growth in these cases, since $ \Phi_{2,M} \= \CQ_1|\CV^{(1)}_{2,M}+\CQ_{M}$ 
and the image of the strongly holomorphic form $\CQ_1$ under $\CV^{(1)}_{2,M}$ is strongly holomorphic. 
This made it particularly interesting to look at the case ${\bf m=210}$, 
the first product of four distinct primes, because in this case, even if the function $\CQ_{210} = \Phi_{2,210}^{----}$
can be chosen to be have optimal growth, then the six middle terms in the decomposition~\eqref{Qdecomp}, which here takes the form 
\bea 8 \, \Phi_{2,210} &\= \CQ_{1}\,|\, \CV^{(1)}_{2,210} \+ \CQ_{6}\,|\,\CV^{(6)}_{2,35} \+ \CQ_{10}\,|\,\CV^{(10)}_{2,21} 
 \+ \CQ_{14}\,|\,\CV^{(14)}_{2,15} & \cr
  &\qquad \+ \CQ_{15}\,|\,\CV^{(15)}_{2,14}  \+ \CQ_{21}\,|\,\CV^{(21)}_{2,10} \+ \CQ_{35}\,|\,\CV^{(35)}_{2,6}\+ \CQ_{210} \; , & \eea
will not be. (Even if $\CQ_{M}$ is chosen to have optimal growth, its image under the operator $\CV^{M}_{2,t}$ will in general satisfy
only  $\Delta \ge - t^{2}$. For $M$ of the form $pq$ we had no problems with this because $\CQ_1$, and hence also its image under $\CV^{(1)}_{2,M}$, is actually holomorphic.) The corresponding computation was somewhat forbidding since the space $\wt J_{2,210}$ of forms 
by which we can change the initial choice $\vstF_{2,210}$ has dimension $3815$, 
and the conditions $C(\Delta,r)=0$ for $\Delta < -1$ give~3813 linear constraints, so that we have to find 
the kernel of a roughly square matrix of size nearly~4000. Moreover, the coefficients of this matrix are 
gigantic, so that we have to work modulo some large prime. Nevertheless, this computation could 
be carried out (using PARI/GP), and it turned out that 
the function $\CQ_{210}$ could again be chosen to have optimal growth! 
This solution is given by $\CQ_{210} \= -  \CF_{210}$, where $\CF_{210} = \sum_{\ell \mypmod{420}} h_{\ell} \, \vth_{210,\ell}$
with the Fourier coefficients of the mock theta functions $h_{\ell}$ given in Table~\ref{fcq210}. 
\begin{table}[h]
{\footnotesize{
\begin{center} 
\begin{tabular}{!{\thvline}c|c!{\thvline}rcccccccccccc!{\thvline}ccc}  
\noalign{\thhline}
       $\ell \text{ with } h_{\ell} = h_{\ell_{0}}$ &$ \ell  \text{ with } h_{\ell} = -h_{\ell_{0}}$ & 
     $n=0$ & 1& 2& 3& 4& 5& 6& 7& 8& 9& 10& 11& 12  \\  
\noalign{\thhline}
       {\bf 1}, 29, 41, 71 & 139, 169, 181, 209 &  $-2$& 13& 28& 34& 49& 46& 71& 59& 83& 77&102& 87&121 \\
       {\bf 11}, 31, 59, 101 & 109, 151, 179, 199  &  0&  9& 21& 27& 36& 41& 51& 44& 75& 62& 62& 82&104 \\
       {\bf 13}, 43, 83, 97 & 113, 127, 167, 197 &  0&  6& 17& 17& 35& 20& 49& 31& 57& 36& 77& 32& 94\\
       {\bf 17}, 53, 67, 73 & 137, 143, 157, 193 &  0&  4& 12& 16& 22& 19& 43& 17& 40& 50& 41& 27& 87 \\
       {\bf 19}, 61, 79, 89 & 121, 131, 149, 191 &  0&  3& 11& 12& 23& 14& 37& 17& 43& 28& 45& 30& 77  \\
       {\bf 23}, 37, 47, 103 & 107, 163, 173, 187 &  0&  1&  7&  4& 20& $-1$& 32&  3& 30& 10& 50&$-16$& 63 \\
\noalign{\thhline}
\end{tabular}
\end{center}
}}
\caption{\small Fourier coefficients of the function~$\CF_{210}$}
\label{fcq210}
\end{table}
In this table, we have listed the values of $\ell$ with $(\ell,210)=1$ and $0 \le \ell < 210$ (since 
$h_{420n\pm\ell} = h_{\ell}$ and $h_{\ell}=0$ for $(\ell,210)>1$) in groups of~8, with $h_{\ell} = h_{\ell_{0}}$
for the first four and  $h_{\ell} = - h_{\ell_{0}}$ for the second four, where $\ell_{0}$ (shown in boldface)
belongs to $\{ 1,11,13,17,19,23 \}$ and $h_{\ell_{0}}=\sum_{n \ge0} C(n) \, q^{n-\ell_{0}^{2}/840}$. 
Thus, the third line of the table means that 
$h_{\pm 13} = h_{\pm 43} = \cdots  = -h_{\pm 167} = -h_{\pm 197} 
\= q^{-{169}/{840}} \big(6q + 17 q^{2} + 17 q^{3} + \dots  \big)$.
A striking property of these coefficients is that they are so small: for instance, 
the largest value of~$c(n,\ell)$ in the table is~$c(12,1)=121$, with corresponding discriminant 
$\D=4\cdot210\cdot12-1^2=10079$, while for the Mathieu form $h^{(2)}$ the table following
eq.~(7.31) gave the much larger coefficient $132825$ for the much smaller discriminant
$4\cdot2\cdot9-1^2=71$.  
%
This behavior is a direct consequence of the ``optimal growth''
property together with the fact that the index is large (cf.~Theorem~\ref{CardyProp} below and the 
discussion in~\S\ref{specialMTF}).

In fact, in this example even more happens: not only the function $\CQ_{210} = \Phi_{2,210}^{- - - -}$, but the whole function $\Phi_{2,M}$ 
(and hence each of its eigencomponents $\Phi_{2,210}^{\pm\pm\pm\pm}$) can be chosen to have optimal growth.
It turns out that this statement holds for all square-free indices~$m$, whether they have an even or an odd
number of prime factors. If~$m$ is not square-free, then this is not true in 
general,\footnote{It is true if and only if~$m$ is the product of a square-free number and a prime power.}
but it becomes true after a very small modification of the functions in question:
\begin{thm} \label{QmOG}
The functions $\{\Phi_{2,m} \}_{m \ge 1}$ can be chosen to have the form 
\be \label{Phi0mimp}
  \Phi_{2,m}(\t,z) \= \sum_{d^{2} \mid m} \Phi^{0}_{2,m/d^{2}} (\t, \, d z)
\ee
where $\{\Phi_{2,m}^{0} \}_{m \ge 1}$ are primitive mock Jacobi forms having optimal growth.  
\end{thm}

\begin{corollary} \label{Qmcor}
The functions $\{\CQ_{M} \}_{M > 1, \, \mu(M)=1}$ can be chosen to have optimal growth. 
\end{corollary}
\ndt (Here we have excluded the case $M=1$ because the function~$\CQ_{1}=-\CH$ has already been
chosen to have no poles at all.) Tables of Fourier coefficients of the functions~$\CQ_{M}$ appearing 
in the corollary for~$M \le 50$ are given in Appendix~\ref{apptabs1}.

If we denote by $\CE_{m}\NH$ the space of elliptic forms of index~$m$ (cf.~\S\ref{hecke}) 
having optimal growth, then we have the sequence of inclusions 
$$ \CE_{m}^{0} \; \subset \; \CE_{m} \; \subset \; \CE_{m}\NH \; \subset \; \wt \CE_{m} \; \subset \; \CE_{m}^{!}  
\qquad \big( \, \v \in \CE_{m}\NH \;  \underset{\rm DEF}\Longleftrightarrow  \; \text{$C_{\v}(\DD, \ell) = 0$ 
if $\DD < -1$} \, \big)
$$
of spaces of different types of elliptic forms (cusp, holomorphic, optimal growth, weak, weakly holomorphic)
with different orders of growth. Theorem~\ref{QmOG} says that we can represent all of our functions 
in terms of new, ``primitive,'' forms $\Phi_{m}^{0} \in \CE_{m}\NH$ for all~$m$. 
The function~$\Phi_{m}^{0}$ is unique up to the addition of an element of the space 
$J_{2,m}\NH=\CE_{m}\NH \cap \wt J_{2,m}$ of Jacobi forms of weight~2 and index~$m$ of optimal growth.  
We will study this latter space in the next subsection.

Theorem~\ref{QmOG} has an immediate consequence for the growth of the Fourier coefficients 
of the optimal choice of~$\Phi_{2,m}$, since the Hardy-Ramanujan circle method implies that 
the Fourier coefficients of a weakly holomorphic Jacobi (or mock Jacobi) form grow asymptotically 
like~$e^{\pi\sqrt{\DD|\DD_{0}|}/m}$, where~$\DD_{0}$ is the minimal value of the discriminants  
occurring in the form. 
We give a more precise version of this statement in the special case of 
forms of weight~2 having optimal growth.

\begin{thm}\label{CardyProp}
Let $\v$ be a weak or mock Jacobi form of weight~$2$ and index~$m$ with optimal growth.
Then the Fourier coefficients of $\v$ satisfy the asymptotic formula
  \be \label{HRR} C_\v(\D,r) \= \kappa_r\,e^{\pi\sqrt\D/m}
   \;+\;\text O\bigl(e^{\pi\sqrt\D/2m}\bigr)  \ee             
as $\D\to\infty$, where the coefficient $\kappa_r=\kappa(\v;r)$ is given by 
  \be \label{kappa} \kappa_r \= -\sum_{\ell\mypmod{2m} \atop \ell^{2}\,\equiv1\mypmod{4m}} 
    C_\v(-1,\ell)\,\cos\Bigl(\frac{\pi\ell r}m\Bigr)\;. \ee
\end{thm}
\begin{proof}
We only sketch the proof, since it is a standard application of the circle method. We write the 
coefficient $c_\v(\D,r)$ as a Fourier integral $\int_{-\half+i\ve}^{\half+i\ve}h_r(\t)\,\E(-\D\t/4m)\,d\t\,$,
with $h_\ell$ defined by~\eqref{defhltau} and~\eqref{jacobi-theta}.  Its value is unchanged if we replace 
$h_r(\t)$ by its completion $\wh h_r(\t)$, since they differ only by terms with negative discriminant
The main contribution to the integral comes from a neighbourhood of $\t=0$, where we can estimate $\wh h_r(\t)$ to high accuracy as
 $$  \wh h_r(\t) \= -\frac{(i/\t)^{3/2}}{\sqrt{2m}}\;
   \sum_{\ell\mypmod{2m}} \E\Bigl(\frac{r\ell}{2m}\Bigr)\; \wh h_\ell\Bigl(-\frac1\t\Bigr)
\= -\frac{\kappa_r}{\sqrt{2m}}\;(i/\t)^{3/2}\,\E\Bigl(\frac1{4m\t}\Bigr)\+\text O\bigl(\t^{1/2}\bigr)\,, $$
the first equation holding by the $S$-transformation equation (\cite{Eichler:1985ja}, p.~59) and the second because 
the ``optimal growth" assumption implies that $\wh h_\ell(-1/\t)=C_\v(-1,\ell)\,\E(1/4m\t)+\text O(1)$ as $\t\to0$
and the contribution from the Eichler integral of the shadow of $\wh h_\ell$ is at most~$\text O(1/\t)$. The rest
follows by a standard calculation, the only point worthy of note being that since the $h_\ell$'s have weight~3/2
the Bessel functions that typically arise in the circle method are simple exponentials here.
\end{proof}
It is interesting to compare the statement of Theorem~\ref{CardyProp} with the tables given in the 
Appendix. This will be discussed in more detail in \S\ref{specialMTF}.

\subsection{Optimal choice of the function~$\Phi_{2,m}$ \label{optimal}}

As we already saw, Theorem~\ref{QmOG} fixes the choice of the mock Jacobi form~$\Phi_{2,m}^{0}$
up to the addition of an element of the space 
$J_{2,m}\NH$ of Jacobi forms of weight~2 and index~$m$ having optimal growth.   
The structure of this latter space is described by the following theorem, proved in~\S\ref{structure}.  
  \begin{thm}\label{chooseKm}   The quotient $J_{2,m}\NH / J_{2,m}$ has dimension $1$ for every~$m$.
  An element $\CK_{m}\in J_{2,m}\NH$ representing a non-trivial element of this quotient space can be chosen to satisfy
    $$ \CK_{m}(\t,z) \= y \, - \, 2 \+ y^{-1} \+ O(q)\quad\in\;\IC[y,y^{-1}][[q]]\,, $$
   and then has polar coefficients given by 
  \be
   C(\CK_{m}; \, \DD, \, r) \= \begin{cases} 
   1 & \text{if $\DD=-1$, $r^{2} \equiv 1 \mypmod{4m}$} \, , \\  0 & \text{if $\DD<-1$} \, . \end{cases} 
   \ee
    The function $\CK_{m}$ can be chosen to be primitive and invariant under all the operators 
    $W_{m_1}$ $(m_1\|m)$.
  \end{thm}

The choice of $\CK_m$ described in the theorem is unique up to the addition of a primitive holomorphic
Jacobi form invariant under all $W_{m}$ operators. This first occurs for $m=25$, so $\CK_m$ is
unique for small values of~$m$. Because of Theorem~\ref{chooseKm}, the function $\Phi^{0}_{m}$ has only a very 
small amount of choice: we can change it without disturbing the ``optimal growth'' property by adding an arbitrary 
multiple of $\CK_{m}$ and an arbitrary holomorphic Jacobi form, but this is the only freedom we have. In particular, we have 
three special choices of ``optimal'' $\Phi^{0}_{2,m}$, each unique for $m<25$, 
corresponding to the three that we already encountered for $m=6$ and $m=10$ in \S\ref{choosephi}: 
\begin{enumerate}
\item[(I)] The $(+,\cdots,+)$-eigencomponent of  $\Phi_{2,m}^{0, \rm I}$ is strongly holomorphic.
\item[(II)] $\Phi_{2,m}^{0, \rm II}$ has $c(0,1)=0$. 
\item[(III)] $\Phi_{2,m}^{0, \rm III}$ has $c(0,0)=0$. 
\end{enumerate}
Each of these three choices has certain advantages over the other ones. 
In case~(I) we can choose~$\Phi_{2,m}^{0, \rm I}$  within its class modulo~$J_{2,m}^{0}$ 
so that its $(+,\cdots,+)$-eigencomponent is precisely equal to $2^{1-\w(m)} \CQ_{1}|\CV^{(1)}_{2,m}$ 
and hence has Fourier coefficients expressible in terms of class numbers, like in the examples for 
small~$m$ that we already saw. 
%
Also,~$\Phi_{2,m}^{0,\rm I}$ is strongly holomorphic when $m$ is a prime power, so that 
this is certainly the best of the three choices in that special case. 
The function occurring in~(II) has a Fourier expansion beginning $\frac{1}{12}(m+1) + O(q)$, 
and the corresponding $\psi^{0,\rm II}_{m+1}$ satisfies~\eqref{coeffCm} and belongs to the 
family of the weight zero, index~$m$ Jacobi forms discussed in~\cite{Gritsenko:1999fk}.  
The choice~(III) looks at first sight like the least natural of the three, since it might seem  
strange to insist on killing a $\DD=0$ term when there are $\DD=-1$ terms that are more singular, 
but in fact  it is the nicest in some ways, since it can be further improved by the addition of a holomorphic 
Jacobi form to have {\it all} of its $\DD=0$ coefficients equal to zero, and its 
polar coefficients are then given by an elegant formula, which we now state. 
\begin{thm} \label{PolarCoeffs}
The mock Jacobi form~$\Phi_{2,m}^{0, \rm III}$ in Theorem~\ref{QmOG} can be chosen, 
uniquely up to the addition of a Jacobi cusp form of weight~$2$ and index~$m$, 
so that all its Fourier coefficients with~$\DD=0$ vanish. Its polar coefficients are 
then given by 
\be \label{Phi0IIIcoeffs}
C(\Phi_{2,m}^{0, \rm III}; -1, r) \= \frac{1}{24} \sum_{d^{2}|m} \, \mu(d) \, 
\Big(\big(\frac{r-1}{2}, \frac{m}{d^{2}} \big) \+ \big(\frac{r+1}{2}, \frac{m}{d^{2}} \big)  \Big) 
\quad \text{for $r^{2} \equiv 1 \mypmod{4m}$.} 
\ee
Moreover, the function~$\Phi_{2,m}^{0, \rm III}$ can be chosen to be primitive without affecting these properties. 
\end{thm}
\ndt This also gives a multiplicative formula for the polar coefficients of all eigencomponents 
of~$\Phi_{2,m}^{0, \rm III}\,$: 
\begin{corollary}
Write $m=\prod_{i=1}^{s} P_{i}$, where $P_{i}=p_{i}^{\nu_{i}}$ with distinct primes $p_{i}$
and exponents $\nu_{i} \ge 1$. Then for any~$\ve_{i} \in \{\pm 1\}$ with~$\ve_{1}\cdots \ve_{s}=1$, 
we have  
\be \label{eigencoeffs}
C\Bigl( \Phi_{2,m}^{0, \rm III} \, \Bigl| \, \prod_{i=1}^{s} \big( 1+ \ve_{i} W_{P_{i}} \big); \, -1 , \, r \Bigr) \=
\frac{\pm 1}{12} \, \prod_{i=1}^{s} \; \begin{cases} \, (p_{i} + \ve_{i}) & \text{if $\nu_{i}=1$,}  \\
\, (p_{i}^{\nu_{i}} - p_{i}^{\nu_{i}-2}) & \text{if $\nu_{i}\ge 2$,}  
\end{cases}
\ee
where the sign is the product of the $\ve_{i}$ for all $i$ for which $r \equiv -1 \mypmod{2 P_{i}}$. 
In particular, for the functions~$\CQ_{M}$ $(M>1$, $\mu(M)=1)$, chosen by the Corollary
to Theorem~\ref{QmOG}
to have optimal growth, the Fourier coefficients with non-positive discriminant are given by
\be \label{Qmpolar}
C(\CQ_{M}; \DD, r) 
\= \begin{cases}  \qquad 0 \,  & \text{if $\DD = 0$} \, , \\
\pm \dfrac{\v(M)}{24} \,  &  \text{if $\DD \= -1$} \, , 
\end{cases} 
\ee
where in the second line $r^{2} \equiv 1 \mypmod{4M}$ and the sign is determined 
by $\mu((\frac{r+1}{2}, M)) = \pm 1$.
\end{corollary}
The final statement of the corollary can be compared with the examples in~\eqref{Qlist}, 
in which each of the $\CF_{M}$'s occurring (for $M>1$)  was normalized with $C(\CF_{M};-1,1)=-1$, 
and the numerical factor relating $\CQ_{M}$ and $\CF_{M}$ equals $-\v(M)/24$ in each case.
To give a feeling for the nature of these functions in general, 
and as an illustration of the formulas in Theorem~\ref{PolarCoeffs} and its corollary, 
we give in Table~\ref{Kmphi0} below the first few Fourier coefficients $C(\Phi; \DD,\ell)$ for these 
and some related functions $\Phi$ for $m=25$ (the first index where $J_{2,m}\neq\{0\}$ and hence
$\CK_{m}$ is not unique), $m=37$~(the first index where $J_{2,m}^{0}\neq\{0\}$ 
and hence $\Phi_{2,m}^{0, \rm III}$ is not unique),
$m=50$~(the first case where \eqref{eigencoeffs} is non-trivial), 
and $m=91$~(the first case where where $J_{2,m}^{0,--}\neq\{0\}$ and hence $\CQ_m$ is not unique). 
The free coefficients~``$a$'', ``$b$'' 
and ``$c$'' in the table arise because $\dim J_{2,25} = \dim J_{2,37}^{0} = \dim J_{2,50}= \dim J_{2,91}^0 =1$. 
In the case of $m=91$ we show the two $\ell$-values corresponding to each value of~$\DD$ in the format 
$\,\ell \, / \,\ell^*\,$ and give in the following line the Fourier coefficient $C(\CQ_{91};\D,\ell)=-C(\CQ_{91};\D,\ell^*)$.
In the cases of $m=25$ and $m=37$, we have given the (non-unique) holomorphic mock Jacobi 
form~$\Phi^{0, \rm I}_{2,2m}$ as well as~$\Phi^{0, \rm III}_{2,2m}$.
\begin{table}[h] 
{\small
\begin{center}
\begin{tabular}{!{\thvline}c!{\thvline}c|ccc|c|c|c|c|c|c!{\thvline}cc}  
\noalign{\thhline}
       $\DD$ & $-1$&  \multicolumn{3}{c|}{0}  & $4$& $11$ & $16$ &$19 $& $24$& $31$   \\  
       $\pm \ell \mypmod{50}$ & $  1$& $0$ & $  10$ & $  20$ & $  14$& $  17$ & $  22$ &$  9$& $  24$& $ 13$   \\  \hline
       $ \CK_{25}$ & $1$ & $-2$ & $a$ & $1-a$ & $1-a$ & $1+2a$ & $2+3a$ & $2-4a$ & $5-2a$ & $3+2a$   \\
       $ \Phi^{0, \rm III}_{2,25} $ & $1$ & $0$ & $0$ & $0$ & $0$ & $1$ & $2$ & $-1$ & $2$ & $1$ \\
       $ \Phi^{0, \rm I}_{2,25} $ & $0$ & $2$ & $c$ & $-1-c$ & $-1-c$ & $2c$ & $3c$ & $-3-4c$ & $-3-2c$ & $-2+2c$ \\
\noalign{\thhline}
\end{tabular} 
\end{center}

\begin{center}
\begin{tabular}{!{\thvline}c!{\thvline}c|c|c|c|c|c|c|c|c|c!{\thvline}ccc}  
\noalign{\thhline}
       $\DD$ & $-1$& $0$ & 3 & 4 & $7$& $11$ & $12$ &$16 $& $27$& $28$   \\  
       $\pm \ell \mypmod{74}$ & $  1$& $0$ & $21$ & $12$ & $17$& $27$ & $32$ &$24$& $11$& $34$   \\  \hline
       $ \CK_{37} $ & $1$ & $-2$ & $a$ & $a$ & $1-a$ & $1+a$ & $2-a$ & $2-2a$ & $1-3a$ & $2+3a$ \\
       $ \Phi_{2,37}^{0, \rm III} $ & $\frac{19}{12}$ & $0$ & $b$ & $-\frac16+b$ & $\frac14-b$ & $\frac{11}{12}+b$ 
       & $\frac32-b$ & $1-2b$ & $-\frac34-3b$ & $\frac{13}{6}+3b$ \\
       $ \Phi_{2,37}^{0, \rm I} $ & $0$ & $\frac{19}{6}$ & $c$ & $-\frac{1}{6}+c$ & $-\frac{4}{3}-c$ & $-\frac{2}{3}+c$ 
       & $-\frac{5}{3}-c$ & $-\frac{13}{6}-2c$ & $-\frac{7}{3}-3c$ & $-1+3c$   \\
\noalign{\thhline}
\end{tabular} 
\end{center}
}

{\small 
\begin{center}
\begin{tabular}{!{\thvline}c!{\thvline}cc|ccc|c|c|c|cc|c|c!{\thvline}cc}  
\noalign{\thhline}
       $\DD$ & \multicolumn{2}{c|}{$-1$} &  \multicolumn{3}{c|}{0} & $4$ & $16$ &$24 $& \multicolumn{2}{c|}{31} & $36$ &39  \\  
       $\pm \ell \mypmod{100}$ & $1$& $49$ & $0$ & $20$ & $40$& $14$ & $28$ &$24$& $13$& $37$ &$42$ & $19$  \\  \hline
       $ \CK_{50} $ & $1$ & $1$ & $-2$ & $a$ & $1-a$ & $-a$ & $1+a$ & $1-2a$ & $1+2a$ & $1+2a$ &$2+3a$& $-2a$\\
       $ \Phi_{2,50}^{0, \rm III} | (1+W_{2}) $ & $3$ & $3$ & $0$ & $0$ & $0$ & $-2$ & $2$ & $-2$ & $3$ & $3$ & $8$ &$-1$ \\
       $ \Phi_{2,50}^{0, \rm III} | (1-W_{2}) $ & $1$ & $-1$ & $0$ & $0$ & $0$ & $0$ & $0$ & $0$ & $-1$ & $1$ & $0$  &$-1$ \\
\noalign{\thhline}
\end{tabular} 
\end{center}


\begin{center}
\begin{tabular}{!{\thvline}c!{\thvline}c|c|c|c|c|c|c|c|c!{\thvline}c|c}  
\noalign{\thhline}
       $\DD$ & $-1$& 3 & 12 & 27 & 40 & 48 & 55 & 68 & 75  \\  
       $\pm \ell \mypmod{182}$ & $1\, / \, 27$& $19\, / \,33$ & $38\, / \,66$ & $57\, / \,83$ & $18\, / \,60$& 
        $50\, / \,76$ & $71\, / \,85$ &$32\, / \,46$&  $17\, / \,87$   \\  \hline
       $\CQ_{91}$ & $3$ & $a$ & $-1+a$ & $-1-2a$ & $-1-3a$ & $-2+2a$ & $3a$ & $1+3a$ & $-5+2a$ \\ 
\noalign{\thhline}
\end{tabular} 
\end{center}

} 
\caption{\small{Examples of non-unique Jacobi and mock Jacobi forms}} 
 \label{Kmphi0}
\end{table}

In summary, the ``optimal growth'' property of~$\Phi_{2,m}^{0}$ and the near-uniqueness of the 
Jacobi form~$\CK_{m}$ have permitted us to pin down both of these functions in many cases, 
but there still remains an additive ambiguity whenever there are Jacobi cusp forms.  
We do not know how to resolve this ambiguity in general, and state this as an open problem:


\smallskip
\ndt{\it Question}: Is there a canonical choice of the forms $\CK_m$ and $\CQ_M$ for all
positive integers~$m$ and all positive square-free integers~$M\,$?

\smallskip
{ \ndt Should the answer to either of these questions be positive, the corresponding forms would be of
considerable interest. In particular, one can speculate that the canonical choice of~$\CQ_M$, if it exists,
might be canonically associated to the quaternion algebra of discriminant~$M$, since there
is a natural bijection between positive square-free integers and quaternion algebras over~$\Q$,
with the latter being indefinite or indefinite according as $\mu(M)=1$ (corresponding to our $k=2$
forms and to mock theta functions of weight~3/2) or $\mu(M)=-1$ (corresponding to $k=1$ and 
to mock theta functions of weight~1/2).
}
%

We end this subsection by describing one possible approach to answering the above question
that we tried.  This approach was not successful in the sense that, although it did produce
a specific choice for $\CK_m$ and $\Phi_{2,m}^0$ (or $\Phi_{1,m}^0$), these did not
seem to have arithmetic coefficients and therefore had to be rejected, but the
method of calculation involves a result that is of interest in its own right and
may be useful in other contexts.  The idea is to fix the choice of the Jacobi
forms (or mock Jacobi forms) in question by demanding that they (or their completions)
be orthogonal to Jacobi cusp forms with respect to the Petersson scalar product.
(As an analogy, imagine that one wanted to 
determine the position of~$E_{12}=1+\cdots$ in $M_{12}=\IC \, E_{4}^{3} + \IC \, \Delta$ 
but did not know how to compute the Fourier expansion of $\sum (m\t+n)^{-12} \,$. Then one  
could simply write $E_{12}=E_{4}^{3}+a\Delta$ and use the orthogonality of Eisenstein series 
and cusp forms to compute the number~$a=-432000/691$ numerically
as~$-(E_{4}^{3},\Delta)/(\Delta,\Delta)$, where $( \, \, , \,)$ denotes the Petersson scalar product.)  
To apply this idea, we need a good way to compute the Petersson scalar product
of Jacobi forms. This is the result of independent interest referred to above.
For simplicity, we state it only in the case~$k=2$.
\begin{thm}\label{JCPSP}
Let $\v_1$ and $\v_2$ be two Jacobi cusp forms of weight~$2$ and index~$m$, and set
 \be \label{AHdef} A(\D) \,=\, \sum_{\ell\mypmod{4m}} C_{\v_{1}}(\D,\ell)\,\overline{C_{\v_{2}}(\D,\ell)}\,,\;\quad 
 H(t) \,=\, \sum_{\D=1}^\infty\,\frac{A(\D)/\sqrt\D}{\exp(2\pi \sqrt{t\D/m})\,-\,1}\;.  \ee
Then the function $\dfrac{tH(t)-H(1/t)}{t-1}$ $(t\ne1)$ has a constant value $c_H$
which is proportional\footnote{The exact constant of proportionality plays no role for our
purposes, since we will only be concerned with ratios of Petersson scalar products.
If the scalar product is normalized as in~\cite{Eichler:1985ja}, its value is 1/2.
}
to the Petersson scalar product $(\v_1,\v_2)$. \end{thm}
\ndt{\it Sketch of proof}: We refer to~\cite{Eichler:1985ja} for the definition of 
$(\v_1,\v_2)$ (eq.~(13), p.~27) and for its expression as a simple multiple of 
$\sum_\ell(h_\ell^{(1)},h_\ell^{(2)})$, where $h_\ell^{(1)}$ and~$h_\ell^{(2)}$ are the coefficients 
in the theta expansions of $\v_1$ and $\v_2$ (Theorem 5.3, p.~61). 
Calculating these weight 3/2 scalar products by the Rankin-Selberg method as explained in \S\ref{modularbasic},
we find that $(\v_1,\v_2)$ can be expressed as a simple multiple of the residue at $s=3/2$
of the Dirichlet series $L(s)=\sum_{\D>0}A(\D)\D^{-s}$.  The  Rankin-Selberg method also shows that 
the function $\wt L(s)=2\,(2\pi)^{-s}m^s\,\Gamma(2s)\,\zeta(2s)\,L(s+\frac12)$ has a meromorphic continuation 
to all~$s$, with poles only at $s=0$ and $s=1$, and is invariant under~$s\to1-s$.  But $\wt L(s)$ is
simply the Mellin transform of $H(t)$, so this implies the proposition by a standard argument. 
\hfill $\square$

A rather surprising aspect of this theorem is that it is actually {\it easier} to compute Petersson scalar 
products for Jacobi forms than it is for ordinary modular forms of integral weight. In that case the analogue 
of $H(t)$ would have an infinite sum of $K$-Bessel functions multiplying each coefficient~$A(\D)$,
whereas here the Bessel functions are replaced by simple exponentials (because the modular forms $h_\ell^{(i)}$
have half-integral weight, so that the usual gamma factor $\Gamma(s)\Gamma(k+s-1)$ in the 
Rankin-Selberg method becomes a single gamma function $\Gamma(2s)$)  and at the same time 
the infinite sum becomes a geometric series that can be summed explicitly.

To apply the proposition to our situation in the case of $\CK_{37}$, we start with an initial choice
(like the one given in Table~\ref{Kmphi0}) and replace it by $\CK_{37}+a\v_{37}$, where~$\v_{37}$
is the generator of~$J_{2,37}^{0}$ (normalized by~$C(\v_{37}\,; 3,21)=1$) and where $a\in\R$ is
chosen to make this function orthogonal to $\v_{37}$, i.e., $a=-(\CK_{37},\v_{37})/(\v_{37},\v_{37})$,
in the hope that this~$a$ will be an integer or simple rational number. 
(Of course $\CK_{37}$ is not a cusp form, but the integral defining the Petersson scalar product 
still converges,\footnote{The precise condition needed for the convergence of the Petersson scalar 
product of two weakly holomorphic Jacobi forms~$\v_{1}$ and~$\v_{2}$ is that $\, \text{ord}_{\infty}(h_{\ell}^{(1)})+
\text{ord}_{\infty}(h_{\ell}^{(2)})>0$ for all~$\ell \in \IZ/2m\IZ$. This is amply satisfied in the 
case~$\v_1=\v_{37}$, $\v_2=\CK_{37}$, since the only negative value of $\, \text{ord}_{\infty}(h_{\ell}^{(2)})$ 
is~$-1/148$ for~$\ell \equiv \pm 1 \mypmod{74}$, and there~$\, \text{ord}_{\infty}(h_{\ell}^{(1)}) = 147/148$.}
the sum defining $H(t)$ converges for $t>\frac{1}{148}$ by Theorem~\ref{CardyProp}, and Theorem~\ref{JCPSP}
still holds.\footnote{The proof given above breaks down, since the numbers~$A(\DD)$ grow exponentially
in~$\sqrt{\DD}$, so that the Dirichlet series~$\sum_{\DD}A(\DD)/\DD^{s}$ does not converge for any 
value of~$s$, but the function~$f(y)=y^{3/2}\sum_{\DD} A(\DD) e^{-\pi\DD y/m}$, which is the constant 
term of an~$SL(2,\IZ)$-invariant function that is small at infinity, still has a Mellin transform with a 
meromorphic continuation to all~$s$, and if we define~$\wt L(s)$ as the product of this Mellin transform
with~$\pi^{-s}\G(s)\zeta(2s)$, then~$\wt L(s)$ has the same analytic properties as before and its inverse Mellin
transform~$H(t)$ is still given by the sum in~\eqref{AHdef} for~$t$ sufficiently large and satisfies the 
same functional equation.})
The numerical calculation for
the case $\v_1=\v_2=\v_{37}$ gives the value $c_H=c_1=-0.95284748097797917403$ (a value that can be
checked independently, since up to a factor $-\pi$ it must be equal to the even period of the elliptic
curve $y^2-y=x^3-x$, as calculated by GP-PARI as {\tt ellinit}([0,0,1,-1,0])[15]). The corresponding
number for $\v_1=\v_{37}$ and $\v_2=\CK_{37}$ (with $\CK_{37}$ chosen as in 
Table~\ref{Kmphi0}, with~$a=0$) is $c_2= 0.26847600706319893417$, which is not a simple 
multiple of $c_1$, but when we calculate again with $\v_2=\Phi^{0,{\rm III}}_{2,37}$ (again as in Table~\ref{Kmphi0},  
with~$b=0$) we find a value $c_3 = 0.107471184190738587618$ that is related to the two previous values
by $12c_3-19c_2=4c_1$ numerically to high precision. Thus in this case we have indeed found
a form orthogonal to cusp forms and with rational coefficients.  This at first looks promising, but
in fact has a simple and disappointing explanation: the linear combination 
$\Phi^{0,{\rm III}}_{2,37}-\frac{19}{12}\CK_{37}-\frac{1}{3}\v_{37} =\Phi^{0,{\rm I}}_{2,37}-\frac{1}{3}\v_{37}$ 
is nothing other than $\CQ_1|V_{37}$ (as one can see in Table~\ref{Kmphi0},
where setting $c=-\frac{1}{3}$ gives entries equal to $-H(\D)$), and since the function $\CQ_1=-\CH$ is
constructed from an Eisenstein series, albeit a non-holomorphic one, it is automatically orthogonal
to all cusp forms. And indeed, when we look at the case of~$\CQ_{91}$ and~$\v_{91} \in J_{2,91}^{0,--}$, 
where there are no Eisenstein series
in sight, then the calculation of the quotient $(\wh\CQ_{91},\v_{91})/(\v_{91},\v_{91})$ failed to 
yield a simple number. (The non-holomorphy of $\wh\CQ_{91}$ does not affect the application 
of Theorem~\ref{JCPSP}, since its non-holomorphic part has Fourier coefficients only for~$\D<0$, 
where the coefficients of $\v_{91}$ vanish.)
As an incidental remark, we mention that the cusp form $\v_{37}$, whose first few Fourier
coefficients were obtained only with difficulty in~\cite{Eichler:1985ja} (pp.~118--120 and 145),
can now be obtained very easily using the ``theta blocks" of~\cite{GSZ} as
$$  \v_{37}(\t,z) \= \frac1{\eta(\t)^6}\,\prod_{i=1}^{10}\vartheta(\t,a_iz)\,,
   \qquad (a_1,\dots,a_{10})\=(1,1,1,2,2,2,3,3,4,5)\,, $$
and similarly $\v_{2,91}$ can be obtained easily as the difference of the two similarly defined theta blocks 
for the 10-tuples $(a_1,\dots,a_{10})=(1,1,2,3,4,4,5,5,6,7)$ and $(1,1,2,2,3,3,4,5,7,8)$.

\subsection{Observations on the weight one family, integrality, and positivity  \label{specialMTF}}

In the last two subsections we discussed experimental properties of the family of weight~2 mock 
Jacobi forms that were introduced in~\S\ref{simplest} and formulated a number of general results that will
be proved in \S\ref{structure}.  In this subsection we report on the experimental data for the 
weight~1 family and discuss various properties of both families that play a role in connection with
other investigations in mathematics and physics (like ``Umbral Moonshine") or that seem to be of
independent interest.  


The discussion of the weight~1 family is quite different in nature from that of the weight~2 case, because
here there do not seem to be general structural results as there were there, but only a small number of
special functions for small values of the index having nice properties.  In particular, the analogue
of Theorem~\ref{QmOG} no longer holds, and we find instead that (at least up to $m=100$, and 
presumably for all~$m$) the appropriately defined primitive mock Jacobi forms $\Phi_{1,m}^0$  
can be chosen to have optimal growth only for the indices   
   \be\label{mList}  m \= 2,\;3,\;4,\;5,\;6,\;7,\;8,\;9,\;10,\;12,\;13,\;16,\;18,\;25\,. \ee
Similarly, although the analogue of Theorem~\ref{phimQm} is still valid in the 
weight~1 case (now with the special mock Jacobi forms~$\CQ_M$ labelled by square-free integers~$M$ having an
odd rather than an even number of prime factors), the functions $\CQ_M$ can apparently be chosen to have 
optimal growth only for 
\be\label{QMList}  M \= 2,\;3,\;5,\;7,\;13,\;30,\;42,\; 70,\;78\;. \ee
(This list has a precise mathematical interpretation; see {\it 3.} below.)
Thus here a small number of special examples is singled out. Interestingly enough, these
turn out to include essentially all of the special mock theta functions that have been 
studied in the literature in the past. 

For convenience, we divide up our discussion into several topics.

\newpage
\bigskip\noindent{\it 1. Optimal choices for the weight one family  \label{wt1optimal}}
\smallskip

As in the weight~2 case, we begin our discussion by looking at small values of~$m$.  We will 
consider the original forms $\Phi_{1,m}$ first, starting with the ``standard" choices~\eqref{defphi1m}
and trying to modify them by the addition of weak Jacobi forms to attain the property of having optimal growth.  
Actually, it is better to look at the primitive mock Jacobi forms, now defined by
  \be \label{DefPhi01m}
  \Phi_{1,m}^0\;\equiv\;\sum_{d^2|m}\frac{\mu(d)}d\,\Phi_{1,m/d^2}|U_d\,, \qquad
  \Phi_{1,m}\;\equiv\;\sum_{d^2|m}\frac1d\,\Phi^0_{1,m/d^2}|U_d\qquad\mypmod{\wt J_{1,m}}\,, \ee
instead of~\eqref{Phi0mimp} and ~\eqref{defPhi02m}, for the same reason for which we introduced the
functions~$\Phi_{2,m}^0$ in~\S\ref{Statements}: residue arguments similar to those that we will give 
in~\S\ref{structure} for the weight~2 case show here that $\Phi_{1,m}$ can never have optimal growth unless
it agrees with $\Phi_{1,m}^0$ (which happens if and only if $m$ is square-free or the square
of a prime number), so that the latter class is certain to be more productive.  Finally, we
will look at the forms~$\CQ_M$ where~$M$ is a product of an odd number of distinct
prime factors, since according to Theorem~\ref{phimQm} all of the forms~$\Phi_{1,m}$ can be constructed from these.
Unlike the situation for~$k=2$, here we do not have to worry about any choices,
because of the following theorem (to be compared with Lemma~2.3 of~\cite{Cheng:2012tq} and to be proved, 
like all its companions, in~\S\ref{structure}) and corollary.
The only issue is therefore whether the ``optimal growth" condition can be fulfilled at all.

\begin{thm} \label{Wt1OSP}
There are no weight one weak Jacobi forms of optimal growth.  
\end{thm}

\begin{corollary} \label{UniqCor}
If any of the mock Jacobi forms $\Phi_{1,m}$, $\Phi_{1,m}^0$, or $\CQ_{M}$ $(\mu(M)=-1)$ can 
be chosen to have optimal growth, then that choice is unique.
\end{corollary}

We now look at some small individual indices.  For {\bf m=1} there is nothing
to discuss, since the function $\Phi^{\rm st}_{1,1}=(C/2A)^{\rm F}$ vanishes identically, as already
explained in~\S\ref{MJExamples} (Example~1).  (This, by the way, is the reason why $\Phi_{1,m}=\Phi_{1,m}^0$
when~$m$ is the square of a prime.) The case of {\bf m=2} has also already been treated 
in the same place (Example~2), where we saw that 
\be  \label{PHI12} 
\PHI_{1,2} \= \frac{B}{24\,A}\,C\,, \qquad \PHIF_{1,2} \= -\frac1{24}\,\CF_2\,, \ee
with $\CF_2$ the mock modular form defined in~\S\ref{Mock} and related to the Mathieu moonshine story.
For~{\bf m=3}, we find that the (unique) choice of $\Phi_{1,m}$ having optimal growth is given by
 \be  \label{PHI13}
\PHI_{1,3} \= \frac{B^2-E_4A^2}{288A}\,C\,, \qquad \PHIF_{1,3} \= \PHI_{1,3}\,-\,\AP{1,3} \= \frac1{12}\,\CF_3\,, \ee
where the normalizing factor 1/12 has been chosen to give $\CF_3$ integral Fourier coefficients. These
coefficients are determined by virtue of their periodicity and oddness by the values $C^{(3)}(\D)=C(\CF_3;\D,r)$ 
for $r\in\{1,2\}$, the first few values of $C^{(3)}(\D)$ being given by
$$ 
\begin{tabular}{c|ccccccccccc}     $\DD$ & $-1$ & 8 & 11 & 20 & 23 & 32 & 35 & 44 & 47 & 56 & 59 \\ 
 \hline  $C^{(3)}(\DD)$ & $-1$ & 10 & 16 & 44 & 55 & 110 & 144 & 280 & 330 & 572 & 704  \\    \end{tabular}  \;\,. 
$$
(A more extensive table of these coefficients, listed separately for the two  values~$r=1$ and~$r=2$,
is given in Appendix~\ref{apptabs2}, where we also give the corresponding values for the other
forms treated here, so that we no longer will give tables of coefficients in the text.)
For {\bf m=4} and {\bf m=5}, the functions $\Phi_{1,m}$ can again be chosen (of course uniquely, 
by the corollary above) to have optimal growth, the corresponding meromorphic Jacobi forms~$\PHI_{1,4}$
and~$\PHI_{1,5} $ being given by
\be  \label{4and5}
 \frac{B^3-3E_4A^2B+2E_6A^3}{2\cdot 12^2\,A}\,C \quad \text{and}\quad
\frac{B^4-6E_4A^2B^2+8E_6A^3B-3E_4^2A^4}{2\cdot 12^3\,A}\,C\,, 
\ee
respectively. Note that in the case of~$m=4$, according to the decomposition result given in
Theorem~\ref{phimQm}, we could also make the choice $\CQ_2|\CV_{1,2}^{(2)}$ for the function~$\Phi_{1,4}$,
but then it would have the worse discriminant bound $\D_{\rm min}=-4$. We also observe that, although~4
is not square-free, we do not have to distinguish between~$\Phi_{1,m}$ and~$\Phi_{1,m}^0$ here, because
they differ by the function $\Phi_{1,1}^0|U_2$ and the function $\Phi_{1,1}$ vanishes identically.

It is now very striking that the numerators of the fractions in~\eqref{PHI13} and~\eqref{4and5}
are the {\it same} as the expressions that we already saw in eqs.~\eqref{PHI21}, \eqref{defX2}
and \eqref{PSIF4} in~\S\ref{choosephi}, meaning that for these values of~$m$ we have
\be \label{coincidence}
\PHI_{1,m} \= \frac CA\,\psi_{0,m-1}^{\rm opt} \= C\,\PHI_{2,m-2} 
\ee
with the same weight~0 weak Jacobi modular forms $\psi_{0,m-1}^{\rm opt}$ as were tabulated in Table~\ref{Cmcoeffs}.
The same thing happens for several other small values of the index.  We do not know the deeper
reason (if any) for this coincidence, but merely remark that the weight~0 weak Jacobi forms obtained
from the weight~1 meromorphic Jacobi forms in our family by dividing by~$C$ also play a prominent role
in the recent work of Cheng, Duncan and Harvey~\cite{Cheng:2012tq} on ``umbral moonshine," and also that
the the forms $\psi_{0,m}^{\rm opt}$ occur in a somewhat different context in earlier work of 
Gritsenko\cite{Gritsenko:1999fk} related to Borcherds products, as was already mentioned in \S\ref{Statements} when we introduced the ``second standard choice" $\Phi_{2,m}^{0,\rm II}$ of the primitive mock Jacobi form $\Phi_{2,m}^0$.

Continuing to larger values, we look first at primes, since for $m$ prime the three functions of interest
$\Phi_{1,m}$, $\Phi_{1,m}^0$ and $\CQ_m$ all coincide.  For {\bf m=7} we find that $\CQ_7=\PHIF_{1,7}$ again has
optimal growth and again satisfies~\eqref{coincidence} with $\PHI_{0,6}$ as given in Table~\ref{Cmcoeffs}.  
However, for {\bf m=11} the function $\CQ_M$ {\it cannot} be chosen to be of optimal growth; the best possible
choice, shown in the table~\ref{apptabs3} in the appendix, has minimal discriminant~$-4$. For {\bf m=13}
the function $\CQ_{13}$ again has optimal growth and is again given by~\eqref{coincidence},
but for larger primes this never seems to happen again, the minimal value of~$\D$ for the optimal 
choices of~$\CQ_M$ for the first few prime indices being given by
$$ 
\begin{tabular}{c|ccccccccccccccc}     $M$ & 2 & 3 & 5 & 7 & 11 & 13 & 17 & 19 & 23 & 29 & 31 & 37 \\ 
 \hline  $\DD_{\rm min}$ & 1 & 1 & 1 & 1 & 4 & 1 & 4 & 4 & 8 & 5 & 8 & 4  \\    \end{tabular}  \;\,. 
$$

The next simplest case is $m=p^2$, since here $\Phi_{1,m}$ is still primitive.  For {\bf m=9} 
and {\bf m=25} this function can be chosen to have optimal growth,
but apparently not after that. The following case to look at is $m=p^\nu$ with $\nu\ge2$, since here the 
eigendecomposition~\eqref{Qdecomp} of $\Phi_{1,m}$ is trivial.  (The eigenvalue of $W_m$ must be~$-1$ 
because the weight is odd, so $\Phi_{1,m}=\Phi_{1,m}^-=\CQ_p|\CV_{p^{\nu-1}}^{(p)}$.)  Here we find optimal
growth only for {\bf m=8} and {\bf m=16}.  Next, for $M$ square-free but not prime, where the decomposition~\eqref{QmOG}
is trivial, the only two cases giving optimal growth seem to be~{\bf m=6} and {\bf m=10}. And 
finally, for $M$ neither square-free nor a power of a prime we also find two examples {\bf m=12} and {\bf m=18}. 
Altogether this leads to 
the list~\eqref{mList} given above.  The first few coefficients of the forms $\Phi_{1,m}^0$ for all of these values
of~$m$ are tabulated in Appendix~\ref{apptabs2}.
 
However, this is not the end of the story.  If $\Phi_{1,m}$ has optimal growth for some~$m$, then of course
all of its eigencomponents with respect to the Atkin-Lehner operators also do, but the converse is not true,
so there still can be further square-free values of~$M$ for which the function 
$\CQ_M=2^{1-\o(M)}\,\Phi_{1,M}^{-,\dots,-}$ has optimal growth, and this indeed happens.  Since primes
 have already been treated and products of five or more primes are beyond the range of our ability to
compute, we should consider products of three primes. Here we find four examples, for the indices
   $$  30\=2\cdot3\cdot5\,, \qquad 42\=2\cdot3\cdot7\,, \qquad 
     70\=2\cdot5\cdot7\,,\qquad 78\=2\cdot3\cdot13\,,   $$ 
but for no other values less than~200 (and probably for no other values at all), giving
the list~\eqref{QMList}.  Again the
Fourier coefficients of the corresponding forms, as well as those for all other $M<150$ of the
form $p_1p_2p_3$, have been listed in the appendix (part~\ref{apptabs3}).  These examples are of 
particular interest, for a reason to which we now turn.

\bigskip\noindent{\it 2. Relationship to classical  mock theta functions} \smallskip

Recall from \S\ref{basicdefs} that a mock modular form is called a mock theta function if its
shadow is a unary theta series, and that this definition really does include all of Ramanujan's
original mock theta functions that began the whole story.  The theta expansion coefficients $h_\ell(\t)$
of our special mock Jacobi forms $\Phi_{1,m}$ and $\Phi_{2,m}$ are always mock theta functions,
since their shadows are multiples of the unary theta series $\v_{m,\ell}^1$ and $\v_{m,\ell}^0$,
respectively.  In fact, it turns out that many of our examples actually coincide with the mock theta
functions studied by Ramanujan and his successors.  This is the ``particular interest" just mentioned.

Let us consider first the case {\bf m=30}, the first of our $\CQ_M$ examples with composite~$M$.
Because it is so highly composite, there are very few orbits of the group of Atkin-Lehner involutions
on the group of residue classes modulo~$2m$ and hence very few components in the theta decomposition
of~$\CQ_{30}$. In fact there are only two, so that up to sign we get only two distinct coefficients
$h_1$ and $h_7$ and $\CQ_{30}$ has a decomposition
\ben -3\,\CQ_{30}  \= &h_1\,\bigl(\vth_{30,1}+\vth_{30,11}+\vth_{30,19}+\vth_{30,29}
   -\vth_{30,31}-\vth_{30,41}-\vth_{30,49}-\vth_{30,59}\bigr) \cr
 + & \,h_7\,\bigl(\vth_{30,7}+\vth_{30,13}+\vth_{30,17}+\vth_{30,23}
 -\vth_{30,37}-\vth_{30,43}-\vth_{30,47}-\vth_{30,53}\bigr)\,.
\een
(The normalizing factor $-3$ is included for convenience: the function $-3\CQ_{30}$ is the function $\CF_{30}$
defined in~\ref{DefOfcM} below.) The mock theta functions $h_1$ and $h_7$ 
have Fourier expansions beginning
\ben  h_1(\t) &\= q^{-1/120}\,(-1\+q\+q^2\+2q^3\+q^4\+3q^5\+\cdots)\,
, \cr
   h_7(\t) &\= q^{71/120}\,(1\+2q\+2q^2\+3q^3\+3q^4\+4q^5\+\cdots)\,. 
\een
These expansions are easily recognized: up to trivial modifications they are two of the ``mock
theta functions of order~5" in the original letter of Ramanujan to Hardy, namely
\bal  
  2\,+\,q^{1/120}\,h_1(\t) &\= \chi_1(q) \= \sum_{n=0}^\infty\frac{q^n}{(1-q^{n+1})\cdots(1-q^{2n})}\;, \notag \\
 q^{-71/120}\,h_2(\t) &\= \chi_2(q) \= \sum_{n=0}^\infty\frac{q^n}{(1-q^{n+1})\cdots(1-q^{2n+1})}\;.\notag
\end{align}
(Actually, Ramanujan gave five pairs of mock theta functions of ``order~5" in his letter, any two differing 
from one another by weakly holomorphic modular forms, but it was pointed out in~\cite{Zagier:2007} that of the 
five associated pairs of mock modular forms only $(q^{-1/120}(2-\chi_1(q)),\,q^{71/120}\chi_2(q))$ has the 
property that its completion transforms as a vector-valued modular form for the full modular group.)  For {\bf m=42},
exactly the same thing happens, and the three distinct components $h_1$, $h_5$ and $h_{11}$ of $\CQ_{42}$,
whose initial coefficients are tabulated in Appendix~\ref{apptabs3}, turn out to be, up to powers of~$q$ and signs, 
identical with Ramanujan's three ``order~7" mock theta functions $\CF_{7,i}$ ($i=1,\,2,\,3$) that were given as our
first examples of mock modular forms in~\S\ref{basicdefs}.  What's more, this behavior continues: in~\cite{Zagier:2007}
it was mentioned that one could construct mock theta functions of arbitrary prime\footnote{actually, only
$(p,6)=1$ is needed} ``order"~$p$ generalizing these two examples as quotients by $\eta(\t)^3$ of certain
weight~2 indefinite theta series, and an explicit example 
$$  M_{11,j}(\t) \=\frac1{\eta(\t)^3}\,\sum_{m>2|n|/11 \atop n\equiv j\mypmod{11}} 
 \Bigl(\frac{-4}m\Bigr)\,\Bigl(\frac{12}n\Bigr)\,\bigl(m\,\text{sgn}(n)\,-\,\frac n6\bigr)\,q^{m^2/8-n^2/264}
  \qquad(j\in\Z/11\Z)$$
was given for~$p=11$.  Computing the beginnings of the Fourier expansions of these five mock theta functions
(only five rather than eleven because $M_{11,j}=-M_{11,-j}$) and comparing with the table in Appendix~\ref{apptabs3}, 
we find perfect agreement with the coefficients in the theta expansion of~$\CQ_{66}$, the first 
example in that table of a form that is not of optimal growth.

This discussion thus clarifies the precise mathematical meaning of the number designated by Ramanujan as the 
``order" of his mock theta functions, but which he never defined: it is indeed related to the {\it level} of 
the completions of the corresponding mock modular forms, as was clear from the results of Zwegers and other authors 
who have studied these functions, but it is even more closely related to the {\it index} of the mock 
Jacobi forms that have these mock theta functions as the coefficients of its theta expansion.  

Of course the above cases are not the only ones where there is a relationship between the special
mock Jacobi forms~$\CQ_M$ and classically treated mock theta functions, but merely the ones where the
relationship between the ``order" and the index is clearest.  Without giving any further details, we
say only that Ramanujan's order~3 functions are related in a similar way\footnote{but not ``on the nose,'' for instance, our function $h_{3,1}$ differs from Ramanujan's 3rd order mock theta function $-3\, q^{1/12} f(q^2)$ by an eta product.} to our~$\CQ_3$ and that his
order~10 functions (not given in his letter to Hardy, but contained in the so-called ``Lost Notebook")
are related to~$\CQ_5$.  We also refer once again to the beautiful recent work of Cheng, Duncan and 
Harvey~\cite{Cheng:2012tq} on ``umbral moonshine," in which many of these same functions are discussed
in detail in connection with representations of special finite groups.

Finally, we should mention that all of Ramanujan's functions were $q$-hypergeometric series but that 
in general no $q$-hypergeometric expression for the theta expansion coefficients of the functions
$\Phi_{k,m}$ or $\CQ_M$ is known.  This seems to be an interesting subject for further research.
One could look in particular for representations of this sort for the functions $\CQ_{13}$ and
$\Phi_{1,25}^0$ as given in Table~\ref{apptabs2}, since these have particularly small and smoothly
growing coefficients that suggest that such a representation, if it exists, might not be too complicated.
It would also be reasonable to look at $\CQ_{70}$ and $\CQ_{78}$, the only two other known weight~1
examples with optimal growth.


\bigskip\noindent{\it 3. Integrality of the Fourier coefficients} \smallskip

In Subsections~\ref{choosephi} and~\ref{Statements} and above, we encountered several
examples of functions $\CQ_{M}$ of optimal growth that after multiplication by a suitable factor gave 
a function $\CF_M$ having all or almost all of the following  nice arithmetic properties:
\begin{enumerate}[itemsep=0pt,topsep=3pt]
  \item [(A)] All Fourier coefficients $C(\CF_M;\D,r)$ are integral.
  \item [(B)] The polar Fourier coefficient $C(\CF_M;-1,1)$ equals $-1$.
  \item [(C)] The non-polar Fourier coefficients $C(\CF_M;\D,r)$ have a sign depending only on~$r$.
  \item [(D)] $C(\CF_M;\D,r)$ is always positive for $r=r_{\rm min}$, the smallest 
     positive number whose square equals $-\D\mypmod{4m}$. \end{enumerate}
In particular, for $k=2$ this occurred for $M=6$, 10, 14 and 15 (compare eqs.~\eqref{defC6}, \eqref{c10}, 
\eqref{c14} and~\eqref{c15}, which describe the coefficients $C^{(M)}(\D)=C(\CF_M;\D,r_{\rm min})$
for these indices) and for $k=1$, for the values $M=2$, 3, 5, 7, 13 and 30 (with ``non-negative" 
instead of ``positive" in~(C) in the case of~$M=13$), as well as for the functions $\Phi_{1,m}^0$ for 
the prime powers $M=4$, 8, 9, 16 and 25 (again with ``positive" replaced by ``non-positive," and 
also allowing half-integral values for the Fourier coefficients with~$\D=0$).
It is natural to ask for what values of~$M$ these properties can be expected to hold. We will discuss
properties~(A) and~(B) here and properties~(C) and~(D) in point {\it 4.} below.


Clearly either (A) or (B) can always be achieved by multiplying $\CQ_M$ by a suitable scalar factor.
We define a renormalized function~$\CF_M$ in all cases by setting
\be\label{DefOfcM}  \CF_M(\t,z) \= -c_M\,\CQ_M(\t,z)\,,  \ee
where $c_M$ is the unique positive rational number such that $\CF_M$ has integer coefficients with no common
factor.  Then~(A) is always true and (B) and (C) will hold in the optimal cases.
This normalization is also convenient for the presentation of the coefficients in tables, since there are
no denominators and very few minus signs.  Tables of the first few Fourier coefficients of all $\CF_M$ up to~$M=50$
in the weight~2 case (i.e., for $M$ with $\mu(M)=1$) and of all known~$\CF_M$ having optimal growth in the weight~1 
case (i.e., for $M$ with $\mu(M)=-1$) are given in~\ref{apptabs1} and~\ref{apptabs3} of the Appendix, respectively.
Looking at these tables, we see
that the values of $M$ for which all three properties above hold are $M=6$, 10, 14, 15, 21 and 26 in the case $k=2$ 
and $M=2$, 3, 5, 7, 13, 30 and 42 in the case $k=2$, and that in the $k=2$ case we also get $\CF_M$ satisfying~(C)
but not~(B) for $M=22$ and $\CF_M$ satisfying~(B) but not~(C) for $M=35$ and 39.  This list is
not changed if we look at a larger range of~$M$ (we have calculated up to $M=210$ in the weight~2 case and $M=100$
in the weight~1 case) or ~$\DD$ (in each case we have calculated about twice as many Fourier coefficients
as those shown in the tables), so that it seems quite likely that it is complete.

For the weight~2 functions, the corollary of Theorem~\ref{Phi0IIIcoeffs} tells us that $C(\CF_M;-1,1)$ will
 always be negative with the normalization~\eqref{DefOfcM} and that we have 
\be\label{tConj}  c_M \= t \cdot {\text{denom}\big(\frac{\v(M)}{24}\big)}\,, \qquad C(\CF_M;-1,1) 
\= -t\cdot{\text{numer}\big(\frac{\v(M)}{24}\big)}\ee
for some positive rational number~$t$. In all cases we have looked at, $t$ turned out to be an integer, 
and in fact usually had the value~1. For instance, the second statement is true for all values of~$M$ tabulated in 
Appendix~\ref{apptabs}, and also for~$M=210$, where Table~\ref{fcq210} in~\S\ref{Statements} shows that 
$c_{M}=1$ and hence~$t=1$, but the table below giving all values up to~$M=100$ shows that it fails for~$M=65$ and $M=85$.

 \small  
\begin{center}
  \begin{tabular}{|c|c|c|c|c|c|c|c|c|c|c|c|c|c|c|c|c|c}  \hline
       $M$ & 6 &10& 14& 15& 21 & 22 & 26 & 33 & 34 & 35 & 38 & 39 &46  \\  \hline
       $c_{M}$ &12 & 6 & 4 & 3 & 2 & 12 & 2 & 6 & 3 & 1 & 4& 1& 12 \\
       $C(\CF_{M};-1,1)$ &1&1&1&1&1&5&1&5&2&1&3&1&11 \\ 
       $t$&1&1&1&1&1&1&1&1&1&1&1&1&1 \\
    \hline \end{tabular}\quad 

\vspace{0.4cm}

  \begin{tabular}{|c|c|c|c|c|c|c|c|c|c|c|c|c|c|c|c|c|c|c|}  \hline
  51 & 55 & 57 & 58 & 62 & 65 & 69 & 74 & 77 & 82 & 85 & 86 & 87 & 91 & 93 & 94 & 95 \\  \hline
    3 & 3 & 2 & 6 & 4 & 2 & 6 & 2 & 2 & 3 & 6 & 4 & 3 & 1 & 2 & 12 & 1 \\
   4&5&3&7&5&4&11&3&5&5&16&7&7 &3 & 5 & 23 & 3\\ 
   1&1&1&1&1&2&1&1&1&1&2&1&1&1&1&1&1 \\
    \hline \end{tabular}\quad 
    \end{center} 
\normalsize

If it is true that $t$ is integral, then equation~\eqref{tConj} shows that the above list of~$\CF_M$
satisfying property~(B) is complete, because we can then only have $|C(\CF_M;-1,1)|=1$ if $\v(M)$ divides~$24$,
which in turn implies that~$M$ is at most~90 and hence within the range of our calculations. We do not know
whether the integrality is true in general, but in any case~$12^{M+1}t$ must be an integer, since 
correcting the initial choice~$\Phi_{2,M}^{\rm stand} = 12^{-M-1} B^{M+1}A^{-1} \, - \, \AP{2,M}$ 
by weak Jacobi forms cannot make its denominator any worse.

Exactly similar comments apply in the weight~1 case.  Here a residue calculation similar to the ones
described in the next section (but with the residue~$\CR$ as defined in~\eqref{RasRes} below replaced by 
its weight~1 analogue as defined in~\eqref{RasReswt1}) tells us that in the case of optimal growth
eq.~\eqref{Qmpolar} still holds, so that we find that eq.~\eqref{tConj} still holds for some postive rational 
number~$t$.  If this~$t$ turns out to be an integer, it again follows that property~(B) can only hold if
$\v(M)|24$, and the list of all integers~$M$ with $\mu(M)=-1$ satisfying this latter condition coincides exactly with the 
list~\eqref{QMList}.  Thus, although we cannot yet prove that that list is complete, we can at least understand
where it comes from.

\bigskip\noindent{\it 4.~Positivity of the Fourier coefficients} \smallskip

We now turn to properties~(C) and (D). Here we can use the asymptotic formula given in Theorem~\ref{CardyProp}
to understand the question of the signs and asymptotics of the Fourier coefficients. 
(We will discuss only the weight~2 case in detail, but a similar theorem and similar discussion could 
equally well be given for the case~$k=1$.)  This theorem implies that~(C) ``almost holds" in the sense
that for any fixed value of~$r$ the Fourier coefficients $C(\D,r)$ eventually have a constant sign, so
that there are at most finitely many exceptions to the statement.  We therefore have to consider~(D)
and also the question of which values of~$r$ are most likely to give rise to the exceptional signs.

For the function~$\CQ_{M}$ (which we remind the reader differs from~$\CF_{M}$ by a negative proportionality factor,
so that the expected behavior~(D) would now say that~$C(\DD,r)$ is always negative when 
$r=r_{\rm min}$ is the smallest positive square root of~$-\DD$), the value of the constant~$\k_{r}$ 
appearing in that theorem as given according to~\eqref{kappa} and~\eqref{Qmpolar} by
\be \label{didnt}
\k(\CQ_{M},r) \= - \frac{\v(M)}{24} \, \cdot \, 2 \! \sum_{0<\ell < M \atop \ell^{2}\,\equiv1\mypmod{4M}}
\mu\Big( \big(\frac{\ell+1}{2}, M \big) \Big) \cos\Bigl(\frac{\pi\ell r}M\Bigr) \, .
\ee
A similar formula applies in the weight~1 case with ``cos" resplaced by ``sin."  In the case when~$M$
is a prime, this immediately gives the positivity of $\k(\CQ_M,r_{\rm min})$, and in the case
when~$M$ has three prime factors, the tables in~\ref{apptabs3} show that~(D) usually fails.  We therefore 
consider the weight~2 case, where $\mu(M)=+1$.

If~$M$ (which we always assume to be larger than~1) is a product 
of two primes, and in particular for all~$M<330$ with~$\mu(M)=1$ except~210, the sum on
the right-hand side has only two terms and reduces to simply~$\cos(\pi r/M)-\cos(\pi r^{*}/M)$, 
and since~$0<r_{\rm min}<r_{\rm min}^{*}<2M$ and~$\cos(\pi x)$ is decreasing on the 
interval~$[0,\pi]$, the value of~$-\k(\CQ_{M}, r_{\rm min})$ is indeed always strictly negative.  
It then follows from Theorem~\ref{CardyProp} that almost all~$C(\CF_{M};\D, r_{\rm min})$ are positive,
but there can be, and in general are, finitely many exceptions. However, for small~$M$ these 
are fairly rare, e.g., the following table lists all triples~$(M,\,r\!=\! r_{\rm min},\, n)$ within the 
range of Appendix~\ref{apptabs} for which~$c(\CF_{M};n,r)$ is negative.

\begin{table}[h]
\begin{center}
  \begin{tabular}{|c|c|c|c|c|c|c|c|cc|c|c|c|c|}  \hline
       $M$ & 6--22 &$33$& $34$& $35$& $38$ & $39$ & $46$ \\  \hline
       $r$ &--- & $8$ & $15$ & $2$ & $17$ & $10$ & $21$ \\
       $n$ &--- & $1,\,5,\,7$ & $5,\,11$ & $1,\,5$ & $5,\,7,\,11$ & $5$ & $5,\,7,\,11,\,13$  \\ 
    \hline \end{tabular}\quad 
    \end{center} 
\end{table}
\ndt (The corresponding $r$-values are marked with a star in the appendix, and a search 
up to a much larger~$n$-limit yields no further negative coefficients for these values of~$M$.)
Here again the asymptotic formula~\eqref{HRR} helps us to understand this behavior:
the negative values are much more likely to occur when~$\k_{r}$ is small. For instance, 
for~$M=33$ the table

\begin{table}[h]
 \begin{center} \begin{tabular}{|c|c|c|c|c|c|c|c|c|c|c|c|c|c|}  \hline
       $r\;\,(\,=\,r_{\rm min}\,)$ & 1& 2& 4& 5 & 7 & 8 &10&13&16&19 \\  \hline
       $\cos(\frac{\pi r}{66})\vphantom{\Bigl(}-\cos(\frac{\pi r^{*}}{66})$ 
       & $1.58$ & $1.31$ & $1.71$ & $0.94$ & $1.71$ & $0.49$ & $1.58$ & $1.31$ & $0.94$ & $0.49$ \\
    \hline \end{tabular}\quad \end{center}
\end{table}

\ndt shows that the residue classes~$r_{\rm min} \mypmod{2M}$ most likely to give some negative 
Fourier coefficients are~$8$ and~19, for which the values of~$\k_{r}$ (which are the same; this 
duplication always occurs for odd~$M$) are minimal, and indeed the class~$r=8$ did lead to three  
negative coefficients.  The class~$r=19$ did not contain any negative values of~$c(n,r)$, but instead 
achieved its small initial average by the initial zeros forced by the optimal growth 
condition, which says that~$c(n,r)=0$ for~$0 \le n\le (r^2-2)/4m$. 

For values of~$M$ for more than two prime factors on the other hand, it does not have to be true 
that~$\k(\CF_{M}, r_{\rm min})>0$. It is true for the first such index, $M=210$, and even in a very strong form:
not only is~$\k_{r}$ positive for each of the six possible values of~$r_{\rm min}$ (which are nothing other
than the values~$\ell_{0}$ printed in boldface in Table~\ref{fcq210}), but of the seven ``partners'' of 
each~$r_{\rm min}$, i.e., the seven other values of~$r\in[0,210]$ with~$r^{2}\equiv r_{\rm min}^{2} \mypmod{840}$, the three that have the same sign of~$C(\DD,r)$ as for~$C(\DD,r_{\rm min})$ (i.e., which differ 
from~$r_{\rm min} \mypmod{2p}$ for exactly two of the four prime divisors~$p$ of~210) all belong to the 
first half of the interval~$[0,210]$, while the four that have the opposite sign (i.e., which differ 
from~$r_{\rm min} \mypmod{2p}$ for one or three of these~$p$) belong to the second half, 
so that all eight terms in the sum in~\eqref{didnt} are positive. The value of~$\k(\CF_{M},r_{\rm min})$ 
is again positive for all~$r_{\rm min}$ for the next two values~$M=330$ and~$M=390$, but for 
$M=462=2\cdot 3\cdot 7\cdot 11$ and $r\; (=r_{\rm min})\, =\, 31$ we have~$\k(\CF_{M},r)<0$,
so that in this case all but finitely many coefficients $C(\CF_{M};\DD,r)$ will be negative rather than positive. 
What does hold for any~$M$ and any~$r \mypmod{2M}$ is that the values of~$C(\CF_{M}; \DD,r)$ 
eventually all have the same sign, because the coefficient~$\k(\CF_{M},r)$ in~\eqref{HRR} can never 
vanish, since the primitive~$(2M)^{\rm th}$ roots of unity are linearly independent over~$\IQ$.

\subsection{Higher weights} \label{highweight}

So far we have studied two families of meromorphic Jacobi forms $\v_{1,m}(\t,z)$ and $\v_{2,m}(\t,z)$ with 
particularly simple singularities, and the corresponding mock Jacobi forms.  We now discuss in much less detail the analogous families for 
higher weights.  
%
%
We mention that other cases of the decomposition $\v = \v^{F} + \v^{P}$ for meromorphic Jacobi forms with 
higher order poles\footnote{These were considered in~\cite{BringFolsom} while answering a question 
due to Kac concerning the modularity of certain characters of affine Lie superalgebras.} 
\cite{BringFolsom, Olivetto}, and the asymptotics of their Fourier 
coefficients~\cite{BringmannFolsom2}, have been studied in the literature.


The starting point is a meromorphic Jacobi form $\v(\t,z)=\v_{k,m}(\t,z)$ of weight $k$ and arbitrary index $m\in\N$ having as 
its unique singularity (modulo the period lattice $\Z\t+\Z$) a pole of order $k$ at $z=0$, which we normalized by requiring the leading term of $\v$ 
to be  $(k-1)!/(2\pi iz)^k$. By subtracting a linear combination of functions $\v_{k',m}$ with $k'<k$ we can then fix the entire principal part of $\v$ at $z=0$ to be
 $$  \v(\t,z) \= \frac{(k-1)!}{(2\pi iz)^k}\,\sum_{0\le2j<k}\Bigl(\frac {m\,E_2(\t)}{12}\Bigr)^j \,\frac{(2\pi iz)^{2j}}{j!} \;+\;\text O(1)\qquad\text{as $z\to 0$.} $$
We can make a standard choice, generalizing the previously defined forms \eqref{defphi1m} and \eqref{defphi2m}, by taking $\v$ to be the unique polynomial 
in $A^{-1}$ and $B$ (resp.~$C$ times such a polynomial if $k$ is odd) with modular coefficients of the appropriate weight having this principal
part, the next three of these being
  $$  \vst_{3,m} \= \frac{B^m\,C}{12^mA^2}\,, \;\quad  \vst_{4,m} \= \frac{6\,B^{m+2}}{12^{m+2}A^2}\,,\;\quad
   \vst_{5,m} \= \frac{(B^2-mE_4A^2/2)\,B^{m-1}\,C}{12^m\,A^3}\,.   $$
Any other allowed choice of $\v_{k,m}$ will differ from this standard one by a weak Jacobi form of weight~$k$ and index~$m$, and any two choices
will have the same polar part.  To write it down, we introduce Euler's functions
$$ \mathcal E_k(y) \= \AvZ\Bigl[\frac{(k-1)!}{(2\pi iz)^k}\Bigr] \= \Bigl(-y\,\frac d{dy}\Bigr)^{k-1}\AvZ\Bigl[\frac1{2\pi iz}\Bigr]
   \= \bigl(-y\,\frac d{dy}\bigr)^{k-1}\biggl(\frac12\,\frac{y+1}{y-1}\biggr)\;\in\;\Q\Bigl[\frac1{y-1}\Bigr] $$
(which for $k=1$ and $k=2$ coincide with the functions $\mathcal R_0$ and $\mathcal R^{(2)}_0$ from \S\ref{MockfromJacobi}) and their index $m$ averages 
 $$   \AP{k,m}(\t,z) \=  \AvL\Bigl[\frac{(k-1)!}{(2\pi iz)^k}\Bigr] \= \Av m{\mathcal E_k(y)} \=\sum_{s\in\Z} q^{ms^2}y^{2ms}\mathcal E_k(q^s y)  $$
generalizing \eqref{AP1AP2}.
These functions all have wall-crossing, with Fourier expansions given by
$$ \mathcal E_k(y) \= \begin{cases}\phantom-\frac12\,\delta_{k,1}\+\sum_{t>0}t^{k-1}y^{-t} &\text{if $|y|>1$} \\
     -\frac12\,\delta_{k,1}\,-\,\sum_{t<0}t^{k-1}y^{-t} &\text{if $|y|<1$} \end{cases}\;\,=\;\,\sum_{t\in\Z}\frac{\sgn(t)-\sgn(z_2)}2\,t^{k-1}y^{-t}$$
and
\be \label{APkm} \AP{k,m}(\t,z) \= \sum_{s,\,t\in\Z}\frac{\sgn(t)\,-\,\sgn(s+z_2/\t_2)}2\,t^{k-1} q^{ms^2-st} y^{2ms-t}\;.  \ee
(Note that in the latter expression the coefficient of $y^r$ is a Laurent polynomial in~$q$ for every~$r$ and that the discriminant $4nm-r^2$ 
is the negative of a perfect square for each monomial $q^ny^r$ occurring.)  Then the polar part of $\v_{k,m}$ is given by
$$  \v_{k,m}^P(\t,z) \= \sum_{0\le2j<k}\frac{(k-1)!}{(k-1-2j)!\,j!}\,\Bigl(\frac {m\,E_2(\t)}{12}\Bigr)^j \,\AP{k-2j,m}(\t,z)\;. $$
(To see this, observe that by~\eqref{APkm} each of the functions $\AP{k,m}$, and hence any linear combination of them with coefficients depending
only on~$\t$, is in the kernel of the operator $\v\mapsto h_\ell$ defined by~\eqref{Defofhl} for all~$\ell\in\Z$.)
Its finite part $\Phi_{k,m}=\v_{k,m}^F=\v_{k,m}-\v_{k,m}^P$ is a weakly holomoprhic elliptic form that depends on the particular choice of $\v_{k,m}$.
From the results of \S\ref{structure} (specifically, Theorem~\ref{bigweight}) it follows that for $k\ge3$ the function $\v_{k,m}$ can always be chosen 
so that $\Phi_{k,m}$ is strongly holomorphic, this choice then obviously being unique up to the addition of a strongly holomorphic Jacobi form.

As an example, consider $k=3$ and $m=2$. (This is the smallest possibility since index~1 cannot occur in odd weight.)  The finite part of our standard choice is
$$    \vstF_{3,2} \= \frac{B^2C}{144A^2} \,-\,\AP{3,2}\,-\,\frac{E_2}3\,\AP{3,1} \= h_{3,2}^{\rm stand}\,\bigl(\vth_{2,1}\,-\,\vth_{2,3}\bigr) $$
with
$$  h_{3,2}^{\rm stand} \,=\,\frac1{144}\, q^{-1/8}\,\bigl(1+51q+2121q^2+14422q^3+60339q^4+201996q^5+588989q^6+\cdots\bigr)\,.$$
This is weakly holomorphic and has large Fourier coefficients, but by subtracting a multiple of the unique weak Jacobi form $E_4C$ of
this weight and index we get a new choice
$$  \PHIF_{3,2} \= \vstF_{3,2} \,-\,\frac{E_4}{144}\,C\= h_{3,2}^{\rm opt}\,\bigl(\vth_{2,1}\,-\,\vth_{2,3}\bigr) $$
that is strongly holomorphic and has much smaller coefficients
$$  h_{3,2}^{\rm opt} \=  h_{3,2}^{\rm stand} \,-\,\frac{E_4}{144\,\eta^3} \= -\frac43\,q^{-1/8}\,\bigl(q+4q^2+5q^3+11q^4+6q^5+21q^6+\cdots\bigr)\,,$$
in accordance with the theorem just cited.  

However, this choice, although optimal from the point of view of the order of growth of its Fourier coefficients, is not the
most enlightening one.  By subtracting a different multiple of $E_4C$ from $\vstF_{3,2}$, we get a third function
$$  \Phi_{3,2} \= \vstF_{3,2} \+\frac{7E_4}{144}\,C \= \PHIF_{3,2} \+ \frac{E_4}{18}\,C\= h_{3,2}\,\bigl(\vth_{2,1}\,-\,\vth_{2,3}\bigr) $$
with
$$  h_{3,2} \=  h_{3,2}^{\rm stand} \+\frac{7E_4}{144\,\eta^3}  \=  h_{3,2}^{\rm opt} \+\frac{E_4}{18\,\eta^3} 
   \= \frac1{18}\,q^{-1/8}\,(1+219q+2793q^2+15262q^3 +\cdots)\,.$$
This function can be identified: it is given by
$$   h_{3,2} \= \frac13\,\dRS{1/2}(h_2)\,, $$
where $h_2(\t)=q^{-1/8}(-1+45q+231q^2+\cdots)$ is the mock theta function defined in~\eqref{defh2} (the one occurring in the Mathieu moonshine story) and
$\dRS{1/2}$ is the Ramanujan-Serre derivative as defined in~\eqref{RSderiv}. We can also write this at the level of mock Jacobi forms as
$$   \Phi_{3,2} \=  \frac1{24}\,\LL_{1,2}\bigl(\CF_2\bigr) \,, $$
where $\CF_2$ is the mock Jacobi form defined in~\eqref{defCF2} and $\LL_{1,2}$ the modified heat operator defined in~\eqref{LRS}.

An analogous statement holds in general, e.g., for weight 3 we have
$$ \Phi_{3,m} \= -\,\LL_{1,m}\bigl(\Phi_{1,m}\bigr) \quad\pmod{\wt J_{3,m}} $$
for any index $m$ and any choice of the functions $\Phi_{3,m}$ and $\Phi_{1,m}$ in their defining classes, and similarly for higher weight,
the specific formulas for $k\in\{3,4\}$ and the standard choices being
\ben \vstF_{3,m}\+\LL_{1,m}\bigl(\vstF_{1,m}\bigr)&\=&-\,\frac{m-1}{12^m}\,\Bigl(4(m-2)E_6A\+(2m+3)E_4B\Bigr)\,B^{m-3}C\,, \\
  \vstF_{4,m}\+\LL_{2,m}\bigl(\vstF_{2,m}\bigr)&\=&-\,\frac{m+1}{72\cdot 12^m}\,\Bigl(4mE_6A\+(2m-3)E_4B\Bigr)\,B^{m-1}\,. \een
Thus in one sense, the higher-weight families are better than the weight~1 and weight~2 families (because they can be made strongly holomorphic by 
subtracting an appropriate weak Jacobi form), but in another, they are much less interesting (because they are always combinations of derivatives of 
lower-weight examples and ordinary Jacobi forms, so that they do not produce any essentially new functions).

\section{Structure theorems for Jacobi and mock Jacobi forms \label{structure}} 
In the last section we introduced two special families of meromorphic Jacobi forms and the
associated mock Jacobi forms and described various observations concerning them.  In this
section we try to give theoretical explanations for some of these. We begin by studying in
\S\ref{relations} how the holomorphic Jacobi forms of weight~$k$ and index~$m$ lie in the
space of weak Jacobi forms of the same weight and index, or equivalently, what the relations
are among the polar Fourier coefficients (i.e., the coefficients of $q^ny^r$ with $4nm < r^2$)
of weak Jacobi forms.  The results described here are relatively easy consequence of results
in the literature but are, as far as we know, new.  The remaining two subsections describe
the application of these results to mock Jacobi forms in general and to our two special 
families in particular, providing fairly complete explanations of the experimentally
found properties of the weight two family and partial explanations in the weight one case.

\subsection{Relations among the polar coefficients of weak Jacobi forms \label{relations}}

The definitions and main properties of holomorphic and weak Jacobi forms were
reviewed in~\S\ref{Jacobi}, following the more detailed treatment in \cite{Eichler:1985ja}.
Here we investigate the relationship of these two spaces in more detail.

By definition, a weak Jacobi form has coefficients $c(n,r)$ which are non-zero only for $n\ge0$,
whereas a holomorphic one has coefficients that vanish unless $4nm\ge r^2$.  By virtue of
the periodicity property~\eqref{cnrprop}, this stronger condition need only be checked for
$r$ running over a set of representatives for $\Z/2m\Z$, say $-m<r\le r$, and since condition~\eqref{modtransform}
with $\sm abcd=\sm{-1}00{-1}$ also implies that $c(n,-r)=(-1)^kc(n,r)$ and consequently $C(\DD,-r)=(-1)^kC(\D,r)$,
it in fact suffices to consider the coefficients with $0\le r\le m$ if $k$ is even and those with $0<r<m$ if $k$ is odd.
In other words, for all positive integers $k$ and $m$ we have an exact sequence
  \be\label{exseqJkm} 0 \longrightarrow J_{k,m} \longrightarrow \wt J_{k,m}
  \stackrel{P}{\longrightarrow}  \IC^{\CN_{\pm,m}} \qquad(\,(-1)^k=\pm1\,)\,,  \ee
where 
\ben
\CN_{+,m} & \= & \big\{ (r,n) \in \IZ^{2} \, \mid \, 0 \le r \le m, \; 0 \le n < r^{2}/4m  \big\} \; , \\
\CN_{-,m} & \= & \big\{ (r,n) \in \IZ^{2} \, \mid \, 0 < r < m, \; 0 \le n < r^{2}/4m  \big\} \; , 
\een
and where the map $P$ (``polar coefficients'') associates to $\v \in \wt J_{k,m}$ the collection of its
Fourier coefficients $c_{\v}(n,r)$ for $(n,r) \in \CN_{\pm,m}$. In particular we have 
\be\label{dimJkmineq}
\dim J_{k,m} \;\ge\; j(k,m) \; := \; \dim \wt J_{k,m} \, - \, \bigl|\CN_{\pm,m}\bigr| \qquad (\,(-1)^{k} \= \pm1\,)\,, 
\ee
with equality if and only if the map $P$ in \eqref{exseqJkm} is surjective.  (Cf.~\cite{Eichler:1985ja}, 
p.~105.)

Using the description of $\wt J_{*,*}$ given in \S\ref{Jacobi} (eqs.~\eqref{jacovermod} and \eqref{WJF}) and 
the definition of $\CN_{\pm,m}$, we can write this lower bound on the dimension explicitly as 
\be \label{jkmformgen}
  j(k,m)\= \begin{cases} \sum_{0\le j\le m} \big(\dim M_{k+2j}-\ceil{j^2/4m}\big) & \text{if $k$ is even,} \\ 
  \sum_{0 < j < m} \big(\dim M_{k+2j-1}-\ceil{j^2/4m}\big) & \text{if $k$ is odd,} \end{cases} 
\ee
where $\ceil x$ denotes the ``ceiling" of  a real number~$x$ (\,= smallest integer $\ge x$).
This formula is computable and elementary since $\dim M_{k} = k/12 + \CO(1)$ with the $\CO(1)$ term 
depending only on $k\mypmod{12}$, but is nevertheless not a good formula because the two 
terms $\sum \dim M_{k+2j}$ and $\sum \ceil{j^2/4m}$ are each equal to $m^{2}/12 + \CO(m)$, 
and $j(k,m)$ is only of the order of~$m$. In \cite{Eichler:1985ja}, pp.~122--125 it was shown that, 
somewhat surprisingly, the numbers $j(k,m)$ can be given by a completely different-looking formula involving 
class numbers of imaginary quadratic fields.  This formula has the form 
\be \label{jkmformula} j(k,m) \= \frac{2k-3}{24} \, (m+1) \+ \frac{b}{4} \, - \, \frac14 \sum_{d|4m} h(-d) \+ \CO(1) \, ,\ee
where $m=ab^{2}$ with $a$ square-free and where the $\CO(1)$ term depends only on $m$ and $k$ modulo~12.  
In particular, $j(k,m)=\frac{2k-3}{24}\,m \+\CO(m^{1/2+\ve})$ as $m\to\infty$ with $k$ fixed. 
The formula \eqref{jkmformula} has an
interesting interpretation as the dimension of a certain space of classical modular forms of level~$m$ and weight~$2k-2$,
discussed in detail on pp.~125--131 of~\cite{Eichler:1985ja}.

In \cite{Eichler:1985ja}, it was shown that the inequality \eqref{dimJkmineq} is an equality for $k\ge m$  
and it was conjectured that this holds for all $k \ge 3$. This was later proved in \cite{SkoruppaZagier} as a consequence
of a difficult trace formula calculation. In fact, a slightly stronger statement was proved there:
It is clear that the exact sequence \eqref{exseqJkm} remains valid if one replaces $J_{k,m}$ by $J_{k,m}^0$
and the map $P:\wt J_{k,m}\to\C^{\CN_{\pm,m}}$ by the analogously defined map $P^0:\wt J_{k,m}\to\C^{\CN^0_{\pm,m}}$, 
where $\CN^0_{\pm,m}$ is defined exactly like $\CN_{\pm,m}$ but with the strict 
inequality $n<r^2/4m$ replaced by $n \le r^2/4m$. 
Then the dimension formula proved in~\cite{SkoruppaZagier} says that $\dim J_{k,m}^0=\dim\wt J_{k,m}-|\CN^0_{\pm,m}|$ and
hence that the map $P^0$ is surjective, which is a stronger statement than the surjectivity of $P$ since
the set $\CN^0_{\pm,m}$ is larger than $\CN_{\pm,m}$. In summary, the results of \cite{Eichler:1985ja} and \cite{SkoruppaZagier} 
together give the following complete description of the ``polar coefficients" maps for both holomorphic and cuspidal Jacobi forms when $k\ge3$:
\begin{thm}\label{bigweight} For $k>2$ and all $m\ge1$ the map $P^0:\wt J_{k,m}\to\C^{\CN^0_{\pm,m}}$ is surjective and we have a
commutative diagram of short exact sequences 
\bea 0 \longrightarrow J_{k,m}^0   \longrightarrow  & \wt J_{k,m} 
  & \stackrel{P^0}{\longrightarrow}  \IC^{\CN_{\pm,m}^0} \longrightarrow{} 0 \cr \cap \; \; \qquad & || & \qquad  
   \twoheaddownarrow \cr
  0 \longrightarrow J_{k,m} \longrightarrow & \wt J_{k,m} & \stackrel{P}{\longrightarrow}\IC^{\CN_{\pm,m}} \longrightarrow 0 
\eea \end{thm}

We now turn to the case $k=2$, the one of main interest for us.  Here formula \eqref{jkmformula} takes the explicit form
\be \label{j2m} j(2,m) \= \frac12\,\Bigl(\Bigl\lfloor\frac{m+1}3\Bigr\rfloor\+\Bigl\lfloor\frac{m+1}3\Bigr\rfloor\Bigr)
   \,-\, \frac{m-b-1}{4} \, - \, \frac14 \sum_{d\mid 4m} h(-d) \,, \ee
but for this value of $k$ it was observed in \cite{Eichler:1985ja} (p.~131) that equality could not always hold 
in \eqref{dimJkmineq} (indeed, the right-hand side is often negative) and that, for reasons explained in detail there, one
should replace the number $j(k,m)$ by the modified number
\be \label{j2mform} j^{*}(2,m) \;:=\; j(2,m) \+ \delta(m)\,, \;\quad {\rm where} \qquad 
  \delta(m) \;:=\;\sum_{d \le \sqrt{m} \atop d|m, \; d^{2} \nmid m} 1 \; . \ee
We note that $\delta(m)$ is a linear combination of three multiplicative functions of $m$, namely
  \be \label{deltam} \delta(m) \= \frac12 \s_{0}(m) \+ \frac12 \,  \delta_{a,1} \, - \, \s_{0}(b) \ee
(with $m=ab^{2}$, $a$ square-free as before), or equivalently in terms of a generating Dirichlet series
 \be \sum_{m=1}^{\infty} \frac{\delta(m)}{m^{s}} \= \frac12\,\zeta(s)^2 \+\frac12\,\zeta(2s)\,-\,\zeta(s) \zeta(2s) \;. \ee
We also note that $\delta(m)$ depends only on the ``shape" of~$m$ (\,= the collection of exponents 
$\nu_i$ if $m$ is written as $\prod p_i^{\nu_i}$), the first three cases being
\bea
\delta(p^{\nu})  & \= & 0  \; , \label{deltapnu} \\
\delta(p_1^{\nu_1}p_2^{\nu_2}) &\= & n_1n_2  \; , \\
\delta(p_1^{\nu_1}p_2^{\nu_2}p_3^{\nu_3}) & \= & 3n_1n_2n_3 + \ve_1n_2n_3 + \ve_2n_1n_3 + \ve_3n_1n_2\; , 
\eea
in which we have written each $\nu_{i}$ as $2n_i+\ve_i-1$ with $n_i\ge1$ and $\ve_i\in\{0,1\}$. 

The conjecture made in \cite{Eichler:1985ja} that $\dim J_{k,m}=j^*(2,m)$ for all $m\ge1$ was proved 
in \cite{SkoruppaZagier}, but here we want to give a sharper statement describing the cokernel of the map $P$ explicitly. 
To do this, we follow \cite{Eichler:1985ja}, pp.~132, where an elementary argument (containing a small error
which will be corrected below) was given to show that $\dim J_{2,m} \ge j^{*}(2,m)$ for square-free $m$. 
Actually, it is more convenient to work with cusp forms, since the corresponding statement here is
  \be\label{dimJ2m0} \dim J^{0}_{2,m} \= \dim\wt J_{2,m} \, - \, |\CN_{+,m}^0|  \,+\, \delta^{0}(m) \ee
with 
  \be \delta^0(m) \= \sum_{d|m,\;d\le\sqrt m}1 \= \sum_{m=m_{1}\cdot m_{2} \atop {\text{up to order}}} 1 \; , \ee
and $\d^0(m)$ is a slightly simpler function than $\d(m)$. (The third term in \eqref{deltam} is missing.)
We want to explain formula~\eqref{dimJ2m0} by finding an explicit map 
\be\label{defDelta0m}
\b \, : \; \IC^{\CN^{0}_{+,m}} \, = \! \bigoplus_{0\le r \le m \atop 0 \le n \le r^{2}/4m} \IC 
\; \, {\longrightarrow}  \; \, \Delta^{0}(m) :=
  \bigoplus_{m=m_{1}\cdot m_{2} \atop {\text{up to order}}} \IC
\ee
such that the sequence
 \be\label{exseqJ2mcusp} 0 \longrightarrow J_{2,m}^{0} \longrightarrow \wt J_{2,m}
  \stackrel{P^0}{\longrightarrow}  \IC^{\CN^{0}_{+,m}} 
  \stackrel{\b}{\longrightarrow}   \Delta^{0}(m) \longrightarrow 0 \; .
 \ee
is exact.

To give a natural definition of the map $\b$, we work, not with $\wt J_{2,m}$ or even with $J^{\,!}_{2,m}$, but with the larger 
vector spaces $\wt\elliptic_{+,m}$ and $\elliptic^{\,!}_{+,m}$   of even weak or weakly holomorphic elliptic forms of index~$m$, 
as defined at the end of \S\ref{hecke}.  We define a ``residue map'' $\CR:\elliptic^{\,!}_{+,m}\to\C$ by 
\be\label{RasRes}\v \=\sum_{n \ge 0,\,r \in \IZ} c(n,r)\, q^{n} \, y^{r} \quad\mapsto\quad 
  \CR[\v]\=\text{Res}_{\t=\infty}\bigl(\v(\t,0)\,d\t\bigr) \= \sum_{r\in\Z} c(0,r)\;. \ee
Then $\CR[\v] =0$ if $\v \in J^{\,!}_{2,m}$ because in that case $\v(\,\cdot\,,0)$ belongs to $M_{2}^{\,!}$ and 
the $q^{0}$ term of any weakly holomorphic modular form $f$ of weight two is zero by 
the residue theorem applied to the $SL(2,\Z)$-invariant 1-form $f(\t) d\t$. (The map $\CR$ was already used in
\cite{Eichler:1985ja}, but it was erroneously stated there that for $\v\in\wt J_{2,m}$ the function $\v(\t,0)$
itself, rather than just its residue at infinity, vanishes.) Now for each decomposition $m=m_1\cdot m_2$ (up to order, 
i.e., without distinguishing the splittings $m_1\cdot m_2$ and $m_2\cdot m_1$) we define $R_{m_1,m_2}:\,\elliptic^{\,!}_{+,m} \to\IC$ 
by setting $R_{m_1,m_2}(\v)=\CR\bigl[\v\bigl|W_{m_1}\bigr]$ if $(m_1,m_2)=1$, where $W_{m_1}$ is the involution 
defined in \S\ref{hecke}, while if $(m_{1},m_{2})=t >1$ then we define $R_{m_1,m_2}(\v)$ as $R_{m_1/t,m_2/t}(\v|u_t)$, 
where $u_t$ is the map from $\elliptic^{\,!}_{+,m}$ to $\elliptic^{\,!}_{+,m/t^2}$ defined in~\S\ref{hecke}.
(It was pointed out there that both $W_{m_1}$ for $(m_1,m/m_1)=1$ and $u_t$ for $t^2|m$
make sense on all of $\elliptic^{\,!}_{+,m}$, and not just on the subspace of  holomorphic Jacobi forms, but that this
is \emph{not} true for weak elliptic or weak Jacobi forms, which is why we are now working with weakly holomorphic forms.) 
We put all of these maps together to a single map $\,R:\,\elliptic^{\,!}_{+,m}\to \Delta^{0}(m)$ and then 
define the desired map~$\b$ by the commutative diagram 
\bea\label{triangle} 
& \qquad \qquad \wt \elliptic_{+,m}  \quad  \stackrel{P^0}{\longrightarrow} \qquad \IC^{\CN_{+,m}^0} & \cr
& \qquad \cap \qquad \qquad \qquad \cap &\cr
& \qquad \qquad  \elliptic^{\,!}_{+,m}  \quad {\longrightarrow} \qquad \IC^{\CN_{+,m}^{!,0}} &\\
\nonumber & \qquad R\;\searrow \qquad  \swarrow\;\b  & \cr
 &   \quad \qquad \Delta^{0}(m)  & \eea
with $ \IC^{\CN_{+,m}^{!,0}}$ defined as in \eqref{defN!} below. 
(To see that such a map $\b$ exists, observe that the coefficients $c(0,r)$ occurring in~\eqref{RasRes} all
have discriminant $\D=-r^2\le0$, that the operators $W_{m_1}$ and $u_t$ can be defined purely in terms of
the coefficients $C_\v(\D,r)$ and preserve the condition $\D\le0$, and finally that the coefficients
$C_\v(\D,r)$ for $\v\in \wt \elliptic_{+,m}$ can all be expressed in terms of the coefficients $c_\v(n,r)$ with
$(n,r)\in\CN_{+,m}$ by virtue of the periodicity and evenness conditions defining $\wt \elliptic_{+,m}$.)
The discussion above shows that the composition $\b\circ P^0$ in~\eqref{exseqJ2mcusp} vanishes, 
since $\CR$ vanishes on weakly holomorphic and hence in particular on weak Jacobi forms, and combining this
statement with the dimension formula~\eqref{dimJ2m0} proved in~\cite{SkoruppaZagier} we obtain 
the main result describing strong and weak Jacobi forms of weight 2 and arbitrary index:
\begin{thm} \label{exactseqallm}
The sequence \eqref{exseqJ2mcusp} with $\b$ defined by \eqref{triangle} is exact for all integers $m \ge 1$. 
\end{thm}

The two theorems above give a complete description of the space of Jacobi forms of all weights $k \ge 2$.
We end this subsection by discussing the case of weight one.  Here the description of Jacobi forms is very
simple to state (though its proof is quite deep), since we have the following theorem of Skoruppa 
(\cite{Skoruppa}, quoted as Theorem~5.7 of \cite{Eichler:1985ja}):
\begin{thm} \label{Skorthm}
We have $J_{1,m} = \{0\}$ for all $m\ge 1$. 
\end{thm}
Combining this with the discussion above, we find that in this case we have an exact sequence 
 \be\label{exseqJ1m} 0 \longrightarrow \wt J_{1,m} 
 \longrightarrow  \bigoplus_{0 < r < m \atop 0 \le n < r^{2}/4m} \IC 
 \longrightarrow  \IC^{|j(1,m)|} \longrightarrow 0 \; .
 \ee
with 
$$ |j(1,m)| \= -j(1,m) \= \frac m{24} \+\CO\bigl(m^{1/2+\ve}\bigr)\;. $$
We have not been able to find a natural description of the last space, or explicit description of the last map,
in \eqref{exseqJ1m}, analogous to the results given for $k=2$.

\subsection{Choosing optimal versions of weak elliptic forms \label{optimalforms}} 

In this subsection, we will apply the theorem just explained to show how to make ``optimal versions''
of elliptic forms by counting actual Jacobi forms to get poles of small order.

Let $\wt\elliptic_{+,m}$ be as in the previous subsection and $\wt\elliptic_{-,m}$ be the similarly defined 
space with $c(n,-r) = -c(n,r)$, and let $\elliptic_{\pm,m}$ and $\elliptic_{\pm,m}^0$ be the corresponding 
spaces with the condition \eqref{weakjacobi} replaced by \eqref{holjacobi} and \eqref{cuspjacobi},
respectively. The discussion of \S\ref{relations} can be summarized as saying that we have a commutative 
diagram of short exact sequences, 
\bea \label{Jkmdiagram}
0 \longrightarrow \elliptic_{\pm,m}  \longrightarrow  & \wt\elliptic_{\pm,m}
& \stackrel{P}{\longrightarrow}  \IC^{\CN_{\pm,m}} \longrightarrow 0 \cr
  \cup \; \; \qquad & \cup & \qquad  || \cr
0 \longrightarrow J_{k,m} \longrightarrow & \wt J_{k,m} & \longrightarrow  
\IC^{\CN_{\pm,m}} \longrightarrow 0
\eea
for $k>2$, $(-1)^k=\pm1$, and similarly with $J_{k,m}$ and $\elliptic_{\pm,m}$ replaced by
$J_{k,m}^0$ and $\elliptic_{\pm,m}^0$, $\CN_{\pm,m}$ by $\CN_{\pm,m}^0$ and $P$ by $P^0$. 
For $k=2$, we have instead: 
\bea \label{J2mdiagram}
0 \longrightarrow \elliptic_{+,m}^0  \longrightarrow  & \wt\elliptic_{+,m} 
& \stackrel{P^0}{\longrightarrow}  \IC^{\CN_{+,m}^0} \longrightarrow 0 \cr
  \cup \; \; \qquad & \cup & \qquad  || \cr
0 \longrightarrow J^{0}_{2,m} \longrightarrow & \wt J_{2,m} & \longrightarrow  
\IC^{\CN_{+,m}^0} \stackrel\b\longrightarrow  \Delta^{0}(m)
 \longrightarrow 0 
\eea
in the cuspidal case, where $\Delta^{0}(m)$ was defined in \eqref{defDelta0m} 
and~$\b$ is defined by~$\b \circ P^{0}=R$. 
As a consequence of these diagrams and of the vanishing of $\delta(p^{\nu})$ ({\it cf.} 
equation \eqref{deltapnu}), which implies that the diagram \eqref{Jkmdiagram} also holds 
for $k=2$ when $m$ is a prime power, we obtain:
\begin{thm} \label{mpnuthm}
Let $k\ge2$, and $\v \in \wt \elliptic_{\pm,m}$, where $(-1)^{k} = \pm 1$. 
Then $\v$ is the sum of an element of $\wt J_{k,m}$ and an element of $ \elliptic^{0}_{\pm,m}$
if and only if either 
\null \quad \newline{\rm(i)\phantom{ii}} $\quad k>2\,,$ or 
\newline{\rm(ii)\phantom{i}} $\quad k=2$ and $R(\v)\;(=\b(P^0(\v)))\;=0\,;$ \\
If $k=2$ and $m$ is a prime power, then $\v$ is the sum of an element of 
$\wt J_{k,m}$ and an element of~$ \elliptic_{\pm,m}$. 
\end{thm}
Theorem \ref{mpnuthm} for $k=2$ makes a negative statement: if $\v \in \wt \elliptic_{+,m}$
has $\CR[\v]\neq 0$, and $m$ is not a prime power, then $\v$ \emph{cannot} in general be decomposed as the 
sum of a weak Jacobi form of weight 2 and a holomorphic elliptic form. 
However, it turns out that we can come very close to this, writing 
$\v$ as a sum of images under Hecke-like operators of 
forms that ``have optimal growth'' in the sense of \S\ref{Statements}.
We now explain this in two steps. We first use the decomposition \eqref{primdecomp}, 
and its variants for weakly holomorphic Jacobi or elliptic forms,  
to reduce to the case of primitive forms. We then show that any primitive even elliptic form 
can be corrected by the addition of a weakly holomorphic weight 2 Jacobi form 
to give a new elliptic form whose non-zero Fourier coefficients 
$c(n,r)$ all have discriminant $4mn-r^{2}\ge -1$.

Since the decomposition \eqref{primdecomp}, as was already observed in \S 4.3 fails for weak Jacobi forms 
but holds for weakly holomorphic ones, we must first extend the above results to 
weakly holomorphic forms. 
\begin{proposition}
All statements in Theorem~\ref{mpnuthm} remain true \emph{mutatis mutandis} if 
all $\wt \elliptic$, $\wt J$ are replaced by $\elliptic^{!}$, $J^{!}$. 
\end{proposition}
The proposition is equivalent to saying that 
the analog of \eqref{J2mdiagram} still holds 
with weakly holomorphic rather than weak forms, i.e. we have a diagram 
\bea \label{wklyholdiagram}
0 \longrightarrow \elliptic_{+,m}^0  \longrightarrow  & \elliptic^{\,!}_{+,m} 
& \stackrel{P^0}{\longrightarrow}  \IC^{\CN_{+,m}^{0,!}} \longrightarrow 0 \cr
  \cup \; \; \qquad & \cup & \qquad  || \cr
0 \longrightarrow J^{0}_{2,m} \longrightarrow & J^{\,!}_{2,m} & \longrightarrow  
\IC^{\CN_{+,m}^{0,!}} \stackrel\b\longrightarrow  \Delta^{0}(m)
 \longrightarrow 0 
\eea
where 
\be\label{defN!}
\CN_{+,m}^{0,!} \=  \big\{ (r,n) \in \IZ^{2} \, \mid \, 0 \le r \le m, \; -\infty < n \le r^{2}/4m  \big\} \; .
\ee
We again prove this by comparing dimensions, even though the middle 
spaces in each row in the diagram \eqref{wklyholdiagram} are infinite-dimensional. 
Since any elliptic form $\v$ has a minimal value of $n$ with $c_{\v}(n,r) \neq 0$
for some $r$, it suffices to prove the exactness of the rows in \eqref{wklyholdiagram} 
with $\elliptic_{+,m}^{\,!}$, $J^{\,!}_{2,m}$, and $\IC^{\CN_{+,m}^{0,!}}$ replaced by the 
spaces $\elliptic_{+,m}^{\ge N}$, $J^{\ge N}_{2,m}$, and $\IC^{\CN_{+,m}^{0,\ge N}}$,  
of forms with non-zero coefficients only for $n \ge -N$ for some $N \ge 0$. 
The only non-trivial part is  the exactness at 
$\IC^{\CN_{+,m}^{0,\ge -N}}$, and since the composition of the maps 
$P^{0}$ and $\b$ vanishes, we 
only need to check that the alternating sum of dimensions is zero. For this,
we observe that the dimensions of $J^{\ge -N}_{2,m}$ and $\IC^{\CN_{+,m}^{0,\ge -N}}$ 
are both larger by $(m+1)N$ than their values at $N=0$: 
for $\IC^{\CN_{+,m}^{0,\ge -N}}$, this is obvious since $\CN_{+,m}^{0,\ge -N}$ 
differs from $\CN_{+,m}^{0}$ by the addition of $N$ rows of length $m+1$, and 
for $J^{\ge -N}_{2,m}$ it follows from the computation
\bea
\dim J_{2,m}^{\ge -N} & \= & \dim \wt J_{12N+2,m} \qquad  
   \; \big(\,\text{because $J_{2,m}^{\ge -N} = \Delta(\t)^{-N} \wt J_{12N+2,m}$}\,\big)\cr 
& \= & \sum_{j=0}^{m} \dim M_{k+12N+2j}  \=  \dim \wt J_{2,m} \+ (m+1)N \, . 
\eea
This completes the proof of the proposition. \hfill $\square$

Now we can apply the decomposition \eqref{primdecomp} to the whole diagram \eqref{wklyholdiagram}
to find that everything can be reduced to primitive forms (and their images under the $U_{t}$ operators). 
The key point is that although $U_{t}$ and $u_{t}$ change the discriminant $\Delta = 4 mn -r^{2}$
by multiplying or dividing it by $t^{2}$, respectively, they do not change its sign, and therefore all 
of the spaces in our diagram can be decomposed as direct sums of primitive parts of lower index 
and their images under $U_{t}$. This reduces everything to case of primitive forms.

For primitive forms our diagram becomes
\bea \label{wklyholdiagramprim}
0 \longrightarrow \elliptic_{+,m}^{0, \rm prim}  \longrightarrow  & \elliptic^{!, \rm prim}_{+,m} 
& \stackrel{P^0}{\longrightarrow}  \IC^{\CN_{+,m}^{0,!, \rm prim}} \longrightarrow 0 \cr
  \cup \; \; \qquad & \cup & \qquad  || \cr
0 \longrightarrow J^{0, \rm prim}_{2,m} \longrightarrow & J^{!, \rm prim}_{2,m} & \longrightarrow  
\IC^{\CN_{+,m}^{0,!, \rm prim}} \stackrel\b\longrightarrow  \Delta^{0, \rm prim}(m)
 \longrightarrow 0 
\eea
where 
\be\label{delta0prim}
\Delta^{0, \rm prim}(m) =  
\bigoplus_{{m=m_{1}\cdot m_{2} \atop (m_{1}, m_{2}) =1} \atop {\text{up to order}}} \IC \;.
\ee
The map $\b$ in \eqref{wklyholdiagramprim} is again defined by the requirement that $\b \circ P^{0} = R$ on 
$\elliptic_{+,m}^{!,\rm prim}$, and lands in $\Delta^{0,\rm prim}(m)$ because $\v|u_{t}=0$ for $\v$ primitive and $t>1$.
The space $\Delta^{0, \rm prim}(m)$ has dimension $\delta^{0, \rm prim}(m) = 2^{s-1}$ if $m>1$, where 
$s$ is the number of primes dividing $m$. (The failure of this formula for $m=1$ corresponds to 
the second term in \eqref{deltam}.)

We also have the non-cusp-form analogs of the diagrams \eqref{wklyholdiagram} and 
\eqref{wklyholdiagramprim}, where all superscripts ``0'' have been dropped and
where $\CN_{+,m}^{\,!}$ is defined as in \eqref{defN!} but  
with the condition $4 mn - r^{2}  \le 0$ replaced by the strict inequality $4 mn - r^{2} <0$, 
while $\Delta^{\rm prim}(m)$ is a space of dimension 
$\delta^{\rm prim}(m)=\delta^{0, \rm prim}(m)-1 = 2^{s-1}-1$ (compare equation \eqref{deltam}). 
This space is the quotient of the space in \eqref{delta0prim} by the one-dimensional 
subspace of vectors with all components equal. This is essentially the idea of the proof of 
Theorem~\ref{chooseKm} given below.

Having completed these preparations, we now proceed to the proofs of our main results.
We consider the case of square-free index first and then give a brief description of the 
modifications needed to extend them to primitive forms of arbitrary index.

For $m>1$ square-free, we will show that any $\v\in\wt\elliptic_{+,m}$ can be corrected by a 
weak Jacobi form of weight 2 to get a another elliptic form whose non-zero Fourier coefficients 
$c(n,r)$ all have discriminant $4mn-r^{2}\ge -1$, i.e., that we have a decomposition
\be \label{thJminone}
\wt\elliptic_{+,m} \= \wt J_{2,m} \+ \elliptic^{\rm OG}_{+,m} \, \qquad \text{($m$ square-free),} 
\ee 
with $\elliptic^{\rm OG}_{+,m}$ defined as in \S\ref{Statements}.
Actually, this is still not quite the statement we want, since it is slightly wasteful:
the codimension of $\elliptic_{+,m}^0$ in $\elliptic_{+,m}\NH$ is $2^{s-1}+1$,  where~$s$
as before is the number of prime factors of~$m$ (there are $2^{s-1}$ pairs
$(r,n)\in\CN_{+,m}^0$ with $\D:=4nm-r^2$ equal to $-1$, and one pair (0,0) with $\D=0$), and there are 
only $\dim \DD^{0}(m)=2^{s-1}$ constraints $R_{m_1,m_2}(\v)=0$ coming from 
diagram~\eqref{J2mdiagram}.  Correspondingly,
the intersection of the two spaces on the right-hand side of~\eqref{thJminone} contains $J_{2,m}$
($=J_{2,m}^{0}$ in this case)
as a subspace of codimension~1.  To get a sharp statement, we should therefore replace~$\elliptic_{+,m}\NH$
by a suitable codimension one subspace.  There are three natural ways to do this, each of which
leads to a sharpening of~\eqref{thJminone}, corresponding exactly (in the square-free case) 
to the various choices of $\Phi_{2,m}^{0}$ introduced in the discussion preceding Theorem~\ref{PolarCoeffs}.
What's more, all statements remain true verbatim for arbitrary~$m>1$ if we restrict to primitive forms. 
We state all three variants in a single theorem.
\begin{thm} \label{Emstar}
Suppose that $m>1$ is square-free and let $\elliptic^{\rm OG,\star}_m\;\,(\star=\rm I,\,\rm II,\, \rm III)$ denote 
the codimension one subspace of $\elliptic_{+,m}\NH$ defined by one of the three conditions 
\newline\null\qquad {\rm(I)} $\quad \sum_{0<r<m,\;r^2\equiv1\pmod{4m}} C_\v(-1;r)=0$,
\newline\null\qquad {\rm(II)} $\quad c_\v(0,1)=0$,
\newline\null\qquad {\rm(III)} $\quad c_\v(0,0)=0$,
\newline respectively.  Then we have an exact sequence
\be
0 \longrightarrow J_{2,m} \stackrel{\text{\rm diag}}\longrightarrow 
\wt J_{2,m} \oplus \elliptic_m^{\rm OG, \star} \stackrel-\longrightarrow \wt\elliptic_{+,m} \longrightarrow 0 \; ,
\ee
i.e., any even weak elliptic form of index~$m$ can be decomposed as the sum of a weak Jacobi form 
of weight~2 and index~$m$ and an element 
of $\elliptic_m^{\rm OG, \star}$, and this decomposition is unique up to a holomorphic Jacobi form 
$($which can be added to one term and subtracted from the other$)$. 

For~$m>1$ the same statement remains true if we replace~$\wt \elliptic_{+,m}$ by its 
subspace~$\wt \elliptic_{+,m}^{\rm res-prim}$ of forms~$\v$ satisfying~$R(\v|u_{t})=0$
for all~$t>1$ and~$\elliptic_m^{\rm OG, \star}$ by its intersection with~$\wt \elliptic_{+,m}^{\rm res-prim}$. 
\end{thm}

\begin{proof}
By the exactness of the top line of diagram \eqref{J2mdiagram}, we can find functions 
$\psi_{m;(r,n)} \in \wt \elliptic_{+,m}$, unique up to the addition of cusp forms, mapping to the canonical basis 
$\d_{(r,n)}$ of $\IC^{\CN_{+,m}^{0}}$. 
Clearly, the space $\elliptic_{+,m}^{\rm OG,\rm III}/\elliptic_{+,m}^{0}$ is spanned by 
$\psi_{m;(r,n)}$ with $0< r < m$ and $4mn-r^{2}=-1$, and has dimension $2^{s-1}$. 
(There are exactly two square roots of 1 modulo~$2P$ for each exact prime power divisor~$P$ of~$m$ and 
these can be combined arbitrarily by the Chinese remainder theorem. This gives~$2^{s}$ 
residue classes $r \mypmod{2m}$ with $r^{2} \equiv 1 \mypmod{4m}$, and if we choose the 
representatives of smallest absolute values then half of them are positive.) 
Moreover, the basis elements can be chosen to be a single orbit of the group 
$\{W_{m_{1}} \}_{m_{1} | m}$ 
introduced in \S\ref{hecke}, whose cardinality is 
the same number $2^{s-1}$. Namely, if $A_{m}$ is any choice of $\psi_{m;(1,0)}$, i.e. 
any function in $\wt \elliptic_{+,m}$ whose unique non-zero coefficient 
in $\CN^{0}_{+.m}$ is $c(0,1)=1$, then the other functions $\psi_{m;(r,n)}$
with $(r,n) \in \CN^{0}_{+,m}$ and $4 mn-r^{2} = -1$, 
can be chosen to be the images of $A_{m}$ under the Atkin-Lehner involutions $W_{m_{1}}$, 
$m=m_{1}m_{2}$. 
More explicitly, with this choice $\psi_{m;(r,n)}$ equals $A_{m}|W_{m_{1}}$ 
with~$m_{1}=\gcd\big(\frac{r+1}{2},m\big)$, $m_{2}=\gcd\big(\frac{r-1}{2},m\big)$,
and conversely~$A_{m}|W_{m_{1}}$ for any decomposition~$m=m_{1}m_{2}$ equals
$\psi_{m;(|r_{0}|,(r_{0}^{2}-1)/4m)}$ 
where $r_{0}$ ($=1^{*}$ in the notation of \S\ref{hecke}) is the unique $r_{0}$ satisfying
\be \label{defr0}
r_{0} \equiv -1 \mypmod{2m_{1}} \; ,  \qquad r_{0} \equiv 1 \mypmod{2m_{2}}  \; , \qquad  |r_{0}| < m\; . 
\ee

On the other hand, for any $\v \in \elliptic_{+,m}^{\rm OG, III}$, we have 
\be\label{Ronem}
R_{1,m}(\v) \= \CR[\v] \=2 c_{\v}(0,1) \; 
\ee
(because $c_{\v}(0,-1) = c_{\v}(0,1)$ and $c_{\v}(0,r)=0$ for all $r \neq \pm 1$
by the defining property of~$\elliptic_{+,m}^{\rm OG, III}$), and more generally, 
\be
R_{m_{1}, m_{2}}(\v) \=2 \, c_{\v}(r_{0},(r_{0}^{2}-1)/4m) \; 
\ee
for each decomposition $m=m_{1}m_{2}$, and with $r_{0}$ as in \eqref{defr0}. 
It follows that, for any two decompositions $m=m_{1}m_{2}$ and $m=m_{1}'m_{2}'$,
we have 
\be
R_{m_{1},m_{2}}\big(A_{m} \! \mid \! W_{m_{1}'} \big)  
\= \begin{cases} 
2 \qquad \text{if $m_{1}'=m_{1}$ or $m_{1}'=m_{2}$} \; , \\
0 \qquad {\rm otherwise} \; .
\end{cases}
\ee
This shows that the map  
\be
\elliptic_{+,m}^{\rm OG, III}/\elliptic_{+,m}^{0} \stackrel{R}{\longrightarrow} 
\bigoplus_{m=m_1\cdot m_2 \atop {\text{up to order}}} \IC \, 
\ee
is an isomorphism, and in combination with Theorem \ref{mpnuthm}, completes the proof 
of the theorem in the case of square-free~$m$ and~$\star=\rm III$: any~$\v \in \wt\elliptic_{+,m}$ 
can be decomposed, uniquely up to Jacobi cusp forms, as the sum of 
$\frac12 \sum_{m=m_{1}m_{2}} R_{m_{1},m_{2}}(\v) A_{m}|W_{m_{1}} \in \elliptic_{+,m}^{\rm OG, III}$ 
and an element of~$\wt J_{2,m}$.

The proofs for $\star=\rm I$ and $\star=\rm II$ are almost identical,
changing the basis $\{\psi_{m;(r,n)} = A_{m}|W_{m_{1}}\}$ of~$\elliptic_{+,m}^{\rm OG, III}/\elliptic_{+,m}^{0}$
in both cases in such a way as to satisfy the new defining condition without changing the residues, i.e., 
in the case~$\star=\rm II$ by replacing the first basis element~$\psi_{m;(1,0)}$ by~$2 \psi_{m;(0,0)}$, and in 
the case~$\star=\rm I$ by adding $2 \psi_{m;(0,0)} - 2^{-s+1} \sum_{r'} \psi_{m;(r',n)}$ to each basis
element. 

The proof of the statement for forms of arbitrary index~$m$ is very similar,  since there are still~$2^{s-1}$ polar 
coefficients of a form in~$\elliptic_{+,m}\NH$, in one-to-one correspondence with the decompositions of~$m$ 
into two coprime factors. The details are left to the reader.
\end{proof}
Theorem~\ref{Emstar} says that the first line in the diagram 
\bea \label{primdiag}
0 \longrightarrow  \; \; J_{2,m}^{\rm prim}  \quad \stackrel{\text{\rm diag}}\longrightarrow 
& \wt J_{2,m}^{\rm prim} \oplus \elliptic_m^{\rm OG, \rm prim, \star} & \stackrel-\longrightarrow 
\wt\elliptic_{+,m}^{\rm prim} \longrightarrow 0  \cr
  \cap  \qquad \qquad &  \;  \cap  \qquad & \qquad  || \\
0 \longrightarrow  J_{2,m}^{\rm OG, prim}   \stackrel{\text{\rm diag}}\longrightarrow 
&\wt J_{2,m}^{\rm prim} \oplus \elliptic_m^{\rm OG, \rm prim} & \stackrel-\longrightarrow 
\wt\elliptic_{+,m}^{\rm prim} \longrightarrow 0  \nonumber
\eea
is exact. The exactness of the second line is a trivial consequence, and the fact 
that~$\elliptic_{+,m}^{\rm OG, prim, \star}$ has codimension exactly one in~$\elliptic_{+,m}^{\rm OG, prim}$ 
implies that~$J_{2,m}^{\rm prim}$ has codimension one in~$J_{2,m}^{\rm OG, prim}$.
Moreover, the proof of the theorem also shows that~$J_{2,m}^{\rm OG, \star} = J_{2,m}$. 
(If~$\v\in J_{2,m}\NH$, then the vanishing of the residues~$R_{m_{1}m_{2}}(\v)$ for all decompositions 
with~$(m_{1},m_{2})=1$ implies that the polar coefficients~$C_{\v}(-1,r)$ are all equal to one another 
and to~$-\half c_{\v}(0,0)$, so that imposing any of the conditions~$(\rm I)$, $(\rm II)$, $(\rm III)$ 
forces~$\v$ to be holomorphic. Conversely, if~$\v\in J_{2,m}$, then the conditions $\rm OG$, (I), (II) are 
trivial, and (III) holds because $\CR[\v]=0$.) It follows that the codimension of~$J_{2,m}$ 
in~$J_{2,m}^{\rm OG}$ is also one.

\subsection{The residues of mock Jacobi forms} \label{MJFappl}

We can apply the general theorems of the previous subsection 
to mock Jacobi forms. In particular, if we have any weak mock 
Jacobi form of weight $k\ge3$ and index $m$, it will be in $\wt \elliptic_{\pm,m}$, and can be 
corrected by a weak Jacobi form to get a holomorphic mock Jacobi form. 
In the case $k=2$, there is an obstruction $R(\v)$ to this being true, but, by the way we defined it, this
obstruction can only depend on the shadow of~$\v$. In this subsection, we show how to compute $R(\v)$ 
explicitly in terms of the shadow. Since $R(\v)$ is the vector with components $\CR\bigl[\v|u_{(m_1,m_2)}|W_{m_1/(m_1,m_2)}\bigr]$,
it suffices to compute $R[\v]$ for a weak mock Jacobi form of weight~2 in terms of the shadow of~$\v$.
 
Let $\v(\t,z)$ be a (strong or weak) mock Jacobi form of weight 2 and index $m$. Then, by definition, $\v$ has a 
theta expansion \eqref{jacobi-theta} where each $h_{\ell}$ is a (in general weakly holomorphic) mock modular form 
of weight~3/2 and has a completion $\wh h_{\ell} = h_{\ell} + g_\ell^*$ for some modular form $g_{\ell}$ of weight~1/2 
such that the completion 
$\wh \v = \sum_{\ell} \, \wh h_{\ell} \, \vth_{m,\ell} = \v + \sum_{\ell} \, g_{\ell}^{*} \, \vth_{m,\ell}$
transforms like a Jacobi form of weight 2 and index $m$.

As usual, we denote by $\vth_{m,\ell}^0(\t)$ the Thetanullwert $\vth_{m,\ell}(\t,0)$. 
\begin{thm} \label{Rphishad}
Let $\v$ be as above. Then
\be \label{Rphivalue}
\CR[\v] \= \frac{\sqrt{\pi}}{6} \, \sum_{\ell \mypmod{2m}} \bigl(\vartheta_{m,\,\ell}^{0}, g_\ell \bigr)\,, 
\ee
where $(\,\cdot\,,\,\cdot\,)$ denotes the Petersson scalar product in weight 1/2.
\end{thm}
\begin{proof}
With the normalization of the Petersson scalar product as given in \eqref{Peterkabeta},
we have
\be
\sum_{\ell \mypmod{2m}}  \bigl(\vartheta_{m,\ell}^{0}, g_\ell \bigr) \= \frac{3}{\pi} \int_{\CF} \, F(\t) \, 
d \mu(\t) \; , 
\ee
where
\be \label{defFtau}
F(\t) \=  \t_{2}^{1/2}\, \sum_{\ell \mypmod{2m}}  \overline{g_{\ell}(\t)} \,  \vth_{m,\ell}^{0}(\t) \; , 
\ee
and $\CF$ denotes a fundamental domain for $\G \equiv SL_{2}(\IZ)$. This makes sense since the function 
$F(\t)$ is $\G$-invariant (the individual terms in \eqref{defFtau} are not, which is why we had to 
be a little careful in the normalization of the scalar product) and is convergent because the weight of the
modular forms $g_\ell$ and $\vth_{m,\ell}^0$ is less than~1. From eq.~\eqref{starinv} it follows that the
$\bar\t$-derivative of the completion
 \be\label{phihatismock}
  \wh\v(\t,0) \= \v(\t,0) \+ \sum_{\ell \mypmod{2m}}  g_{\ell}^{*}(\t) \,  \vth_{m,\ell}^{0}(\t) \, , \ee
 of $\v(\t,0)$ satisfies
\be \label{Fisshadow} \frac{\p}{\p \bar\t} \, \wh \v (\t,0) \= \frac{1}{4i\sqrt\pi} \frac{F(\t)}{\t_2^{2}} \,. \ee

The fact that $\wh \v(\t,0)$ transforms like a modular form of weight 2 on $\G$ 
implies that the differential one-form $\w = \wh \v(\t,0) \, d\t$ is $\G$ invariant, and 
equation \eqref{Fisshadow} implies 
\be  d\o \= d\big(\wh\v(\t,0)\,d\t\big) \= -\frac{\p\,\wh\v(\t,0)}{\p \tbar} \, d\t \, d\tbar 
\= \frac{1}{2 \sqrt\pi} \, F(\t) \, d \mu(\t) \; .
\ee
Therefore, by Stokes's theorem, we have 
\be
\frac{\sqrt\pi}{6} \, \sum_{\ell \mypmod{2m}}  \bigl(\vartheta_{m,\ell}^{0}, g_\ell \bigr) \= \int_{\CF} d \o \= \int_{\p \CF} \w 
\= \lim_{T \to \infty} \int_{\p \CF(T)} \w
\ee
where $\CF(T)$ is the truncated fundamental domain 
\be
\{\t \in \CH \mid |\t|>1\, , \; |\t_{1}| < \frac12 \, , \; \t_{2} < T \} \; .
\ee
On the other hand, we have 
\be
\int_{\p \CF(T)} \w \= \int_{-\frac12 + iT}^{\frac12 + iT} \w
\= \int_{-\frac12 + iT}^{\frac12 + iT} \v(\t,0) \, d\t  \+ O\big(\frac{1}{\sqrt{T}}\big) \; .
\ee
The first equality holds because $\w$ is $\G$-invariant, and 
all the edges of $\CF(T)$ except the top edge $\t_{2}=T$
come in pairs which are $\G$-equivalent, but of opposite orientation. 
The second equality follows from \eqref{phihatismock}, and 
because $\vth^{0}_{m,\ell}$ is $O(1)$ and $g_{\ell}^{*}(\t)$ is $O(T^{-1/2})$ by \eqref{defstar} 
with $k=3/2$. The theorem then follows.
\end{proof}
\ndt The theorem immediately implies a more general formula for the residue 
$R_{m_1,m_2}(\v) = \CR\bigl[\v\bigl|W_{m_1}\bigr]$ when $m=m_{1}m_{2}$ with 
$(m_{1}, m_{2})=1$  in terms of Petersson scalar products, namely, 
\be \label{Rm1m2ans}
R_{m_1,m_2}(\v)  \= \frac{\sqrt{\pi}}{6} \, \sum_{\ell \mypmod{2m}} \bigl(\vartheta_{m,\,\ell}^{0}, g_{\ell^{*}} \bigr)\,, 
\ee
where $\ell \mapsto \ell^{*}$ is the involution on $(\IZ/2m\IZ)^{*}$ appearing in the definition of~$W_{m_{1}}$.

\subsection{Remaining proofs of theorems from \S\protect\ref{flyJac} \label{proofs}}

In this subsection, we apply the results proved so far to the special family studied in~\S\ref{flyJac}, namely,
the meromorphic Jacobi forms~$\v_{2,m}$ having a pole~$1/(2 \pi i z)^{2}$ at the origin
and their associated mock Jacobi forms~$\Phi_{2,m}$. Specifically, we shall prove 
Theorems \ref{phimQm}, \ref{QmOG}, \ref{chooseKm}, \ref{PolarCoeffs}, and \ref{Wt1OSP}.

We begin with a preliminary remark.
Two of the results we need to prove are Theorem~\ref{phimQm} and the statement that the 
functions~$\Phi_{2,m}^{0}$ 
defined implicitly by~\eqref{Phi0mimp} are primitive modulo weak Jacobi forms. Both of these have the form 
that a certain elliptic form is in fact a (weak) Jacobi form, i.e., that its shadow vanishes. Since the forms 
in question are defined using the Hecke-like operators defined in~\S\ref{hecke}, we first have to 
explain how the action of these operators extends to non-holomorphic elliptic forms, since then 
we can apply them to the completions of our mock Jacobi forms and verify the necessary identities 
between shadows. The definitions of the Hecke-like operators were given in terms of the Fourier coefficients~$c(n,r)$ 
or $C(\DD,\ell)$ defined by~\eqref{fourierjacobi} and \eqref{cnrprop}.
In the case of the completions of mock modular forms~$\Phi$, when we apply equation~\eqref{defUl} or~\eqref{defVl} 
for the action of $U_{t}$ or $V_{t}$ with the~$\v$ there replaced by~$\Phi^{C}:=\wh \Phi - \Phi$ 
(the ``correction term'' discussed at the end of~\S\ref{MockDecomposition}), 
we must interpret~$c(n,r)$ as the 
coefficient of $\b(|n|\t_{2}) \, q^{n} \, y^{r}$ rather than simply $q^{n} y^{r}$ as in the holomorphic case,
where $\b(t)=\sqrt{\pi/4m} \; \erfc(2\sqrt{\pi t})$ in the weight~2 case, 
and similarly for equation~\eqref{Ucech2} and $C(\DD, \, r\! \mypmod{2m})$. 
In view of~\eqref{phi2mhat} and~\eqref{defstar}, these coefficients in the case of~$\Phi_{2,m}$ are given by  
\be \label{phiCcfs}
C(\Phi^{C}_{2,m} ; \DD , \, \ell \mypmod{2m}) \= \sum_{\l^{2} = - \DD \atop  \l \equiv \ell \mypmod{2m}} |\l| \; . 
\ee

\bigskip

\ndt {\bf Proof of Theorem~\ref{phimQm}}:
We have to show that for $m=P_{1} \cdots P_{r}$ ($P_{i} = p_{i}^{\nu_{i}}$, $p_{i}$ distinct primes) 
we have 
\be
\Phi_{2,m} \, \big| \, \prod_{i=1}^{r} (1+\ve_{i} W_{P_{i}}) \, \equiv \, 
\Bigl(\Phi_{2,M} \, \big| \, \prod_{p|M} (1-W_{p}) \Bigr) \,\bigl|\,\CV_{2,m/M}^{\,(M)} \, , 
\ee
where $M$ is the product of the primes~$p_{i}$ for which~$\ve_{i}=-1$ and ``$\, \equiv \, $''
means that the two sides of the equation differ by a weak Jacobi form. Because of the 
multiplicativity of the Hecke-like operators $U_{d}$, $V_{t}$ and $W_{m_{1}}$, this formula 
can be rewritten as 
\be
\Phi_{2,m} \, \big| \, \prod_{i=1}^{r} (1+\ve_{i} W_{P_{i}}) \, \equiv \, 
\Phi_{2,M} \, \big| \,  \prod_{\ve_{i}=-1} (1-W_{p_{i}}) \CV_{P_{i}/p_{i}}^{(p_{i})} \, \big| \, 
\prod_{\ve_{i}=+1} \CV_{P_{i}}^{(1)} \, , 
\ee
so by induction on $r$, it suffices to show that if $P=p^{\nu} || m$ then
\be\label{induction}
\Phi_{2,m} \! \mid \! (1\+W_{P}) \, \equiv \, \Phi_{2,m/P}  \! \mid \! \CV_{2,P}^{(1)} \, , \qquad 
\Phi_{2,m} \! \mid \! (1\, - \, W_{P}) \, \equiv \, \Phi_{2,mp/P}  \! \mid \! (1-W_{p}) \, \CV_{2,P/p}^{(p)} \, .
\ee 
For this, it is enough to show that in each the difference of the left-hand side and the right-hand side 
has shadow zero. 

For the first equation in \eqref{induction}, eq.~\eqref{phiCcfs} implies that 
the ``correction term'' for the left-hand side is given by 
\bea\label{LHS1}
&& C(\Phi^{C}_{2,m} | (1+W_{P}) \, ; \DD , \ell \mypmod{2m})  =  C(\Phi^{C}_{2,m} \, ; \DD , \ell \mypmod{2m}) + C(\Phi^{C}_{2,m} \, ; \DD , \ell^{*} \mypmod{2m}) \cr
&& \; =  \begin{cases} D (\d^{(2m)}_{D,\ell} + \d^{(2m)}_{D,-\ell} + \d^{(2m)}_{D,\ell^{*}}
+ \d^{(2m)}_{D,-\ell^{*}}) & \text{if $\DD=-D^{2}$, for some $D \in \IN$} \\
0 & \text{if $-\DD$ is not a perfect square}
\end{cases} 
\eea
where $\ell^{*} \mypmod{2m}$ is defined by 
$\ell^{*} \equiv  -\ell \mypmod{2P}$, $\ell^{*} \equiv +\ell \mypmod{2m/P}$, and the notation~$\d^{(n)}_{a,b}$
means $\d_{a,b \mypmod{n}}$.
The correction term for the right-hand side, on the other hand, is given 
(since $\CV_{2,P}^{(1)}=V_{2,P}-pV_{2,P/p^{2}} U_{p}$, where the second term is omitted if $p^{2} \nmid P$) 
by
\bea
&& \sum_{d|(\frac{\DD+\ell^{2}}{4m},\ell,P)} \! d \, C\Big(\Phi^{C}_{2,\frac{m}{P}} \, ; 
\frac{\DD}{d^2}, \frac{\ell}{d}\mypmod{2\frac{m}{P}} \Big) \; \quad -  p \sum_{d|(\frac{\DD+\ell^{2}}{4m},
\ell,\frac{P}{p^{2}})} \! d \, C\Big(\Phi^{C}_{2,\frac{m}{P}} \, ; \frac{\DD }{p^{2} d^2}, 
\frac{\ell}{pd}\mypmod{2\frac{m}{P}} \Big) \cr
&& \quad \= \sum_{d|(\frac{\DD+\ell^{2}}{4m},\ell,P)} 
\sum_{\l^{2} = - \DD  \atop  \l \equiv \ell \mypmod{\frac{2md}{P}}} |\l| \; \quad - 
\sum_{d|(\frac{\DD+\ell^{2}}{4m},\ell,\frac{P}{p^{2}})} 
\sum_{\l^{2} = - \DD  \atop \l \equiv \ell \mypmod{\frac{2md}{P}}} |\l| \; , \cr
&& \quad \= \sum_{d \in \{P,P/p\} \atop d|(\frac{\DD+\ell^{2}}{4m},\ell,P)} 
\sum_{\l^{2} = - \DD  \atop  \l \equiv \ell \mypmod{\frac{2md}{P}}} |\l|  \; ,
\eea
where the last line holds because the terms in the first sum in the second line 
with~$d|Pp^{-2}$ cancel with the terms in the second sum.
The terms with~$d=P$ in the third line give the first two terms $\sqrt{-D}=\pm \ell^{*}$ 
of the right hand side of~\eqref{LHS1},
while the terms with~$d=P$ in the third line give the second two terms $\sqrt{-D}=\pm \ell^{*}$.

The correction term for the left-hand side of the second equation in \eqref{induction} is given 
by the the right-hand side of~\eqref{LHS1} with the signs of the last two terms changed. 
This is now equal to zero unless~$(\ell,m)=1$, in which case it is equal to~$\sqrt{-\DD}$ (resp.~$-\sqrt{-\DD}$) 
if $D \equiv \pm \ell \mypmod{2m}$ (resp.~$D \equiv \pm \ell^{*} \mypmod{2m}$).
The correction term for the right-hand side, on the other hand, is given by
\bea\label{RHS2}
&& \sum_{d|(\ell,\frac{P}{p},\frac{\DD+\ell^{2}}{4m})} \! d \, \bigg( C\Big(\Phi^{C}_{2,mp/P} \, ; 
\frac{\DD}{d^2} \,, \frac{\ell}{d}\mypmod{\frac{2mp}{P}} \Big) \, - \, C\Big(\Phi^{C}_{2,mp/P} \, ; 
\frac{\DD}{d^2} \,, \frac{\ell^{*}}{d}\mypmod{\frac{2mp}{P}} \Big) \bigg)\cr
&& \qquad \=  \sum_{d|(\ell,\frac{P}{p},\frac{\DD+\ell^{2}}{4m})} 
\Big(\sum_{\l^{2} = - \DD  \atop  \frac{\l}{d} \equiv \frac{\ell}{d} \mypmod{\frac{2mp}{P}}}
\, - \, \sum_{\l^{2} = - \DD  \atop  \frac{\l}{d} \equiv \frac{\ell^{*}}{d} \mypmod{\frac{2mp}{P}}} \Big) |\l| \; ,
\eea
where $\ell^{*} \mypmod{2mp/P}$ is now defined by 
$\ell^{*} \equiv  -\ell \mypmod{2p}$, $\ell^{*} \equiv +\ell \mypmod{2m/P}$.
If~$p|\ell$, then~$\ell^{*}=\ell$ and the expression vanishes, so that only~$d=1$
contributes in the first summation. The condition~$4m|(\DD+\ell^{2})$ combined 
with~$\l \equiv \ell \mypmod{2mp/P}$ implies~$\l \equiv \ell \mypmod{2m}$ (and similarly 
for $\ell^{*}$ since $\ell^{2} \equiv \ell^{*2} \mypmod{2mp/P}$), thus proving the second 
equation in~\eqref{induction}.  \hfill $\square$

\bigskip

\ndt {\bf Proof of Theorem~\ref{QmOG}}:
By M\"obius inversion, the mock Jacobi forms~$\Phi_{2,m}$  associated to any choice of minimal meromorphic 
Jacobi forms~$\v_{2,m}$ can be represented in the form~\eqref{Phi0mimp}, where~$\Phi_{2,m}^{0}$ 
is defined by 
  \be \label{defPhi02m}
  \Phi^{0}_{2,m} \= \sum_{d^{2} \mid m} \mu(d) \; \Phi_{2,m/d^{2}}\,|\,U_{d}  \, .
  \ee
We have to show that $\Phi^{0}_{2,m}$ can be chosen to be primitive and to have optimal growth.

To prove the second of these statements, it is enough to show that $\Phi^{0}_{2,m}| u_{p}$ has no shadow 
for all primes $p$ with $p^{2} | m$. 
Then, by Theorem~\ref{Emstar}, the weak elliptic form~$\Phi^{0}_{2,m}$ can be chosen (by adding 
a weak Jacobi form) to have optimal growth.
Writing the square divisors of~$m$ as $d^{2} p^{2i}$ with $p \nmid d$ and noting that $\mu(d p^{i})=0$ for $i>1$ 
and that $U_{dp}=U_{d} U_{p}$, $U_{d} u_{p} = u_{p} U_{d}$, and $U_{p} u_{p} = p$, 
we find from~\eqref{defPhi02m} that  
\bea
\Phi^{0}_{2,m} \! \mid \! u_{p} &\=& \sum_{d^{2}\mid m, \, p \nmid d} \mu(d) 
\big[\Phi_{2,m/d^{2}} \! \mid \! U_{d} \, - \, \Phi_{2,m/d^{2} p^{2}} \! \mid \! U_{d}  \! \mid \! U_{p} \big] \mid u_{p} \cr
& \= & \sum_{d^{2}\mid m, \, p \nmid d} \mu(d) 
\big[\Phi_{2,m/d^{2}} \! \mid \! u_{p} \, - \, p \, \Phi_{2,m/d^{2} p^{2}} \big] \mid U_{d} \, , \nonumber
\eea
so it suffices  to show that $\Phi_{2,m} | u_{p}$ and $p\, \Phi_{2,m/p^{2}}$ 
have the same shadow for any $m$ with~$p^{2}|m$. But the shadow of $\Phi_{2,m} = \v^{\rm F}_{2,m}$ 
is determined by $\v^{C}_{2,m} = \wh \Phi_{2,m} - \Phi_{2,m}$ (cf.~comments at the end of \S\ref{MockDecomposition}), 
and by \eqref{phiCcfs}, we have: 
\be
C(\v^{C}_{2,m} \, ; \DD p^{2}, \, rp  \mypmod{2m} ) 
 \=  \begin{cases} p |\l| & \text{if $\DD = - \l^{2}$, \,  $\l \equiv r \mypmod{2m/p}$} \, ,  \\ 
\; 0 & \text{otherwise} \, .
\end{cases}
\ee
Then, from \eqref{Ucech2}, we have:
\bea
C(\v^{C}_{2,m}|u_{p} \, ; \DD, \, \ell \mypmod{2m/p^{2}}) 
&\= &\sum_{r \mypmod{2m/p} \atop r \equiv \ell \mypmod{2m/p^{2}}} 
C(\v^{C}_{2,m} \, ;  \DD p^{2} ,\, rp  \mypmod{2m} ) \; , \cr
& \= & \begin{cases} p |\l| & \text{if $\DD= - \l^{2}$, \,  $\l \equiv \ell \mypmod{2m/p^{2}}$} \, ,  \\ 
\; 0 & \text{otherwise} \, .
\end{cases} \cr
& \= & p \, C(\v^{C}_{2,m/p^{2}} \, ; \DD,\ell) \, .
\eea
This completes the proof of the fact that $\Phi^{0}_{2,m}$ can be chosen to have optimal growth. 
Now, by formula~\eqref{Ucech2}, we have that~$\Phi^{0}_{2,m} | u_{t}$, and hence 
also~$\Phi^{0}_{2,m} | u_{t}|U_{t}$, is holomorphic for $t >1$, so by replacing~$\Phi^{0}_{2,m}$ 
by~$\pi^{\rm prim}(\Phi^{0}_{2,m})$ as defined in~\eqref{primproj} we can also assume 
that~$\Phi^{0}_{2,m}$ is primitive.  \hfill $\square$

\bigskip

\ndt {\bf Proof of Theorem~\ref{chooseKm}}:
The proof of Theorem~\ref{Emstar} given in~\S\ref{optimalforms} the following discussion already showed 
that $J_{2,m}\NH/J_{2,m}$ is one-dimensional and that any representative $\CK_{m}$ of this 
quotient space, if we normalize it to have $c(\CK_{m}\,; 0,0)=2$, has polar coefficients as given in 
the statement of Theorem~\ref{chooseKm}. Since all the polar coefficients $C(\CK_{m} \, ; -1,r)$ are equal, 
we have that~$\CK_m|W_{m_{1}}-\CK_m \in J_{2,m}$ for all $m_{1} || m$, so (by averaging) we can 
choose~$\CK_m$ to be invariant under all Atkin-Lehner operators. Finally, we can make~$\CK_{m}$ 
primitive by the argument used at the end of the proof of Theorem~\ref{QmOG}. 
\hfill{$\square$}
\vspace{0.2cm}

\ndt {\bf Remark}. For square-free $m$, since all holomorphic Jacobi forms are in fact cusp forms, $\CK_{m}$
is unique up to Jacobi cusp forms, and its coefficients with $\DD=0$ are unique. In this case the only such 
coefficient, up to translation, is~$c(0,0)$. If~$m$ is not square-free, there is more than one coefficient~$c(n,r)$ such 
that~$\DD=0$. Notice that $r^{2}\equiv 0 \mypmod{4m} \Leftrightarrow (r,2m)=2m/t$ with~$t^{2}|m$.
For each~$t$ with $t^{2}|m$, we have 
\be\label{ckm00}
\sum_{j \mypmod{t}} C(\CK_{m} \, ; 0, \frac{2mj}{t} \mypmod{2m}) \= 
\begin{cases} -2   & \text{if $t=1$} \, , \\ \; 0  & \text{if $t>1$}\, . \end{cases} 
\ee
(For~$t=1$, the sum on the left reduces to the single term~$c(\CK_{m}\,;0,0)=-2$, and for~$t>1$,
it equals~$C(\CK_{m}|u_{t}; 0,0)$, which vanishes because~$\CK_m | u_{t} \in J_{2,m/t^{2}}$.)
We can choose the~$\DD=0$ coefficients arbitrarily subject to this one constraint, and once we have 
made a choice, then~$\CK_{m}$ is unique up to cusp forms even in the non-square-free case. 
For instance, at the expense of introducing denominators, we could fix~$\CK_{m}$ uniquely up to cusp forms by requiring 
that~$C(\CK_{m};0,r)$ depends only on the gcd of~$r$ and~$2m$. (This would correspond to choosing~$a=\half$
in the first and third tables at the end of~\S\ref{Statements}.) Another choice, more similar to our ``$\star = {\rm II}$''
condition, would be to choose~$C(\CK_{m};0,r)=0$ for all~$r$ with~$4m|r^{2}$ and~$r \nmid 2m$, in which case
the remaining coefficients~$C(\CK_{m};0,2m/t)$ ($t^{2}|m$) are uniquely determined by~\eqref{ckm00}. (This would 
correspond to choosing~$a=1$ in the first of the tables at the end of~\S\ref{Statements}.) The reader can easily
work out the complete formula for all of the coefficients~$C(\CK_{m};0,r)$ for either of these two special choices.

Our next task is to prove Theorem~\ref{PolarCoeffs}. For this we need to calculate the residues of~$\Phi_{2,m}$,
using the results of the previous subsection. We therefore have to apply the formula~\eqref{Rm1m2ans} to the mock modular 
form~$\Phi_{2,m}=\v_{2,m}^{\rm F}$. (Note that the residue of this form will be independent 
of the specific choice of~$\v_{2,m}$, since any two choices differ by a Jacobi form
and the residues of Jacobi forms vanish.) Since the completion of $\Phi_{2,m}$ 
involves the weight~1/2 unary theta series~$\vth_{m,\ell}^{0}\,$, the first step is to 
compute the scalar products of these theta series. 
\begin{proposition} \label{thetapeter}
For $m \in \IN$, and $\ell_{1} , \ell_{2} \in \IZ/2m\IZ$, we have
\be
(\vartheta_{m, \, \ell_{1}}^0 \,, \,\vartheta_{m, \, \ell_{2}}^0) \= \frac{1}{2 \sqrt{m}} \big( 
\delta_{\ell_{1}, \ell_{2} \! \mypmod{2m}}
+ \delta_{\ell_{1}, -\ell_{2} \! \mypmod{2m}} \big)
\ee
\end{proposition}
\begin{proof}
Formula~\eqref{RS} gives 
\bes
(\vartheta_{m, \, \ell_{1}}^0 \,, \,\vartheta_{m, \, \ell_{2}}^0) \= \frac{1}{2} \, \text{Res}_{s=\half} 
\biggl( \sum_{r>0} \frac{(\delta_{r, \, \ell_{1}\;(2m)}+ \delta_{r, \, -\ell_{1}\;(2m)})
(\delta_{r, \, \ell_{2}\;(2m)}+ \delta_{r, \, -\ell_{2}\;(2m)})}{(r^{2}/4m)^{s}} \biggr) \; . 
\ees
The propostition follows.
\end{proof}

\begin{proposition} \label{RPhi2mval}
Let $\v_{2,m}$ be any meromorphic Jacobi form of weight 2 and index $m$ with pole
$1/(2 \pi i z)^{2}$ and its translates, and let $\Phi_{2,m}$ be its finite part $\v^{\rm F}_{2,m}$.
Then
\be \label{Rm1m2}
R_{m_{1},m_{2}}(\Phi_{2,m}) \= \frac{m_{1}+m_{2}}{12}  \, .
\ee
\end{proposition}
\begin{proof}
By equation \eqref{phi2mhat}, 
the $\ell^{\rm th}$ component  $g_{\ell}$ of the shadow for $\Phi_{2,m}$ is $-\sqrt{\frac{m}{4 \pi}} \vth^{0}_{m,\ell}$. 
When $(m_{1}, m_{2})=1$, $R_{m_{1},m_{2}}(\Phi_{2,m}) = \CR[\Phi_{2,m} \! \mid \! W_{m_{1}}]$. 
Theorem~\ref{Rphishad} then gives:
\bea 
\CR[\Phi_{2,m} \! \mid \! W_{m_{1}}] &\= & \frac{\sqrt{m}}{\sqrt{4 \pi}} \frac{\sqrt{\pi}}{6}  
\sum_{\ell  \! \mypmod{2m}} \bigl( \vth^{0}_{m,\ell}, \vartheta_{m,\ell^{*}}^{0} \bigr) \cr
& \= & \frac{1}{24} \, \sum_{\ell  \! \mypmod{2m}}  \big( \delta_{\ell, \ell^{*} \! \mypmod{2m}}
+ \delta_{\ell, -\ell^{*} \! \mypmod{2m}} \big) \cr
&\=& \frac{1}{12} \, (m_{1}+m_{2}) \, . \nonumber 
\eea
In this calculation, we used Proposition~\ref{thetapeter} to get the second line. For 
the third line, recall that $\ell^{*}$ is defined by 
$(\ell^{*} \equiv -\ell \mypmod{2m_{1}}, \ell^{*} \equiv \ell \mypmod{2m_{2}})$, so that 
$\ell^{*} \equiv \ell \mypmod{2m} \Leftrightarrow m_{1} \! \mid \! \ell$ (which is true for~$2 m_{2}$ values of~$\ell$ in 
$\IZ/2m\IZ$), 
and similarly, $\ell^{*} \equiv -\ell \mypmod{2m} \Leftrightarrow m_{2} \! \mid \! \ell$ (which is true for~$2m_{1}$ values of~$\ell$).  
When $(m_{1},m_{2})=t>1$, we have 
$R_{m_{1},m_{2}}(\Phi_{2,m}) = R_{m_{1}/t,m_{2}/t}(\Phi_{2,m} \! \mid \! u_{t})$.
As we have shown in the proof of Theorem~\ref{QmOG},  
the shadow of $\Phi_{2,m} \! \mid \! u_{t}$
is the same as the shadow of $t \Phi_{2,m/t^{2}}$, so $R_{m_{1}/t,m_{2}/t}(\Phi_{2,m} \! \mid \! u_{t}) = 
t R_{m_{1}/t,m_{2}/t}(\Phi_{2,m/t^{2}}) = t (\frac{m_{1}}{t} + \frac{m_{2}}{t})/12 = (m_{1}+m_{2})/12$ by the 
special case already proved.
\end{proof}

The reader might find it helpful to compare the statement of this proposition with the table 
in \S\ref{choosephi} presenting the data for the case~$m=6$, 
in which the function $\Phi=\Phi_{2,6}$ (in any of its three versions I, II, or III) has 
\ben 
R_{1,6}(\Phi) &\= 2 C(\Phi ; -1, \, \pm 1\! \mypmod{12}) + C(\Phi ; 0, \, 0\! \mypmod{12}) \= 7/12 \, , \\
R_{2,3}(\Phi) &\= 2 C(\Phi ; -1, \, \pm 5\! \mypmod{12}) + C(\Phi ; 0, \, 0\! \mypmod{12}) \= 5/12 \, . 
\een

\vspace{0.2cm}

\ndt {\bf Proof of Theorem~\ref{PolarCoeffs}}: We can now complete the proof of Theorem~\ref{PolarCoeffs}}
and its corollary. For $m$ square-free, any choice of~$\Phi=\Phi_{2,m}^{0}$ in $\CE_{m}^{\rm OG}$, and 
any~$r$ with $r^{2} \equiv 1 \mypmod{4m}$, we have 
\be
2 C(\Phi; \, -1, \ell) + C(\Phi; \, 0,0) \= \CR[\Phi \!\mid \! W_{m_{1}}] \= \frac{m_{1}+m_{2}}{12} 
\ee
by Proposition~\ref{RPhi2mval}, where $m=m_{1} m_{2}$ is the decomposition of~$m$ for which~$r^{*}=1$.
Since $C(\Phi^{\rm III}; 0,0)=0$, this proves Theorem~\ref{PolarCoeffs} and also gives the 
corresponding formulas 
$$
C(\Phi^{\rm I}; -1,r) \= -\frac{(m_{1}-1)(m_{2}-1)}{12} \, , \qquad  C(\Phi^{\rm II}; -1,r) \= -\frac{1}{12} \, 
\biggl( \frac{m_{1}+m_{2}}{2} \, - \, \prod_{p|m} \frac{p+1}{2} \biggr)  
$$
for $\Phi^{\rm I}$ and $\Phi^{\rm II}$, as well as formula~\eqref{Qmpolar} for the polar coefficients of~$\CQ_{M}$
when $\mu(M)=1$. 
If~$m$ is not square-free, then the solutions~$r \mypmod{2m}$ of~$r^{2} \equiv 1 \mypmod{4m}$
are still in~1:1 correspondence with the decompositions of~$m$ into coprime factors~$m_{1}$ and~$m_{2}$. 
To compute the corresponding Fourier coefficient of~$\Phi_{2,m}^{0}$, we look at each term 
in~\eqref{defPhi02m} separately, writing each~$d$ as $d_{1}d_{2}$ with~$d_{i}^{2}|m_{i}$ and noting 
that $\frac{m_{1}}{d_{1}^{2}}=(\frac{r-1}{2},\frac{m}{d^{2}})$,  $\frac{m_{2}}{d_{2}^{2}}=(\frac{r+1}{2},\frac{m}{d^{2}})$,
after which~\eqref{Phi0IIIcoeffs} follows from~\eqref{Rm1m2} and the fact that $\CR[\v  \! \mid \! U_{d}] \= \CR[\v]$. 
\hfill $\square$

\vspace{0.2cm}

\ndt {\bf Proof of Theorem~\ref{Wt1OSP}}: It is enough to show that any weak Jacobi form of optimal growth 
is actually holomorphic since, as we have seen in Theorem \ref{Skorthm}, there are no holomorphic Jacobi forms 
of weight~1. To show this, we use the ``residue operator on weight~1 forms'' defined by
\be\label{RasReswt1}  \CR^{(1)}[\v]\=\text{Res}_{\t=\infty}\bigl(\v'(\t,0)\,d\t\bigr) \= \sum_{r\in\Z} r \, c_{\v}(0,r)\,\ee
instead of~\eqref{RasRes}, where, as usual,~$\v'(t,z) = \frac{1}{2 \pi i} \frac{d}{dz} \,\v(t,z)$. 
This residue vanishes if~$\v\in \wt J_{1,m}$, since then~$\v'(\cdot, 0) \in M_{2}^{!}$. 
If, further, $\v$ has optimal growth, then $2  c_{\v}(0,1)=\CR^{(1)}[\v]=0$, so~$C_{\v}(-1,1)$ vanishes. 
By the same argument with~$\v$ replaced by $\v$ acted upon by each of the Atkin-Lehner operators, 
one has that~$C_{\v}(-1,r)=0$ for all~$r$ with~$r^{2} \equiv -1 \mypmod{4m}$,  
and thus~$\v \in J_{1,m}$. \hfill $\square$

\section{Quantum black holes and mock modular forms \label{MockforDyons}}

We now have all the ingredients required to address the questions posed
in the introduction. These answers follow from applying our results of
\S{\ref{MockfromJacobi}}, in particular theorem~\ref{Thmdouble}, to the quarter-BPS partition 
function of the the $\CN=4$ theory. It turns out that in this case, many of the mathematical objects 
in the theorem simplify. A reader who has not gone through  the details of \S{\ref{MockfromJacobi}} could still
fruitfully read this section to obtain the relevant results for this example.

In particular,  we show in  \S{\ref{Degen}} that the immortal black hole degeneracies 
 extracted from the asymptotic partition function (\ref{inverse3})   are 
Fourier coefficients of a (mixed) mock Jacobi form. The completion transforms as a true Jacobi form and obeys a first order partial differential equation~\eqref{holanom}
which can be understood as a  holomorphic anomaly equation. 
In \S{\ref{Anomaly}} we discuss the physical origin of the meromorphy of the asymptotic counting function in this example and how it is related
to the noncompactness of the microscopic brane dynamics. 
In \S{\ref{Mlimit}} we analyze the contour corresponding to the M-theory limit to embed our analysis in the context of  $AdS_{3}/CFT_{2}$ holography near the horizon of a black string. 
We conclude in \S{\ref{generalization}} with  comments and open problems.

\subsection{Mock Jacobi forms, immortal black holes, and $AdS_{2}$ holography \label{Degen}}

We first discuss the consequences of the  $S$-duality symmetry and its relation to the spectral flow symmetry for the moduli-independent immortal  degeneracies  $d^{*}(n,\ell,m)$ defined by \eqref{inverse3}.
Recall that under the S-duality transformations
\be\label{11b1}
\bem 1 & b  \\
    0 & 1\eem \ 
\ee 
with integer $b$, 
the S-modulus \eqref{Sdef} transforms as:
\be
S  \to  S +  b \, . 
\ee
which means that the axion transforms as $a \to a +b$. Since the axion arises from a dimensional reduction of  the 3-form field in the M-theory frame (discussed in \S\ref{Mlimit}), this subgroup \eqref{11b1} of $S$-duality 
has its origins  in large gauge transformations of the 3-form field. When the axion transforms as above, the effective theta angle changes and as a result the charge vector of the dyon  transforms as $(Q, P) \rightarrow (Q + P b, P)$  due to the Witten effect \cite{Witten:1979ey}.
Consequently, the charge-invariants transform as
\be\label{nmltrans}
 n  \to  n + m b^2 + b \ell \, , \qquad  \ell  \to  \ell + 2 m b \, , \qquad m \rightarrow m .
\ee
which are nothing but the spectral flow transformations.  Note that $m$ is left fixed by this subgroup of S-duality whereas  other elements of the S-duality group would change $m$. Thus,  only this subgroup is expected to act on the sector with a fixed $m$. The covariance of the spectrum under these $S$-duality transformations implies
\be
d(n, \ell, m) |_{S, \mu} \= d(n + m b^2 + b, \ell + 2 m b, m) | _{S + b, \mu} \, .
\ee
Since the degeneracies of immortal black holes are independent of the moduli by definition, 
they are actually invariant under~$S$-duality:
\be\label{spectral-immortal}
d^*(n + m b^2 + b, \ell + 2 m b, m) \=  d^* (n, \ell, m) \, .
\ee 
{}From the mathematical point of view, this is precisely the action of the elliptic part of the Jacobi transformations \eqref{elliptic} on the Fourier coefficients of a Jacobi (or mock Jacobi) form. 

The immortal degeneracies  $d^{*}(n,\ell,m)$ are computed using the attractor contour~\eqref{contourattract}, 
for which the imaginary parts of the potentials are proportional to the charges and scale as 
\begin{equation}\label{Attractorcontour}
\Im (\sigma) \= 2n/\varepsilon \, , \qquad 
\Im (\tau) \= 2m/\varepsilon \, ,\qquad 
\Im (z) \= -\ell/\varepsilon \, ,
\end{equation}
with $\varepsilon$ very small and positive. In other words,  
\bea\label{attractor2}
|p| = \l^{2n}\, , \quad |q| = \l^{2m} \, , \quad  |y| = \l^{-\ell} \, , \qquad \text{with $\l  = \exp(-2\pi/\e) \rightarrow 0$}
\eea
on the attractor contour.
We  assume that $n > m$ without loss of generality, otherwise one can 
simply exchange  $\s$ and $\tau$. Moreover, using the spectral flow symmetry \eqref{spectral-immortal} we can always bring $\ell$ to the window $0 \leq \ell <2m$ for the immortal degeneracies. 

To extract the Fourier coefficient from  the product representation \eqref{final2} of the counting 
function~\eqref{igusa} using \eqref{attractor2}, we need to expand a typical term in the product of the form
\be
\frac{1}{(1- p^{r}q^{s} y^{t})^{2C_0(4rs -t^{2})}} \, .
\ee
where $2 C_0$ is the Fourier coefficient of the elliptic genus of a single copy 
of $K3$. Depending on whether~$|p^{r}q^{s}y^{t}|$ is less than one or greater than one, we expand 
the denominator using 
\be\label{xsmall}
\frac{1}{1-x} \= \begin{cases} 1+ x + x^{2} + \ldots & \text{for $|x| < 1$}\,, \\
 -(x^{-1} + x^{-2} + x^{-3} + \ldots)  & \text{for $|x| > 1$}\,.  
  \end{cases}
\ee
We see that it is not \textit{a priori} obvious  that expanding first in $p$  is the same as using the 
attractor contour because if  the powers of $t$ are sufficiently large and positive, the attractor contour may 
correspond to expanding some terms effectively in   $p^{-1}$ instead of in $p$. 

To address this question  we now show that for the terms that appear in the product representation,  $|p^{r}q^{s}y^{t}|$ is always less than one using the fact that
\be
C_{0}(\DD) = 0 \qquad \text{for $\DD  <-1$} \, .
\ee
To begin with, the $p$-dependent terms with $\DD=0, -1$   arise only for $s=0$ and  are of the form
\be
\frac{1}{(1-py )^{2}(1- py^{-1})^{2} (1- p^{r})^{20}} \,
\ee
These involve only single  powers $y$ or $y^{-1}$ and positive powers of $p$.  The absolute value of 
$py$ or $py^{-1}$ for the attractor contour goes as $\lambda^{2m\pm l}$  which is always much less than one for $0 \leq |l|  < 2m$.   
The remaining  $p$-dependent terms with $s > 0$  have $\DD= 4rs -t^{2}>0 $ and $r >0$. The absolute value of $p^{r}q^{s} y^{t}$  goes as
\be
\l^{2nr+ 2ms -\ell t} \, .
\ee
The exponent is most negative when $\ell$ takes its maximum positive value $\ell_{max} = \sqrt{2mn}$ and $t$ takes its  maximum positive value $t_{max}=\sqrt{2rs}$. 
For black hole configurations with a smooth  horizon,   $4mn- \ell ^{2}>0 $ because the classical area of the horizon is the square-root of this quantity.  Hence the most negative value of the exponent  would be
\be
2nr+ 2ms - 2 \sqrt{mnrs} \= 2(\sqrt{nr} -\sqrt{ms})^{2} >  0 \, , 
\ee
and the absolute value of $p^{r}q^{s} y^{t}$ is always less than or equal to
\be
\l^{ 2(\sqrt{nr} -\sqrt{ms})^{2}} 
\ee
which is always much less than one. This implies that to extract the attractor Fourier coefficients, we expand all terms in the product in the small $p$ expansion using~\eqref{xsmall} thus giving us the Jacobi form 
$\psi_{m}$ of~\eqref{reciproigusa}  as a  partition function in the remaining  two chemical potentials.

The degeneracies extracted from $\psi_{m}(\t, z)$ still experience many wall-crossings, and we should now 
use the rest of the attractor contour to analyze the immortal part of the degeneracies of 
$\psi_{m}$.  We do this in two steps by first performing the inverse Fourier transform in $z$, 
and then in $\t$. 
For the first step, we need to specify a contour for the imaginary part of $z$ relative to that of $\t$. 
The conditions~\eqref{Attractorcontour}  on the imaginary parts of $(\t,z)$ implies that 
$\Im(z) = -\frac{\ell}{2m} \Im(\t)$. 
Since by moving the contour in a purely horizontal direction we do not cross any poles, 
we can add an arbitrary real part to $z$, which we can choose so that $z = -\ell\t/2m$. This 
gives the result of the first integral to be:
\be\label{fmlpsim}
f_{m,\ell}^*(\t)  \=  e^{-\pi i l^2 \t / 2m} \int_{P}^{P+1} \psi_m(\t,z) \, e^{- 2 \pi i \ell z} \, dz \, , 
\qquad P = - \ell \t/2m \, . 
\ee
We then have to perform a~$\tau$ integral with the contour specified by the attractor 
values~\eqref{Attractorcontour}, but since the integrand~$f_{m,\ell}^{*}(\t)$ is holomorphic 
in~$\t$, the answer is independent of the contour. We thus get the degeneracies $d^{*}(n,\ell,m)$ to be:
\be\label{dstarfml}
d^*(n,\ell,m) \= \int e^{\pi i \ell^2 \t / 2m}\, f_{m,\ell}^*(\t) \, d\tau \, , 
\ee
where the integral is over an interval of length~$1$ for the fixed imaginary value~\eqref{Attractorcontour} of $\tau$.

To extend to other values of $\ell$, we use the spectral flow invariance \eqref{spectral-immortal} and sum over all values of~$\ell$ to get a canonical partition function in the two variables~$(\t,z)$. 
Because of the factor~$e^{\pi i \ell^2 \t / 2m}$ in \eqref{dstarfml}, the sum over spectral-flow images yields the function~$\vartheta_{m, l}(\t,z)$. 
Putting all this together, and comparing the above equations \eqref{fmlpsim},~\eqref{dstarfml} 
with our definition \eqref{phiF} of the finite part of a meromorphic Jacobi form, 
we see that the single centered black hole degeneracies are Fourier coefficients of 
precisely the {\it finite} part of $\psi_{m}$ as defined in~\S\ref{MockfromJacobi}:
\be
\sum_{\ell \, \mod \, (2m)} f_{m,\ell}^*(\t) \, \vartheta_{m,\ell}(\t,z) 
\= \psi_{m}^{\rm F}(\t,z) \, .
\ee
We see that the seemingly unconventional charge-dependent choice for the contour 
in~\S\ref{MockfromJacobi} for defining the finite part is completely natural from the point 
of view of the attractor mechanism (indeed, it was motivated there by these physical considerations).
According to our analysis in~\S\ref{MockfromJacobi}, this finite part $\psi_{m}^{\rm F}$ is a {\it mock} Jacobi form. 
This is one of the main conclusions of our analysis, so we summarize the entire discussion by the 
following concise statement:
\begin{quote}
{\it The degeneracies of single centered (immortal) black holes with magnetic charge 
invariant $M^2/2=m$ are Fourier coefficients of  a mock Jacobi form of index $m$.} 
\end{quote}

We stress here that the converse is not true. 
The single centered black hole solution exists only when the 
discriminant $4mn-\ell^2$ is positive, and its degeneracy grows exponentially 
for large values of the discriminant. On the other hand, as we saw in \S\ref{Jacobi}, 
any Jacobi form with exponential growth necessarily has negative discriminant states. 
This means that the partition function for single centered black holes alone with fixed magnetic charge 
and varying electric charge cannot have good modular properties. Indeed the partition 
function $\psi_{m}^{\rm F}$ does contain negative discriminant states for example for $n<m$. These low-lying states 
do not correspond to single centered black holes~\cite{Dabholkar:2007vk,  Sen:2011mh}.

We now turn to the physical interpretation of the polar part $\psi_{m}^{\rm P}(\t,z)$ of the meromorphic Jacobi form defined in~\eqref{phiPdoub}. As mentioned in the introduction,   the only multi-centered gravitational configurations that contribute 
to the supersymmetric index of quarter-BPS dyons in our~$\CN=4$ string theory have exactly two centers \cite{Dabholkar:2009dq} each of which is a half-BPS state.  At a wall, one of these two-centered configurations decays. A basic decay at a wall-crossing  is when  one center is purely electric and the other is 
purely magnetic. Other decays are related to this basic decay by duality \cite{Sen:2007vb}.
For this   two-centered configuration, with one electric and one magnetic center, the indexed 
partition function after summing over the electric charges is:
\be\label{basicmulti}
p_{24}(m+1) \cdot \frac{1}{\eta^{24}(\tau)} \, \cdot \sum_{\ell > 0} \, \ell \, y^{\ell} \, , 
\ee
where the first factor is  the degeneracy of the half-BPS magnetic center, 
the second factor is the partition function of the electric center ~\cite{Dabholkar:1989jt},  
 and the third factor is the partition function counting the dimensions of the $SU(2)$  multiplets with  angular momentum $(\ell -1)/2$ contained in the electromagnetic field generated by the two centers~\cite{Denef:2000nb} for which $N \cdot M = \ell$.
The third factor is an expansion of the function
\be\label{basicmulti2}
 \frac{y}{(1-y)^{2}}
 \ee
in the range  $|q| < |y| < 1$ (see equation~\eqref{BBm}). 
Here we already see the basic wall-crossing of the theory because different 
ways of expanding this meromorphic function for~$|y|<1$ and~$|y|>1$ will give different degeneracies. Other wall-crossings  are related to this basic wall-crossing by S-duality which as we have explained is nothing but spectral flow for a given fixed $m$. Thus  the full partition function that captures all wall-crossings  is obtained by `averaging'  over the spectral-flow images of ~\eqref{basicmulti2}. 
This averaging  is nothing but the  averaging operation defined in \eqref{AP1AP2} and  \eqref{AP1AP2again}
\be
 \AvB m{ \frac{y}{(y-1)^2} } \, = \, \frac{q^{ms^2 +s}y^{2ms+1}}{(1 -q^s y)^2}  \,
\ee
It is now easy to see that this averaging over  spectral-flow  images gives 
exactly the polar part~$\psi_{m}^{\rm P}$ of the meromorphic Jacobi form~$\psi_{m}$:
\be\label{Tm} 
 \frac{p_{24}(m+1) }{\eta^{24}(\tau)} \, \sum_{s\in\Z} \, \frac{q^{ms^2 +s}y^{2ms+1}}{(1 -q^s y)^2} 
 \= \psi_{m}^{\rm P}(\t,z) \, .
\ee
This follows from the fact that the poles of~$\psi_{m}(\t,z)$ are exactly at~$z \in \IZ \tau +\IZ$ and the residue at~$z=0$ is $$\frac{p_{24}(m+1) }{\eta^{24}(\tau)} \, .$$  as mentioned in  \eqref{psifromphim}.

We thus see that  the decomposition theorem~(\ref{Thmdouble}) 
as applied to our dyon partition function~$\psi_{m}(\t,z)$ has a natural physical interpretation. It is simply the statement  that  the full partition function is a sum of 
its immortal and decaying constituents:
\be \label{singlepart} 
\psi_{m}(\tau,z) \; = \;  \psi_{m}^{\rm F}(\tau,z) \, + \, \psi^{\rm P}_{m}(\tau,z) \, .
\ee
The nontrivial part is of course  the implication of such a decomposition for modularity. By separating part of the function $\psi_{m}$, we have evidently broken modularity, and $\psi_{m}^{\rm F}(\t,z)$  
is not a Jacobi form any more. However,  the decomposition theorem~(\ref{Thmdouble}) guarantees 
that~$\psi_{m}^{\rm F}$ still has a very special modular behavior in that 
it is a mock Jacobi form. The mock behavior can be summarized by the following partial differential equation
 obeyed by its completion: 
\be \label{holanom}
 \tau_2^{3/2} \; \frac{\partial} {\partial \bar{\tau}}   \, \wh \psi_{m}^{\rm F}(\tau,z) \; = \; 
\sqrt{\frac{m}{8 \pi i}} \; \frac{ p_{24}(m+1)}{ \Delta(\tau)} \,  
 \sum_{\ell \, \mod \, (2m)}  {\overline{\vth_{m,\ell}(\tau,0)}} \, \vartheta_{m,\ell} (\tau,z) \, .
\ee
Note that~$\psi_{m}^{\rm P}$ also admits a completion~$\wh{\psi_{m}^{\rm P}}$ which is modular.
However, the function~$\wh{\psi_{m}^{\rm P}}$, unlike~$\wh{\psi_{m}^{\rm F}}$, is meromorphic,
and therefore has wall-crossings. We refer the reader to the summary at the end of~\S\ref{MockDecomposition} 
for a detailed discussion of the relationship between meromorphy (in~$z$), holomorphy (in~$\t$) and modularity.

Let us now discuss the implications of these results for $AdS_{2}/CFT_{1}$ holography. 
The near horizon  geometry of a supersymmetric black hole is a two-dimensional anti de Sitter space $AdS_{2}$ and the dual theory is a one-dimensional conformal field theory  $CFT_{1}$. The partition function of this $CFT_{1}$ is nothing but the integer $d^*(n, \ell, m)$ which gives the number of microstates of the black hole~\cite{Sen:2008yk,Sen:2008vm}. Our results show that that these integers are the Fourier coefficients of a mock modular form and as a result there is a hidden modular symmetry. Such a modular symmetry would make it possible to use
powerful techniques such as the  Hardy-Ramanujan-Rademacher expansion to process  the integers $d^*(n, \ell, m)$ into a form that can be more readily identified with a bulk partition function of string theory in $AdS_{2}$. There has been recent progress in evaluating the bulk partition function using localization techniques which makes it possible to study the subleading nonperturbative corrections in the Rademacher expansion in a systematic way~\cite{Dabholkar:2010uh,Dabholkar:2011ec,Gupta:2012cy,Murthy:2013xpa}. It would be interesting to see if the additional terms in the Rademacher expansion  arising from the `mock' modular nature can be given a physical bulk interpretation. 

\subsection{Meromorphy, mock modularity,  and noncompactness \label{Anomaly}}

From the mathematical point of view, we have seen how the mock modularity of
$\psi_{m}^{F}$ is a consequence of meromorphy of the Jacobi form $\psi_{m}$. In
the $\CN=4$ example above, we can understand the physical origin of this
meromorphy as coming from the non-compactness of the underlying brane dynamics. To see this
more clearly, we write the meromorphic Jacobi form $\psi_m(\t, z)$ as
  \be\label{meroformtwo}  \psi_m(\t, z) =  \frac 1{\v_{-2,1}(\t,z)} \, \frac{1}{\Delta(\t)} \,
   \, \chi (\t,z; \textrm{Sym}^{m+1}(K3)) \, . \ee
The three factors in the above equations have the following physical interpretation~\cite{David:2006yn}. 
The function $\psi_m(\t, z)$
can be viewed as an indexed partition function of the following $(0, 4)$ SCFT
  \be\label{scft}  \sigma(\textrm{TN}) \times \sigma_L(\widetilde K\text{-}n) \times \s(\textrm{Sym}^{m+1}(K3)) \,. \ee
Each CFT factor has a natural interpretation from the microscopic derivation of
the dyon counting formula \cite{David:2006yn} using the 4d-5d lift~\cite{Gaiotto:2005gf}. 
The $\sigma_L(\widetilde K$-$n)$ factor is the purely left-moving bosonic
CFT associated with $\widetilde K$-$n$ bound states of the Kaluza-Klein monopole
with momentum. The $ \sigma(\textrm{TN})$ factor is a $(4, 4)$ SCFT associated
with the motion of the center of mass of motion in the Kaluza-Klein background in
going from five dimensions to four dimensions in the 4d-5d lift. The
$\s(\textrm{Sym}^{m+1}(K3))$ is the $(4, 4)$ SCFT that is associated with the
five-dimensional Strominger-Vafa $Q1$-$Q5$-$n$ system. Note that the SCFT
$\sigma(\textrm{TN})$ has a \textit{noncompact} target space and the double pole
of  $1/{\v_{-2,1} (\t,z)}$ can be traced back to this noncompactness.

Going over from the partition function to the Fourier coefficients corresponds to going over from the
canonical ensemble with fixed $z$ to the micro-canonical ensemble with fixed $\ell$.
The Fourier coefficient  is supposed to count the number of
right-moving ground-states in this fixed charge sector. Right-moving states with nonzero energy cancel in pairs in this
index. Now, counting ground states for a compact theory is a well-defined operation.
However, when the target-space is noncompact, the wavefunctions are not properly
normalized without putting an infrared regulator in the target space. In the
present context, the essential physics of the noncompactness is captured by the point particle motion of the center of
mass in the Taub-NUT geometry.  A simple regulator can be introduced by turning on the axion
field  which according to \eqref{hetBrel} corresponds  in the Type-IIB frame  to  taking the modulus $U_{B} = U_{1B}+ iU_{2B}$ to have a small 
nonzero real part $U_{1B}$ \cite{David:2006yn}. This introduces a nonzero potential on the
Taub-NUT geometry proportional to $U_{1B}^2 U_{2B}$  that regulates the infrared behavior to
give a well defined counting problem \cite{Gauntlett:1999vc,Pope:1978zx}.
However, the index thus computed depends on the sign of $U_{1B}$ and jumps as $U_{1B}$ goes
from being slightly negative to slightly positive.  Indeed the partition function for these modes is precisely the function
  \be\label{pole-part}  \frac{y}{(1-y)^2}  \ee
which can be expanded either in power of $y$ or in powers of $y^{-1}$ depending on whether $U_{1B}$ is positive or negative. 
For $\ell$ positive, there are $\ell$ normalizable states  when $U_{1B} < 0$ and none when $U_{1B} >0$.  This is the essential source
of the wall-crossing behavior in this context.  We should emphasize that even though turning on $U_{1B}$ appears to make the
problem effectively compact, this is true only up to some maximal $\ell_{max}$ which scales as  
$U_{1B}^2 U_{2B}$ \cite{Gauntlett:1999vc,Pope:1978zx}.
For a given $\ell $ positive, one can always make the problem well-defined with a compact target space by choosing the moduli
so that $\ell_{max}$ is larger than the $\ell$ of interest. The partition function in that case is well-defined without any poles and goes as
  \be   y + 2y^{2} + \ldots +  \ldots + \ell_{max} y^{\ell_{max}}  \, . \ee
However,  if one wishes to compute the canonical partition function that sums over all $\ell$, one has to essentially take 
$U_{1B}^2 U_{2B}$ to go to infinity which is what leads to the pole at $y=1$ in (\ref{pole-part}).
This in essence is the origin of the meromorphy of the canonical counting 
function of the asymptotic states, and therefore indirectly also of the mock modularity of the immortal near-horizon counting function.

As we have seen, this meromorphy of the asymptotic counting function $\psi_{m}$ implies that the counting 
function of the near-horizon immortal states cannot be both modular and holomorphic
at the same time. One can either have the holomorphic function $\psi_{m}^{\rm F}$ which is not quite 
modular, or a modular function $\wh \psi_{m}^{\rm F}$ which is not quite holomorphic. 
{}From this point of view, the shadow can be regarded as a `modular anomaly' in $\psi_{m}^{\rm F}$ 
or as a `holomorphic anomaly' in the completion~$\wh \psi_{m}^{\rm F}$. 

A natural question that arises from this discussion is if the  mock Jacobi form and the holomorphic anomaly can be given a direct physical interpretation rather than arrive at it via a meromorphic Jacobi form as we have done. To address this question, it is useful to view the problem from the  perspective of $AdS_{3}$ holography in M-theory as we now discuss.

\subsection{The M-Theory limit  and $AdS_{3}$ holography \label{Mlimit}}

 {}From the  perspective of $AdS_{2}/CFT_{1}$ holography discussed above, there is no \textit{a priori} reason why the black hole degeneracy should have  anything to do with modularity. It is only when we can  view the black hole as an excitation of a black string that we can have an \textit{a priori} physical expectation of modularity. This is because the near horizon geometry of a black string is a three-dimensional anti de Sitter space $AdS_{3}$. The Euclidean thermal $AdS_{3}$ has a conformal torus as a boundary whose complex structure parameter is $\tau$. The partition function of the dual boundary conformal field theory $CFT_{2}$ depends on on $\tau$.  In this framework the  modular symmetry is identified with the $SL(2, \mathbb{Z})$ mapping class group of the boundary 
torus~\cite{Maldacena:1998bw}, whereas the  elliptic symmetry is identified with large gauge transformations of the 
3-form field~\cite{deBoer:2006vg} as we have already discussed in \S\ref{Degen}. 
  
We now describe how the quarter-BPS black holes considered  thus far can be regarded  as excitations of a black string in an appropriate M-theory frame. This  will enable us to make contact with $AdS_3/CFT_2$ holography. The M-theory framework is 
 natural also from the point of view of generalizing  these considerations to BPS black holes  in  ${\cal N}=2$  compactifications on a
general Calabi-Yau manifold $X_6$.

To go to the M-theory frame, we consider the $T^2$ to be  a product
of two circles $S^1 \times \tS^1$ and consider the charge configuration \eqref{final charges} in the Type-IIB frame.   We first  $T$-dualize the $\tilde S$ circle to the $\hat S^{1}$ circle to go to the Type-IIA frame and then use  mirror symmetry of $K3$ to map the Type-IIB configuration \eqref{final charges} to  a Type-IIA configuration of D4-F1-n-NS5 charges. We then  lift it to M-theory  to obtain a charge configuration consisting of M5-M2 bound states carrying momentum along $S^{1}$.  
This configuration now has  M5-branes 
wrapping  $D \times S^{1}$ where $D$ is a divisor in $K3 \times \hat T^2$, an M2-brane wrapping the $\hat T^2 = \hat S^{1} \times S^{1}_{m}$,  with some momentum  along  $S^{1}$. This  is a precisely a  configuration of the kind considered by Maldacena, Strominger, and Witten \cite{Maldacena:1997de}. Since all M5-branes wrap $S^{1}$ and the momentum is along $S^{1}$, it is natural  flip $S^{1}_{m}$ and $S^1$  and regard the $S^{1}$ as the new M-circle.   In the limit when the radius $R$ of the  new M-circle $S^{1}$ is very large, the low-lying excitations of this M5-brane are described by the MSW string wrapping the M-circle. 
To write the final configuration in the M-theory frame,  let $C^1$ to be the homology 2-cycle of the $\hat T^2$, and let   $C^2$ and $C^3$ be two 2-cycles in $K3$ with intersection matrix 
  \be \left(  \begin{array}{cc}  0 & 1 \\  1 & 0 \\  \end{array}\right) \, . \ee
The intersection number of these three 2-cycles is 
  \be \int_{K3 \times \hat T^2} C^1 \wedge C^2 \wedge C^3 =1 \,. \ee
Let  $\{ D_{a}\}$ be the  4-cycles   dual to $ \{ C^{a} \}$,  $D_{a} \cap C^{b} = \delta_{a}^{b}$ .
In this basis, the M5-brane is wrapping the  4-cycle $$D =  \sum_{a=1}^3p^{{a}} D_{a} $$ with $p^1=\tilde K$, $p^2=Q_1$ and $p^3=Q_5$.
The M2-brane charges are $q_2=q_3=0$ and $q_1=\tilde n$, and M-momentum or equivalently the D0-brane charge from the Type-IIA
perspective is $q_0=n$.  Using this notation we see that
\begin{myitemize} 
\item The electric charges in \eqref{final charges} map to $n$ units of momentum along the M-circle
$S^1$, and $\tilde K$  M5-branes wrapping   $D_1 \times S^1$.
  \item The magnetic charges in \eqref{final charges} map to $Q_1$ M5-branes  wrapping $D_2 \times S^1$, $Q_5$
M5-branes  wrapping $D_3 \times S^1$,  and $\tilde n$ M2-branes wrapping $\hat T^2$.
\end{myitemize}
In summary, the charge configuration of our interest is
  \be\label{M-charges}  \G = \left[ \begin{array}{c}  N \\ M \end{array}\right] 
     \= \left[ \begin{array}{cccc} 0,& n ; & 0, & \tilde K \\ Q_1 ,& \tilde n; & Q_5, & 0\\ \end{array} \right]_{B}  
     \=  \left[ \begin{array}{cccc} 0,& q_0 ; & 0, &  p^1 \\ p^2 ,&  q_1; & p^3, & 0\\ \end{array} \right]_{M}\, \ee
Reduction of this charge configurations along $S^{1}$ gives a configuration in Type-IIA frame consisting of only D4-D2-D0-branes.

Since the effective MSW string wraps the M-circle, we have to take the radius $R$  of this circle to be large  keeping other scales fixed to obtain an $AdS_3$ near-horizon geometry of a long black string.
This implies that  in the original  heterotic frame in which we have labeled the charges \eqref{chargevector},  the radius $R$ goes to infinity keeping other scales 
and moduli fixed. Since the 2-torus in the heterotic frame  is also a simple product of two circles $ S^1 \times \tilde S^{1}$,  the moduli in \eqref{metric-B}   take the form
  \be \label{M-moduli} G_{ij} =  \left( \begin{array}{cc} \tilde R^2 & 0\\  0 & R^2 \\  \end{array} \right) \, , \quad B_{ij} =0 \ee
which implies that for the three moduli $S, T, U$ we have 
  \be S\= a + i \frac{R\tilde R}{g_6^2}\,, \qquad T \= i R \tilde R\,, \qquad U = i\,\frac{R}{\tilde R} \,,  \ee
where $g_6$ is the six-dimensional string coupling constant of heterotic compactified on $T^4$ and we have allowed a nonzero
axion field $a$. This means in particular that
  \be S_2 \sim R, \quad T_2 \sim R, \quad U_2 \sim R \ee
and for the charge configuration (\ref{M-charges}),  the central charge matrix scales as
  \be\label{MSW-central-charge}  {\cal Z} \sim  \left( \begin{array}{cc} R &   0 \\ 
     0 & 0 \\ \end{array} \right) + \frac 1{R}\left( \begin{array}{cc} 2n \quad & \quad   \ell + {a g_6^2}/{\tilde R} \\ 
     \ell + {a g_6^2}/{\tilde R} & \quad 2m + (Q_1^2/{\tilde R}^2 + Q_5^{2 }{\tilde R}^2 )\\ \end{array} \right) \,  \, . \ee
From the contour prescription (\ref{contour}), we conclude that 
the imaginary parts  for the Fourier integral scale as 
  \begin{eqnarray}\label{Mcontour} \Im(\sigma) &\sim & R \\ \Im (\tau) &\sim & 1/R \\ \Im (z) & \sim & 1/R \, . \end{eqnarray}
Therefore, in the region of the moduli space corresponding to the M-theory limit,  $p := \exp{(2\pi i \sigma)}$ is becoming much
smaller compared to $q:= \exp{(2\pi i \tau)}$ and $y:= \exp{(2\pi i z)}$. Hence one must first expand around $p=0$:
  \be\label{defpsin} \frac 1{\Phi_{10}(\Omega)} \= \sum_{m=-1}^{\infty} p^m \psi_m (\t,z) \ . \ee
Since the function $\Phi_{10}(\Omega)$ has double zeros at $z=0$, the function $\psi_m(\t,z)$ is meromorphic with double
 poles at $z=0$ and its translates.

We conclude that the  function $\psi_{m}(\tau,z)$  can  be interpreted as the asymptotic  partition function for counting the BPS excitation of this MSW M5-brane  for fixed value of the 
magnetic charge $M^2/2= m$ and is  a \textit{meromorphic} Jacobi 
form of weight $-10$ and index $m$.
 This is the essential new ingredient compared to other situations encountered earlier such as the partition function of the  
D1-D5 string \cite{Strominger:1996sh} which is a \textit{holomorphic} weak Jacobi form\footnote{Here we mean holomorphic 
or meromorphic in the $z$ variable that is conjugate to $\ell$. This  should not be confused with the nomenclature used  in \S\ref{Jacobi} while discussing  growth conditions \eqref{holjacobi}.}.

 Note that $\psi_{m}(\tau,z)$ is the asymptotic counting function  in the canonical ensemble for a fixed chemical potential $z$ where one sums
over all M2-brane charges $\ell$ for a fixed chemical potential $z$. To obtain the micro-canonical partition function for a fixed M2-brane
charge, one has to perform the Fourier integral in $z$.  There is an infinity of double poles in the $z$ plane at $z = \a\t $ for
$\a \in \Z$. The $z$ contour can be made to lie between the poles at $(\a-1)\t$ and $\a\t$ depending on how $\Im(z)$ compares with
 $\Im(\tau)$ given  the $O(1/R)$ terms in the central charge matrix (\ref{MSW-central-charge}).   It is easy to see that by varying
the moduli, the contour can be made to cross all poles. Thus, there are an infinity of walls still accessible in the regime
when the radius of the M-circle becomes large. 

\begin{figure}[h]
\begin{center}
\includegraphics[height=6cm]{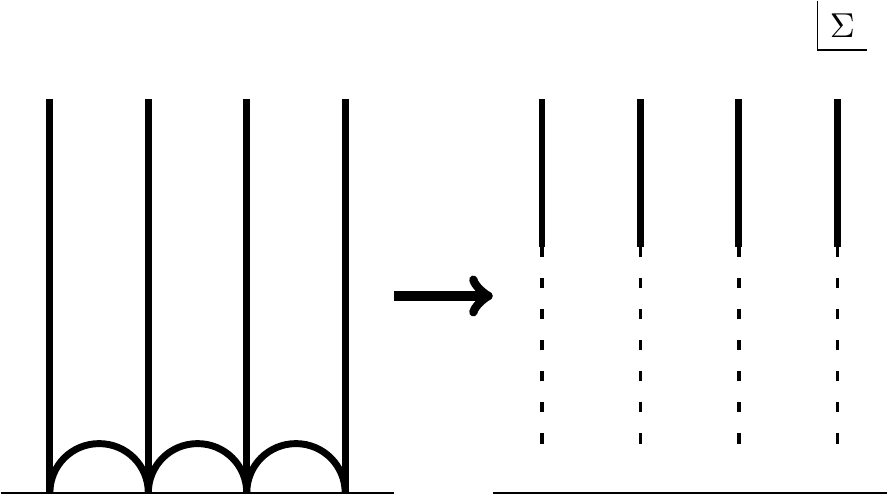}
\end{center}
\caption{
The diagram on the left shows the moduli space projected onto the $\Sigma$ upper-half plane divided into chambers separated 
by walls where  a quarter-BPS state decays into two half-BPS states. The M-theory limit, shown on the right, corresponds to taking
$\Im(\Sigma)$ to be very large. In this limit several walls are no longer visible.  There are still an infinite number of
walls  which can be crossed by varying $\Re(\Sigma)$. Spectral flow transformation maps one chamber to another chamber.
}
\label{walls}
\end{figure}
 
To display this graphically it is useful to define a complex scalar field $\Sigma = \Sigma_1 + i \Sigma_2$ by
 \be 
 {\CZ } \sim   \frac 1{\Sigma_2}\left(\begin{array}{cc} |\Sigma | ^2& \Sigma_1 \\  \Sigma_1& 1\\   \end{array}\right) \, . \ee
For fixed charges, the MSW region which corresponds to taking $\Sigma_2$ large with $\Sigma_1$ fixed corresponding to 
vanishing axion. Varying $\Sigma_1$  allows one to access the infinite number of chambers separated by walls 
as shown in Fig. \ref{walls}. 
Degeneracies in these different chambers can be obtained simply by choosing the Fourier contour in $z$ for an appropriate
value of $\alpha$. Note that the field $\Sigma_1$ depends both on the moduli  and    on the charges and thus varies as
moduli are varied for fixed charges or as the charges are varied for fixed moduli.   In particular, even for fixed values of
moduli, one can cross all walls by varying $\ell$ appropriately.  

With this embedding into M-theory, we  now outline what we regard as a consistent physical picture to interpret $\wh \psi_{m}^{F}$ even though we do not yet fully understand the details.   In the  limit of large radius $R$, we effectively have a five-dimensional theory obtained from compactifying M-theory on $K3 \times \hat T^{2}$. For a given a set of M5-brane charges $\{p^{i}\}$ at asymptotic infinity,  the five-dimensional supergravity   admits a single-centered black string solution carrying magnetic charges $\{p^{i}\}$. The near horizon geometry of this solution is an $AdS_{3}$ geometry where all vector multiplet moduli of the five-dimensional supergravity are fixed to their attractor values.   Following the usual reasoning of holography, we  expect that the superconformal field theory describing the low energy excitations of the M5-brane will be dual to this near horizon $AdS_{3}$. The indexed partition function of this SCFT is expected to describe the BPS gravitational configurations inside the $AdS_{3}$.  

What is the indexed partition function of this SCFT? It  cannot possibly be the meromorphic Jacobi form $\psi_{m}$ that exhibits many wall-crossings, because  all vector-multiplet moduli are fixed in the attractor $AdS_{3}$ geometry and no wall-crossings are possible.   It cannot possibly be $\psi_{m}^{F}$, because it is not modular as would be expected from the symmetry under the $SL(2, \mathbb{Z})$ mapping class group of the boundary torus. We propose that  the most natural  object to identify with   the indexed partition function of the $SCFT_{2}$ dual to the $AdS_{3}$ is the modular completion $\wh \psi_{m}^{F}(\t, z)$.  It satisfies (in a rather nontrivial way) both requirements of being properly modular and  not having any wall-crossings. The non-holomorphy is not something that we would naively expect for the indexed partition function but this can be a reflection of noncompactness as we comment upon later. The quarter-BPS black holes discussed in this paper are  excitations inside this single $AdS_{3}$ represented as BTZ black holes and their degeneracies should then  be captured by the Fourier coefficients   of $\psi_{m}^{F}(\t, z)$ which is indeed the case for $n>m$.

This \textit{microscopic} picture is consistent with the \textit{macroscopic}  analysis of
multi-centered black string solutions \cite{deBoer:2008fk} of  five dimensional supergravity.  These solutions can be obtained by lifting  multi-centered black hole solutions of four-dimensional supergravity. According to this supergravity analysis, given a set of asymptotic magnetic charges $\{p^{i}\}$ of a wrapped M5-brane,  the single-centered $AdS_{3}$ throat described above is not the only allowed solution \cite{deBoer:2008fk}.   There exist multi-centered solutions with multiple $AdS_{3}$ throats with the same total asymptotic charge.  Now,  as we have seen from the microscopic analysis, 
 there are still  infinite chambers separated by walls visible in the M-theory limit as shown in Fig~\ref{walls}. These chambers  can thus  correspond to chambers in which different   multi-string configurations exist as stable solutions.  The  wall-crossings can result  not from decays of multi-centered black holes inside a single $AdS_{3}$ throat but from the decays of these multi-string solutions. 
Thus, it is reasonable to identify~$\wh\psi_{m}^{F}$,  which has no wall-crossings, with the (indexed) partition function of the CFT dual to the single~$AdS_{3}$ throat.

\subsection{Open problems and comments\label{generalization}}

The interpretation  of $\wh \psi_{m}^{F}(\t, z)$ proposed above raises a number of interesting questions. The essentially new feature of our proposal  is the fact  that the modular partition function $\wh \psi_{m}^{F}(\t, z)$  is nonholomorphic and correspondingly the holomorphic counting function $\psi_{m}^{F}$ is mock modular. It would clearly be important to understand the physical origin of this nonholomorphy and mock modularity from the boundary and bulk perspectives.

\begin{itemize}

\item 
\textit{Mock modularity from  the boundary perspective}:
\vspace{2mm}

 We have seen in \S\ref{Anomaly} that the mock modularity of $\psi^{F}_{m}$  is tied  to the noncompactness of the asymptotic SCFT. A more interesting question is to understand this mock modularity directly from the point of view of the near horizon SCFT. It is natural to assume that the near horizon SCFT is also noncompact and this noncompactness is what is responsible for the mock modularity.  The connection between mock modularity/nonholomorphy and noncompactness has been noted  earlier in physics. For example, the partition function
of topological $SO(3)$  $\CN =4$ Yang-Mills theory on $\mathbb{CP}^2$ is a mock modular form and it was suggested in \cite{Vafa:1994tf} that it is  related to  the noncompactness of  flat directions in field space. The nonholomorphy of the elliptic genus has also been noted in the context of superconformal 
field theories \cite{Eguchi:2004yi,Troost:2010ud} with the noncompact $SL(2, \mathbb{R})/U(1)$ target 
space.

In this context,  the nonholomorphy  can be understood in complete detail \cite{Troost:2010ud, Ashok:2011cy} as a consequence of the continuum of spectrum of  the  noncompact SCFT. 
It would be interesting to see if  the holomorphic anomaly \eqref{holanom} of the completion of $\wh \psi_{m}^{F}$ 
can be understood from such a path integral perspective of the conformal field theory. 

In the framework of  $AdS_{2}$ holography, the quantity of interest  is the integer $d^{*}(n, m, \ell)$ which is identified with the partition function of the boundary $CFT$. Using the fact that the  completion $\wh \psi_{m}^{F}(\t, z)$ is modular, it should be possible to develop a Rademacher-like expansion for the $d^{*}(n, m, \ell)$ 
along the lines of \cite{BringOno, BringRicht, Bringmann:2010sd}. This method has been used to prove 
that the black hole degeneracies are positive integers~\cite{Bringmann:2012zr}.

\item 
\textit{Mock modularity from the bulk perspective}:

If mock modularity is indeed a consequence of the noncompactness of the  boundary $CFT$, what are the implications of this noncompactness in the bulk $AdS_{3}$? In particular, if the spectrum of conformal dimensions is continuous, then the bulk theory must have a continuum of masses. 
Perhaps these are to be identified with some extended long-string like excitations. It would be interesting to see if the mock modularity of the elliptic genus can be understood  in terms of such bulk excitations. 
Indeed, recent investigations seem to indicate that the bulk string partition function in 
string theories~\cite{Harvey:2013mda} 
 based on the non-compact $SL(2, \mathbb{R})/U(1)$ CFT~\cite{Giveon:1999px, Murthy:2003es} 
show a mock modular behavior. 

The M-theory limit corresponds to the attractor contour  for $n>m$ and $0 \leq \ell <2m$ but the converse is not necessarily true, as explained in \S\ref{Degen}. For $n \leq m$, we have the possibility that there will be multi-centered configurations. In particular, if  $n \leq m$, then even if $\ell <2m$, it is possible to have states with non-positive discriminant ($\Delta = 4mn - \ell^{2} \leq 0)$.  Since, the discriminant must be positive for a black hole horizon, such states cannot possibly correspond to single-centered black holes and must be realized in supergravity  as multi-centered configurations if there are smooth solutions corresponding to them at all.  Hence, if  
$\wh \psi_{m}^{F}$ is indeed  the elliptic genus of a single-centered $AdS_{3}$, then there must be supergravity configurations corresponding to these  states for $n \leq m$ that fit inside a single  $AdS_{3}$. 

This microscopic prediction is  plausible  from the supergravity perspective. According to \cite{deBoer:2008fk}, all two-centered black hole solutions where each center is a D4-D2 bound state,  are scaled out of  a single~$AdS_{3}$ throat. However,  the two-centered solutions that have one center carrying D6-brane charge and the other carrying anti-D6-brane charge can continue to exist as stable solutions inside a single $AdS_{3}$.  Such configurations can give the supergravity realization of the  negative discriminant states.  Further analysis  is required to verify these consequences of  our proposed interpretation of $\wh \psi_{m}^{F}$. 

\item
\textit{Generalization to $\CN=2$ compactifications}:

One motivation for the present investigation is to abstract some general lessons which
 could be applicable to the more
challenging case of dyons in $\CN=2$ compactifications on a general Calabi-Yau
three-fold. 
In the $\CN=4$ case, the mock modularity of $\psi_{m}^{F}$ was a consequence of meromorphy of $\psi_{m}$. As we have seen in \S\ref{Anomaly}, this meromorphy is in turn a consequence of noncompactness of Taub-NUT factor $ \sigma(\textrm{TN})$ in   the asymptotic SCFT (\ref{scft})\footnote{Note that this is not to be confused with $\mathbb{R}^{3}$ factor corresponding to the center of mass motion. It corresponds to the relative motion of the KK-monopole and the D1D5 system \cite{David:2006yn} and thus  would  belong to  what is called the ``entropic" factor
\cite{Minasian:1999qn} for  the MSW string in the M-theory frame.}.  This suggests that the SCFT dual to the near horizon $AdS_{3}$ is also noncompact. In the  $\CN=2$
context also, there is no a priori reason to
expect that the MSW SCFT will be compact in general  everywhere in the moduli space.
If this SCFT is noncompact, then our results suggest that mock Jacobi forms would provide the right framework for making the  holographic
$SL(2, \Z)$ symmetry manifest for the indexed partition function considered, for example, in \cite{Dijkgraaf:2000fq,deBoer:2006vg,Manschot:2007ha}.  
In the $\CN=4$ case, the  mock modularity  is closely
correlated with the wall-crossing phenomenon. The wall-crossing phenomenon in $\CN=2$ case corresponds to jumps in the
Donaldson-Thomas invariants  and is known to be much
more complicated \cite{Denef:2007vg, Konstevich:2008, Joyce:2008, Manschot:2010qz}. 
It would be interesting though to define the  analog of the counting function for the immortal degeneracies in this case. Since these degeneracies  do not change under wall-crossing, such a counting function is likely to have an interesting mathematical interpretation. Moreover, it is expected to have nice (mock) modular properties from the considerations of holography discussed above.

\item
\textit{Mock modularity and indefinite theta series}:

Mock modular forms appear in Zwegers's work as members of  three quite different families of functions:
\begin{myenumerate}
\item Fourier coefficients of meromorphic Jacobi forms, 
\item Appell-Lerch sums,  
\item Indefinite theta series. 
\end{myenumerate}
We have seen that the first two families appear naturally in the context of black hole physics with 
natural physical interpretations. This  suggests that the indefinite theta series could also play a role 
in this physical context. It  would be  interesting to develop this connection further. 
For some earlier work in this direction, see~\cite{Manschot:2009ia, Alexandrov:2012au}.
\end{itemize}

\smallskip

\section*{Acknowledgments}
%
%
%
We would like to acknowledge useful conversations with Kathrin Bringmann, Miranda Cheng, Jan de
Boer, Rajesh Gopakumar, Jeff Harvey, Jan Manschot, Shiraz Minwalla, Boris Pioline, Jan Troost, 
Erik Verlinde, Sander Zwegers, and especially Ashoke Sen. We would also like to thank 
Alexander Weisse for technical help. It is a pleasure to thank the
organizers of the 2007 workshop on `Black Holes and Modular Forms' in Paris and of the
2008 workshop on `Number Theory and Physics at Crossroads' in Banff for hospitality
where this research was initiated. The work of A.~D. was supported in part by the
Excellence Chair of the Agence Nationale de la Recherche (ANR). The work of S.~M.
was supported by the ERC Advanced Grant no.~246974,
{\it ``Supersymmetry: a window to non-perturbative physics''}. \\
v2: The work of A.~D.~was conducted within the framework of  the ILP LABEX (ANR-10-LABX-63)  supported by French state funds managed by the ANR within the Investissements d'Avenir programme under reference ANR-11-IDEX-0004-02. S.~M.~would like to thank the SFTC for support from Consolidated grant number ST/J002798/1.

\newpage

\appendix

\section{Appendix: Tables of special mock Jacobi forms \label{apptabs}}

\smallskip

\small

{\normalsize 
This appendix contains three collections of tables. The first gives the first Fourier coefficients 
of the weight~2 mock Jacobi forms~$\CQ_{M}$ for all products of two primes less than~50, the second gives
coefficients of the primitive mock Jacobi forms~$\Phi_{1,m}^{0}$ for all known cases having optimal growth, 
and the third gives coefficients of the weight~1 mock Jacobi forms~$\CQ_{M}$ for  all products of three primes
less than~150.

\bigskip

\subsection{Table of $\CQ_M$ (weight 2 case)\label{apptabs1}}

For each~$M$ with~$\mu(M)=1$ with~$1<M<50$ we tabulate the value of~$c_{M}$ and the Fourier
coefficients~$c(\CF_{M}; n, r_{\rm min})$ for~$n$ ranging from~0 to some limit and all values of~$r_{\rm min}$,
(defined as the minimal positive representatives prime to~$M$ of their squares modulo~$4M$). 
Here $c_{M}$ is the smallest positive rational number such that~$\CF_{M}=-c_{M}\CQ_{M}$ has integral 
coefficients, as defined in~\eqref{DefOfcM}.  The asterisk before a value of~$r_{\rm min}$
means that the corresponding row contains at least one negative non-polar coefficient. 
}

\vskip.2cm {\bf M = 6} \qquad{$c_M\=12$}
\begin{center} \begin{tabular}{|c|ccccccccccccc|cccc}
\hline
\diagbox[width=1.5cm,height=0.8cm]{$r_{\text{min}}$}{$n$} & 0 & 1 & 2 & 3 & 4 & 5 & 6 & 7 & 8 & 9 & 10 & 11 &12 \\ \hline
1 & $-1$ & $35$ & $130$ & $273$ & $595$ & $1001$ & $1885$ & $2925$ & $4886$ & $7410$ & $11466$ & $16660$ & $24955$ \\ 
\hline \end{tabular}  \end{center}

\vskip.2cm {\bf M = 10}\qquad{$c_M\=6$}
\begin{center} \begin{tabular}{|c|cccccccccccccc|cccc}
\hline
\diagbox[width=1.5cm,height=0.8cm]{$r_{\text{min}}$}{$n$} & 0 & 1 & 2 & 3 & 4 & 5 & 6 & 7 & 8 & 9 & 10 & 11 & 12 & 13 \\ \hline
1 & $-1$ & $21$ & $63$ & $112$ & $207$ & $306$ & $511$ & $693$ & $1071$ & $1442$ & $2037$ & $2709$ & $3766$ & $4788$   \\ 
3 & $0$ & $9$ & $35$ & $57$ & $126$ & $154$ & $315$ & $378$ & $625$ & $819$ & $1233$ & $1491$ & $2268$ & $2772$ \\ 
\hline
\end{tabular}  \end{center}

\vskip.2cm {\bf M = 14} \qquad $c_M\=4$
\begin{center} \begin{tabular}{|c|cccccccccccccc|cccc}
\hline
\diagbox[width=1.5cm,height=0.8cm]{$r_{\text{min}}$}{$n$}  & 0 & 1 & 2 & 3 & 4 & 5 & 6 & 7 & 8 & 9 & 10 & 11 & 12 & 13 \\ \hline
1 & $-1$ & $15$ & $42$ & $65$ & $120$ & $150$ & $255$ & $312$ & $465$ & $575$ & $819$ & $975$ & $1355$  & $1605$ \\ 
3 & $0$ & $10$ & $30$ & $51$ & $85$ & $120$ & $195$ & $230$ & $360$ & $465$ & $598$ & $765$ & $1065$ & $1235$ \\ 
5 & $0$ & $3$ & $15$ & $20$ & $51$ & $45$ & $113$ & $105$ & $195$ & $215$ & $345$ & $348$ & $595$ & $615$ \\ 
\hline
\end{tabular}  \end{center}

\vskip.2cm {\bf M = 15} \qquad $c_M\=3$
\begin{center} \begin{tabular}{|c|ccccccccccccccc|cc}
\hline
\diagbox[width=1.5cm,height=0.8cm]{$r_{\text{min}}$}{$n$} & 0 & 1 & 2 & 3 & 4 & 5 & 6 & 7 & 8 & 9 & 10 & 11 & 12 & 13 & 14 \\ \hline
1 & $-1$ & $14$ & $29$ & $54$ & $83$ & $128$ & $172$ & $230$ & $331$ & $430$ & $537$ & $726$ & $924$ & $1116$ & $1409$  \\ 
2 & $0$ & $4$ & $24$ & $28$ & $58$ & $52$ & $136$ & $120$ & $224$ & $232$ & $376$ & $368$ & $626$ & $632$ & $920$   \\ 
4 & $0$ & $10$ & $28$ & $40$ & $88$ & $98$ & $176$ & $208$ & $304$ & $386$ & $552$ & $632$ & $888$ & $1050$ & $1372$ \\ 
7 & $0$ & $2$ & $9$ & $26$ & $36$ & $50$ & $84$ & $116$ & $155$ & $226$ & $269$ & $340$ & $488$ & $580$ & $712$ \\ 
\hline
\end{tabular}  \end{center}

\newpage
\vskip.2cm {\bf M = 21} \qquad $c_M\=2$
\begin{center} \begin{tabular}{|c|cccccccccccccccc|cc}
\hline
\diagbox[width=1.5cm,height=0.8cm]{$r_{\text{min}}$}{$n$}  & 0 & 1 & 2 & 3 & 4 & 5 & 6 & 7 & 8 & 9 & 10 & 11 & 12 & 13 & 14 & 15 \\ \hline
1 & $-1$ & $8$ & $15$ & $32$ & $40$ & $56$ & $79$ & $96$ & $135$ & $176$ & $191$ & $248$ & $336$ & $368$ & $434$ & $576$   \\ 
2 & $0$ & $8$ & $26$ & $28$ & $56$ & $64$ & $112$ & $104$ & $180$ & $184$ & $272$ & $300$ & $416$ & $432$ & $606$ & $636$   \\ 
4 & $0$ & $2$ & $12$ & $8$ & $32$ & $14$ & $60$ & $40$ & $80$ & $64$ & $148$ & $96$ & $208$ & $166$ & $272$ & $268$   \\ 
5 & $0$ & $8$ & $17$ & $32$ & $49$ & $64$ & $85$ & $120$ & $150$ & $200$ & $241$ & $288$ & $374$ & $448$ & $534$ & $656$   \\ 
8 & $0$ & $2$ & $12$ & $20$ & $32$ & $38$ & $76$ & $72$ & $112$ & $138$ & $176$ & $200$ & $304$ & $302$ & $408$ & $472$   \\ 
11 & $0$ & $0$ & $3$ & $8$ & $11$ & $24$ & $26$ & $32$ & $59$ & $64$ & $74$ & $112$ & $138$ & $144$ & $189$ & $248$   \\ 
\hline
\end{tabular}  \end{center}

\vskip.2cm {\bf M = 22} \qquad $c_M\=12$
\begin{center} \begin{tabular}{|c|cccccccccccccc|ccc}
\hline
\diagbox[width=1.5cm,height=0.8cm]{$r_{\text{min}}$}{$n$}  & 0 & 1 & 2 & 3 & 4 & 5 & 6 & 7 & 8 & 9 & 10 & 11 & 12 & 13  \\ \hline
1 & $-5$ & $49$ & $132$ & $176$ & $308$ & $357$ & $572$ & $604$ & $951$ & $1061$ & $1391$ & $1572$ & $2227$ & $2325$   \\ 
3 & $0$ & $44$ & $103$ & $169$ & $272$ & $301$ & $536$ & $580$ & $771$ & $976$ & $1307$ & $1380$ & $2004$ & $2121$  \\ 
5 & $0$ & $25$ & $91$ & $116$ & $223$ & $229$ & $427$ & $405$ & $716$ & $700$ & $1071$ & $1121$ & $1608$ & $1633$ \\ 
7 & $0$ & $13$ & $48$ & $79$ & $127$ & $180$ & $272$ & $272$ & $452$ & $549$ & $632$ & $764$ & $1120$ & $1133$ \\ 
9 & $0$ & $1$ & $23$ & $24$ & $91$ & $25$ & $180$ & $92$ & $247$ & $205$ & $403$ & $229$ & $715$ & $453$ \\ 
\hline
\end{tabular}  \end{center}

\vskip.2cm {\bf M = 26} \qquad $c_M\=2$
\begin{center} \begin{tabular}{|c|cccccccccccccccc|c}
\hline
\diagbox[width=1.5cm,height=0.8cm]{$r_{\text{min}}$}{$n$}  & 0 & 1 & 2 & 3 & 4 & 5 & 6 & 7 & 8 & 9 & 10 & 11 & 12 & 13 & 14 & 15 \\ \hline
1 & $-1$ & $9$ & $21$ & $30$ & $51$ & $54$ & $87$ & $96$ & $138$ & $149$ & $207$ & $219$ & $308$ & $315$ & $399$ & $468$   \\ 
3 & $0$ & $7$ & $21$ & $27$ & $43$ & $51$ & $90$ & $76$ & $132$ & $144$ & $183$ & $201$ & $297$ & $283$ & $393$ & $423$   \\ 
5 & $0$ & $6$ & $15$ & $23$ & $39$ & $42$ & $69$ & $75$ & $102$ & $128$ & $165$ & $165$ & $243$ & $258$ & $324$ & $371$   \\ 
7 & $0$ & $3$ & $12$ & $14$ & $33$ & $27$ & $57$ & $51$ & $90$ & $84$ & $132$ & $126$ & $201$ & $183$ & $261$ & $267$   \\ 
9 & $0$ & $1$ & $6$ & $12$ & $15$ & $21$ & $36$ & $34$ & $57$ & $63$ & $70$ & $93$ & $138$ & $118$ & $162$ & $210$   \\ 
11 & $0$ & $0$ & $3$ & $2$ & $12$ & $3$ & $26$ & $6$ & $30$ & $24$ & $57$ & $15$ & $86$ & $45$ & $87$ & $81$ \\
\hline
\end{tabular}  \end{center}

\vskip.2cm {\bf M = 33} \qquad $c_M\=6$
\begin{center} \begin{tabular}{|r|ccccccccccccccc|cc}
\hline
\diagbox[width=1.5cm,height=0.8cm]{$r_{\text{min}}$}{$n$}  & 0 & 1 & 2 & 3 & 4 & 5 & 6 & 7 & 8 & 9 & 10 & 11 & 12 & 13 & 14  \\ \hline
1 & $-5$ & $33$ & $55$ & $127$ & $127$ & $160$ & $242$ & $259$ & $347$ & $468$ & $468$ & $534$ & $749$ & $826$ & $903$ \\ 
2 & $0$ & $21$ & $66$ & $43$ & $142$ & $120$ & $252$ & $174$ & $340$ & $262$ & $526$ & $394$ & $658$ & $580$ & $844$  \\ 
4 & $0$ & $24$ & $88$ & $88$ & $165$ & $154$ & $286$ & $244$ & $418$ & $376$ & $572$ & $596$ & $869$ & $728$ & $1122$   \\ 
5 & $0$ & $14$ & $36$ & $47$ & $80$ & $113$ & $113$ & $160$ & $215$ & $259$ & $229$ & $372$ & $460$ & $386$ & $507$  \\ 
7 & $0$ & $33$ & $44$ & $99$ & $143$ & $167$ & $209$ & $308$ & $365$ & $440$ & $506$ & $541$ & $748$ & $847$ & $915$  \\ 
* \, 8 & $0$ & $-4$ & $29$ & $18$ & $62$ & $-4$ & $146$ & $-8$ & $124$ & $76$ & $234$ & $54$ & $336$ & $116$ & $380$  \\ 
10 & $0$ & $11$ & $50$ & $66$ & $116$ & $110$ & $248$ & $220$ & $292$ & $358$ & $474$ & $435$ & $738$ & $622$ & $892$  \\ 
13 & $0$ & $0$ & $21$ & $54$ & $77$ & $120$ & $131$ & $153$ & $230$ & $273$ & $317$ & $405$ & $482$ & $548$ & $613$ \\ 
16 & $0$ & $0$ & $3$ & $22$ & $58$ & $36$ & $94$ & $124$ & $127$ & $146$ & $248$ & $182$ & $350$ & $380$ & $408$  \\ 
19 & $0$ & $0$ & $0$ & $7$ & $18$ & $7$ & $40$ & $40$ & $58$ & $113$ & $58$ & $113$ & $186$ & $153$ & $120$ \\
\hline
\end{tabular}  \end{center}

\newpage
\vskip.2cm {\bf M = 34} \qquad $c_M\=3$
\begin{center} \begin{tabular}{|r|ccccccccccccccc|cc}
\hline
\diagbox[width=1.5cm,height=0.8cm]{$r_{\text{min}}$}{$n$}  & 0 & 1 & 2 & 3 & 4 & 5 & 6 & 7 & 8 & 9 & 10 & 11 & 12 & 13 & 14 \\ \hline
1 & $-2$ & $14$ & $36$ & $43$ & $69$ & $82$ & $123$ & $117$ & $183$ & $195$ & $253$ & $252$ & $364$ & $352$ & $481$  \\ 
3 & $0$ & $13$ & $29$ & $44$ & $77$ & $61$ & $125$ & $116$ & $164$ & $185$ & $244$ & $234$ & $381$ & $334$ & $427$   \\ 
5 & $0$ & $11$ & $28$ & $39$ & $59$ & $67$ & $106$ & $107$ & $162$ & $157$ & $214$ & $242$ & $321$ & $306$ & $391$  \\ 
7 & $0$ & $7$ & $26$ & $28$ & $53$ & $54$ & $100$ & $74$ & $149$ & $144$ & $194$ & $177$ & $292$ & $252$ & $378$  \\ 
9 & $0$ & $5$ & $14$ & $27$ & $41$ & $43$ & $79$ & $72$ & $92$ & $125$ & $165$ & $141$ & $236$ & $234$ & $262$  \\ 
11 & $0$ & $1$ & $13$ & $12$ & $36$ & $26$ & $57$ & $43$ & $100$ & $69$ & $124$ & $100$ & $196$ & $140$ & $236$  \\ 
13 & $0$ & $0$ & $5$ & $12$ & $16$ & $18$ & $34$ & $32$ & $54$ & $65$ & $53$ & $85$ & $128$ & $80$ & $139$  \\ 
* \, 15 & $0$ & $0$ & $2$ & $1$ & $13$ & $-2$ & $35$ & $0$ & $29$ & $14$ & $59$ & $-2$ & $94$ & $28$ & $72$ \\ 
\hline
\end{tabular}  \end{center}

\vskip.2cm {\bf M = 35} \qquad $c_M\=1$
\begin{center} \begin{tabular}{|r|cccccccccccccccc|c}
\hline
\diagbox[width=1.5cm,height=0.8cm]{$r_{\text{min}}$}{$n$}  & 0 & 1 & 2 & 3 & 4 & 5 & 6 & 7 & 8 & 9 & 10 & 11 & 12 & 13 & 14 & 15 \\ \hline
1 & $-1$ & $10$ & $14$ & $18$ & $31$ & $44$ & $49$ & $60$ & $76$ & $86$ & $98$ & $134$ & $147$ & $168$ & $200$ & $226$ \\ 
* \, 2 & $0$ & $-2$ & $8$ & $4$ & $14$ & $-2$ & $24$ & $8$ & $30$ & $12$ & $56$ & $0$ & $70$ & $22$ & $72$ & $44$ \\ 
3 & $0$ & $2$ & $6$ & $16$ & $11$ & $18$ & $24$ & $26$ & $43$ & $50$ & $44$ & $58$ & $84$ & $70$ & $93$ & $124$ \\ 
4 & $0$ & $6$ & $10$ & $16$ & $28$ & $26$ & $46$ & $40$ & $56$ & $70$ & $90$ & $96$ & $128$ & $128$ & $152$ & $176$ \\ 
6 & $0$ & $6$ & $16$ & $16$ & $30$ & $32$ & $56$ & $52$ & $78$ & $80$ & $104$ & $108$ & $152$ & $150$ & $208$ & $216$ \\ 
8 & $0$ & $2$ & $10$ & $8$ & $20$ & $12$ & $42$ & $24$ & $52$ & $40$ & $74$ & $48$ & $104$ & $90$ & $122$ & $120$ \\ 
9 & $0$ & $2$ & $3$ & $8$ & $14$ & $18$ & $17$ & $16$ & $30$ & $42$ & $39$ & $58$ & $69$ & $60$ & $71$ & $100$ \\ 
11 & $0$ & $2$ & $5$ & $16$ & $22$ & $20$ & $30$ & $42$ & $52$ & $68$ & $67$ & $84$ & $114$ & $118$ & $121$ & $166$ \\ 
13 & $0$ & $0$ & $5$ & $10$ & $9$ & $18$ & $24$ & $34$ & $33$ & $42$ & $59$ & $52$ & $79$ & $92$ & $99$ & $126$ \\ 
16 & $0$ & $0$ & $2$ & $0$ & $12$ & $12$ & $16$ & $8$ & $28$ & $18$ & $50$ & $40$ & $56$ & $36$ & $88$ & $56$ \\ 
18 & $0$ & $0$ & $0$ & $4$ & $6$ & $6$ & $24$ & $8$ & $30$ & $24$ & $32$ & $40$ & $64$ & $46$ & $80$ & $80$ \\ 
23 & $0$ & $0$ & $0$ & $0$ & $1$ & $2$ & $5$ & $8$ & $7$ & $16$ & $11$ & $8$ & $33$ & $26$ & $19$ & $50$ \\ 
\hline
\end{tabular}  \end{center}

\vskip.2cm {\bf M = 38} \qquad $c_M\=4$
\begin{center} \begin{tabular}{|r|ccccccccccccccc|cc}
\hline
\diagbox[width=1.5cm,height=0.8cm]{$r_{\text{min}}$}{$n$}  & 0 & 1 & 2 & 3 & 4 & 5 & 6 & 7 & 8 & 9 & 10 & 11 & 12 & 13 & 14 \\ \hline
1 & $-3$ & $21$ & $44$ & $59$ & $103$ & $97$ & $159$ & $149$ & $238$ & $240$ & $322$ & $300$ & $465$ & $445$ & $540$  \\ 
3 & $0$ & $16$ & $46$ & $62$ & $80$ & $102$ & $166$ & $138$ & $226$ & $231$ & $287$ & $318$ & $449$ & $391$ & $559$  \\ 
5 & $0$ & $15$ & $43$ & $48$ & $91$ & $75$ & $151$ & $148$ & $194$ & $208$ & $303$ & $269$ & $427$ & $354$ & $512$ \\ 
7 & $0$ & $14$ & $29$ & $43$ & $82$ & $83$ & $119$ & $117$ & $179$ & $197$ & $256$ & $252$ & $360$ & $358$ & $433$  \\ 
9 & $0$ & $6$ & $30$ & $36$ & $61$ & $55$ & $124$ & $77$ & $179$ & $165$ & $201$ & $193$ & $360$ & $256$ & $404$  \\ 
11 & $0$ & $4$ & $16$ & $31$ & $47$ & $45$ & $88$ & $91$ & $108$ & $136$ & $168$ & $155$ & $257$ & $254$ & $300$  \\ 
13 & $0$ & $0$ & $14$ & $14$ & $40$ & $26$ & $72$ & $45$ & $103$ & $64$ & $169$ & $105$ & $207$ & $137$ & $249$  \\ 
15 & $0$ & $0$ & $5$ & $11$ & $14$ & $29$ & $43$ & $25$ & $57$ & $75$ & $58$ & $76$ & $160$ & $88$ & $133$  \\ 
*\,17 & $0$ & $0$ & $1$ & $1$ & $20$ & $-7$ & $33$ & $-5$ & $47$ & $21$ & $60$ & $-14$ & $109$ & $15$ & $83$ \\ 
\hline
\end{tabular}  \end{center}

\newpage

\vskip.2cm {\bf M = 39} \qquad $c_M\=1$
\begin{center} \begin{tabular}{|r|cccccccccccccccc|c}
\hline
\diagbox[width=1.5cm,height=0.8cm]{$r_{\text{min}}$}{$n$}  & 0 & 1 & 2 & 3 & 4 & 5 & 6 & 7 & 8 & 9 & 10 & 11 & 12 & 13 & 14 & 15 \\ \hline
1 & $-1$ & $6$ & $8$ & $16$ & $22$ & $26$ & $31$ & $42$ & $47$ & $60$ & $72$ & $72$ & $94$ & $110$ & $115$ & $154$ \\ 
2 & $0$ & $4$ & $14$ & $16$ & $26$ & $20$ & $52$ & $32$ & $64$ & $60$ & $86$ & $76$ & $128$ & $100$ & $148$ & $152$ \\ 
4 & $0$ & $2$ & $12$ & $6$ & $20$ & $16$ & $36$ & $22$ & $48$ & $32$ & $68$ & $46$ & $96$ & $66$ & $124$ & $86$ \\ 
5 & $0$ & $6$ & $11$ & $18$ & $22$ & $34$ & $37$ & $44$ & $64$ & $66$ & $71$ & $102$ & $116$ & $120$ & $137$ & $170$ \\ 
7 & $0$ & $2$ & $3$ & $10$ & $12$ & $12$ & $16$ & $18$ & $28$ & $42$ & $28$ & $40$ & $64$ & $44$ & $57$ & $88$ \\ 
8 & $0$ & $4$ & $12$ & $12$ & $28$ & $22$ & $44$ & $40$ & $56$ & $62$ & $88$ & $74$ & $120$ & $108$ & $156$ & $148$ \\ 
*\,10 & $0$ & $0$ & $4$ & $0$ & $10$ & $-4$ & $22$ & $4$ & $18$ & $0$ & $36$ & $0$ & $46$ & $16$ & $42$ & $24$ \\ 
11 & $0$ & $2$ & $7$ & $16$ & $17$ & $24$ & $34$ & $36$ & $47$ & $66$ & $62$ & $76$ & $107$ & $104$ & $117$ & $154$ \\ 
14 & $0$ & $0$ & $6$ & $8$ & $14$ & $20$ & $28$ & $24$ & $46$ & $40$ & $62$ & $60$ & $86$ & $72$ & $118$ & $116$ \\ 
17 & $0$ & $0$ & $1$ & $6$ & $10$ & $12$ & $15$ & $24$ & $31$ & $34$ & $35$ & $44$ & $61$ & $72$ & $71$ & $86$ \\ 
20 & $0$ & $0$ & $0$ & $2$ & $4$ & $6$ & $16$ & $8$ & $16$ & $20$ & $24$ & $30$ & $48$ & $32$ & $44$ & $60$ \\ 
23 & $0$ & $0$ & $0$ & $0$ & $1$ & $4$ & $3$ & $2$ & $11$ & $12$ & $4$ & $18$ & $19$ & $6$ & $25$ & $32$ \\ 
\hline
\end{tabular}  \end{center}

\vskip.2cm {\bf M = 46} \qquad $c_M\=12$
\begin{center} \begin{tabular}{|r|cccccccccccccc|ccc}
\hline
\diagbox[width=1.5cm,height=0.8cm]{$r_{\text{min}}$}{$n$}  & 0 & 1 & 2 & 3 & 4 & 5 & 6 & 7 & 8 & 9 & 10 & 11 & 12 & 13  \\ \hline
1 & $-11$ & $57$ & $161$ & $184$ & $276$ & $299$ & $483$ & $414$ & $667$ & $644$ & $874$ & $897$ & $1219$ & $1023$ \\ 
3 & $0$ & $69$ & $115$ & $184$ & $299$ & $253$ & $493$ & $437$ & $585$ & $690$ & $838$ & $736$ & $1288$ & $1127$ \\ 
5 & $0$ & $46$ & $138$ & $147$ & $299$ & $285$ & $414$ & $368$ & $667$ & $607$ & $782$ & $676$ & $1205$ & $1044$ \\ 
7 & $0$ & $46$ & $123$ & $161$ & $215$ & $230$ & $414$ & $353$ & $560$ & $575$ & $744$ & $790$ & $1043$ & $905$ \\ 
9 & $0$ & $30$ & $92$ & $138$ & $230$ & $168$ & $421$ & $329$ & $467$ & $444$ & $766$ & $559$ & $1065$ & $766$ \\ 
11 & $0$ & $23$ & $75$ & $98$ & $167$ & $230$ & $282$ & $276$ & $443$ & $443$ & $558$ & $535$ & $817$ & $765$ \\ 
13 & $0$ & $5$ & $69$ & $69$ & $166$ & $120$ & $281$ & $143$ & $442$ & $350$ & $465$ & $424$ & $810$ & $516$ \\ 
15 & $0$ & $0$ & $27$ & $73$ & $119$ & $92$ & $211$ & $188$ & $211$ & $349$ & $376$ & $303$ & $606$ & $487$ \\ 
17 & $0$ & $0$ & $26$ & $26$ & $92$ & $26$ & $187$ & $95$ & $256$ & $121$ & $325$ & $233$ & $466$ & $236$ \\ 
19 & $0$ & $0$ & $2$ & $23$ & $25$ & $71$ & $94$ & $48$ & $163$ & $163$ & $96$ & $140$ & $349$ & $188$ \\ 
*\,21 & $0$ & $0$ & $0$ & $1$ & $47$ & $-22$ & $93$ & $-22$ & $70$ & $1$ & $208$ & $-68$ & $324$ & $-21$ \\ 
\hline
\end{tabular}  \end{center}

\newpage

{\normalsize 

\bigskip\bigskip

\subsection{Table of $\Phi^0_{1,m}$ of optimal growth \label{apptabs2}}

We give the  Fourier coefficients of the primitive weight~1 mock Jacobi form~$\Phi^0_{1,m}$ for the
fourteen known values of~$m$ for which it has optimal growth, in the same format as was used in the $k=2$
 case above.  The number~$c$ is the factor relating the true value of~$\Phi_{1,m}^0$ to the form whose
coefficients have been tabulated, and which has been normalized to have integral Fourier coefficients
with no common factor except for the values $C(\D=0,r)$, which are allowed to have a denominator of~2.
(This happens for the indices 9, 16, 18, and 25.)
}

\vskip.8cm {\bf m = 2} \qquad $c\=24$
\begin{center} \begin{tabular}{|r|cccccccccccc|ccccc}\hline
\diagbox[width=1.5cm,height=0.8cm]{$r_{\text{min}}$}{$n$}  & 0 & 1 & 2 & 3 & 4 & 5 & 6 & 7 & 8 & 9 & 10 & 11  \\ \hline
$1$& $-1$ & $45$ & $231$ & $770$ & $2277$ & $5796$ & $13915$ & $30843$ & $65550$ & $132825$ & $260568$ & $494385$   \\ 
\hline \end{tabular}  \end{center}

\vskip.5cm {\bf m = 3} \qquad $c\=12$
\begin{center} \begin{tabular}{|r|cccccccccccccc|ccc}
\hline \diagbox[width=1.5cm,height=0.8cm]{$r_{\text{min}}$}{$n$}  & 0 & 1 & 2 & 3 & 4 & 5 & 6 & 7 & 8 & 9 & 10 & 11 & 12 & 13 \\ \hline
$1$& $-1$ & $16$ & $55$ & $144$ & $330$ & $704$ & $1397$ & $2640$ & $4819$ & $8480$ & $14509$ & $24288$ & $39765$ & $63888$   \\ 
$2$& $0$ & $10$ & $44$ & $110$ & $280$ & $572$ & $1200$ & $2244$ & $4180$ & $7348$ & $12772$ & $21330$ & $35288$ & $56760$   \\ 
\hline \end{tabular}  \end{center}

\vskip.5cm {\bf m = 4} \qquad $c\=8$
\begin{center} \begin{tabular}{|r|ccccccccccccccc|cc}
\hline \diagbox[width=1.5cm,height=0.8cm]{$r_{\text{min}}$}{$n$}  & 0 & 1 & 2 & 3 & 4 & 5 & 6 & 7 & 8 & 9 & 10 & 11 & 12 & 13 & 14 \\ \hline
$1$& $-1$ & $7$ & $21$ & $43$ & $94$ & $168$ & $308$ & $525$ & $882$ & $1407$ & $2255$ & $3468$ & $5306$ & $7931$ & $11766$   \\ 
$2$& $0$ & $8$ & $24$ & $56$ & $112$ & $216$ & $392$ & $672$ & $1128$ & $1840$ & $2912$ & $4536$ & $6936$ & $10416$ & $15456$   \\ 
$3$& $0$ & $3$ & $14$ & $28$ & $69$ & $119$ & $239$ & $393$ & $693$ & $1106$ & $1806$ & $2772$ & $4333$ & $6468$ & $9710$   \\ 
\hline \end{tabular}  \end{center}

\vskip.5cm {\bf m = 5} \qquad $c\=6$
\begin{center} \begin{tabular}{|r|cccccccccccccccc|c}
\hline \diagbox[width=1.5cm,height=0.8cm]{$r_{\text{min}}$}{$n$}  & 0 & 1 & 2 & 3 & 4 & 5 & 6 & 7 & 8 & 9 & 10 & 11 & 12 & 13 & 14 & 15 \\ \hline
$1$& $-1$ & $4$ & $9$ & $20$ & $35$ & $60$ & $104$ & $164$ & $255$ & $396$ & $590$ & $864$ & $1259$ & $1800$ & $2541$ & $3560$  \\ 
$2$& $0$ & $5$ & $15$ & $26$ & $54$ & $90$ & $156$ & $244$ & $396$ & $590$ & $905$ & $1320$ & $1934$ & $2751$ & $3924$ & $5456$  \\ 
$3$& $0$ & $4$ & $11$ & $24$ & $45$ & $80$ & $135$ & $220$ & $350$ & $540$ & $810$ & $1204$ & $1761$ & $2524$ & $3586$ & $5040$  \\ 
$4$& $0$ & $1$ & $6$ & $10$ & $25$ & $36$ & $76$ & $110$ & $189$ & $280$ & $446$ & $636$ & $970$ & $1360$ & $1980$ & $2750$  \\ 
\hline \end{tabular}  \end{center}

\newpage
\vskip.2cm {\bf m = 6} \qquad $c\=24$
\begin{center} \begin{tabular}{|r|ccccccccccccccc|cc}
\hline \diagbox[width=1.5cm,height=0.8cm]{$r_{\text{min}}$}{$n$} & 0 & 1 & 2 & 3 & 4 & 5 & 6 & 7 & 8 & 9 & 10 & 11 & 12 & 13 & 14  \\ \hline
$1$& $-5$ & $7$ & $26$ & $33$ & $71$ & $109$ & $185$ & $249$ & $418$ & $582$ & $858$ & $1184$ & $1703$ & $2291$ & $3213$  \\ 
$2$& $0$ & $24$ & $48$ & $96$ & $168$ & $264$ & $432$ & $672$ & $984$ & $1464$ & $2112$ & $2976$ & $4200$ & $5808$ & $7920$ \\ 
$3$& $0$ & $12$ & $36$ & $60$ & $120$ & $180$ & $312$ & $456$ & $720$ & $1020$ & $1524$ & $2124$ & $3036$ & $4140$ & $5760$   \\ 
$4$& $0$ & $12$ & $36$ & $72$ & $120$ & $216$ & $348$ & $528$ & $816$ & $1212$ & $1752$ & $2520$ & $3552$ & $4920$ & $6792$   \\ 
$5$& $0$ & $1$ & $13$ & $14$ & $51$ & $53$ & $127$ & $155$ & $291$ & $382$ & $618$ & $798$ & $1256$ & $1637$ & $2369$  \\ 
\hline \end{tabular}  \end{center}

\vskip.2cm {\bf m = 7} \qquad $c\=4$
\begin{center} \begin{tabular}{|r|cccccccccccccccc|c}
\hline \diagbox[width=1.5cm,height=0.8cm]{$r_{\text{min}}$}{$n$}  & 0 & 1 & 2 & 3 & 4 & 5 & 6 & 7 & 8 & 9 & 10 & 11 & 12 & 13 & 14 & 15 \\ \hline
$1$& $-1$ & $2$ & $3$ & $5$ & $10$ & $15$ & $21$ & $34$ & $48$ & $65$ & $94$ & $129$ & $175$ & $237$ & $312$ & $413$  \\ 
$2$& $0$ & $2$ & $6$ & $10$ & $16$ & $24$ & $40$ & $54$ & $84$ & $116$ & $164$ & $222$ & $310$ & $406$ & $552$ & $722$  \\ 
$3$& $0$ & $3$ & $6$ & $11$ & $18$ & $29$ & $45$ & $66$ & $95$ & $137$ & $192$ & $264$ & $361$ & $486$ & $650$ & $862$  \\ 
$4$& $0$ & $2$ & $6$ & $8$ & $18$ & $24$ & $42$ & $58$ & $90$ & $122$ & $180$ & $240$ & $338$ & $448$ & $612$ & $794$  \\ 
$5$& $0$ & $1$ & $3$ & $7$ & $11$ & $18$ & $28$ & $41$ & $63$ & $91$ & $125$ & $177$ & $245$ & $328$ & $441$ & $590$  \\ 
$6$& $0$ & $0$ & $2$ & $2$ & $6$ & $6$ & $16$ & $18$ & $32$ & $40$ & $66$ & $80$ & $126$ & $156$ & $224$ & $286$  \\ 
\hline \end{tabular}  \end{center}

\vskip.2cm {\bf m = 8} \qquad $c\=4$
\begin{center} \begin{tabular}{|r|cccccccccccccccc|c}
\hline \diagbox[width=1.5cm,height=0.8cm]{$r_{\text{min}}$}{$n$}  & 0 & 1 & 2 & 3 & 4 & 5 & 6 & 7 & 8 & 9 & 10 & 11 & 12 & 13 & 14 & 15 \\ \hline
$1$& $-1$ & $1$ & $2$ & $4$ & $6$ & $7$ & $14$ & $18$ & $25$ & $35$ & $48$ & $63$ & $87$ & $110$ & $146$ & $190$  \\ 
$2$& $0$ & $2$ & $4$ & $6$ & $10$ & $16$ & $22$ & $32$ & $46$ & $62$ & $86$ & $116$ & $152$ & $202$ & $264$ & $340$  \\ 
$3$& $0$ & $2$ & $5$ & $7$ & $13$ & $18$ & $29$ & $38$ & $58$ & $77$ & $108$ & $141$ & $195$ & $250$ & $333$ & $424$  \\ 
$4$& $0$ & $2$ & $4$ & $8$ & $12$ & $18$ & $28$ & $40$ & $56$ & $80$ & $108$ & $144$ & $196$ & $258$ & $336$ & $440$  \\ 
$5$& $0$ & $1$ & $4$ & $5$ & $11$ & $14$ & $24$ & $32$ & $50$ & $63$ & $94$ & $122$ & $170$ & $215$ & $294$ & $371$  \\ 
$6$& $0$ & $0$ & $2$ & $4$ & $6$ & $10$ & $16$ & $22$ & $32$ & $46$ & $62$ & $86$ & $116$ & $152$ & $202$ & $264$  \\ 
$7$& $0$ & $0$ & $1$ & $1$ & $4$ & $3$ & $9$ & $8$ & $17$ & $20$ & $33$ & $36$ & $60$ & $70$ & $101$ & $124$  \\ 
\hline \end{tabular}  \end{center}


\vskip.4cm {\bf m = 9} \qquad $c\=3$
\begin{center} \begin{tabular}{|r|cccccccccccccccc|c}
\hline \diagbox[width=1.5cm,height=0.8cm]{$r_{\text{min}}$}{$n$}  & 0 & 1 & 2 & 3 & 4 & 5 & 6 & 7 & 8 & 9 & 10 & 11 & 12 & 13 & 14 & 15 \\ \hline
$1$& $-1$ & $0$ & $2$ & $2$ & $3$ & $6$ & $8$ & $10$ & $15$ & $20$ & $27$ & $36$ & $46$ & $58$ & $78$ & $98$  \\ 
$2$& $0$ & $2$ & $3$ & $4$ & $8$ & $9$ & $15$ & $20$ & $27$ & $37$ & $51$ & $63$ & $87$ & $109$ & $142$ & $178$  \\ 
$3$& $0$ & $2$ & $3$ & $6$ & $9$ & $12$ & $18$ & $26$ & $36$ & $48$ & $63$ & $84$ & $111$ & $144$ & $183$ & $234$  \\ 
$4$& $0$ & $1$ & $4$ & $5$ & $9$ & $13$ & $21$ & $25$ & $39$ & $50$ & $70$ & $90$ & $122$ & $151$ & $202$ & $251$  \\ 
$5$& $0$ & $2$ & $3$ & $6$ & $8$ & $12$ & $18$ & $26$ & $34$ & $48$ & $65$ & $84$ & $112$ & $146$ & $186$ & $240$  \\ 
$6$& $0$ & $\frac12$ & $3$ & $3$ & $7$ & $9$ & $15$ & $18$ & $30$ & $36$ & $54$ & $66$ & $93$ & $115$ & $156$ & $192$  \\ 
$7$& $0$ & $0$ & $1$ & $2$ & $4$ & $6$ & $9$ & $12$ & $19$ & $26$ & $33$ & $46$ & $62$ & $78$ & $103$ & $134$  \\ 
$8$& $0$ & $0$ & $1$ & $1$ & $3$ & $1$ & $6$ & $5$ & $10$ & $10$ & $18$ & $19$ & $33$ & $35$ & $52$ & $60$  \\ 
\hline \end{tabular}  \end{center}


\vskip.8cm 
{\bf m = 10} \qquad $c\=8$
\begin{center} \begin{tabular}{|r|cccccccccccccccc|c}
\hline \diagbox[width=1.5cm,height=0.8cm]{$r_{\text{min}}$}{$n$} & 0 & 1 & 2 & 3 & 4 & 5 & 6 & 7 & 8 & 9 & 10 & 11 & 12 & 13 & 14 & 15 \\ \hline
$1$& $-3$ & $3$ & $1$ & $2$ & $5$ & $10$ & $9$ & $15$ & $19$ & $26$ & $31$ & $41$ & $52$ & $76$ & $87$ & $108$  \\ 
$2$& $0$ & $4$ & $8$ & $16$ & $16$ & $28$ & $40$ & $48$ & $72$ & $96$ & $120$ & $160$ & $208$ & $256$ & $328$ & $416$  \\ 
$3$& $0$ & $1$ & $7$ & $7$ & $16$ & $14$ & $27$ & $32$ & $51$ & $57$ & $87$ & $101$ & $144$ & $166$ & $226$ & $269$  \\ 
$4$& $0$ & $8$ & $12$ & $16$ & $28$ & $40$ & $56$ & $80$ & $104$ & $136$ & $184$ & $240$ & $304$ & $392$ & $496$ & $624$  \\ 
$5$& $0$ & $2$ & $6$ & $8$ & $14$ & $18$ & $30$ & $34$ & $54$ & $68$ & $92$ & $114$ & $158$ & $190$ & $252$ & $308$  \\ 
$6$& $0$ & $4$ & $8$ & $16$ & $24$ & $32$ & $48$ & $64$ & $88$ & $124$ & $160$ & $208$ & $272$ & $348$ & $440$ & $560$  \\ 
$7$& $0$ & $0$ & $5$ & $3$ & $11$ & $8$ & $22$ & $23$ & $40$ & $39$ & $69$ & $75$ & $113$ & $128$ & $182$ & $210$  \\ 
$8$& $0$ & $0$ & $4$ & $8$ & $8$ & $16$ & $24$ & $32$ & $44$ & $64$ & $80$ & $104$ & $144$ & $176$ & $232$ & $296$  \\ 
*\,$9$& $0$ & $0$ & $1$ & $-1$ & $5$ & $2$ & $9$ & $2$ & $13$ & $11$ & $23$ & $18$ & $43$ & $37$ & $60$ & $60$  \\ 
\hline \end{tabular}  \end{center}

\vskip1.4cm {\bf m = 12} \qquad $c\=8$
\begin{center} \begin{tabular}{|r|cccccccccccccccc|clc}
\hline \diagbox[width=1.5cm,height=0.8cm]{$r_{\text{min}}$}{$n$} & 0 & 1 & 2 & 3 & 4 & 5 & 6 & 7 & 8 & 9 & 10 & 11 & 12 & 13 & 14 & 15 \\ \hline
$*\,1$& $-3$ & $-2$ & $1$ & $-1$ & $-2$ & $-2$ & $1$ & $-5$ & $-1$ & $-7$ & $-5$ & $-6$ & $-5$ & $-16$ & $-10$ & $-16$  \\ 
$2$& $0$ & $4$ & $4$ & $8$ & $12$ & $12$ & $20$ & $28$ & $32$ & $44$ & $56$ & $68$ & $88$ & $112$ & $132$ & $164$  \\ 
$3$& $0$ & $4$ & $8$ & $8$ & $16$ & $20$ & $28$ & $32$ & $48$ & $60$ & $80$ & $92$ & $124$ & $148$ & $188$ & $224$  \\ 
$4$& $0$ & $2$ & $4$ & $6$ & $8$ & $12$ & $16$ & $20$ & $28$ & $36$ & $44$ & $58$ & $72$ & $88$ & $112$ & $136$  \\ 
$5$& $0$ & $3$ & $5$ & $9$ & $13$ & $14$ & $24$ & $31$ & $40$ & $51$ & $66$ & $80$ & $109$ & $130$ & $162$ & $199$  \\ 
$6$& $0$ & $4$ & $8$ & $12$ & $16$ & $24$ & $32$ & $44$ & $56$ & $72$ & $96$ & $120$ & $152$ & $188$ & $232$ & $288$  \\ 
$7$& $0$ & $1$ & $6$ & $5$ & $11$ & $14$ & $22$ & $21$ & $39$ & $43$ & $61$ & $71$ & $98$ & $111$ & $152$ & $174$  \\ 
$8$& $0$ & $0$ & $2$ & $4$ & $6$ & $8$ & $12$ & $16$ & $20$ & $28$ & $36$ & $44$ & $58$ & $72$ & $88$ & $112$  \\ 
$9$& $0$ & $0$ & $4$ & $4$ & $12$ & $8$ & $20$ & $20$ & $36$ & $36$ & $56$ & $64$ & $92$ & $100$ & $144$ & $160$  \\ 
$10$& $0$ & $0$ & $0$ & $4$ & $4$ & $8$ & $12$ & $12$ & $20$ & $28$ & $32$ & $44$ & $56$ & $68$ & $88$ & $112$  \\ 
$*\,11$& $0$ & $0$ & $0$ & $-1$ & $1$ & $-3$ & $1$ & $-2$ & $0$ & $-5$ & $0$ & $-9$ & $2$ & $-8$ & $-7$ & $-14$  \\ 
\hline \end{tabular}  \end{center}

\newpage

\vskip.8cm {\bf m = 13} \qquad $c\=2$
\begin{center} \begin{tabular}{|r|cccccccccccccccc|c}
\hline \diagbox[width=1.5cm,height=0.8cm]{$r_{\text{min}}$}{$n$}  & 0 & 1 & 2 & 3 & 4 & 5 & 6 & 7 & 8 & 9 & 10 & 11 & 12 & 13 & 14 & 15 \\ \hline
$1$& $-1$ & $1$ & $1$ & $0$ & $1$ & $1$ & $2$ & $3$ & $3$ & $4$ & $6$ & $7$ & $7$ & $10$ & $12$ & $14$  \\ 
$2$& $0$ & $0$ & $1$ & $2$ & $2$ & $3$ & $4$ & $4$ & $6$ & $8$ & $10$ & $12$ & $16$ & $18$ & $24$ & $28$  \\ 
$3$& $0$ & $1$ & $1$ & $2$ & $3$ & $4$ & $5$ & $7$ & $9$ & $11$ & $13$ & $17$ & $22$ & $26$ & $32$ & $39$  \\ 
$4$& $0$ & $1$ & $2$ & $2$ & $4$ & $4$ & $6$ & $8$ & $11$ & $12$ & $18$ & $20$ & $26$ & $32$ & $40$ & $46$  \\ 
$5$& $0$ & $1$ & $2$ & $3$ & $3$ & $5$ & $7$ & $8$ & $11$ & $15$ & $18$ & $23$ & $29$ & $34$ & $43$ & $53$  \\ 
$6$& $0$ & $1$ & $2$ & $2$ & $4$ & $4$ & $8$ & $8$ & $12$ & $14$ & $19$ & $22$ & $30$ & $34$ & $44$ & $52$  \\ 
$7$& $0$ & $1$ & $1$ & $2$ & $4$ & $5$ & $6$ & $8$ & $11$ & $13$ & $17$ & $22$ & $27$ & $34$ & $41$ & $51$  \\ 
$8$& $0$ & $0$ & $2$ & $2$ & $3$ & $3$ & $6$ & $6$ & $10$ & $12$ & $16$ & $18$ & $26$ & $28$ & $38$ & $44$  \\ 
$9$& $0$ & $0$ & $1$ & $2$ & $2$ & $3$ & $4$ & $6$ & $7$ & $10$ & $13$ & $15$ & $20$ & $25$ & $30$ & $37$  \\ 
$10$& $0$ & $0$ & $1$ & $0$ & $2$ & $2$ & $4$ & $4$ & $6$ & $6$ & $10$ & $10$ & $16$ & $17$ & $24$ & $26$  \\ 
$11$& $0$ & $0$ & $0$ & $1$ & $1$ & $1$ & $2$ & $2$ & $4$ & $5$ & $5$ & $8$ & $10$ & $11$ & $14$ & $18$  \\ 
$12$& $0$ & $0$ & $0$ & $0$ & $1$ & $0$ & $2$ & $0$ & $2$ & $1$ & $4$ & $2$ & $6$ & $4$ & $8$ & $8$  \\ 
\hline \end{tabular}  \end{center}

\vskip1.4cm {\bf m = 16} \qquad $c\=2$
\begin{center} \begin{tabular}{|r|cccccccccccccccc|c}
\hline \diagbox[width=1.5cm,height=0.8cm]{$r_{\text{min}}$}{$n$}  & 0 & 1 & 2 & 3 & 4 & 5 & 6 & 7 & 8 & 9 & 10 & 11 & 12 & 13 & 14 & 15 \\ \hline
$1$& $-1$ & $0$ & $0$ & $0$ & $1$ & $0$ & $1$ & $2$ & $1$ & $1$ & $2$ & $3$ & $3$ & $4$ & $5$ & $5$  \\ 
$2$& $0$ & $0$ & $1$ & $1$ & $1$ & $1$ & $2$ & $2$ & $3$ & $4$ & $4$ & $5$ & $7$ & $7$ & $9$ & $11$  \\ 
$3$& $0$ & $1$ & $1$ & $1$ & $2$ & $2$ & $3$ & $3$ & $4$ & $5$ & $7$ & $7$ & $9$ & $11$ & $13$ & $15$  \\ 
$4$& $0$ & $1$ & $1$ & $1$ & $2$ & $3$ & $3$ & $4$ & $5$ & $6$ & $8$ & $9$ & $11$ & $14$ & $16$ & $19$  \\ 
$5$& $0$ & $0$ & $1$ & $2$ & $2$ & $2$ & $4$ & $4$ & $6$ & $7$ & $9$ & $10$ & $14$ & $14$ & $19$ & $22$  \\ 
$6$& $0$ & $1$ & $1$ & $2$ & $2$ & $3$ & $4$ & $5$ & $6$ & $8$ & $9$ & $11$ & $14$ & $17$ & $20$ & $24$  \\ 
$7$& $0$ & $1$ & $2$ & $1$ & $3$ & $3$ & $4$ & $5$ & $7$ & $7$ & $10$ & $11$ & $15$ & $17$ & $21$ & $23$  \\ 
$8$& $0$ & $\frac12$ & $1$ & $2$ & $2$ & $3$ & $4$ & $4$ & $6$ & $8$ & $9$ & $12$ & $14$ & $16$ & $20$ & $24$  \\ 
$9$& $0$ & $0$ & $1$ & $1$ & $2$ & $2$ & $4$ & $4$ & $6$ & $6$ & $9$ & $10$ & $14$ & $14$ & $19$ & $22$  \\ 
$10$& $0$ & $0$ & $1$ & $1$ & $2$ & $2$ & $3$ & $4$ & $5$ & $6$ & $8$ & $9$ & $11$ & $14$ & $17$ & $20$  \\ 
$11$& $0$ & $0$ & $1$ & $1$ & $2$ & $2$ & $3$ & $2$ & $5$ & $5$ & $7$ & $7$ & $11$ & $11$ & $15$ & $16$  \\ 
$12$& $0$ & $0$ & $0$ & $1$ & $1$ & $1$ & $2$ & $3$ & $3$ & $4$ & $5$ & $6$ & $8$ & $9$ & $11$ & $14$  \\ 
$13$& $0$ & $0$ & $0$ & $0$ & $1$ & $0$ & $2$ & $1$ & $3$ & $2$ & $4$ & $4$ & $6$ & $6$ & $9$ & $9$  \\ 
$14$& $0$ & $0$ & $0$ & $0$ & $0$ & $1$ & $1$ & $1$ & $1$ & $2$ & $2$ & $3$ & $4$ & $4$ & $5$ & $7$  \\ 
$15$& $0$ & $0$ & $0$ & $0$ & $1$ & $0$ & $1$ & $0$ & $1$ & $1$ & $2$ & $0$ & $3$ & $1$ & $3$ & $2$  \\ 
\hline \end{tabular}  \end{center}

\newpage

\vskip.8cm {\bf m = 18} \qquad $c\=3$
\begin{center} \begin{tabular}{|r|cccccccccccccccc|c}
\hline \diagbox[width=1.5cm,height=0.8cm]{$r_{\text{min}}$}{$n$}  & 0 & 1 & 2 & 3 & 4 & 5 & 6 & 7 & 8 & 9 & 10 & 11 & 12 & 13 & 14 & 15 \\ \hline
*\,$1$& $-2$ & $0$ & $-1$ & $0$ & $-1$ & $2$ & $1$ & $0$ & $0$ & $1$ & $0$ & $3$ & $1$ & $3$ & $2$ & $4$  \\ 
$2$& $0$ & $3$ & $0$ & $3$ & $3$ & $3$ & $3$ & $6$ & $6$ & $6$ & $9$ & $9$ & $12$ & $15$ & $15$ & $21$  \\ 
*\,$3$& $0$ & $-1$ & $2$ & $1$ & $1$ & $0$ & $3$ & $1$ & $4$ & $2$ & $5$ & $5$ & $7$ & $4$ & $9$ & $8$  \\ 
$4$& $0$ & $0$ & $3$ & $3$ & $3$ & $6$ & $6$ & $6$ & $12$ & $12$ & $15$ & $18$ & $21$ & $24$ & $30$ & $36$  \\ 
$5$& $0$ & $2$ & $1$ & $1$ & $2$ & $3$ & $3$ & $3$ & $4$ & $6$ & $5$ & $6$ & $9$ & $11$ & $12$ & $13$  \\ 
$6$& $0$ & $3$ & $3$ & $3$ & $6$ & $6$ & $9$ & $12$ & $12$ & $15$ & $18$ & $24$ & $27$ & $33$ & $39$ & $45$  \\ 
*\,$7$& $0$ & $-1$ & $1$ & $0$ & $3$ & $1$ & $3$ & $2$ & $6$ & $4$ & $8$ & $7$ & $11$ & $8$ & $15$ & $15$  \\ 
$8$& $0$ & $3$ & $3$ & $6$ & $6$ & $6$ & $9$ & $12$ & $12$ & $18$ & $21$ & $24$ & $30$ & $36$ & $39$ & $48$  \\ 
$9$& $0$ & $0$ & $1$ & $1$ & $2$ & $2$ & $3$ & $3$ & $5$ & $5$ & $7$ & $7$ & $10$ & $11$ & $14$ & $15$  \\ 
$10$& $0$ & $0$ & $3$ & $3$ & $3$ & $6$ & $9$ & $9$ & $12$ & $15$ & $18$ & $21$ & $27$ & $30$ & $39$ & $45$  \\ 
$11$& $0$ & $0$ & $2$ & $1$ & $3$ & $0$ & $4$ & $3$ & $5$ & $3$ & $7$ & $5$ & $10$ & $8$ & $14$ & $12$  \\ 
$12$& $0$ & $0$ & $\frac32$ & $3$ & $3$ & $6$ & $6$ & $6$ & $9$ & $12$ & $15$ & $18$ & $21$ & $24$ & $30$ & $36$  \\ 
*\,$13$& $0$ & $0$ & $0$ & $-1$ & $1$ & $1$ & $2$ & $0$ & $3$ & $3$ & $4$ & $3$ & $8$ & $6$ & $9$ & $8$  \\ 
$14$& $0$ & $0$ & $0$ & $3$ & $3$ & $3$ & $3$ & $6$ & $6$ & $9$ & $9$ & $9$ & $15$ & $18$ & $18$ & $24$  \\ 
*\,$15$& $0$ & $0$ & $0$ & $0$ & $2$ & $-1$ & $1$ & $1$ & $3$ & $0$ & $4$ & $1$ & $5$ & $2$ & $5$ & $4$  \\ 
$16$& $0$ & $0$ & $0$ & $0$ & $0$ & $0$ & $3$ & $0$ & $3$ & $3$ & $3$ & $6$ & $6$ & $6$ & $9$ & $12$  \\ 
*\,$17$& $0$ & $0$ & $0$ & $0$ & $1$ & $0$ & $2$ & $0$ & $2$ & $-1$ & $1$ & $0$ & $3$ & $1$ & $3$ & $0$  \\ 
\hline \end{tabular}  \end{center}

\newpage

\vskip.2cm {\bf m = 25} \qquad $c\=1$
\begin{center} \begin{tabular}{|r|cccccccccccccccc|c}
\hline \diagbox[width=1.5cm,height=0.8cm]{$r_{\text{min}}$}{$n$}  & 0 & 1 & 2 & 3 & 4 & 5 & 6 & 7 & 8 & 9 & 10 & 11 & 12 & 13 & 14 & 15 \\ \hline
$1$& $-1$ & $0$ & $0$ & $0$ & $0$ & $0$ & $0$ & $0$ & $0$ & $0$ & $1$ & $0$ & $0$ & $1$ & $1$ & $1$  \\ 
$2$& $0$ & $0$ & $0$ & $1$ & $0$ & $0$ & $1$ & $0$ & $1$ & $1$ & $1$ & $1$ & $1$ & $1$ & $1$ & $2$  \\ 
$3$& $0$ & $0$ & $0$ & $0$ & $1$ & $0$ & $1$ & $1$ & $1$ & $1$ & $1$ & $1$ & $2$ & $2$ & $2$ & $2$  \\ 
$4$& $0$ & $1$ & $1$ & $0$ & $1$ & $1$ & $1$ & $1$ & $1$ & $1$ & $2$ & $2$ & $2$ & $3$ & $3$ & $3$  \\ 
$5$& $0$ & $0$ & $1$ & $1$ & $0$ & $1$ & $1$ & $1$ & $1$ & $2$ & $2$ & $2$ & $3$ & $3$ & $3$ & $4$  \\ 
$6$& $0$ & $0$ & $0$ & $0$ & $1$ & $1$ & $1$ & $1$ & $2$ & $2$ & $2$ & $2$ & $3$ & $3$ & $4$ & $4$  \\ 
$7$& $0$ & $1$ & $0$ & $1$ & $1$ & $1$ & $1$ & $2$ & $2$ & $2$ & $2$ & $3$ & $3$ & $4$ & $4$ & $5$  \\ 
$8$& $0$ & $0$ & $1$ & $1$ & $1$ & $1$ & $1$ & $1$ & $2$ & $2$ & $3$ & $3$ & $4$ & $3$ & $5$ & $5$  \\ 
$9$& $0$ & $1$ & $1$ & $1$ & $1$ & $1$ & $2$ & $2$ & $2$ & $2$ & $3$ & $3$ & $4$ & $4$ & $5$ & $6$  \\ 
$10$& $0$ & $\frac12$ & $1$ & $0$ & $1$ & $1$ & $2$ & $2$ & $2$ & $2$ & $3$ & $3$ & $4$ & $4$ & $5$ & $5$  \\ 
$11$& $0$ & $0$ & $0$ & $1$ & $1$ & $1$ & $1$ & $1$ & $2$ & $3$ & $2$ & $3$ & $4$ & $4$ & $5$ & $6$  \\ 
$12$& $0$ & $0$ & $1$ & $1$ & $1$ & $1$ & $2$ & $1$ & $2$ & $2$ & $3$ & $3$ & $4$ & $4$ & $5$ & $5$  \\ 
$13$& $0$ & $0$ & $1$ & $0$ & $1$ & $1$ & $1$ & $2$ & $2$ & $2$ & $3$ & $3$ & $3$ & $4$ & $5$ & $5$  \\ 
$14$& $0$ & $0$ & $1$ & $1$ & $1$ & $1$ & $2$ & $1$ & $2$ & $2$ & $3$ & $3$ & $4$ & $3$ & $5$ & $5$  \\ 
$15$& $0$ & $0$ & $0$ & $1$ & $1$ & $1$ & $1$ & $1$ & $2$ & $2$ & $2$ & $2$ & $3$ & $4$ & $4$ & $5$  \\ 
$16$& $0$ & $0$ & $0$ & $0$ & $1$ & $0$ & $1$ & $1$ & $2$ & $1$ & $2$ & $2$ & $3$ & $3$ & $4$ & $4$  \\ 
$17$& $0$ & $0$ & $0$ & $1$ & $0$ & $1$ & $1$ & $1$ & $1$ & $2$ & $2$ & $2$ & $3$ & $3$ & $3$ & $4$  \\ 
$18$& $0$ & $0$ & $0$ & $0$ & $1$ & $0$ & $1$ & $1$ & $1$ & $1$ & $2$ & $1$ & $3$ & $2$ & $3$ & $3$  \\ 
$19$& $0$ & $0$ & $0$ & $0$ & $1$ & $1$ & $1$ & $1$ & $1$ & $1$ & $1$ & $2$ & $2$ & $2$ & $2$ & $3$  \\ 
$20$& $0$ & $0$ & $0$ & $0$ & $\frac12$ & $0$ & $1$ & $0$ & $1$ & $1$ & $1$ & $1$ & $2$ & $1$ & $3$ & $2$  \\ 
$21$& $0$ & $0$ & $0$ & $0$ & $0$ & $0$ & $0$ & $0$ & $0$ & $1$ & $1$ & $1$ & $1$ & $1$ & $1$ & $2$  \\ 
$22$& $0$ & $0$ & $0$ & $0$ & $0$ & $0$ & $1$ & $0$ & $1$ & $0$ & $1$ & $0$ & $1$ & $1$ & $1$ & $1$  \\ 
$23$& $0$ & $0$ & $0$ & $0$ & $0$ & $0$ & $0$ & $0$ & $1$ & $0$ & $0$ & $0$ & $1$ & $0$ & $1$ & $1$  \\ 
$24$& $0$ & $0$ & $0$ & $0$ & $0$ & $0$ & $1$ & $0$ & $0$ & $0$ & $1$ & $0$ & $1$ & $0$ & $1$ & $0$  \\ 
\hline \end{tabular}  \end{center}

\bigskip\bigskip

{\normalsize 
\subsection{Table of $\CQ_M$ (weight 1 case)  \label{apptabs3}}


 We give the coefficients of $\CQ_M$ for all products of three primes less than 150 for which $\CQ_M$.
We also show the factor $c_M$ relating $\CQ_M$ to $\CF_M$ and the minimal discriminant of~$\CF_M$,
since it is no longer assumed to be~$-1$.  This time an asterisk means simply that the corresponding
row contains (non-polar) coefficients of opposite signs, since here there are many examples of
values of~$r\mypmod{2m}$ for which the scalar factor~$\k_r$ that determines the sign of the coefficients
$C(\D,r)$ for large~$\D$ is negative rather than positive for~$r=r_{\rm min}$.

\newpage

\vskip.2cm {\bf M = 30} \qquad $c_M\=3$   \qquad $\DD_{\rm min}\=-1$
\begin{center} \begin{tabular}{|r|cccccccccccccccc|cccc}
\hline
\diagbox[width=1.5cm,height=0.8cm]{$r_{\text{min}}$}{$n$} 
      & 0 & 1 & 2 & 3 & 4 & 5 & 6 & 7 & 8 & 9 & 10 & 11 & 12 & 13 & 14 & 15 \\ \hline
$1$& $-1$ & $1$ & $1$ & $2$ & $1$ & $3$ & $2$ & $3$ & $3$ & $5$ & $3$ & $6$ & $5$ & $7$ & $7$ & $9$  \\
$7$& $0$ & $1$ & $2$ & $2$ & $3$ & $3$ & $4$ & $4$ & $6$ & $5$ & $7$ & $8$ & $9$ & $9$ & $12$ & $12$  \\ 
\hline  \end{tabular}  \end{center}

\vskip.2cm {\bf M = 42} \qquad $c_M\=2$   \qquad $\DD_{\rm min}\=-1$
\begin{center} \begin{tabular}{|r|cccccccccccccccc|cccc}
\hline
\diagbox[width=1.5cm,height=0.8cm]{$r_{\text{min}}$}{$n$} 
      & 0 & 1 & 2 & 3 & 4 & 5 & 6 & 7 & 8 & 9 & 10 & 11 & 12 & 13 & 14 & 15 \\ \hline
$1$& $-1$ & $-1$ & $0$ & $-1$ & $-1$ & $-1$ & $0$ & $-2$ & $-1$ & $-2$ & $-1$ & $-2$ & $-1$ & $-3$ & $-2$ & $-3$ \\
$5$& $0$ & $1$ & $1$ & $1$ & $2$ & $1$ & $2$ & $2$ & $2$ & $3$ & $3$ & $2$ & $4$ & $4$ & $4$ & $4$ \\
 $11$& $0$ & $1$ & $1$ & $2$ & $1$ & $2$ & $2$ & $3$ & $2$ & $3$ & $3$ & $4$ & $4$ & $5$ & $4$ & $6$  \\ 
\hline  \end{tabular}  \end{center}

\vskip.2cm {\bf M = 66} \qquad $c_M\=12$   \qquad $\DD_{\rm min}\=-25$
\begin{center} \begin{tabular}{|r|ccccccccccc|cccc}
\hline
\diagbox[width=1.5cm,height=0.8cm]{$r_{\text{min}}$}{$n$} 
      & 0 & 1 & 2 & 3 & 4 & 5 & 6 & 7 & 8 & 9 & 10  \\ \hline
$1$& $-5$ & $-5$ & $-15$ & $-60$ & $-125$ & $-313$ & $-620$ & $-1270$ & $-2358$ & $-4394$ & $-7698$  \\
 $5$& $-1$ & $15$ & $65$ & $175$ & $450$ & $989$ & $2105$ & $4140$ & $7930$ & $14508$ & $25915$ \\
$7$& $0$ & $16$ & $55$ & $155$ & $385$ & $852$ & $1816$ & $3597$ & $6880$ & $12645$ & $22627$ \\
 $13$& $0$ & $5$ & $22$ & $60$ & $155$ & $357$ & $781$ & $1567$ & $3070$ & $5725$ & $10367$  \\
$19$& $0$ & $0$ & $-5$ & $-33$ & $-99$ & $-268$ & $-605$ & $-1320$ & $-2623$ & $-5104$ & $-9398$  \\ 
\hline  \end{tabular}  \end{center}

\vskip.2cm {\bf M = 70} \qquad $c_M\=1$   \qquad $\DD_{\rm min}\=-1$
\begin{center} \begin{tabular}{|r|cccccccccccccccc|cccc}
\hline
\diagbox[width=1.5cm,height=0.8cm]{$r_{\text{min}}$}{$n$} 
      & 0 & 1 & 2 & 3 & 4 & 5 & 6 & 7 & 8 & 9 & 10 & 11 & 12 & 13 & 14 & 15 \\ \hline
$1$& $-1$ & $0$ & $-1$ & $-1$ & $0$ & $0$ & $-1$ & $0$ & $-1$ & $-1$ & $-1$ & $-1$ & $-1$ & $0$ & $-1$ & $-1$ \\
 $3$& $0$ & $0$ & $1$ & $0$ & $1$ & $0$ & $0$ & $1$ & $1$ & $0$ & $1$ & $1$ & $1$ & $0$ & $1$ & $1$ \\
$9$& $0$ & $1$ & $0$ & $1$ & $0$ & $1$ & $1$ & $1$ & $0$ & $1$ & $1$ & $1$ & $1$ & $1$ & $1$ & $2$ \\
 $11$& $0$ & $-1$ & $0$ & $0$ & $0$ & $-1$ & $0$ & $-1$ & $0$ & $0$ & $0$ & $-1$ & $0$ & $-1$ & $0$ & $-1$ \\
 $13$& $0$ & $1$ & $1$ & $0$ & $1$ & $1$ & $1$ & $1$ & $1$ & $1$ & $1$ & $1$ & $1$ & $2$ & $2$ & $1$ \\
 $23$& $0$ & $0$ & $1$ & $1$ & $1$ & $1$ & $1$ & $1$ & $1$ & $1$ & $2$ & $1$ & $2$ & $1$ & $2$ & $2$  \\ 
\hline  \end{tabular}  \end{center}

\vskip.2cm {\bf M = 78} \qquad $c_M\=1$   \qquad $\DD_{\rm min}\=-1$
\begin{center} \begin{tabular}{|r|cccccccccccccccc|cccc}
\hline
\diagbox[width=1.5cm,height=0.8cm]{$r_{\text{min}}$}{$n$} 
      & 0 & 1 & 2 & 3 & 4 & 5 & 6 & 7 & 8 & 9 & 10 & 11 & 12 & 13 & 14 & 15 \\ \hline
$1$& $-1$ & $0$ & $0$ & $-1$ & $-1$ & $0$ & $0$ & $-1$ & $0$ & $0$ & $-1$ & $0$ & $-1$ & $-1$ & $-1$ & $-1$ \\
$5$& $0$ & $1$ & $0$ & $1$ & $0$ & $1$ & $0$ & $1$ & $1$ & $1$ & $0$ & $1$ & $1$ & $1$ & $1$ & $1$ \\
$7$& $0$ & $-1$ & $0$ & $0$ & $0$ & $0$ & $0$ & $-1$ & $0$ & $-1$ & $0$ & $0$ & $0$ & $-1$ & $0$ & $0$ \\
$11$& $0$ & $1$ & $1$ & $0$ & $1$ & $0$ & $1$ & $1$ & $1$ & $1$ & $1$ & $1$ & $1$ & $1$ & $1$ & $1$ \\
$17$& $0$ & $1$ & $0$ & $1$ & $1$ & $1$ & $1$ & $1$ & $0$ & $1$ & $1$ & $1$ & $1$ & $2$ & $1$ & $2$ \\
$23$& $0$ & $0$ & $1$ & $1$ & $1$ & $0$ & $1$ & $1$ & $1$ & $1$ & $1$ & $1$ & $2$ & $1$ & $1$ & $1$  \\ 
\hline  \end{tabular}  \end{center}

\vskip.2cm {\bf M = 102} \qquad $c_M\=6$   \qquad $\DD_{\rm min}\=-25$
\begin{center} \begin{tabular}{|r|cccccccccccc|cccc}
\hline
\diagbox[width=1.5cm,height=0.8cm]{$r_{\text{min}}$}{$n$} 
      & 0 & 1 & 2 & 3 & 4 & 5 & 6 & 7 & 8 & 9 & 10 & 11\\ \hline
$1$& $-3$ & $-1$ & $-8$ & $-14$ & $-31$ & $-56$ & $-102$ & $-173$ & $-293$ & $-461$ & $-732$ & $-1129$ \\
$5$& $-1$ & $7$ & $16$ & $41$ & $76$ & $153$ & $262$ & $454$ & $745$ & $1215$ & $1886$ & $2941$ \\
 $7$& $0$ & $3$ & $12$ & $21$ & $45$ & $76$ & $146$ & $236$ & $408$ & $636$ & $1024$ & $1555$  \\
$11$& $0$ & $3$ & $13$ & $26$ & $59$ & $103$ & $197$ & $327$ & $560$ & $890$ & $1427$ & $2187$  \\
 $13$& $0$ & $6$ & $15$ & $31$ & $68$ & $125$ & $229$ & $389$ & $657$ & $1055$ & $1688$ & $2603$  \\
 $19$& $0$ & $2$ & $7$ & $16$ & $35$ & $61$ & $120$ & $205$ & $349$ & $574$ & $924$ & $1433$ \\
 $*\,25$& $0$ & $0$ & $1$ & $-1$ & $-4$ & $-11$ & $-16$ & $-42$ & $-61$ & $-115$ & $-178$ & $-303$  \\
 $31$& $0$ & $0$ & $0$ & $-3$ & $-12$ & $-27$ & $-61$ & $-113$ & $-209$ & $-358$ & $-605$ & $-972$  \\ 
\hline  \end{tabular}  \end{center}

\vskip.8cm {\bf M = 105} \qquad $c_M\=1$   \qquad $\DD_{\rm min}\=-4$
\begin{center} \begin{tabular}{|r|cccccccccccccccc|cccc}
\hline
\diagbox[width=1.5cm,height=0.8cm]{$r_{\text{min}}$}{$n$} 
      & 0 & 1 & 2 & 3 & 4 & 5 & 6 & 7 & 8 & 9 & 10 & 11 & 12 & 13 & 14 & 15 \\ \hline
$1$& $0$ & $1$ & $1$ & $1$ & $2$ & $1$ & $2$ & $3$ & $3$ & $3$ & $4$ & $4$ & $5$ & $6$ & $7$ & $7$ \\
 $2$& $-1$ & $-1$ & $-1$ & $-1$ & $-1$ & $-1$ & $-2$ & $-2$ & $-2$ & $-3$ & $-3$ & $-3$ & $-4$ & $-4$ & $-5$ & $-6$ \\
$4$& $0$ & $1$ & $0$ & $1$ & $1$ & $1$ & $1$ & $1$ & $2$ & $2$ & $2$ & $2$ & $3$ & $4$ & $3$ & $4$ \\
$8$& $0$ & $0$ & $1$ & $1$ & $0$ & $1$ & $2$ & $1$ & $2$ & $2$ & $2$ & $3$ & $3$ & $3$ & $4$ & $5$ \\
$11$& $0$ & $0$ & $1$ & $0$ & $1$ & $1$ & $1$ & $1$ & $2$ & $1$ & $3$ & $2$ & $3$ & $2$ & $4$ & $4$ \\
$13$& $0$ & $1$ & $0$ & $1$ & $1$ & $1$ & $1$ & $2$ & $1$ & $2$ & $2$ & $3$ & $3$ & $4$ & $3$ & $5$ \\
$16$& $0$ & $1$ & $1$ & $1$ & $2$ & $2$ & $2$ & $3$ & $3$ & $4$ & $4$ & $5$ & $6$ & $7$ & $8$ & $8$ \\
$17$& $0$ & $-1$ & $0$ & $-1$ & $0$ & $-1$ & $0$ & $-2$ & $0$ & $-1$ & $-1$ & $-2$ & $-1$ & $-3$ & $-1$ & $-3$ \\
$19$& $0$ & $1$ & $1$ & $2$ & $1$ & $2$ & $2$ & $3$ & $3$ & $4$ & $4$ & $5$ & $5$ & $7$ & $7$ & $9$ \\
$23$& $0$ & $0$ & $1$ & $0$ & $1$ & $1$ & $2$ & $1$ & $2$ & $1$ & $3$ & $3$ & $4$ & $3$ & $5$ & $4$ \\
$32$& $0$ & $0$ & $0$ & $0$ & $0$ & $1$ & $0$ & $0$ & $1$ & $1$ & $1$ & $1$ & $1$ & $1$ & $2$ & $2$ \\
$34$& $0$ & $0$ & $0$ & $1$ & $1$ & $1$ & $1$ & $2$ & $2$ & $2$ & $2$ & $3$ & $4$ & $4$ & $4$ & $6$  \\ 
\hline  \end{tabular}  \end{center}

\newpage

\vskip.2cm {\bf M = 110} \qquad $c_M\=3$   \qquad $\DD_{\rm min}\=-9$
\begin{center} \begin{tabular}{|r|ccccccccccccccc|cccc}
\hline
\diagbox[width=1.5cm,height=0.8cm]{$r_{\text{min}}$}{$n$} 
      & 0 & 1 & 2 & 3 & 4 & 5 & 6 & 7 & 8 & 9 & 10 & 11 & 12 & 13 & 14  \\ \hline
$1$& $-2$ & $-1$ & $-2$ & $-3$ & $-3$ & $-5$ & $-5$ & $-9$ & $-9$ & $-13$ & $-15$ & $-22$ & $-25$ & $-32$ & $-37$  \\
$3$& $-1$ & $0$ & $2$ & $1$ & $4$ & $4$ & $6$ & $7$ & $11$ & $14$ & $18$ & $21$ & $28$ & $34$ & $44$  \\
$7$& $0$ & $1$ & $1$ & $3$ & $2$ & $3$ & $4$ & $5$ & $6$ & $9$ & $9$ & $13$ & $16$ & $20$ & $21$ \\
$9$& $0$ & $2$ & $3$ & $2$ & $5$ & $5$ & $9$ & $10$ & $13$ & $15$ & $22$ & $24$ & $33$ & $39$ & $50$ \\
$13$& $0$ & $1$ & $2$ & $4$ & $6$ & $7$ & $10$ & $14$ & $18$ & $23$ & $31$ & $37$ & $49$ & $58$ & $74$  \\
$17$& $0$ & $2$ & $3$ & $4$ & $5$ & $8$ & $10$ & $12$ & $17$ & $22$ & $27$ & $34$ & $43$ & $52$ & $65$ \\
$19$& $0$ & $1$ & $2$ & $2$ & $2$ & $3$ & $5$ & $5$ & $6$ & $9$ & $10$ & $13$ & $16$ & $19$ & $24$ \\
$23$& $0$ & $0$ & $2$ & $1$ & $4$ & $4$ & $6$ & $8$ & $13$ & $13$ & $20$ & $23$ & $31$ & $37$ & $49$  \\
$*\,29$& $0$ & $0$ & $1$ & $0$ & $1$ & $0$ & $1$ & $-2$ & $0$ & $-3$ & $-2$ & $-4$ & $-2$ & $-9$ & $-6$  \\
$*\,39$& $0$ & $0$ & $0$ & $0$ & $1$ & $-1$ & $-1$ & $-2$ & $-3$ & $-6$ & $-5$ & $-10$ & $-11$ & $-16$ & $-20$ \\ 
\hline \end{tabular}  \end{center}

\vskip.8cm {\bf M = 114} \qquad $c_M\=4$   \qquad $\DD_{\rm min}\=-25$
\begin{center} \begin{tabular}{|r|cccccccccccccc|cccc}
\hline
\diagbox[width=1.5cm,height=0.8cm]{$r_{\text{min}}$}{$n$} 
      & 0 & 1 & 2 & 3 & 4 & 5 & 6 & 7 & 8 & 9 & 10 & 11 & 12 & 13   \\ \hline
$1$& $-1$ & $4$ & $9$ & $18$ & $37$ & $62$ & $110$ & $181$ & $291$ & $449$ & $695$ & $1034$ & $1537$ & $2235$  \\
 $5$& $-1$ & $0$ & $4$ & $3$ & $11$ & $16$ & $32$ & $41$ & $83$ & $113$ & $188$ & $266$ & $413$ & $574$  \\
$7$& $0$ & $6$ & $12$ & $30$ & $53$ & $99$ & $166$ & $282$ & $444$ & $706$ & $1067$ & $1622$ & $2387$ & $3498$  \\
$11$& $0$ & $3$ & $11$ & $21$ & $42$ & $74$ & $133$ & $216$ & $351$ & $547$ & $849$ & $1271$ & $1897$ & $2757$   \\
$13$& $0$ & $3$ & $8$ & $16$ & $37$ & $59$ & $111$ & $178$ & $295$ & $457$ & $718$ & $1062$ & $1606$ & $2327$\\
 $17$& $0$ & $2$ & $9$ & $19$ & $38$ & $70$ & $127$ & $205$ & $345$ & $538$ & $840$ & $1267$ & $1904$ & $2778$  \\
$23$& $0$ & $0$ & $3$ & $5$ & $15$ & $24$ & $48$ & $75$ & $135$ & $204$ & $333$ & $496$ & $765$ & $1110$  \\
 $29$& $0$ & $0$ & $-1$ & $-2$ & $-6$ & $-10$ & $-20$ & $-37$ & $-64$ & $-100$ & $-165$ & $-257$ & $-388$ & $-583$\\
$35$& $0$ & $0$ & $0$ & $-3$ & $-7$ & $-18$ & $-32$ & $-67$ & $-111$ & $-194$ & $-310$ & $-501$ & $-767$ & $-1184$  \\ 
\hline  \end{tabular}  \end{center}

\newpage

\vskip.2cm {\bf M = 130} \qquad $c_M\=2$   \qquad $\DD_{\rm min}\=-9$
\begin{center} \begin{tabular}{|r|ccccccccccccccc|cccc}
\hline
\diagbox[width=1.5cm,height=0.8cm]{$r_{\text{min}}$}{$n$} 
      & 0 & 1 & 2 & 3 & 4 & 5 & 6 & 7 & 8 & 9 & 10 & 11 & 12 & 13 & 14 \\ \hline
$1$& $-1$ & $0$ & $-1$ & $0$ & $-1$ & $-1$ & $-2$ & $-2$ & $-3$ & $-3$ & $-5$ & $-5$ & $-7$ & $-7$ & $-12$ \\
 $3$& $-1$ & $-1$ & $-1$ & $-2$ & $-1$ & $-3$ & $-3$ & $-5$ & $-5$ & $-7$ & $-7$ & $-11$ & $-11$ & $-16$ & $-17$\\
$7$& $0$ & $1$ & $2$ & $2$ & $3$ & $3$ & $6$ & $6$ & $8$ & $9$ & $13$ & $14$ & $19$ & $21$ & $27$ \\
$9$& $0$ & $1$ & $2$ & $2$ & $4$ & $4$ & $6$ & $8$ & $10$ & $12$ & $16$ & $18$ & $24$ & $28$ & $35$ \\
$11$& $0$ & $1$ & $1$ & $2$ & $2$ & $3$ & $4$ & $4$ & $6$ & $8$ & $9$ & $10$ & $14$ & $17$ & $19$  \\
$17$& $0$ & $1$ & $1$ & $2$ & $2$ & $3$ & $4$ & $5$ & $5$ & $8$ & $9$ & $12$ & $14$ & $17$ & $19$  \\
$19$& $0$ & $1$ & $2$ & $2$ & $4$ & $5$ & $6$ & $8$ & $11$ & $13$ & $17$ & $20$ & $26$ & $31$ & $38$  \\
$21$& $0$ & $1$ & $2$ & $3$ & $3$ & $5$ & $6$ & $8$ & $10$ & $13$ & $15$ & $20$ & $23$ & $29$ & $35$  \\
$27$& $0$ & $0$ & $1$ & $0$ & $1$ & $1$ & $2$ & $0$ & $3$ & $1$ & $3$ & $3$ & $5$ & $3$ & $7$ \\
$29$& $0$ & $0$ & $1$ & $1$ & $2$ & $2$ & $4$ & $4$ & $6$ & $6$ & $10$ & $11$ & $15$ & $16$ & $22$ \\
$37$& $0$ & $0$ & $0$ & $0$ & $0$ & $-1$ & $0$ & $-1$ & $-2$ & $-2$ & $-2$ & $-4$ & $-4$ & $-6$ & $-7$  \\
$47$& $0$ & $0$ & $0$ & $0$ & $0$ & $-1$ & $0$ & $-2$ & $-2$ & $-3$ & $-3$ & $-6$ & $-5$ & $-9$ & $-10$   \\ 
\hline  \end{tabular}  \end{center}

\vskip.8cm {\bf M = 138} \qquad $c_M\=2$   \qquad $\DD_{\rm min}\=-49$
\begin{center} \begin{tabular}{|r|cccccccccccc|cccc}
\hline
\diagbox[width=1.5cm,height=0.8cm]{$r_{\text{min}}$}{$n$} 
      & 0 & 1 & 2 & 3 & 4 & 5 & 6 & 7 & 8 & 9 & 10 & 11 \\ \hline
$1$& $-5$ & $5$ & $20$ & $57$ & $141$ & $323$ & $658$ & $1308$ & $2449$ & $4450$ & $7786$ & $13373$  \\
$5$& $-2$ & $9$ & $27$ & $56$ & $125$ & $235$ & $451$ & $812$ & $1442$ & $2463$ & $4182$ & $6884$ \\
$7$& $-1$ & $13$ & $44$ & $111$ & $265$ & $555$ & $1138$ & $2170$ & $4032$ & $7183$ & $12541$ & $21237$ \\
$11$& $0$ & $14$ & $46$ & $115$ & $261$ & $546$ & $1091$ & $2055$ & $3787$ & $6702$ & $11610$ & $19591$\\
$13$& $0$ & $9$ & $29$ & $68$ & $151$ & $308$ & $612$ & $1139$ & $2080$ & $3675$ & $6323$ & $10653$  \\
$17$& $0$ & $6$ & $28$ & $68$ & $174$ & $354$ & $753$ & $1419$ & $2678$ & $4775$ & $8407$ & $14255$ \\
$*\,19$& $0$ & $2$ & $6$ & $-2$ & $-12$ & $-69$ & $-157$ & $-397$ & $-784$ & $-1581$ & $-2884$ & $-5240$  \\
$25$& $0$ & $0$ & $-3$ & $-23$ & $-69$ & $-192$ & $-427$ & $-926$ & $-1817$ & $-3473$ & $-6310$ & $-11189$  \\
$31$& $0$ & $0$ & $-1$ & $-11$ & $-33$ & $-100$ & $-220$ & $-491$ & $-973$ & $-1890$ & $-3456$ & $-6203$ \\
$*\,37$& $0$ & $0$ & $0$ & $-1$ & $-1$ & $-5$ & $6$ & $8$ & $49$ & $110$ & $263$ & $490$  \\
$43$& $0$ & $0$ & $0$ & $0$ & $2$ & $13$ & $47$ & $128$ & $309$ & $677$ & $1369$ & $2653$  \\ 
\hline  \end{tabular}  \end{center}

\newpage

\providecommand{\bysame}{\leavevmode\hbox to3em{\hrulefill}\thinspace}
\providecommand{\MR}{\relax\ifhmode\unskip\space\fi MR }
\providecommand{\MRhref}[2]{%
  \href{http://www.ams.org/mathscinet-getitem?mr=#1}{#2}
}
\providecommand{\href}[2]{#2}


\begin{thebibliography}{10}

\bibitem{Alexandrov:2012au}
Sergei Alexandrov,  Jan Manschot, and Boris Pioline, 
\emph{D3-instantons, Mock Theta Series and Twistors}, (2012),
 \href{http://arXiv.org/abs/1207.1109}{{\tt arxiv:1207.1109}}.

\bibitem{Alim:2010cf}
Murad Alim, Babak Haghighat, Michael Hecht, Albrecht Klemm, Marco Rauch, and Thomas Wotschke, 
\emph{Wall-crossing holomorphic anomaly and mock modularity of multiple M5-branes}, (2010), 
 \href{http://arXiv.org/abs/1012.1608}{{\tt arxiv:1012.1608}}.      
      
\bibitem{Alvarez:1987wg}
O.~Alvarez, T.~P. Killingback, M.~L. Mangano, and P.~Windey,  \emph{String theory
  and loop space index theorems}, {Commun. Math. Phys.} {\bf 111} (1987) 1.

\bibitem{Ashok:2011cy}
  S.~K.~Ashok and J.~Troost, \emph{A Twisted Non-compact Elliptic Genus},
  JHEP {\bf 1103} (2011) 067, 
   \href{http://arXiv.org/abs/1101.1059}{{\tt arxiv:1101.10593}}.

\bibitem{Banerjee:2007sr}
Shamik Banerjee and Ashoke Sen,  \emph{Duality orbits, dyon spectrum and gauge theory limit
  of heterotic string theory on $T^6$}, JHEP {\bf 0803} (2008) 022,
  \href{http://arXiv.org/abs/0712.0043}{{\tt arxiv:0712.0043}}.

\bibitem{Banerjee:2008ri}
\bysame, \emph{S-duality action on discrete {T}-duality invariants}, JHEP {\bf 0804} (2008) 012,
  \href{http://arXiv.org/abs/0801.0149}{{\tt arxiv:0801.0149}}.


\bibitem{Banerjee:2008pv}
S.~Banerjee, A.~Sen, and Y.~K. Srivastava,  \emph{Generalities of quarter BPS dyon
  partition function and dyons of torsion two}, {JHEP} {\bf 0805} (2008)
  101, \href{http://arXiv.org/abs/0802.0544}{{\tt arxiv:0802.0544}}.

\bibitem{Banerjee:2008pu}
\bysame,   \emph{Partition functions of torsion
  $>$ 1 dyons in heterotic string theory on $T^6$}, {JHEP} {\bf 0805}
  (2008) 098, \href{http://arXiv.org/abs/0802.1556}{{\tt arxiv:0802.1556}}.

\bibitem{Banerjee:2008yu}
\bysame,  \emph{Genus two surface and quarter BPS
  dyons: The contour prescription}, {JHEP} {\bf 0903} (2009) 151,
  \href{http://arXiv.org/abs/0808.1746}{{\tt arxiv:0808.1746}}.
  
\bibitem{Banerjee:2009uk}
  N.~Banerjee, I.~Mandal and A.~Sen,
 \emph{Black hole hair removal},
  JHEP {\bf 0907} (2009) 091, 
\href{http://arXiv.org/abs/0901.0359}{\tt arxiv:0901.0359}.
  
\bibitem{Bates:2003vx}
  B.~Bates and F.~Denef,
 \emph{Exact solutions for supersymmetric stationary black hole composites},
  JHEP {\bf 1111} (2011) 127, 
   \href{http://arXiv.org/abs/hep-th/0304094}{{\tt hep-th/0304094}}.
  
  \bibitem{Breckenridge:1996is}
J.~C. Breckenridge, Robert~C. Myers, A.~W. Peet, and C.~Vafa, \emph{{D-branes
  and spinning black holes}}, Phys.Lett. {\bf B391} (1997) 93--98,
  \href{http://arXiv.org/abs/hep-th/9602065}{{\tt hep-th/9602065}}.

\bibitem{BringOno}
K.~Bringmann and K.~Ono,
\emph{Coefficients of harmonic Maass forms}, 
Proceedings of the 2008 University of Florida Conference on Partitions, q-series, and modular forms.

\bibitem{BringRicht}
K.~Bringmann and O.~Richter, \emph{{Zagier-type dualities and lifting maps for
  harmonic Maass-Jacobi forms}}, Adv. Math. 225 \textbf{no. 4,} (2010),
  2298--2315.

\bibitem{Bringmann:2010sd}
  K.~Bringmann and J.~Manschot,
  \emph{From sheaves on $P^2$ to a generalization of the Rademacher expansion},
  Amer.~J.~Math.~\textbf{135} (2013), no.~4, 1039-1065. 
  \href{http://arXiv.org/abs/1006.0915}{{\tt  arXiv:1006.0915}}.

\bibitem{BringMahlburg}
K.~Bringmann and K.~Mahlburg,
\emph{An extension of the Hardy-Ramanujan Circle Method and applications to partitions without sequences},
American Journal  of Math  133  (2011), pages 1151-1178.

\bibitem{BringFolsom}
K.~Bringmann and A.~Folsom,
\emph{Almost harmonic Maass forms and Kac-Wakimoto characters},
Journal f\"ur die Reine und Angewandte Mathematik, doi:10.1515/crelle-2012-0102, 
\href{http://arXiv.org/abs/1112.4726}{{\tt   arXiv:1112.4726}}.
 
\bibitem{BringmannRaum}
  K.~Bringmann, M.~Raum, and O.~Richter,
  \emph{Harmonic Maass-Jacobi forms with singularities and a theta-like decomposition},
  to appear in Transactions AMS, 
  \href{http://arXiv.org/abs/1207.5600}{{\tt   arXiv:1207.5600}}.

\bibitem{Bringmann:2012zr}
  K.~Bringmann and S.~Murthy,
  \emph{On the positivity of black hole degeneracies in string theory},
  Commun.~Number Theory Phys.~\textbf{7} (2013), no.~1, 15-56. 
  \href{http://arXiv.org/abs/1208.3476}{{\tt  arXiv:1208.3476}}.

\bibitem{BringmannFolsom2}
K.~Bringmann and A.~Folsom,
  \emph{On the asymptotic behavior of Kac-Wakimoto characters},
Proceedings of the American Mathematical Society 141 (2013), pages 1567--1576. 

\bibitem{Mod123}
J.~H.~Bruinier, G.~van der~Geer, G.~Harder and D.~Zagier, \emph{The $1$-$2$-$3$ of
  Modular Forms}, Universitext, Springer-Verlag, Belin--Heidelberg (2008).

\bibitem{Cardoso:2008ej}
G.~L. Cardoso, J.~R. David, B.~de~Wit, and S.~Mahapatra, \emph{{The mixed black
  hole partition function for the STU model}}, JHEP \textbf{12} (2008), 086, 
 \href{http://arXiv.org/abs/0810.1233}{{\tt arxiv:0810.1233}}.

\bibitem{Castro:2008ys}
Alejandra Castro and Sameer Murthy, \emph{{Corrections to the statistical
  entropy of five dimensional black holes}}, JHEP \textbf{0906} (2009), 024,
  \href{http://arXiv.org/abs/0807.0237}{{\tt arxiv:0807.0237}}.

\bibitem{Cheng:2010pq}
Miranda~C.N. Cheng, \emph{{K3 surfaces, $N=4$ Dyons, and the Mathieu group M24}},
(2010),  \href{http://arXiv.org/abs/1005.5415}{{\tt arxiv:1005.5415}}.


  \bibitem{Cheng:2012tq}
  Miranda C.~N.~Cheng, John F.~R.~Duncan and Jeffrey A.~Harvey,
 \emph{Umbral Moonshine},
 \href{http://arXiv.org/abs/1204.2779}{{\tt arXiv:1204.2779}}.


\bibitem{Cheng:2007ch}
Miranda C.~N. Cheng and Erik Verlinde, \emph{Dying dyons don't count},   {JHEP} {\bf
  0709} (2007) 070, \href{http://arXiv.org/abs/0706.2363}{{\tt arxiv:0706.2363}}.


\bibitem{Choie}
Young Ju Choie and Subong Lim, 
\emph{The heat operator and mock Jacobi forms}, 
The Ramanujan Journal Volume 22, Number 2 (2010), 209-219. 


\bibitem{Dabholkar:2004yr}
Atish Dabholkar, \emph{Exact counting of black hole microstates}, Phys. Rev.
  Lett. \textbf{94} (2005), 241301,
\href{http://arXiv.org/abs/hep-th/0409148}{{\tt hep-th/0409148}}.

  
\bibitem{Dabholkar:2005by}
Atish Dabholkar, Frederik Denef, Gregory~W. Moore, and Boris Pioline,
  \emph{Exact and asymptotic degeneracies of small black holes}, JHEP
  \textbf{08} (2005), 021,
\href{http://arXiv.org/abs/hep-th/0502157}{{\tt hep-th/0502157}}.

\bibitem{Dabholkar:2005dt}
\bysame, \emph{Precision counting of small black holes,} 
  JHEP {\bf 0510} (2005) 096, 
 \href{http://arXiv.org/abs/hep-th/0507014} {{\tt hep-th/0507014}}.



\bibitem{Dabholkar:2007vk}
Atish Dabholkar, Davide Gaiotto, and Suresh Nampuri, \emph{{Comments on the
  spectrum of CHL dyons}}, JHEP \textbf{01} (2008), 023,
\href{http://arXiv.org/abs/hep-th/0702150}{{\tt hep-th/0702150}}.


  \bibitem{Dabholkar:1990yf}
Atish Dabholkar, Gary~W. Gibbons, Jeffrey~A. Harvey, and Fernando Ruiz~Ruiz,
  \emph{Superstrings and solitons}, Nucl. Phys. \textbf{B340} (1990), 33--55.

\bibitem{Dabholkar:2008zy}
Atish Dabholkar, Joao Gomes, and Sameer Murthy, 
\emph{{Counting all dyons in $N = 4$ string theory}},  
    JHEP {\bf 1105} (2011) 059
\href{http://arXiv.org/abs/0803.2692}{{\tt arxiv:0803.2692}}.

\bibitem{Dabholkar:2010uh}
\bysame, 
  \emph{Quantum black holes, localization and the topological string},
  JHEP {\bf 1106} (2011) 019
     \href{http://arXiv.org/abs/1012.0265}{{\tt arxiv:1012.0265}}.
  
\bibitem{Dabholkar:2011ec}
\bysame, 
\emph{Localization \& Exact Holography},
  JHEP {\bf 1304} (2013) 062
    \href{http://arXiv.org/abs/1111.1161}{{\tt arxiv:1111.1161}}.
  
\bibitem{DGMKloos}
\bysame, 
\emph{Nonperturbative black hole entropy and Kloosterman sums}, (2014), 
    \href{http://arXiv.org/abs/1404.0033}{{\tt arxiv:1404.0033}}.


  
\bibitem{Dabholkar:2010rm}
Atish Dabholkar, Joao Gomes, Sameer Murthy, and Ashoke Sen,
  \emph{{Supersymmetric index from black hole entropy}},  (2010),
    JHEP {\bf 1104} (2011) 034
   \href{http://arXiv.org/abs/1009.3226}{{\tt arxiv:1009.3226}}.

 \bibitem{Dabholkar:2009dq}
  A.~Dabholkar, M.~Guica, S.~Murthy and S.~Nampuri,
 \emph{{No entropy enigmas for $N=4$ dyons}},
  JHEP {\bf 1006} (2010) 007
 \href{http://arXiv.org/abs/0903.2481}{{\tt arxiv:0903.2481}}.

\bibitem{Dabholkar:1989jt}
Atish Dabholkar and Jeffrey~A. Harvey, \emph{Nonrenormalization of the
  superstring tension}, Phys. Rev. Lett. \textbf{63} (1989), 478.

\bibitem{Dabholkar:2004dq}
Atish Dabholkar, Renata Kallosh, and Alexander Maloney, \emph{A stringy cloak
  for a classical singularity}, JHEP \textbf{12} (2004), 059,
\href{http://arXiv.org/abs/hep-th/0410076}{{\tt hep-th/0410076}}.


\bibitem{Dabholkar:2006xa}
Atish Dabholkar and Suresh Nampuri, \emph{{Spectrum of dyons and black Holes in
  CHL orbifolds using Borcherds lift}}, JHEP \textbf{11} (2007), 077,
\href{http://arXiv.org/abs/hep-th/0603066}{{\tt hep-th/0603066}}.


\bibitem{David:2006yn}
Justin~R. David and Ashoke Sen, \emph{{CHL dyons and statistical entropy
  function from D1-D5 system}}, JHEP \textbf{11} (2006), 072,
\href{http://arXiv.org/abs/hep-th/0605210}{{\tt hep-th/0605210}}.

\bibitem{Dixon:1985jw}
  L.~J.~Dixon, J.~A.~Harvey, C.~Vafa and E.~Witten,
\emph{Strings on Orbifolds},
  Nucl.\ Phys.\ B {\bf 261} (1985) 678.


\bibitem{deBoer:2006vg}
Jan de~Boer, Miranda C.~N. Cheng, Robbert Dijkgraaf, Jan Manschot, and Erik
  Verlinde, \emph{{A Farey tail for attractor black holes}}, JHEP \textbf{11}
  (2006), 024,
\href{http://arXiv.org/abs/hep-th/0608059}{{\tt hep-th/0608059}}.

\bibitem{deBoer:2008fk}
  J.~de Boer, F.~Denef, S.~El-Showk, I.~Messamah and D.~Van den Bleeken,
  \emph{Black hole bound states in  $AdS_{3} \times  S^{2}$},
  JHEP {\bf 0811} (2008) 050, 
 \href{http://arXiv.org/abs/0802.2257} {{\tt arXiv:0802.2257}}.

\bibitem{Denef:2000nb}
  F.~Denef,
\emph{Supergravity flows and D-brane stability},
  JHEP {\bf 0008} (2000) 050, 
    \href{http://arXiv.org/abs/hep-th/0005049} {{\tt hep-th/0005049}}.

\bibitem{Denef:2002ru}
  F.~Denef,
 \emph{Quantum quivers and Hall/hole halos},
  JHEP {\bf 0210} (2002) 023, 
   \href{http://arXiv.org/abs/hep-th/0206072} {{\tt hep-th/0206072}}.

\bibitem{Denef:2007vg}
Frederik Denef and Gregory~W. Moore, \emph{Split states, entropy enigmas, holes
  and halos},  (2007), \href{http://arXiv.org/abs/hep-th/0702146}{{\tt hep-th/0702146}}.


\bibitem{Dijkgraaf:2000fq}
Robbert Dijkgraaf, Juan~Martin Maldacena, Gregory~Winthrop Moore, and Erik~P.
  Verlinde, \emph{{A black hole Farey tail}},  (2000), 
\href{http://arXiv.org/abs/hep-th/0005003}{{\tt hep-th/0005003}}.

\bibitem{Dijkgraaf:1996xw}
Robbert Dijkgraaf, Gregory~W. Moore, Erik~P. Verlinde, and Herman~L. Verlinde,
  \emph{Elliptic genera of symmetric products and second quantized strings},
  Commun. Math. Phys. \textbf{185} (1997), 197--209,
\href{http://arXiv.org/abs/hep-th/9608096}{{\tt hep-th/9608096}}.

\bibitem{Dijkgraaf:1996it}
Robbert Dijkgraaf, Erik~P. Verlinde, and Herman~L. Verlinde, \emph{Counting
  dyons in $N = 4$ string theory}, Nucl. Phys. \textbf{B484} (1997), 543--561,
\href{http://arXiv.org/abs/hep-th/9607026}{{\tt hep-th/9607026}}.



\bibitem{Eguchi:2010ej}
Tohru Eguchi, Hirosi Ooguri, and Yuji Tachikawa, \emph{{Notes on the $K3$ surface
  and the Mathieu group $M_{24}$}},  (2010), 
\href{http://arXiv.org/abs/1004.0956}{{\tt arxiv:1004.0956}}.

\bibitem{EOTY:1989}
Tohru Eguchi, Hirosi Ooguri, S-K.~Yang and Anne Taormina, 
\emph{{Superconformal algebras and string compactification on manifolds with $SU(N)$ Holonomy}},  
Nucl. Phys. \textbf{B315} (1989), 193.

\bibitem{Eguchi:2004yi}
  Tohru Eguchi and Yuji Sugawara,
 \emph{$SL(2,R) / U(1)$ supercoset and elliptic genera of noncompact Calabi-Yau manifolds},
  JHEP \textbf{0405} (2004) 014.
  \href{http://arXiv.org/abs/hep-th/0403193}{\tt hep-th/0403193}

\bibitem{EST:2007}
Tohru Eguchi, Yuji Sugawara and Anne Taormina, \emph{Liouville field, modular forms 
and elliptic genera}, JHEP \textbf{0703} (2007), 119.
\href{http://arXiv.org/abs/hep-th/0611338}{\tt hep-th/0611338}


\bibitem{Eguchi:1988af}
Tohru Eguchi and Anne Taormina, 
\emph{{On the unitary representation of $N=2$ and 
$N=4$ superconformal algebras}}, Phys.Lett. \textbf{B210} (1988),125.

\bibitem{Eichler:1985ja}
M.~Eichler and D.~Zagier, \emph{The Theory of Jacobi Forms}, Birkh\"auser,  1985.

\bibitem{Ferrara:1995ih}
Sergio Ferrara, Renata Kallosh, and Andrew Strominger, \emph{$N=2$ extremal black
  holes}, Phys. Rev. \textbf{D52} (1995), 5412--5416,
\href{http://arXiv.org/abs/hep-th/9508072}{{\tt hep-th/9508072}}.

\bibitem{Gaiotto:2005hc}
Davide Gaiotto, \emph{Re-recounting dyons in $N = 4$ string theory},  (2005),
\href{http://arXiv.org/abs/hep-th/0506249}{{\tt hep-th/0506249}}.

\bibitem{Gaiotto:2005gf}
Davide Gaiotto, Andrew Strominger, and Xi~Yin, 
\emph{{New connections between 4D and 5D black holes}}, JHEP
  \textbf{02} (2006), 024,
\href{http://arXiv.org/abs/hep-th/0503217}{{\tt hep-th/0503217}}.

%


\bibitem{Gaiotto:2006wm}
\bysame, \emph{{The M5-brane elliptic genus: Modularity and BPS states}}, JHEP
  \textbf{08} (2007), 070,
\href{http://arXiv.org/abs/hep-th/0607010}{{\tt hep-th/0607010}}.

\bibitem{Gauntlett:1999vc}
Jerome~P.~Gauntlett, Nakwoo Kim, Jaemo Park, and Piljin Yi, \emph{{Monopole
  dynamics and BPS dyons in $N = 2$ super-Yang-Mills theories}}, Phys. Rev.
  \textbf{D61} (2000), 125012,
\href{http://arXiv.org/abs/hep-th/9912082}{{\tt hep-th/9912082}}.

\bibitem{Ginsparg:1988ui}
Paul~H.~Ginsparg, \emph{{Applied conformal field theory}},  (1988),
\href{http://arXiv.org/abs/hep-th/9108028}{{\tt hep-th/9108028}}.

\bibitem{Giveon:1999px}
Amit Giveon and David Kutasov, \emph{{Little string theory in a double scaling limit}}, 
JHEP {\bf 9910} (1999) 034, 
\href{http://arXiv.org/abs/hep-th/9909110}{{\tt hep-th/9909110}}.

\bibitem{Goettsche:1990go}
L.~Goettsche, \emph{The Betti numbers of the Hilbert scheme of points on a smooth protective surface}, 
Math. Ann. \textbf{286} (1990), 193.

\bibitem{Gritsenko:1999fk}
  V.~Gritsenko,
  ``Elliptic genus of Calabi-Yau manifolds and Jacobi and Siegel modular forms,''
\href{http://arxiv.org/abs/math/9906190}{\tt math/9906190}

\bibitem{GritNik1}
V.~A.~Gritsenko and V.~V.~Nikulin, \emph{Automorphic forms and Lorentzian 
Kac--Moody Algebras. Part I}, (1996)
\href{http://arxiv.org/abs/alg-geom/9610022}{\tt alg-geom/9610022}.


\bibitem{GritNik2}
V.~A.~Gritsenko and V.~V.~Nikulin, \emph{Automorphic forms and Lorentzian 
Kac--Moody Algebras. Part II}, (1996)
\href{http://arxiv.org/abs/alg-geom/9611028}{\tt alg-geom/9611028}.

\bibitem{GSZ}
V.~A.~Gritsenko, N.-P.~Skoruppa, D.~Zagier, \emph{Theta blocks I: Elementary theory 
and examples}, in preparation.

\bibitem{Gupta:2012cy}
  Rajesh Kumar Gupta and Sameer Murthy, 
  \emph{All solutions of the localization equations for N=2 quantum black hole entropy},
  JHEP {\bf 1302} (2013) 141, 
  \href{http://arXiv.org/abs/1208.6221"}{{\tt arXiv:1208.6221}}.


\bibitem{Harvey:2013mda}
  Jeffrey A.~Harvey and Sameer Murthy, 
  \emph{ Moonshine in Fivebrane Spacetimes},
  JHEP {\bf 1401} (2014) 146, 
  \href{http://arXiv.org/abs/1307.7717"}{{\tt arXiv:1307.7717}}.

\bibitem{Hickerson}
D.~R.~Hickerson, \emph{On the seventh order mock theta functions}, Inv. Math. {\bf 94},
(1988), 661--677.

\bibitem{Hirzebruch:1976}
F.~Hirzebruch and D.~Zagier, \emph{Intersection numbers on curves on Hilbert
  modular surfaces and modular forms of Nebentypus}, Inv. Math. \textbf{36}
  (1976), 57.

\bibitem{Hull:1994ys}
C.~M.~Hull and P.~K.~Townsend, \emph{Unity of superstring dualities}, Nucl.
  Phys. \textbf{B438} (1995), 109--137,
\href{http://arXiv.org/abs/hep-th/9410167}{{\tt hep-th/9410167}}.

\bibitem{Jatkar:2009yd}
  D.~P.~Jatkar, A.~Sen and Y.~K.~Srivastava,
  \emph{ Black hole hair removal: non-linear analysis},
  JHEP {\bf 1002} (2010) 038, 
  \href{http://arXiv.org/abs/0907.0593}{{\tt arXiv:0907.0593}}.

\bibitem{Joyce:2008}
D.~Joyce, Y.~Song, 
\emph{A Theory of Generalized Donaldson-Thomas Invariants}, 
Memoirs of the AMS 217 (2012), 
\href{http://arXiv.org/abs/0810.5645}{{\tt arXiv:0810.5645}}.



\bibitem{Kawai:1993jk}
Toshiya Kawai, Yasuhiko Yamada, and Sung-Kil Yang, \emph{Elliptic genera and $N=2$ super\-conformal field theory}, Nucl. Phys. \textbf{B414} (1994), 191--212,
\href{http://arXiv.org/abs/hep-th/9306096}{{\tt hep-th/9306096}}.

\bibitem{Konstevich:2008}
M.~Kontsevich and Y.~Soibelman, \emph{Stability structures, motivic Donaldson-Thomas 
invariants and cluster transformations}, 
\href{http://arXiv.org/abs/0811.2435}{{\tt arXiv:0811.2435}}.


\bibitem{Kraus:2005vz}
Per Kraus and Finn Larsen, \emph{Microscopic black hole entropy in theories
  with higher derivatives}, JHEP \textbf{09} (2005), 034,
\href{http://arXiv.org/abs/hep-th/0506176}{{\tt hep-th/0506176}}.

\bibitem{Kraus:2006nb}
\bysame, \emph{Partition functions and elliptic genera from supergravity},  (2006), 
\href{http://arXiv.org/abs/hep-th/0607138}{{\tt hep-th/0607138}}.

\bibitem{LopesCardoso:1998wt}
Gabriel Lopes~Cardoso, Bernard de~Wit, and Thomas Mohaupt, \emph{Corrections to
  macroscopic supersymmetric black-hole entropy}, Phys. Lett. \textbf{B451}
  (1999), 309--316,
\href{http://arXiv.org/abs/hep-th/0412287}{{\tt hep-th/0412287}}.

\bibitem{LopesCardoso:1999xn}
\bysame, \emph{Area law corrections from state counting and supergravity},
  Class. Quant. Grav. \textbf{17} (2000), 1007--1015,
\href{http://arXiv.org/abs/hep-th/9910179}{{\tt hep-th/9910179}}.

\bibitem{LopesCardoso:1999cv}
\bysame, \emph{Deviations from the area law for supersymmetric black holes},
  Fortsch. Phys. \textbf{48} (2000), 49--64,
\href{http://arXiv.org/abs/hep-th/9904005}{{\tt hep-th/9904005}}.

\bibitem{LopesCardoso:1999ur}
\bysame, \emph{Macroscopic entropy formulae and non-holomorphic corrections for
  supersymmetric black holes}, Nucl. Phys. \textbf{B567} (2000), 87--110,
\href{http://arXiv.org/abs/hep-th/9906094}{{\tt hep-th/9906094}}.

\bibitem{LopesCardoso:2004xf}
Gabriel Lopes~Cardoso, B.~de~Wit, J.~Kappeli, and T.~Mohaupt, \emph{Asymptotic
  degeneracy of dyonic $N = 4$ string states and black hole entropy}, JHEP
  \textbf{12} (2004), 075.

\bibitem{Maldacena:1998bw}
  J.~M.~Maldacena and A.~Strominger,
 \emph{AdS(3) black holes and a stringy exclusion principle},
  JHEP {\bf 9812} (1998) 005
  \href{http://arXiv.org/abs/hep-th/9804085}{{\tt hep-th/9804085}}.

\bibitem{Maldacena:1997de}
Juan~M. Maldacena, Andrew Strominger, and Edward Witten, \emph{{Black hole
  entropy in M-theory}}, JHEP \textbf{12} (1997), 002,
\href{http://arXiv.org/abs/hep-th/9711053}{{\tt hep-th/9711053}}.

\bibitem{Manschot:2010qz}
  J.~Manschot, B.~Pioline and A.~Sen,
  \emph{Wall Crossing from Boltzmann Black Hole Halos},
  JHEP \textbf{1107} (2011) 059
  \href{http://arXiv.org/abs/1011.1258}{{\tt arXiv:1011.1258}}.

\bibitem{Manschot:2007ha}
Jan Manschot and Gregory~W. Moore, \emph{{A Modern Farey Tail}},  (2007),
\href{http://arXiv.org/abs/0712.0573}{{\tt arxiv:0712.0573}}.

\bibitem{Manschot:2009ia}
 Jan Manschot, \emph{Stability and duality in N=2 supergravity},
Commun.Math.Phys., {\bf 299}, 651-676  (2010),
\href{http://arXiv.org/abs/0906.1767}{{\tt arxiv:0906.1767}}.

\bibitem{Minahan:1998vr}
J.~A. Minahan, D.~Nemeschansky, C.~Vafa, and N.~P. Warner, \emph{{E-strings and
  $N = 4$ topological Yang-Mills theories}}, Nucl. Phys. \textbf{B527} (1998),
  581--623,
\href{http://arXiv.org/abs/hep-th/9802168}{{\tt hep-th/9802168}}.

\bibitem{Minasian:1999qn}
Ruben Minasian, Gregory~W. Moore, and Dimitrios Tsimpis, \emph{{Calabi-Yau
  black holes and $(0,4)$ sigma models}}, Commun. Math. Phys. \textbf{209}
  (2000), 325--352,
\href{http://arXiv.org/abs/hep-th/9904217}{{\tt hep-th/9904217}}.

\bibitem{Moore:2004fg}
Gregory~W. Moore, \emph{Les Houches lectures on strings and arithmetic},
  (2004), \href{http://arXiv.org/abs/hep-th/0401049}{{\tt hep-th/0401049}}.

\bibitem{Murthy:2003es}
Sameer Murthy, \emph{{Notes on noncritical superstrings in various dimensions}}, JHEP
  \textbf{0311} (2003), 056, \href{http://arXiv.org/abs/hep-th/0305197}{{\tt hep-th/0305197}}. 

\bibitem{Murthy:2009dq}
Sameer Murthy and Boris Pioline, \emph{{A Farey tale for $N=4$ dyons}}, JHEP
  \textbf{0909} (2009), 022, \href{http://arXiv.org/abs/0904.4253}{{\tt arxiv:0904.4253}}. 


\bibitem{Murthy:2013xpa}
Sameer Murthy and Valentin Reys, \emph{{Quantum black hole entropy and the holomorphic
                        prepotential of N=2 supergravity}}, JHEP \textbf{10} (2013) 099, 
\href{http://arXiv.org/abs/1306.3796}{{\tt arXiv:1306.3796}}.

\bibitem{Ochanine:1987}
Serge Ochanine, \emph{Sur les genres multiplicatifs definis par des integrales
  elliptiques}, Topology \textbf{26} (1987), 143.

\bibitem{Olivetto}
Ren\'e Olivetto, 
\emph{{On the Fourier coefficients of meromorphic Jacobi forms}}, 
\href{http://arXiv.org/abs/1210.7926}{{\tt arXiv:1210.7926}}.

\bibitem{Ooguri:1989fd}
H.~Ooguri, \emph{{Superconformal symmetry and geometry of Ricci flat 
K\"ahler manifolds}}, Int. J. Mod. Phys. \textbf{A4} (1989), 4303--4324.

\bibitem{Pope:1978zx}
C.~N. Pope, \emph{{Axial vector anomalies and the index theorem in charged
  Schwarzschild and Taub - NUT spaces}}, Nucl. Phys. \textbf{B141} (1978), 432.

\bibitem{Sen:1994yi}
Ashoke Sen, \emph{Dyon-monopole bound states, selfdual harmonic forms on the
  multi-monopole moduli space, and $Sl(2,\IZ)$ invariance in string theory},
  Phys. Lett. \textbf{B329} (1994), 217--221,
\href{http://arXiv.org/abs/hep-th/9402032}{{\tt hep-th/9402032}}.

\bibitem{Sen:1994fa}
\bysame, \emph{{Strong-weak coupling duality in four-dimensional string
  theory}}, Int. J. Mod. Phys. \textbf{A9} (1994), 3707--3750,
\href{http://arXiv.org/abs/hep-th/9402002}{{\tt hep-th/9402002}}.

\bibitem{Sen:1995in}
\bysame, \emph{Extremal black holes and elementary string states}, Mod. Phys.
  Lett. \textbf{A10} (1995), 2081--2094,
\href{http://arXiv.org/abs/hep-th/9504147}{{\tt hep-th/9504147}}.

\bibitem{Sen:2005ch}
\bysame, \emph{Black holes and the spectrum of half-BPS states in $N=4$
  supersymmetric string theory},  (2005),
  \href{http://arXiv.org/abs/hep-th/0504005}{{\tt hep-th/0504005}}.

\bibitem{Sen:2007vb}
\bysame, \emph{Walls of marginal stability and dyon spectrum in $N=4$
  supersymmetric string theories}, JHEP \textbf{05} (2007), 039,
\href{http://arXiv.org/abs/hep-th/0702141}{{\tt hep-th/0702141}}.

\bibitem{Sen:2008yk}
\bysame, \emph{{Entropy Function and $AdS(2)/CFT(1)$ Correspondence}}, JHEP
  \textbf{11} (2008), 075,
\href{http://arXiv.org/abs/0805.0095}{{\tt arxiv:0805.0095}}.

\bibitem{Sen:2008vm}
\bysame, \emph{{Quantum Entropy Function from $AdS(2)/CFT(1)$ Correspondence}},
  (2008), \href{http://arXiv.org/abs/0809.3304}{{\tt arxiv:0809.3304}}.

\bibitem{Sen:2011mh}
\bysame, \emph{{Negative discriminant states in $N=4$ supersymmetric string theories}},
  (2011), \href{http://arXiv.org/abs/0809.3304}{{\tt arxiv:1104.1498}}.

\bibitem{Shih:2005uc}
David Shih, Andrew Strominger, and Xi~Yin, \emph{{Recounting dyons in $N = 4$
  string theory}}, JHEP \textbf{10} (2006), 087,
\href{http://arXiv.org/abs/hep-th/0505094}{{\tt hep-th/0505094}}.

\bibitem{Shih:2005he}
David Shih and Xi~Yin, \emph{Exact black hole degeneracies and the topological
  string}, JHEP \textbf{04} (2006), 034,
\href{http://arXiv.org/abs/hep-th/0508174}{{\tt hep-th/0508174}}.

\bibitem{Skoruppa}
N.-P. Skoruppa, \emph{\"Uber den Zusammenhang zwischen Jacobi-Formen und Modulformen halbganzen 
Gewichts}, Dissertation, Universit\"at Bonn (1984). 

\bibitem{Skoruppa1}
\bysame, \emph{Developments in the theory of Jacobi forms}.
In Automorphic functions and their applications Khabarovsk, 1988,
Acad. Sci. USSR Inst. Appl. Math., Khabarovsk 1990, 167--185

\bibitem{Skoruppa2}
\bysame, \emph{Binary quadratic forms and the Fourier coefficients of elliptic and Jacobi modular forms}.
J. reine Angew. Math. 411 (1990), 66--95 

\bibitem{SkoruppaZagier2}
N.-P.~Skoruppa, D.~Zagier, \emph{A trace formula for Jacobi forms}
J. reine Angew. Math. 393 (1989) 168-198.


\bibitem{SkoruppaZagier}
\bysame, \emph{Jacobi forms and a certain space of modular forms}.
Inventiones mathe\-maticae 94, 113--146 (1988).



\bibitem{Strominger:1996kf}
Andrew Strominger, \emph{Macroscopic entropy of $N=2$ extremal black holes},
  Phys. Lett. \textbf{B383} (1996), 39--43,
\href{http://arXiv.org/abs/hep-th/9602111}{{\tt hep-th/9602111}}.


\bibitem{Strominger:1998yg}
\bysame,  \emph{$AdS(2)$ quantum gravity and string theory}, JHEP \textbf{01}
  (1999), 007,
\href{http://arXiv.org/abs/hep-th/9809027}{{\tt hep-th/9809027}}.

\bibitem{Strominger:1996sh}
Andrew Strominger and Cumrun Vafa, \emph{Microscopic origin of the
  Bekenstein-Hawking entropy}, Phys. Lett. \textbf{B379} (1996), 99--104,
\href{http://arXiv.org/abs/hep-th/9601029}{{\tt hep-th/9601029}}.

\bibitem{Troost:2010ud}
Jan Troost, \emph{{The non-compact elliptic genus: mock or modular}}, JHEP
  \textbf{1006} (2010), 104, \href{http://arXiv.org/abs/1004.3649}{{\tt arxiv:1004.3649}}.

\bibitem{Vafa:1994tf}
Cumrun Vafa and Edward Witten, \emph{A strong coupling test of S-duality},
  Nucl. Phys. \textbf{B431} (1994), 3--77,
\href{http://arXiv.org/abs/hep-th/9408074}{{\tt hep-th/9408074}}.

\bibitem{Wendland:2000}
Katrin Wendland, 
Ph.D.~Thesis, Bonn University (2000).

 \bibitem{Witten:1979ey}
 Edward Witten,
  \emph{Dyons of charge $e \theta/2 \pi$},
  Phys.\ Lett.\ B {\bf 86} (1979) 283. 
  

\bibitem{Witten:1986bf}
\bysame,  \emph{{Elliptic genera and quantum field theory}}, Commun. Math.
  Phys. \textbf{109} (1987), 525.

\bibitem{Witten:1995ex}
\bysame, \emph{String theory dynamics in various dimensions}, Nucl. Phys.
  \textbf{B443} (1995), 85--126,
\href{http://arXiv.org/abs/hep-th/9503124}{{\tt hep-th/9503124}}.

\bibitem{Witten:1978mh}
Edward Witten and David~I.~Olive, \emph{Supersymmetry algebras that include
  topological charges}, Phys. Lett. \textbf{B78} (1978), 97.

\bibitem{Zagier:1975}
D.~Zagier, \emph{Nombres de classes et formes modulaires de poids $3/2$}, {C.~R.~Acad. Sc.~Paris } \textbf{{281}} (1975), 883.

\bibitem{Zagier:1982}
\bysame, \emph{The Rankin-Selberg method for automorphic functions which are not of rapid decay}, 
J. Fac. Sci. Tokyo 28 (1982) 415-438.

\bibitem{Zagier:2002}
\bysame, \emph{Traces of singular moduli}, in ``Motives, Polylogarithms and Hodge Theory" (eds. F.~Bogomolov, L.~Katzarkov), 
 Lecture Series {\bf 3}, International Press, Somerville (2002), 209--244.

\bibitem{Zagier:2007}
\bysame, \emph{Ramanujan's mock theta functions and their applications $[$d'apr\`es Zwegers and Bringmann-Ono$]$}, S\'eminaire
  Bourbaki, 60\`eme ann\'ee, 2006-2007, \textbf{{986}} (2007).

\bibitem{Zwegers:2002}
S.~P.~Zwegers, \emph{{Mock theta functions}}, {Thesis, Utrecht} (2002), 
{\tt http://igitur-archive.library.uu.nl/dissertations/2003-0127-094324/inhoud.htm}.


\end{thebibliography}
\end{document}